\documentclass[11pt,english]{article}
\usepackage{geometry}
\usepackage{setspace}
\usepackage[T1]{fontenc}
\usepackage[utf8]{inputenc}
\usepackage{babel}
\usepackage{varioref}
\usepackage{amsmath,nccmath}
\numberwithin{equation}{section}
\usepackage{upgreek}
\usepackage{amssymb}
\usepackage{mathrsfs}  
\usepackage{subcaption}
\usepackage{graphicx}
\usepackage{floatrow}
\usepackage[labelfont=bf]{caption}
\usepackage{appendix}
\usepackage{slashed}
\usepackage{svg}
\usepackage{enumitem}

\usepackage [autostyle, english = american]{csquotes}
\MakeOuterQuote{"}
\makeatletter
\usepackage{hyperref}
\usepackage{xcolor}
\usepackage{bm}
\usepackage{authblk}
\usepackage{amsthm}
\usepackage{graphicx}
\usepackage[normalem]{ulem}
\usepackage{hyperref}
\hypersetup{
    colorlinks,
    linkcolor={blue!60!black},
    citecolor={green!60!black},
    urlcolor={blue!80!black}
}

\usepackage{accents}    

\usepackage{mathtools}
\usepackage{comment}
\usepackage{xargs}   
\usepackage[colorinlistoftodos,prependcaption,textsize=tiny]{todonotes}
\newcommandx{\typo}[2][1=]{\todo[linecolor=red,backgroundcolor=red!25,bordercolor=red,#1]{#2}}
\newcommandx{\change}[2][1=]{\todo[linecolor=blue,backgroundcolor=blue!25,bordercolor=blue,#1]{#2}}
\newcommandx{\answer}[1]{\todo[linecolor=pink,backgroundcolor=pink!25,bordercolor=pink]{#1}}
\newcommandx{\unsure}[2][1=]{\todo[linecolor=green,backgroundcolor=green!25,bordercolor=green,#1]{#2}}
\newcommandx{\improve}[2][1=]{\todo[linecolor=violet,backgroundcolor=violet!25,bordercolor=violet,#1]{#2}}
\newcommandx{\thiswillnotshow}[2][1=]{\todo[disable,#1]{#2}}
\allowdisplaybreaks

\AtBeginDocument{}

\newtheorem{thm}{Theorem}[section]

\newtheorem{conjecture}{Conjecture}
\newtheorem{obs}{Observation}
\newtheorem{lemma}{Lemma}[section]
\newtheorem{defi}{Definition}[section]

\newtheorem{cor}{Corollary}[section]
\theoremstyle{remark}
\newtheorem{rem}{Remark}[section]
\theoremstyle{plain}
\newtheorem{prop}{Proposition}[section]

\newenvironment{nalign}{
    \begin{equation}
    \begin{aligned}
}{
    \end{aligned}
    \end{equation}
    \ignorespacesafterend
}
\theoremstyle{remark}
\newcommand{\slashednabla}{\slashed{\nabla}}

\newtheorem{subclaim}{Sublemma}
\newcommand{\conj}[1]{\overline{#1}}
\newcommand{\dd}{\mathop{}\!\mathrm{d}}
\newcommand{\dw}{\sin\theta \dd\theta \dd\varphi}

\renewcommand{\O}{\mathcal{O}}
\newcommand{\e}{\mathrm{e}}
\newcommand{\GG}[1]{}

\renewcommand{\(}{\left(}
\renewcommand{\)}{\right)}
\newcommand{\red}[1]{\textcolor{red}{#1}}
\newcommand{\purple}[1]{\textcolor{purple}{#1}}
\newcommand{\overone}{\stackrel{\mbox{\scalebox{0.4}{(1)}}}}
\newcommand{\fullsystem}{\eqref{eq:lin:trg}--\eqref{eq:lin:al3}}
\renewcommand{\sl}{\mathring{\slashed{\nabla}}}
\newcommand{\sll}{\slashed{\nabla}}
\newcommand{\sls}{{}^{\ast}\mathring{\slashed{\nabla}}}

\newcommand{\lap}{\mathring{\slashed{\Delta}}}
\newcommand{\lapp}{\slashed{\Delta}}
\renewcommand{\div}{\mathring{\slashed{\mathrm{div}}}\,}
\newcommand{\curl}{\mathring{\slashed{\mathrm{curl}}}\,}
\newcommand{\divv}{\slashed{\mathrm{div}}}

\newcommand{\D}[1]{\mathring{\slashed{\mathcal{D}}}_{#1}}

\newcommand{\Ds}[1]{\mathring{\slashed{\mathcal{D}}}^{\ast}_{#1}}

\newcommand{\tr}{\mathrm{tr}}

\newcommand{\pu}{\partial_u}
\newcommand{\pv}{\partial_v}
\newcommand{\su}{\slashed{\nabla}_3}
\newcommand{\sv}{\slashed{\nabla}_4}
\newcommand{\Du}{\slashed{\nabla}_u}
\newcommand{\Dv}{\slashed{\nabla}_v}

\newcommand{\rDv}{\frac{r^2}{D}\Dv}

\newcommand{\rDu}{\frac{r^2}{D}\Du}

\newcommand{\gs}{\accentset{\scalebox{.6}{\mbox{\tiny (1)}}}{\slashed{g}}}
\newcommand{\gsh}{\accentset{\scalebox{.6}{\mbox{\tiny (1)}}}{\hat{\slashed{g}}}}
\newcommand{\trg}{{\tr{\gs}}}
\newcommand{\Om}{\accentset{\scalebox{.6}{\mbox{\tiny (1)}}}{\Omega}}
\newcommand{\Omm}{\left(\frac{\Om}{\Omega}\right)}
\renewcommand{\b}{\accentset{\scalebox{.6}{\mbox{\tiny (1)}}}{b}}
\newcommand{\trx}{\accentset{\scalebox{.6}{\mbox{\tiny (1)}}}{\left(\Omega\tr\chi\right)}}
\newcommand{\trxb}{\accentset{\scalebox{.6}{\mbox{\tiny (1)}}}{\left(\Omega\tr\underline{\chi}\right)}}
\newcommand{\xh}{\accentset{\scalebox{.6}{\mbox{\tiny (1)}}}{{\hat{\chi}}}}
\newcommand{\xhb}{\accentset{\scalebox{.6}{\mbox{\tiny (1)}}}{\underline{\hat{\chi}}}}

\newcommand{\et}{\accentset{\scalebox{.6}{\mbox{\tiny (1)}}}{{\eta}}}
\newcommand{\etb}{\accentset{\scalebox{.6}{\mbox{\tiny (1)}}}{\underline{\eta}}}
\newcommand{\om}{\accentset{\scalebox{.6}{\mbox{\tiny (1)}}}{\omega}}
\newcommand{\omb}{\accentset{\scalebox{.6}{\mbox{\tiny (1)}}}{\underline{\omega}}}
\newcommand{\al}{\accentset{\scalebox{.6}{\mbox{\tiny (1)}}}{{\alpha}}}
\newcommand{\be}{\accentset{\scalebox{.6}{\mbox{\tiny (1)}}}{{\beta}}}
\newcommand{\rh}{\accentset{\scalebox{.6}{\mbox{\tiny (1)}}}{\rho}}
\newcommand{\sig}{\accentset{\scalebox{.6}{\mbox{\tiny (1)}}}{{\sigma}}}
\newcommand{\beb}{\accentset{\scalebox{.6}{\mbox{\tiny (1)}}}{\underline{\beta}}}
\newcommand{\alb}{\accentset{\scalebox{.6}{\mbox{\tiny (1)}}}{\underline{\alpha}}}
\newcommand{\K}{\accentset{\scalebox{.6}{\mbox{\tiny (1)}}}{K}}
\newcommand{\Rlin}{\accentset{\scalebox{.6}{\mbox{\tiny (1)}}}{R}}
\newcommand{\Gamlin}{\accentset{\scalebox{.6}{\mbox{\tiny (1)}}}{\Gamma}}
\newcommand{\glin}{\accentset{\scalebox{.6}{\mbox{\tiny (1)}}}{g}}
\newcommand{\plin}{\accentset{\scalebox{.6}{\mbox{\tiny (1)}}}{{\psi}}}
\newcommand{\pblin}{\accentset{\scalebox{.6}{\mbox{\tiny (1)}}}{{\underline{\psi}}}}
\newcommand{\ps}{\accentset{\scalebox{.6}{\mbox{\tiny (1)}}}{{\psi}}}
\newcommand{\psb}{\accentset{\scalebox{.6}{\mbox{\tiny (1)}}}{{\underline{\psi}}}}
\newcommand{\Ps}{\accentset{\scalebox{.6}{\mbox{\tiny (1)}}}{{\Psi}}}
\newcommand{\Psb}{\accentset{\scalebox{.6}{\mbox{\tiny (1)}}}{\underline{\Psi}}}
\newcommand{\rad}[2]{\EndsWiths{#1}_{#2}}
\ExplSyntaxOn
\NewDocumentCommand{\EndsWiths}{m}
 {
  \use:c { \cs_to_str:N #1 s }
 }
\ExplSyntaxOff
\newcommand{\radc}[1]{\rad{#1}{\Cin}}
\newcommand{\radi}[1]{\rad{#1}{\Scrim}}
\newcommand{\rads}[1]{\rad{#1}{\Sone}}
\newcommand{\radsinf}[1]{\rad{#1}{\Sinfty}}
\newcommand{\radlp}[1]{\rad{#1}{\ell,\Scrip}}
\newcommand{\restr}[2]{\left.{#1}\right|_{#2}}

\newcommand{\albs}{\accentset{\scalebox{.6}{\mbox{\tiny (1)}}}{\underline{\upalpha}}}

\newcommand{\Ks}{\accentset{\scalebox{.6}{\mbox{\tiny (1)}}}{\mathtt{K}}}

\newcommand{\alphass}[1][s]{\upalpha^{[#1]}}
\newcommand{\Psiss}[1][s]{{\bm{\Uppsi}}^{[#1]}}
\renewcommand{\P}{{\accentset{\scalebox{.6}{\mbox{\tiny (1)}}}{\mathtt{P}}}}

\newcommand{\Sb}{\underline{\overone{S}}}

\newcommand{\C}{\mathcal{C}}
\newcommand{\Cbar}{\underline{\mathcal{C}}}
\newcommand{\Cin}{{\underline{\mathcal{C}}}}
\newcommand{\Stwo}{\mathbb{S}^2}
\newcommand{\Scrim}{\mathcal{I}^{-}}
\newcommand{\Scrimv}{\mathcal{I}^{-}_{v\geq v_1}}
\newcommand{\Scrip}{\mathcal{I}^{+}}
\newcommand{\DoD}{\mathcal{D}}

\newcommand{\Schw}{\mathcal{M}_M}
\newcommand{\SchwEx}{\overline{\mathcal{M}}_M}

\newcommand{\Sone}{{\mathcal{S}_1}}
\newcommand{\Sinfty}{{\mathcal{S}_\infty}}

\newcommand{\functions}{\Gamma^{\infty}{(\Stwo)}}

\newcommand{\oneforms}{\Gamma^\infty(T^{(0,1)}\Stwo)}
\newcommand{\oneformsLtwo}{L^2(T^{(0,1)}\Stwo)}

\newcommand{\stfs}{\Gamma^{\infty}(T_{\mathrm{stf}}^{(0,2)}(\Stwo))}
\newcommand{\stfsLtwo}{L^2(T_{\mathrm{stf}}^{(0,2)}(\Stwo))}

\newcommand{\pqbundletensors}[3]{\slashed{\Pi}\,{T}^{(#1,#2)}#3}
\newcommand{\bundleoneforms}[1]{\slashed{\Pi}\,{T}^{(0,1)}#1}
\newcommand{\bundlestfs}[1]{\slashed{\Pi}\,{T}^{(0,2)}_{stf}#1}
\newcommand{\bundlestfss}[1]{\slashed{\Pi}\,{T}^{(0,|s|)}_{stf}#1}

\newcommand{\custombundleoneforms}[1]{\slashed{\Pi}\,{T}^{(0,1)}#1}
\newcommand{\custombundlestfs}[1]{\slashed{\Pi}\,{T}^{(0,2)}_{\mathrm{stf}}#1}

\newcommand{\stffields}[1]{\Gamma^{\infty}(\bundlestfs{#1})}
\newcommand{\oneformsfields}[1]{\Gamma^{\infty}(\bundleoneforms{#1})}

\newcommand{\customfieldstf}[1]{\Gamma^{\infty}(\custombundlestfs{#1})}

\newcommand{\customfieldoneforms}[1]{\Gamma^\infty(\custombundleoneforms{#1})}

\newcommand{\Ylm}{Y_{\ell,m}}
\newcommand{\YlmE}[1]{\Ylm^{\mathrm{E},#1}}
\newcommand{\YlmH}[1]{\Ylm^{\mathrm{H},#1}}

\newcommand{\alphas}[1][s]{\alpha^{[#1]}}
\newcommand{\Psis}[1][s]{\Psi^{[#1]}}
\newcommand{\laps}[1][s]{\lap_{[#1]}}
\newcommand{\lambdas}[1][s]{\Lambda^{[#1]}_\ell}
\newcommand{\Asl}[1][s]{\mathrm{A}^{[#1]}_\ell}
\newcommand{\NP}[2]{\mathbf{I}^{#1,#2}_{\ell}}
\newcommand{\NPpq}{\mathbf{I}^{p,q}_{\ell}}
\renewcommand{\a}[3]{a^{[#1]}_{#2,#3}}
\newcommand{\bb}[3]{b^{[#1]}_{#2,#3}}
\renewcommand{\c}[3]{c^{[#1]}_{#2,#3}}
\newcommand{\at}[4][\ell]{\tilde{a}^{[#2]}_{#3,#4,#1}}
\newcommand{\x}[3]{x^{[#1]}_{#2,#3}}
\newcommand{\aRW}[3]{a^{[#1],\mathrm{RW}}_{#2,#3}}
\newcommand{\bRW}[3]{b^{[#1],\mathrm{RW}}_{#2,#3}}
\newcommand{\cRW}[3]{c^{[#1],\mathrm{RW}}_{#2,#3}}
\newcommand{\xRW}[3]{x^{[#1],\mathrm{RW}}_{#2,#3}}
\newcommand{\albdata}{\underline{\mathscr{B}}}
\newcommand{\aldata}{\mathscr{A}}
\newcommand{\A}{\mathscr{A}}
\newcommand{\B}{\mathscr{B}}
\newcommand{\BB}{\mathscr{C}}
\newcommand{\AP}[1][s]{\mathscr{P}^{1,[#1]}}
\newcommand{\BP}[1][s]{\mathscr{P}^{2,[#1]}}
\newcommand{\CP}[1][s]{\mathscr{P}^{3,[#1]}}
\newcommand{\Aa}{\underline{\mathscr{A}}}
\newcommand{\Bb}{\underline{\mathscr{B}}}
\newcommand{\BBb}{\underline{\mathscr{C}}}
\newcommand{\ctrbar}[2][p]{\underline{C}^{\nearrow}_{(#2,#1)}}
\newcommand{\ctr}[2][p]{C^{\nearrow}_{(#2,#1)}}

\newcommand{\Teuk}{\mathrm{Teuk}^{+2}}
\newcommand{\Teukb}{\mathrm{Teuk}^{-2}}
\newcommand{\RW}{\mathrm{RW}}

\title{The Case Against Smooth Null Infinity V:\\Early-Time Asymptotics of Linearised Gravity Around Schwarzschild for Fixed Spherical Harmonic Modes} 

\author[1]{Lionor Kehrberger\thanks{kehrberger@mis.mpg.de}}
\author[2]{Hamed Masaood\thanks{h.masaood22@imperial.ac.uk}} 
\affil[1]{Max Planck Institute for Mathematics in the Sciences,  Inselstraße 22, 04103 Leipzig, Germany}
\affil[2]{Imperial College London, Department of Mathematics,
South Kensington Campus, London~SW7~2AZ, United Kingdom}
\date{January 8, 2024} 

\makeatother
\begin{document}
\pagenumbering{roman}

\maketitle 
\begin{abstract}
In this work, starting from the predictions of the Post-Newtonian theory for a system of $N$ infalling masses from the infinite past $i^-$, we formulate and solve a scattering problem for the system of linearised gravity around Schwarzschild in a double null gauge as introduced in \cite{DHR16}. The scattering data are posed on a null hypersurface $\Cbar$ emanating from a section of past null infinity $\Scrim$, and on the part of $\Scrim$ that lies to the future of this section:
Along $\Cbar$, we implement the Post-Newtonian theory-inspired hypothesis that the gauge-invariant components of the Weyl tensor $\al$ and $\alb$ (a.k.a.~$\Psi_0$ and $\Psi_4$) decay like $r^{-3}$, $r^{-4}$, respectively, and we exclude incoming radiation from $\Scrim$ by demanding the News function to vanish along $\Scrim$. 	

We also show that compactly supported gravitational perturbations along $\Scrim$ induce very similar data, with $\al$, $\alb$ decaying like $r^{-3}$, $r^{-5}$.

After constructing the unique solution to this scattering problem, we then provide a complete analysis of the asymptotic behaviour of projections onto fixed spherical harmonic number $\ell$ near $\Scrim$, spacelike infinity $i^0$ and future null infinity $\Scrip$, crucially exploiting a set of approximate conservation laws enjoyed by $\al$ and $\alb$. 
Having obtained a clear understanding of the asymptotics of linearised gravity around Schwarzschild,  we also give constructive corrections to popular historical notions of asymptotic flatness such as Bondi coordinates or asymptotic simplicity. 
In particular, confirming earlier heuristics due to Damour and Christodoulou, we find that the peeling property is violated both near $\Scrim$ and near $\Scrip$, with e.g.~$\al$ near $\Scrip$ only decaying like $r^{-4}$ instead of $r^{-5}$. We also find that the resulting solution decays slower towards $i^0$ than often assumed, with $\al,\alb$ both decaying like $r^{-3}$ towards $i^0$.

The issue of summing up the estimates obtained for fixed angular modes in $\ell$ in order to obtain asymptotics for the full solution is dealt with in forthcoming work.
%

\end{abstract}

    \begingroup
\hypersetup{linkcolor=black}
    \tableofcontents{}
    \endgroup
\newpage
\pagenumbering{arabic}

\section{Introduction}\label{sec:intr}
This paper is the fifth in a series of papers (initiated with \cite{I,II,III}) dedicated to the rigorous analysis of the asymptotic properties of gravitational radiation in astrophysical contexts.
More specifically, we study the precise \textit{early-time asymptotics} of \textit{scattering solutions} \cite{DRSR18,Masaood22,Masaood22b} to the system of linearised gravity around Schwarzschild (as introduced in \cite{DHR16}), with an emphasis on solutions with particular physical relevance in the context of $N$-body scattering constructions \cite{WalkerWill79,Damour86,Chr02}.

Here, the expression \textit{early-time asymptotics} is used to denote the asymptotics up until some finite retarded time. Thus, our results complement and directly affect the impressive body of literature on asymptotics of solutions to wave equations on black hole geometries at \textit{late times}, see e.g.~\cite{DR05,LukOh162,AAG21, AAG23, Hintz22,MZ22} and references therein, as well as \cite{GK} for the relation between early- and late-time asymptotics. 

 Describing the asymptotic properties of gravitational radiation, our results also affect popular notions of asymptotic flatness such as that of a smooth null infinity \cite{Penrose65} or Bondi coordinates \cite{SeriesVIII}. More details will be given below.

\paragraph{Structure of the introduction and relation to the overview paper \cite{IV}}
This paper is also the longer companion paper to the much shorter overview paper~\cite{IV}. 
The latter provides a thorough introduction to and discussion of the history of the problem studied in the series \textit{The Case Against Smooth Null Infinity}, as well as a sketch of the derivation of some of the main results of the present paper, and we recommend reading it before reading the present paper---indeed, it is a good starting point not only as an introduction to the present paper, but also to the series.
Nevertheless, we will still give a completely self-contained introduction here, using, however, the existence of the overview paper to allow ourselves to be brief at some places while going into more detail at other places. In general, the present paper is written in an entirely self-contained manner.

The reader who has read the overview paper may feel free to skip directly to \S\ref{sec:intro:results1}, as the contents of \S\ref{sec:intro:motivation} and \S\ref{sec:intro:setup}, which contain motivation and historical background for the problem studied as well as a basic summary of the setup and the notation employed in the paper,  are discussed in more detail in~\cite{IV}. 

Sections \S\ref{sec:intro:results1} and \S\ref{sec:intro:results2} then give a thorough overview of the main results of the paper, namely a complete asymptotic analysis near spacelike infinity of physically well-motivated solutions to the system of linearised gravity around Schwarzschild.
In \S\ref{sec:intro:reflect}, we offer some reflection on how our results affect the various historical 
notions of asymptotic flatness discussed in \S\ref{sec:intro:motivation}. 
In \S\ref{sec:intro:conjecture}, we state a conjecture on the \textit{global} asymptotics (including the late-time asymptotics near future timelike infinity) for solutions to the nonlinear Einstein vacuum equations.
Finally, in \S\ref{sec:intro:structure}, we explain the structure of the remainder of the paper (which is closely aligned to that of \cite{IV}).

\subsection{Motivation and context--Asymptotic flatness, isolated systems and the \texorpdfstring{$N$}{N}-body problem}\label{sec:intro:motivation}
Like we said before, this paper's object of study is the \textit{precise asymptotic analysis of solutions to the semi-global scattering problem near spacelike infinity $i^0$ for linearised gravity around Schwarzschild}, with an emphasis on those solutions that arise from scattering data that have particular physical relevance in the context of the $N$-body problem. 
More precisely, we consider the linearisation of the Einstein vacuum equations,
\begin{equation}\label{eq:intro:EVE}
    \mathrm{Ric}[\bm{g}]=0,
\end{equation}
around the Schwarzschild solution
\begin{equation}\label{eq:intro:SS}
g_M=-\Omega^2\dd t^2+\Omega^{-2}\dd r^2+r^2(\dd \theta^2+\sin^2\theta\dd\varphi^2)  ,\quad\quad \Omega^2:=D:=1-\frac{2M}{r}  
\end{equation}
recast in a double null gauge (see already \eqref{eq:intro:DNGmetric} or \S\ref{sec:lin} for details). 
We formulate for the arising system a scattering problem with \textit{seed scattering data} posed on an ingoing null cone~$\Cin$ (truncated away from the event horizon) and on the part of past null infinity $\Scrim$ that lies to the future of this cone. We prove existence and uniqueness of solutions arising from the seed scattering data (cf.~Thm.~\ref{thm:intro:LEE Scattering wp}).
We then analyse the semi-global\footnote{
We use the word semi-global to denote the entire domain of dependence of $\Cin\cup\Scrim$. Since the cone $\Cin$ is truncated away from the event horizon, this region excludes future timelike infinity $i^+$. See, however, \S\ref{sec:intro:conjecture} for conjectures on the \textit{global} asymptotic behaviour (including $i^+$).} asymptotic properties of these solutions (cf.~Thms.~\ref{thm:intro:main}--\ref{thm:intro:para}). 
In doing so, we single out three particular classes of scattering data to which we prescribe specific physical relevance, namely that they each capture a system of $N$ infalling objects from the infinite past (cf.~Def.~\ref{defi:intro:N}).

Before diving into the details, let us motivate why we study this problem. 
From a mathematical point of view, there is the desire to develop a deeper understanding of scattering properties and asymptotics near spacelike infinity of solutions to (systems of) wave equations such as \eqref{eq:intro:EVE}.
This should be seen as complementing the ongoing research program concerning the \textit{late-time asymptotics} of such systems in the context of black hole geometries~\cite{DR05,LukOh162,AAG18a,AAG18b,AAG21,AAG23,Gajic22Inverse,Hintz22,MZ21,MZ22}; indeed, the two are intimately connected, see \cite{GK} and \S\ref{sec:intro:conjecture}.

On the other hand, there is the desire to slowly approach the difficult question of developing even a basic rigorous understanding of the $N$-body problem in general relativity, cf.~Open Problems 1 and 2 from the introduction of \cite{mythesis}.

Now, it turns out that both of these directions of motivation can be naturally included within the more conceptual motivation related to the much broader debate on how to model isolated systems in general relativity, or how to choose or justify different notions of asymptotic flatness. Since much has been said on this in \cite{IV} already, we shall here be quite brief:
The central difficulty of modelling isolated systems in relativity is, to a large extent, rooted in the problem that one inevitably has to model the gravitational radiation emitted by the corresponding system in the asymptotic regime. 
Historically, there have been various approaches and proposals how to tackle this difficulty; here, we give a schematic list of some of them:
\begin{enumerate}[leftmargin=*,label=\arabic*)]
    \item \textbf{Bondi coordinates/Bondi--Sachs formalism:} One notion for a system to be asymptotically flat and to potentially capture the physics of isolated systems is that such a system should admit \textit{Bondi coordinates} near future null infinity $\Scrip$; these are coordinates that ensure, on the one hand, that $\Scrip$ itself behaves like the Minkowskian $\Scrip$, and, on the other hand, that $\Scrip$ is approached in inverse powers of a radial coordinate $r_{\mathrm{BS}}$ along outgoing null geodesics. 
    The connection of this latter requirement, loosely referred to as \textbf{the peeling property} or \textit{the outgoing radiation condition}, to physics consisted of it holding true for the linearised theory around Minkowski provided that there is \textit{no incoming radiation from past null infinity~$\Scrim$}.
    A definition of the peeling property in our context is given in Def.~\ref{defi:intro:peeling}.
    The groundwork for this notion was set in   \cite{Bondi60Nature,SeriesVI,SeriesVII,SeriesVIII,Sachs62BMS}.
    \item \textbf{Asymptotic simplicity/Smooth null infinity:} In a similar spirit, but from a geometrical point of view much more appealing, is Penrose's notion of asymptotic simplicity \cite{Penrose63,Penrose65}, defining a class of spacetimes essentially via the requirement that they should admit a smooth conformal compactification (with smooth boundaries $\Scrim$ and $\Scrip$). This requirement, just like Bondi coordinates did, implies that gravitational radiation decays in inverse powers of a certain radial coordinate along outgoing null geodesics towards $\Scrip$ (and, in addition, along ingoing null geodesics towards $\Scrim$). Cf.~Def.~\ref{defi:intro:peeling}. 
   
    \item \textbf{Asymptotic flatness at the Riemannian level/Initial data decay assumptions:} The two previous notions were global in spirit (although one can certainly make dynamical sense of them in the context of a characteristic initial value problem, see~\cite{Friedrich86}). Viewing GR as an evolutionary theory, with solutions arising from initial Cauchy data $(\Sigma, \bar{g}, \bar{k})$, one can also introduce Riemannian notions of asymptotic flatness by demanding that $\bar{g}$, $\bar{k}$, respectively, measured with respect to a suitable coordinate chart, approach the Euclidean metric $\delta$, and 0, respectively, at a suitable rate near their flat ends.
    This naturally induces a notion of asymptotic flatness on the arising solution. 
    More generally, in any initial value problem, (i.e.~also characteristic IVPs), the corresponding assumptions on the asymptotic behaviour of the data can be viewed as a notion for asymptotic flatness.
    See, for instance, \cite{CK93} and Remark 1.1 of~\cite{I}, or, in the case of the characteristic IVP, Holzegel's notion of asymptotic extendability \cite{Holz16}. 
    \item \textbf{Past-Stationarity:} A fourth notion how to model isolated systems is to simply say that there was no radiation in the infinite past because the system was stationary up until some finite time.
\end{enumerate}
Each of these notions have in common that they introduce some \textit{ad hoc} assumptions on the asymptotic behaviour of gravitational radiation. Historically, in the case of the first two, these assumptions were partially, but incorrectly, motivated by the \textit{no incoming radiation condition} (we will see in this paper that, already for linearised gravity around Schwarzschild, there is no logical connection between no incoming radiation and $1/r$-expansions towards null infinity).

In the case of the third notion, the connection to physics is less clear. One could certainly entertain sufficiently weak assumptions on the decay rate of the data towards spatial infinity so as to \textit{not exclude} any physics. In the context of proving the stability of the Minkowski spacetime, research in this direction has been undertaken in \cite{bieri} and, more recently, in \cite{Shen23}. In fact, we will see later on that, with the exception of \cite{IonescuPausader22}, all other Minkowski  stability proofs starting from Cauchy data (see~\cite{CK93,LR10,Hintz23} and references therein) consider initial data decay that is too fast for the spacetimes constructed in the present work to be admitted.
At the same time, working with notions of asymptotic flatness as general as in \cite{bieri} will likely wash away a lot of the physically interesting information.
Indeed, we will see that a more appropriate approach would be to impose a certain decay rate \textit{together with extra structure} that captures the no incoming radiation condition. See already \S\ref{sec:intro:noincoming}. 

The fourth notion, finally, is by far the simplest, and the prohibition of radiation near spacelike infinity has the clear physical interpretation of an approximation. But when saying that radiation only started at some finite retarded time $u_{-\infty}$, one should really be keeping track of the parameter $u_{-\infty}$ in order to understand the error made by assuming the past to be stationary; this, however,  is tantamount to allowing for radiation at arbitrarily early times. In other words, justifying past-stationarity requires some quantitative understanding of radiation in the infinite past, and, indeed, as was shown in \cite{GK}, the effects of radiation in the infinite past will dominate everything else at sufficiently late times, cf.~\S1(b) of \cite{IV}.

Now, clearly, it would be desirable to not have to make \textit{ad hoc} assumptions on the spacetime or the initial data that one is studying. 
The right framework to meet this desire is given by the \textbf{\textit{scattering framework}}, where, here, we mean by scattering simply that the (scattering) data are posed in the infinite past, with a scattering solution being defined as a solution that attains these (scattering) data in the limit.
For linearised gravity around Schwarzschild, a global scattering theory has been developed by the second author in \cite{Masaood22,Masaood22b}, following earlier works on scattering for the Einstein vacuum equations \cite{Chr09,DHR13} and for the linear wave equation \cite{AAG20scattering,DRSR18,nicolas,Bachelot97,DK87,Friedlander80,Vish70} etc.

The approach that the present paper is based on, and which is motivated by the works \cite{WalkerWill79,Damour86,Chr02} and described in detail in \cite{IV}, is the following:
Wanting to ultimately understand the asymptotic properties (near $\Scrim$, $i^0$ and $\Scrip$) of a system of $N$ infalling masses from past timelike infinity $i^-$, we resort to Post-Newtonian methods to understand the radiation emitted by such a system in the infinite past up until some finite advanced time~$v_1$, i.e.~up until some null cone $\Cin=\Cbar_{v_1}$, cf.~Fig.~\ref{fig:intro:generation}. The idea is that beyond this advanced time, the spacetime is purely vacuum.
On the null cone $\Cin$, we then implement the information from the previous Post-Newtonian analysis as initial data for the linearised Einstein vacuum equations around Schwarzschild, and complete these data with scattering data on $\Scrim$ that capture the no incoming radiation condition. Cf.~Fig.~\ref{fig:intro:propagation}.
\begin{figure}[htpb]
\includegraphics[width=340pt]{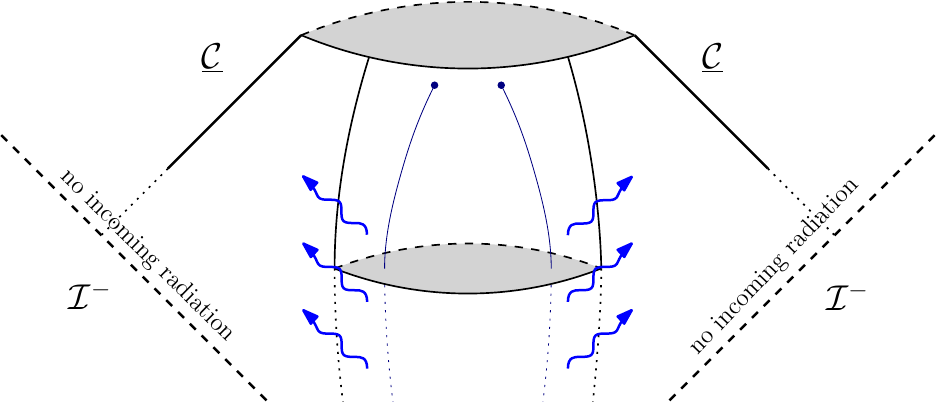} 
\caption{Depiction of infalling masses from the infinite past. We use heuristic arguments to understand the gravitational radiation emitted by these masses up until some null cone $\Cin$.}\label{fig:intro:generation}
\end{figure}

While our scattering data choice is thus motivated from \textit{heuristic} analyses of infalling \textit{masses}, we can also provide a \textit{rigorous} basis for this choice by, instead of considering infalling masses, considering compactly supported gravitational perturbations along $\Scrim$, see already Def.~\ref{defi:intro:N} and Fig.~\ref{fig:intro:gravitons}. As we will see, the behaviour of the corresponding solutions near $i^0$ and~$\Scrip$ will be almost indistinguishable from the case of $N$ infalling \textit{masses}, see already \S\ref{sec:intro:massless}.

In this way, we have set up a mathematical scattering problem that we can solve and asymptotically analyse. In particular, we can now gain a complete  understanding of the asymptotic behaviour of gravitational radiation in a neighbourhood of spacelike infinity (including the null infinities) in the context of physically justified data. 
The results of the corresponding analysis will, in particular, show that the notions of asymptotic flatness in the list above are inadequate, and also give constructive corrections to all these notions, see~\S\ref{sec:intro:reflect}.

Our setup is  motivated by the heuristic works of Walker--Will~\cite{WalkerWill79}, Damour~\cite{Damour86} and Christodoulou~\cite{Chr02}. In the first, the Post-Newtonian analysis referred to above is carried out (for the lowest angular modes, the generalisation to higher modes is given in section 2 of \cite{IV}) to obtain decay rates near $\Scrim$ inconsistent with peeling, while the latter two also put forth heuristic arguments against the failure of peeling towards $\Scrip$, their rates towards $\Scrip$ being in agreement with Thm.~\ref{thm:intro:main} of the present work. See, however, \S\ref{sec:intro:CHr}, where we discuss some accounts of Christodoulou's argument in more detail.

\subsection{Preliminary description of the setup and the system of linearised gravity around Schwarzschild}\label{sec:intro:setup}
\paragraph{Notation: }
We now collect a small subset of the notation employed in this paper so that we'll be able to state the main results more clearly. 
We will mostly be working with standard Eddington--Finkelstein double null coordinates $(u,v)$ (cf.~\eqref{eq:SS:definition of u and v}) on a submanifold of the Schwarzschild manifold $\Schw$ close to spacelike infinity (i.e.~away from the event horizons):
\begin{equation}
    \DoD:=\{(u,v,\theta,\varphi)|\,u\leq u_0<0, v\geq v_1>0\}=D^+(\Cin\cup \Scrimv).
\end{equation}
This region is the domain of dependence of the union of the ingoing null cone $\Cin=\{u\leq u_0,\,v=v_1\}$ and the part of $\Scrim$ with $v\geq v_1$, see Figure~\ref{fig:intro:propagation}. 
The set $\Cin\cup \Scrimv$ is where we will pose our scattering data. We denote by $\Sone,\,\Sinfty$ the future and past end sphere of $\Cin$, respectively, and we further define
\begin{equation}
    r_0(u):=r|_{\Cin}(u)=|u|-2M\log|u|+\O(1).
\end{equation}

\begin{figure}[htpb]
\includegraphics[width=180pt]{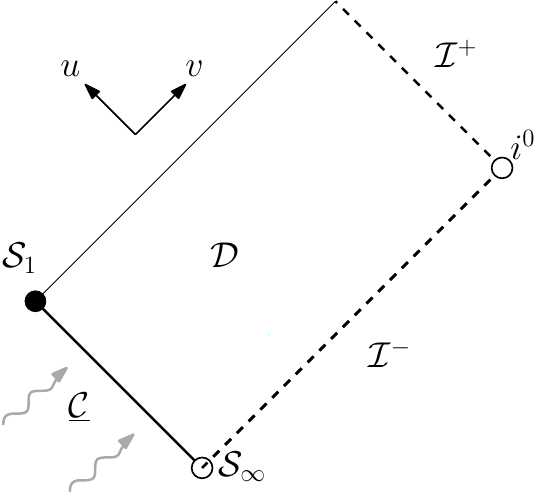} 
\caption{Depiction of the Penrose diagram of $\DoD$. In \S\ref{sec:intro:results1}, we will pose data on $\Cin$ whose choice is informed by the heuristic analysis of the setup of Fig.~\ref{fig:intro:generation}. }\label{fig:intro:propagation}
\end{figure}

Limits "towards $\Scrip$", or "as $v\to\infty$" will always be taken along constant $u$, and vice versa.
We will denote by $\partial_u$ and $\partial_v$ the partial derivatives with respect to $(u,v)$-coordinates, and by $\Du$, $\Dv$ the corresponding projected covariant derivatives. Furthermore, any operator adorned with a~$\mathring{}$ (e.g.~$\D1,\sl,\Ds2$) will denote a differential operator on the unit sphere.
We will mostly be working in a double null orthornormal frame so that, for any covariant $\mathcal{S}_{u,v}$-tangent tensor $\Xi$, 
\begin{equation*}
    (\Du\Xi)(e_A,\dots,e_B)=\pu(\Xi(e_A,\dots,e_B)),
\end{equation*}
$e_A,\dots, e_B$ denoting the components of the frame on the sphere. Whenever we write the integral of a tensor field (e.g.~$\int_{u_1}^{u_2}\Xi \dd u$), this, too, is understood componentwise.

We make use of the standard $\O$ notation, see \S\ref{sec:SS:O} for definitions.

Finally, for any covariant $\mathcal{S}_{u,v}$-tangent tensor $\Xi$, we will use the notation $\Xi_{\ell}$ to denote the projection onto spherical harmonics with angular frequency $\ell$, we write $\Xi^{\mathrm{E}}$ and $\Xi^{\mathrm{H}}$ to denote electric and magnetic part, and we write $\conj{\Xi}=\conj{\Xi^{\mathrm{E}}}+\conj{\Xi^{\mathrm{H}}}=\Xi^{\mathrm{E}}-\Xi^{\mathrm{H}}$ to denote the magnetic conjugate (cf.~Def.~\ref{def:SS:magneticconjugate}).
All these notions and notations are explained in detail in \S\ref{sec:SS}.

\paragraph{The system of linearised gravity around Schwarzschild,} whose stability was shown in \cite{DHR16} (see also \cite{Johnson19}) and whose scattering theory was constructed in \cite{Masaood22, Masaood22b}, is introduced in detail in \S\ref{sec:lin}, see Def.~\ref{def:solution to LEE}, and arises from linearising solutions to \eqref{eq:intro:EVE} in the double null form 
\begin{align}\label{eq:intro:DNGmetric}
    \bm{g}(\epsilon)=-4\bm{\Omega}^2(\epsilon)\dd{u}\dd{v}+\bm{\slashed{g}}_{ab}(\epsilon)(\dd{\theta}^a-\bm{b}^a(\epsilon)\dd{v})(\dd{\theta}^b-\bm{b}^b\dd{v}).
\end{align}
around Schwarzschild, writing $\bm{\Omega}=\Omega+\epsilon\Om+\O(\epsilon^2)$ etc. 
(We here used lowercase Latin indices $a,\, b$ to distinguish from the orthonormal frame $e_A,\,e_B$.)

The arising linearised equations can be regarded as an independent system of equations on Schwarzschild that consists of transport and elliptic equations for the linearised metric components $$\gsh,\,\trg,\,\Om,\,\b,$$
which we will schematically refer to as $\glin$, 
transport and elliptic equations for the connection coefficients $$\trx,\, \trxb\,\xh,\,\xhb,\,\et,\,\etb,\,\om,\,\omb,$$
which we will schematically refer to as $\Gamlin$, and of hyperbolic equations for the Weyl curvature coefficients $$\al,\,\be,\,\rh,\,\sig,\,\beb,\,\alb.$$ We will schematically refer to these as $\Rlin$, and we will also include the Gaussian curvature $\K$ in $\Rlin$.

Some examples for equations from the system of linearised gravity are given by

\smallskip
\begin{subequations}
    \noindent\begin{minipage}{0.5\textwidth}
\begin{equation}
\Du \gsh=2\Omega\xhb\label{eq:intro:gsh3},
\end{equation}
    \end{minipage}%
     \begin{minipage}{0.1\textwidth}
    \centering
    \end{minipage}
    \begin{minipage}{0.5\textwidth}
\begin{equation}
\div\xhb=\Omega\et+r\beb+\frac{1}{2\Omega}\sl\trxb\label{eq:intro:divxhb},
\end{equation}
    \end{minipage}%
    \vskip1em
\end{subequations}
\begin{subequations}
    \noindent\begin{minipage}{0.5\textwidth}
\begin{equation}
 \Du\left(\frac{r^2\xhb}{\Omega}\right)=-r^2\alb,\label{eq:intro:xhb3}
\end{equation}
    \end{minipage}%
     \begin{minipage}{0.1\textwidth}
    \centering
    \end{minipage}
    \begin{minipage}{0.5\textwidth}
\begin{equation}
\Du\left(r\Omega\xh\right)=-2\Omega^2\Ds2\et-{\Omega^2}\Omega \xhb\label{eq:intro:xh3},
\end{equation}
    \end{minipage}%
    \vskip1em
\end{subequations}
\begin{subequations}
    \noindent\begin{minipage}{0.5\textwidth}
\begin{equation}
\Dv(r\Omega^2\alb)=2\Ds2 \Omega^2\Omega\beb +\frac{6M\Omega^2}{r^2}\Omega\xhb\label{eq:intro:alb4},
\end{equation}
    \end{minipage}%
     \begin{minipage}{0.1\textwidth}
    \centering
    \end{minipage}
    \begin{minipage}{0.5\textwidth}
\begin{equation}
\Du\left(\frac{r^4\beb}{\Omega}\right)=-\div r^3\alb.\label{eq:intro:beb3}
\end{equation}
    \end{minipage}%
    \vskip1em
\end{subequations}

For the definition of these quantities as well as the full set of equations, we again refer to~\S\ref{sec:lin}. Let us already emphasize that the curvature coefficients above correspond exactly to the celebrated Newman--Penrose scalars $\Psi_0,\dots,\Psi_4$, i.e., we have
\begin{align}
    \{\al,\be,(\rho,\sig),\beb,\alb\}\longleftrightarrow\{\Psi_0,\Psi_1,\Psi_2,\Psi_3,\Psi_4\}.
\end{align}
Furthermore, $\xh$ and $\xhb$, the out- and ingoing null shear, respectively, correspond to $\sigma $ and $\lambda$ in the notation of the Newman--Penrose formalism. 
For a dictionary between the formalism employed here (which goes back to \cite{CK93}) and that of Newman and Penrose \cite{NP62Approach}, we refer the reader to Appendix~\ref{app:dictionaryCKNP}. 
In particular, we can now define the peeling property:
\begin{defi}\label{defi:intro:peeling}
    A solution to the system of linearised gravity around Schwarzschild satisfies peeling/peels towards $\Scrip$ if $r^5\al$ attains a limit towards $\Scrip$. This implies that $r^4\be,\,r^3(\rho,\sig),\,r^2\beb$ and $r\alb$ also attain limits towards $\Scrip$.

    Similarly, a solution satisfies peeling towards $\Scrim$ if $r^5\alb$ attains a limit towards $\Scrim$, which implies that $r^4\beb,\,r^3(\rh,\sig),\,r^2\be$ and $r\al$ attain limits towards $\Scrim$.
\end{defi}
\begin{rem}
    One could additionally require that the quantities $r^5\al$, $r^5\alb$, respectively, admit $1/r$ expansions towards $\Scrip$, $\Scrim$, respectively, and refer to this as \textit{strong peeling}.
\end{rem}

Let us recall the following insights from \cite{DHR16}: The usual complexity one encounters when studying \eqref{eq:intro:EVE}, namely that it constitutes a set of \textit{geometric} PDE, is still inherited in the linearised system in form of the existence of \textit{linearised gauge solutions} that do not carry physical information. See already~\S\ref{sec:gauge}. 
It is therefore helpful to identify certain \textit{gauge invariant} quantities, i.e.~quantities that remain unchanged under addition of such a gauge solution. The quantity $\alb$ turns out to be such a gauge invariant quantity; moreover, as can already be seen from the equations above (multiply \eqref{eq:intro:alb4} with $\Omega^{-4}r^4$ and then act with $\Du$, using \eqref{eq:intro:xhb3} and \eqref{eq:intro:beb3}), it satisfies a decoupled wave equation, called the Teukolsky equation (of spin $-2$), see \eqref{eq:lin:Teukalb}. Similarly, $\al$ is gauge invariant and satisfies the Teukolsky equation of spin $+2$, see \eqref{eq:lin:Teukal}.
What's more is that both of these equations can, by commutation with suitable vector fields, be transformed into the Regge--Wheeler equation, which is the standard tensorial linear wave equation on Schwarzschild with a potential:
\begin{equation}\label{eq:intro:RW}
    \Du\Dv\Psi-\frac{\Omega^2}{r^2}(\lap-4)\Psi=\frac{6M\Omega^2}{r^3}\Psi.
\end{equation}

We conclude this discussion by remarking that the \textit{dynamics} of the system of linearised gravity are entirely described by the projection to $\ell\geq 2$-angular modes; see \S\ref{sec:gauge:ell01} for details.
In particular, the Teukolsky quantities are supported on $\ell\geq2$.
\subsection{The semi-global scattering theory and the choice of seed scattering data}\label{sec:intro:results1}
\subsubsection{The general semi-global scattering theory}\label{sec:intro:data}
We now provide a summary of the main results obtained in Part~\hyperlink{V:part1}{I} of the paper:
First, we identify a notion of seed scattering data. The precise definitions being given in Definitions~\ref{def:setup:scattering data}--\ref{def:setup:scattering solution}, the specification of seed scattering data along $\Cin\cup\Scrim$ consists of specifying
\begin{itemize}
    \item the limits of $r\xh$, $r^{-1}\b$ and $\Om$ towards $\Scrimv$,
    \item the restrictions of $\gsh$ and $\omb$ to $\Cin$,
    \item the restrictions of $\beb$, $\trg$, $\trx$ and $\trxb$ to $\Sone$, the final sphere of $\Cin$.    
\end{itemize}

The main result that is proved over the course of sections~\ref{sec:setup}--\ref{sec:construct} is the following (see Thm.~\ref{thm:setup:LEE Scattering wp} for the precise statement):
\begin{thm}\label{thm:intro:LEE Scattering wp}
    Suppose that a smooth scattering data set on $\Cin\cup\Scrim$ as above is given, and suppose there exist positive numbers $\epsilon,\delta\in\mathbb R_{>0}$ as well as a smooth symmetric and tracefree two-tensor field $\rad{\gsh}{\Sinfty}$ such that, as $u\to-\infty$, these seed data satisfy along $\Cin$ 
\begin{align}\label{eq:intro:general decay along C}
\lim_{u\to-\infty}\rad{\gsh}{\Cin}=\rad{\gsh}{\Sinfty},&&  \rad{\gsh}{\Cin}-\rad{\gsh}{\Sinfty}=\O_{2}\left(r^{-\frac{1}{2}-\delta}\right),&&
    \rad{\omb}{\Cin}=\O\left(r^{-1-\epsilon}\right).
\end{align}
Then there exists a unique scattering solution realising this seed scattering data set.
\end{thm}
An overview of the proof of Theorem~\ref{thm:intro:LEE Scattering wp} can be found in \S\ref{sec:setup:overview}.
\begin{rem}
    The precise requirement on the initial data captured by \eqref{eq:intro:general decay along C} is that they induce finite Regge--Wheeler energy. In view of the  recent work of Shen~\cite{Shen23}, one might expect to be able to prove this result under weaker decay assumptions ($\delta>-1/2$) by working directly with Bianchi pairs instead, but the rates \eqref{eq:intro:general decay along C} are sufficiently general for the present work.
\end{rem}
\begin{rem}
There are, of course, many different but equivalent choices one can make to specify seed data. For instance, in view of \eqref{eq:intro:divxhb} we could have specified $\etb$ instead of $\beb$. 
Notice, however, that it is, in general, not possible to replace the specifications along $\Sone$ with specifications along $\Sinfty$, unless the initial data decay very fast. 
One way of seeing this is that the decay rate in \eqref{eq:intro:general decay along C} corresponds, by \eqref{eq:intro:gsh3} and \eqref{eq:intro:xhb3}, to a $r^{-5/2-\epsilon}$-decay rate for $\alb$, which means that \eqref{eq:intro:beb3} cannot be integrated from $\Sinfty$: In order to obtain $\beb$ along $\Cin$, we are thus forced to specify its values along $\Sone$.
\end{rem}

The seed scattering data carry a large number of unphysical degrees of freedom. Although not necessary for the proof of the theorem, it is helpful to understand this freedom in more detail.
One way of entirely removing this freedom is as follows:
\begin{defi}\label{defi:intro:initialgauge}
It is possible to subtract pure gauge solutions such that $\trg$, $\trx$ and $\trxb$ vanish on $\Sone$, $\omb$ vanishes along $\Cin$, the limit $\radsinf\gsh$ vanishes along $\Sinfty$, and the limits $\Om$ as well as $r^{-1}\b$ vanish along $\Scrimv$, cf.~Prop.~\ref{prop:setup:Bondi} and Corollary~\ref{cor:setup:Bondiplus}.
A solution for which all these quantities vanish is said to be in \textbf{the initial data gauge.}
\end{defi}
The initial data gauge indeed entirely fixes the unphysical degrees of freedom of the data in the sense that any nontrivial addition of a \textit{linearised gauge solution} will make at least one of the quantities in Def.~\ref{defi:intro:initialgauge} nonvanishing.

The initial data gauge reveals that the physical degrees of freedom of the data are then entirely carried by $r\xh$ along $\Scrimv$, by $\Du\gsh=2\Omega\xhb$ along $\Cin$, and by $\beb$ (which, by \eqref{eq:intro:alb4} is equivalent to $\Dv\alb$) along $\Sone$.

Later on, we will fix the gauge by a double Bondi normalisation condition, which is defined as follows:
\begin{defi}\label{defi:intro:Bondi}
    A solution is said to be \textbf{Bondi-normalised towards $\Scrim$} if the limits of $\Om,\,\gsh$ and $r^2\K$ towards $\Scrim$ vanish.
    
    A solution is said to be \textbf{Bondi-normalised towards $\Scrip$} if the limits of 
    $\Om,\,\gsh$ and $r^2\K$ towards $\Scrip$ vanish.

    A solution is \textbf{doubly Bondi-normalised} if it is Bondi-normalised towards both $\Scrim$ and~$\Scrip$.
\end{defi}
The effect of Bondi normalisation is that the limit of the linearised metric perturbation vanishes identically along $\Scrim$, $\Scrip$, respectively.
\begin{rem}
    Our initial data gauge in particular implies the corresponding solution to be Bondi-normalised towards $\Scrim$.
\end{rem}

At this point, we believe it is appropriate to briefly contrast the present work to the second author's \cite{Masaood22b}: In the present work, we write down specific scattering data that model a certain class of physics; so the central object are the \textit{scattering data and their asymptotics}. 
On the other hand, \cite{Masaood22b} studies the \textit{global scattering map} and its properties in the entire domain of outer communications of Schwarzschild. But in order to study properties of the scattering map, it suffices to only analyse the scattering data in some convenient dense subset (namely that of smooth, compactly supported data), and to then extend the results by density arguments. 
Since the semi-global scattering problem is simpler, but somewhat different, from the global scattering problem, we use the opportunity to give an entirely self-contained, alternative presentation of the former, which, in particular, does not use a gauge fixing procedure and clearly elucidates how to implement physically motivated assumptions on the scattering data.

\subsubsection{The seed data describing the exterior of the \texorpdfstring{$N$}{N}-body problem.}
Having discussed a fairly general class of seed scattering data so far, we now want to establish the connection of these seed data to some actual physics, namely to the physics of $N$ infalling bodies from the infinite past.
For such systems, we always want to exclude radiation coming in from the infinite past, at least for sufficiently late times. 
The gauge-invariant mathematical version of this condition is given by the following requirement:
\begin{equation}\label{eq:intro:noincomingradiation}
    \lim_{u\to-\infty}(r\div\xh-r^2\sl\K)=0\quad \text{for all}\,v\geq v_1,
\end{equation}
see already Def.~\ref{def:noincomingradiation} in  \S\ref{sec:physd} for details and further commentary.

While the results of this paper provide the asymptotics for general seed scattering data, we here single out three specific cases, each of them capturing a different version of $N$ infalling bodies from the infinite past:
A loose definition of each of these three different cases is as follows (see Definitions~\ref{defi:physd:Nbodyseed}--\ref{defi:physd:graviton} for precise versions).
\begin{defi}\label{defi:intro:N}
Suppose that a given seed scattering data set has no incoming radiation along $\Scrimv$. Then we say that it describes the exterior of a system of\textbf{ $N$ infalling {\textbf{masses}} following approximately \textbf{hyperbolic} orbits near the infinite past} if \ref{intro:item:defNBodyalb}--\ref{intro:item:defNBodyal} hold:
    \begin{enumerate}[label=(\Roman*)]
\item There exists $\delta>0$ and some nonvanishing $\albdata$ with $\Du\albdata=0$ (and supported on all angular modes) such that 
\begin{equation}\label{intro:eq:albdatadecay}
\radc\alb=-6\albdata r^{-4}+\O(r^{-4-\delta}).
\end{equation}\label{intro:item:defNBodyalb}
\item The limit $\lim_{u\to-\infty} r^2\radc\xhb=\radsinf\xhb$ exists and is non-vanishing (and is supported on all angular modes).\label{intro:item:defNBodyxhb}
\item The limit $\lim_{u\to-\infty} r^3\radc\al=\aldata$ exists and is non-vanishing (and is supported on all angular modes).\label{intro:item:defNBodyal}
\end{enumerate} 

If, on the other hand, we replace in the above \eqref{intro:eq:albdatadecay} by
\begin{equation}
    \radc\alb=\mathscr{C} r^{-5}+\O(r^{-6})
\end{equation}
for some $\mathscr{C}$ supported on all angular modes with $\Du\mathscr{C}=0$, moreover assuming that $\radc\alb$ admits an asymptotic expansion in powers of $1/r$, and otherwise leave items \ref{intro:item:defNBodyxhb} and \ref{intro:item:defNBodyal} unchanged, then we say that the seed scattering data describe the exterior of a \textbf{compactly supported gravitational perturbation} along $\Scrim$. (The reader may think of this perturbation as being localised around $N$ points along $\Scrim$.)

If, finally, we instead replace in the above \eqref{intro:eq:albdatadecay} by
\begin{equation}
\radc\alb=-\albdata_{\mathrm{par}}r^{-11/3}+\O(r^{-4}), 
\end{equation}
with $\albdata_{\mathrm{par}}$ supported only on $\ell=2$ (and $\Du\albdata_{\mathrm{par}}=0$), and demand the limits in items \ref{intro:item:defNBodyxhb} and \ref{intro:item:defNBodyal} to vanish, then we say that the seed scattering data describe the exterior of a system of \textbf{$N$ infalling \textbf{masses} following approximately \textbf{parabolic} orbits}.
\end{defi}
We note that these definitions do not entirely fix the seed scattering data, but they fix the relevant leading order asymptotic behaviour of the data and the resulting solutions!

The definitions for hyperbolic orbits and parabolic orbits are both merely justified by the Post-Newtonian theory (cf.~Fig.~\ref{fig:intro:generation}); where, for instance, the projection $\mathscr{A}_\ell$ attains the physical interpretation of corresponding to the $\ell$-th Newtonian matter multipole  moment of the infalling masses. This has been discussed in detail in \S2 of~\cite{IV}, see also \S\ref{sec:physd} for a brief recap.  

In contrast, the definition for a compactly supported gravitational perturbation is, at the same time, a mathematical proposition that generic compactly supported gravitational perturbations (this means that the limit on the LHS of \eqref{eq:intro:noincomingradiation} is compactly supported along~$\Scrim_{v<v_1}$) induce such seed scattering data, cf.~Fig.~\ref{fig:intro:gravitons}. 
To further the comparison to the case of infalling masses following hyperbolic orbits, we note that $\mathscr{A}_{\ell}$ is then given by \eqref{eq:physd:graviton:al}, which can be viewed as the radiative analogue to the Newtonian matter multipole moment. Cf.~Prop.~\ref{prop:physd:graviton} as well as Remark~\ref{rem:physd:moments}. 
\begin{figure}[htpb]
\includegraphics[width=110pt]{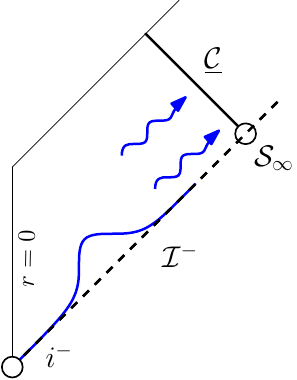} 
\caption{Instead of using heuristic methods to understand the gravitational radiation emitted by infalling masses as in Fig.~\ref{fig:intro:generation} up until some null cone $\Cin$, we can \textit{rigorously} understand the radiation emitted by compactly supported gravitational perturbations along $\Scrim$. }\label{fig:intro:gravitons}
\end{figure}

\subsection{Semi-global asymptotics for seed data describing the exterior of the \texorpdfstring{$N$}{N}-body problem}\label{sec:intro:results2}
Having discussed the general scattering theory for linearised gravity around Schwarzschild, as well as motivated our choices for the seed scattering data, we now move on to discuss the asymptotics of solutions to the semi-global scattering problem.
The way we find these asymptotics is by exploiting a class of approximate conservation laws (related to the Newman--Penrose charges, see already \S\ref{sec:intro:NPcharges}) that are valid for fixed angular mode solutions to the Teukolsky equations of general spin $s$, $\al$ corresponding to $s=+2$ and $\alb$ corresponding to $s=-2$.
Very schematically, in the case $s=2$, these look like
\begin{equation}\label{eq:intro:cons:scheme}
    \Du(r^{-2\ell}\Dv(r^2\Dv)^{\ell}(r^5\al_{\ell}))\approx r^{-2\ell-3}\cdot M (r^2\Dv)^{\ell}(r^5\al_{\ell}).
\end{equation}
See \S\ref{sec:cons} for a general discussion and \eqref{eq:cons:cons0} for the precise version. See also \cite{AAG21,III,MZ21,GK}, where such conservation laws have been used in similar contexts.
Essentially, as a consequence of the extra power in decay on the RHS of \eqref{eq:intro:cons:scheme}, this equation can be integrated along characteristics an appropriate amount of times to obtain an asymptotic expression for $\al_{\ell}$, see also \S5 of \cite{IV} for a sketch.
Roughly speaking, the general correspondence coming out of this is then as follows:
\par\smallskip\smallskip\noindent
\centerline{\fbox{\begin{minipage}{0.95\textwidth}
\textit{Let $p\in\mathbb{R}_{>-1/2}$ and $q\in\mathbb N$. Consider solutions to the Teukolsky equation \eqref{eq:lin:Teukal} for which $\al$ decays like $r^{-3-p}\log ^q r$ along $\Cin$ and for which $\lim_{u\to-\infty}\Dv(r\al)=0$. If $p< 1$, then $\al_{\ell}\sim r^{-4-p}\log ^q r$ towards $\Scrip$, whereas, if $p\geq 1$, then the first term in the asymptotic expansion of $\al_{\ell}$ towards $\Scrip$ that is not regular in $1/r$ is given by $Mr^{-4-p}(1+\delta_{p\in\mathbb N_{\geq 1}}\log r)\cdot \log^q r$.
}
\end{minipage}}}
\par\smallskip\smallskip
See Theorem~\ref{thm:al:0} for the precise version for general spin.

Below, we will content ourselves with stating the  asymptotic results associated to the specific data choices of Def.~\ref{defi:intro:N}.
\begin{rem}
In the boxed correspondence above, the assumption $\lim_{u\to-\infty}\Dv(r\al)=0$, on its own, is, in fact, consistent with nontrivial incoming radiation. However, together with the assumption that $\al$ along $\Cin$ decays faster than $r^{-5/2}$, incoming radiation is excluded. Cf.~Rem.~\ref{rem:physd:noincomingDUvanishes}.
\end{rem}

\subsubsection{The case of hyperbolic orbits}
We employ the notation $\O^u$ with a $u$-superscript to mean that the precise error is a function of only the variable $u$, see \S\ref{sec:SS:O}. 
\begin{thm}\label{thm:intro:main}
    Given a smooth seed scattering data set describing the exterior of $N$ infalling masses following approximately \textbf{hyperbolic orbits} near the infinite past, then the corresponding solution satisfies the following estimates throughout $\DoD$ for any $\ell\geq 2$:
    
  \paragraph{(I)} \hypertarget{thm1}{} The projections $\al_{\ell}$ exhibit the following asymptotics, where $\B_{\ell}:=\tfrac{(\ell+2)!}{(\ell-2)!}\conj{\albdata}_{\ell}$:
    \begin{multline}\label{eq:intro:al_asy}
\frac{r^5\al_{\ell}}{\Omega^2}= \sum_{n=0}^{\ell-2} \left(\frac{r_0}{r}\right)^n\{\A_\ell r_0^2 S_{\ell,0,\ell-2,n,2}+\B_\ell r_0 \log r_0 S_{\ell,1,\ell-2,n,2}+\O^u(r_0) \}\\
+ \frac{(-1)^{\ell+1}2M\A_\ell r}{(\ell-1)(\ell+2)}+\frac{(-1)^{\ell+1}{3}M\B_\ell \log ^2 r}{\ell(\ell+1)}
+M\O(r_0\log r),
\end{multline}
where the nonvanishing constants $S_{\ell,0,\ell-2,n,2}$ are defined in~\eqref{eq:al:SLPJN}.
In particular, the following limits are conserved
\begin{equation}\label{eq:intro:thm:limitsconserved}
    \lim_{v\to\infty}r^4\al_\ell=(-1)^{\ell+1} 2M\A_{\ell}\frac{1}{(\ell-1)(\ell+2)}= \frac{(\ell-2)!}{(\ell+2)!}\lim_{v\to\infty} \frac{1}{r}(\rDv)^{4}(\Omega^2r\conj{\alb}_{\ell})
\end{equation}

As $\al$ and $\alb$ are both gauge invariant, these statements are gauge-independent.
We now fix the gauge by imposing the \textbf{initial data gauge} defined in Def.~\ref{defi:intro:initialgauge}.

\paragraph{(II)}\hypertarget{thm2}{}  All curvature coefficients decay like $r^{-3}$ along $t=0$, i.e.~the limits of $r^3\Rlin_{\ell}$ along $t=0$ towards $i^0$ are all finite and generically nonvanishing.
Similarly, all connection coefficients $\Gamlin_{\ell}$ decay like $r^{-2}$ towards $i^0$, and all metric coefficients $\glin_{\ell}$ decay like $r^{-1}$ towards $i^0$.

\paragraph{(III)} \hypertarget{thm3}{}The curvature coefficients $r\alb_{\ell}$, $r^2\beb_{\ell}$, $r^3(\rh_\ell,\sig_{\ell})$ all attain finite limits towards $\Scrip$, and we have:
\begin{equation}\label{eq:intro:thm:limitofalphabar}
\rad{\alb}{\ell,\Scrip}:=\lim_{v\to\infty}r\alb_{\ell}=(-1)^\ell \left(\frac{12\Bb_\ell}{\ell(\ell+1)}-\frac{4M(\ell-2)!\conj{\A}_\ell}{(\ell+2)!}\right) r_0^{-3}+\O(r_0^{-3-\delta}+Mr_0^{-4}),
\end{equation}
as well as the relations
\begin{align}\label{eq:intro:beblimit}
   \radlp{\beb}:&=    \lim_{v\to\infty}r^2\beb_{\ell}=-\int_{-\infty}^u\div\radlp{\alb}\dd u',\\\label{eq:intro:rhosiglimit}
        (\radlp\rh,\radlp\sig):&=\lim_{v\to\infty}(r^3\sig_{\ell},r^3\rh_{\ell})=(-1)^{\ell}(\radsinf{\sig},\radsinf\rh)_\ell-\int_{-\infty}^u(\div\radlp\beb,\curl\radlp\beb)\dd u'.
\end{align}
Here, $\radsinf\rh,\,\radsinf\sig$ are the unique $\ell\geq2$ scalar functions such that $\A=\Ds2\Ds1(\radsinf\rh,-\radsinf\sig)$.

Moreover, we have
\begin{equation}\label{eq:intro:belimit}
    \lim_{v\to\infty} \frac{r^4}{\log r}\be_{\ell}=\frac{(-1)^{\ell}2M}{(\ell-1)(\ell+2)}\div \A_{\ell}.
\end{equation}
\paragraph{(IV)}\hypertarget{thm4}{}
The connection coefficients $\frac{r^3}{\log r}\om_\ell,\,r\et_{\ell},\,r^2\etb_{\ell}$, $r^2\trx_{\ell},\,r\trxb_\ell$, $r^2\xh_\ell,\,r\xhb_\ell$, and, in addition, the components $r^3\K_{\ell},\,r\trg_{\ell}, r\gsh_{\ell}$ as well as $\b_{\ell},\,\omb_{\ell}$ and $\Om_{\ell}$ all attain finite limits towards $\Scrip$, with the latter limit decaying like $r_0^{-2}\log r_0 $ as $u\to-\infty$. 

We can use the limit of $\Om$ to define and subtract a linearised gauge solution such that the solution becomes \textbf{doubly Bondi-normalised.}
\paragraph{(V)}\hypertarget{thm5}{}
The doubly Bondi-normalised form of the solution exhibits the following semi-global asymptotics for $\xh_{\ell}$ and $\xhb_{\ell}$:
\begin{align}\label{eq:intro:thm:xh}
     \frac{{r^2\xh}_{\ell}}{\Omega}&=\radlp{\xh}-\frac{(-1)^{\ell}2M\A_{\ell}}{(\ell-1)(\ell+2)}\frac1r-\frac{(-1)^{\ell}3M\B_{\ell}}{2\ell(\ell+1)}\frac{\log^2 r}{r^2}+\O\left(M \frac{r_0\log r}{r^2}+\frac{r_0^2}{r^2}\right),\\
     r\Omega\xhb_{\ell}&=\radlp\xhb-\frac{1}{r}\lim_{v\to\infty}r^2\Dv(r\Omega\xhb_{\ell})+(-1)^{\ell}\cdot M \conj{\A_{\ell}}\frac{r_0\log r}{r^2}+\O(r_0/r^2),
\end{align}
with the limits of $\xh$ and $\xhb$ satisfying
\begin{align}\label{eq:intro:xhblimit}
    \radlp\xh=\radlp\xh^--\int_{-\infty}^u \radlp\xhb,&&\radlp\xhb=-\int_{-\infty}^u\radlp\alb,\\
    \curl\div\radlp\xh=\radlp\sig,&&\div\radlp\xhb=\radlp\beb.\label{eq:intro:curldivlimit}
    \end{align}
The precise expression for $\radlp\xh^-=\lim_{u\to-\infty}\radlp\xh$ is given in \eqref{eq:xh:prop:xhlimits}. In particular, we have, for any $\ell\geq 2$:
\begin{equation}\label{eq:intro:thm:antipodal}
    \lim_{u\to-\infty}\lim_{v\to\infty}r^2\xh_{\ell}^{\mathrm{H}}=(-1)^{\ell}\lim_{v\to\infty}\lim_{u\to-\infty}r^2\xhb_{\ell}^{\mathrm{H}}.
\end{equation}
\paragraph{(VI)}\hypertarget{thm6}{}
The quantities $\frac{r^3}{\log r}\om_{\ell}$, $\frac{r^2}{\log r}\Om_{\ell}$ and $\frac{r^2}{\log r}\omb_{\ell}$ attain a finite, conserved limit towards $\Scrip$ iff $M\A_{\ell}\neq0$.
Finally, the quantities $r^2\et_{\ell},\,r^2\etb_{\ell}$, $r^2\trx_{\ell},\,r^2\trxb_{\ell}$, $r^3\K_{\ell}$, $r\trg_{\ell}$, $r\gsh_{\ell}$ and $r\b_{\ell}$ all attain a finite limit at $\Scrip$ that schematically behaves like "$\mathrm{constant}+\O(r_0^{-1}\log r_0)$" as $u\to-\infty$, with the limits of $r^2\trx_{\ell}$, $r\trg_{\ell}$ and $r\b_{\ell}$ being exactly conserved.
\end{thm}
Let us provide some commentary:
The  two most important results of the theorem are \eqref{eq:intro:al_asy} and \eqref{eq:intro:thm:limitofalphabar} of points \hyperlink{thm1}{\textbf{(I)}} and \hyperlink{thm3}{\textbf{(III)}}---they form the basis for all remaining statements of the theorem.
Indeed, at the heart of the proof of the theorem lies the asymptotic analysis of the Teukolsky variables $\al$ and $\alb$, which is presented in \S\ref{sec:alp} and \S\ref{sec:alb} (and which is a generalisation of the analysis presented in \cite{III}).
Notice, in particular, that equation~\eqref{eq:intro:al_asy} shows that the peeling property does not hold (cf.~Def.~\ref{defi:intro:peeling}).

At this point, we already want to make the important remark that the methods employed in the present work are tailored specifically to fixed angular mode analysis, and the control on the $\ell$-dependent error terms hidden in the $\O$-terms of e.g.~\eqref{eq:intro:al_asy} \textit{\textbf{does not yet allow for summation of any of these estimates.}} 
This issue is resolved in upcoming work \cite{X}, see also \S\ref{sec:sum} for further commentary.

The main point of part \hyperlink{thm4}{\textbf{(IV)}} of the theorem is to show that, if the solution is Bondi-normalised towards $\Scrim$, then it is \textit{almost} Bondi-normalised towards $\Scrip$, with only $\Om$ not vanishing towards $\Scrip$. We can then eliminate $\Om$ along $\Scrip$ by subtracting a gauge solution that leaves the Bondi-normalisation of $\Scrim$ in tact. 
As in \cite{Masaood22b}, one can also show that this gauge solution can be controlled in $L^2$ in terms of an initial data energy. (See Theorem III of Section 8.2.4 therein for an example of such an estimate for the \textit{global} scattering problem. See also Section 10.4 and Remark 11.2 of \cite{Masaood22b} for the detailed derivation of the estimate of Theorem III of \cite{Masaood22b}.)

Part \hyperlink{thm5}{\textbf{(V)}} then provides the expansions of the shears towards $\Scrip$, along with the \textit{laws of gravitational radiation} \eqref{eq:intro:xhblimit}--\eqref{eq:intro:curldivlimit} (in integral form) and what can be interpreted as the \textit{antipodal matching condition} \cite{Strominger14}, namely~\eqref{eq:intro:thm:antipodal}  (see also \cite{PI22}, \cite{CNP22}).

Part \hyperlink{thm6}{\textbf{(VI)}} provides a complete description of the asymptotic behaviour of the doubly Bondi-normalised solution towards $\Scrip$. These rates are relevant for e.g.~studying the black hole stability problem starting from characteristic initial data extending towards $\Scrip$ (as in \cite{DHR16,KlSz21,DHRT21}), cf.~Holzegel's definition of asymptotic extendability \cite{Holz16} and the discussion in \S\ref{sec:intro:noincoming}.

\subsubsection{The massless case}\label{sec:intro:massless}
The theorem above describes the semi-global asymptotic properties of solutions describing the exterior of $N$ infalling masses following approximately hyperbolic orbits near the infinite past. We can make a similar statement in the massless case (cf.~Def.~\ref{defi:intro:N}):
\begin{thm}\label{thm:intro:main2}
    All results of Theorem~\ref{thm:intro:main} also apply to the case of seed scattering data describing the exterior of a compactly supported gravitational perturbation along $\Scrim$ when setting $\underline{\B}$ to vanish.
\end{thm}
In view of the linearity of the system, one can of course also consider a situation that describes the exterior of both infalling masses and compactly supported gravitational perturbations. (For instance, one could add a fine-tuned choice of compactly supported data for the gravitational field to the data of Thm.~\ref{thm:intro:main} such that the limit \eqref{eq:intro:thm:limitsconserved} vanishes.)
It is somewhat remarkable that the asymptotic behaviour near $\Scrip$ of a massive and a massless scattering process is so similar.  This is in stark contrast to the behaviour on Minkowski: To illustrate this, we record the following curious
\begin{obs}\label{Observation}
    Consider the case of seed scattering data describing the exterior of a compactly supported gravitational perturbation along $\Scrim$. 
If $M\neq 0$, then, for any $\ell\geq2$,
\begin{equation}\label{eq:intro:observationMneq0}
    \lim_{v\to\infty}r\alb_{\ell}=(-1)^{\ell+1}\frac{(\ell-2)!}{(\ell+2)!}4M\conj{\A}_{\ell}r_0^{-3}+\O(r_0^{-4}).
\end{equation}

On the other hand, if $M=0$, then there exist generically nonvanishing $\mathscr{E}_{\ell}$ such that, for any $\ell\geq2$:\footnote{Recall that $\alb$ is only supported on $\ell\geq2$.}
\begin{equation}\label{eq:intro:observationM=0}
    \lim_{v\to\infty}r\alb_{\ell}=\mathscr{E}_{\ell}r_0^{-3-\ell}+\O(r^{-4-\ell}).
\end{equation}
\end{obs}
See also Fig.~\ref{fig:M=0}. 
While \eqref{eq:intro:observationMneq0} follows directly from \eqref{eq:intro:thm:limitofalphabar}, \eqref{eq:intro:observationM=0} is the cause of many special cancellations present on Minkowski and is discussed at the end of \S\ref{sec:alb}, cf.~Thm.~\ref{thm:alb:M=0}.
The observation thus shows that the structure afforded by these cancellations\footnote{These cancellations are generalisations of the following: Solutions to the wave equation $\Box\phi=0$ on Minkowski with $r\phi|_{\Cin}=r^{-p}$ and with $\pv(r\phi)|_{\Scrimv}=0$ satisfy $\lim_{v\to\infty} r\phi_{\ell}=0$ if $\ell\geq p$ and $p\in\mathbb N_{>0}$, cf.~Remark 1.4 of \cite{III}.\label{foot:intro:cancel}} is destroyed by the presence of nonzero mass, changing the decay rate by $\ell$ powers, and making the asymptotic behaviour towards and along $\Scrip$ look almost identical to that of the case of masses following hyperbolic orbits. 

\subsubsection{The case of parabolic orbits}
In the case of parabolic orbits, on the other hand, the difference to hyperbolic orbits is more notable. The following theorem summarises the main differences. Recall that we assumed the quantity $\albdata_{\mathrm{par}}$ to be supported only on $\ell=2$, nevertheless, we state the theorem for $\albdata_{\mathrm{par}}$ supported on all $\ell$.
\begin{thm}\label{thm:intro:para}
    Given a smooth seed scattering data set describing the exterior of a system of $N$ infalling masses following approximately parabolic orbits near the infinite past, then we have the following estimates throughout $\DoD$ for any $\ell\geq 2$:
    \begin{equation}
    \lim_{v\to\infty}r^{14/3}\al_{\ell}=(-1)^{\ell+1}\sqrt{3}\pi (\ell-\tfrac23)(\ell+\tfrac53)\conj{\albdata}_{\mathrm{par},\ell}.
    \end{equation}
    Moreover, we have
    \begin{equation}    \lim_{v\to\infty}r\alb_{\ell}=C_{\ell}\albdata_{\mathrm{par},\ell}r_0^{-\frac{8}{3}}+\O(r_0^{-\frac83-\delta}+Mr_0^{-\frac{11}{3}}),
    \end{equation}
    with $C_{\ell}=-1/2$ for $\ell=2$. 
\end{thm}
Note that the decay of $r\alb_{\ell}$ along $\Scrip$ is slower in the parabolic case than in the hyperbolic case (cf.~\eqref{eq:intro:thm:limitofalphabar}) even though, at any finite $v$, $\alb$ decays faster towards $\Scrim$ in the parabolic case! This is because the hyperbolic case features a cancellation that only occurs for integer decay rates, cf.~Footnote~\ref{foot:intro:cancel}.
There is also a simple physical explanation for this fact: The quadrupole formula predicts, roughly speaking, that the limit of $r\alb_{\ell=2}$ behaves like the fourth time derivative of the quadrupole moment of the matter distribution under consideration. In the case of hyperbolic orbits, where the matter distribution grows like $|t|+\log |t|$, the leading order term is annihilated by four derivatives, and only the $\log$-term survives. This is not the case for parabolic orbits, for which the matter distribution grows like $|t|^{2/3}$.

\subsection{Revisiting previous notions of asymptotic flatness and further remarks}\label{sec:intro:reflect}
Recall from \S\ref{sec:intro:motivation} that a central motivation of the present work is rooted in the desire to understand how to model isolated systems/asymptotic flatness in general relativity.
Now, we certainly do not labour under any illusions to the account of having satisfied this desire entirely, not least because our work completely sidesteps the question of the cosmological constant.

Nevertheless, we still believe that the present work marks definitive progress towards the question of modelling isolated systems. 
This is because the constructions presented in this work should not just be viewed as counter-examples to, say, a smooth null infinity---instead, the work constructively provides an alternative approach (as well as corrections) to the notions listed in~\S\ref{sec:intro:motivation}.

Furthermore, even though our constructions are merely examples, they may potentially serve to cover a quite general range of physics: For, on the one hand, there is a certain sense in which one might view e.g.~hyperbolic orbits in the infinite past as \textit{worst-case scenarios} in a class of "tame" physics, as they correspond to the fastest possible growth in size of a system near the infinite past. This would suggest that the decay rate $\al\sim r^{-4}$ towards $\Scrip$ is as bad as it gets. 
Here, tame physics loosely means physics that are well-described by the Post-Newtonian approximation, and would exclude, for instance, hypothetical systems that remain bounded and exhibit oscillations near the infinite past and are somehow stabilised by incoming gravitational radiation.\footnote{The construction of an interactive system where the masses remain confined in a bounded region and are stabilised by radiation would be an interesting problem in its own right!} On the other hand, within this class of tame physics, the cases described in Def.~\ref{defi:intro:N} are the only ones we can reasonably expect, so long as we exclude incoming radiation at late advanced times or non-gravitational interactions.

We therefore believe that the asymptotics found in \S\ref{sec:intro:results2} have a high degree of physical relevance.

In view of this, we will list below the modifications to the notions mentioned in \S\ref{sec:intro:motivation} that are induced by the results of the present work. 
We hope that this will provide those who commonly work with these notions (Bondi coordinates etc.) with clearly motivated alternative assumptions.
\subsubsection{Modifications to Bondi coordinates}\label{sec:intro:sub:Bondi}
The metric in Bondi coordinates as originally written down in \cite{SeriesVIII} is given by
\begin{equation}
   \dd s^2= \frac{V\e^{2\beta}}{r}\dd u^2-2\e^{2\beta}\dd u \dd r+r^2h_{AB}(\dd x^A-U^A\dd u)(\dd x^B-U^B\dd u),
\end{equation}
with (frequently, this metric is instead written down with $(\gamma,\delta)\mapsto\frac12(\gamma+\delta,\gamma-\delta)$, see \cite{MWini16}) 
\begin{equation}
    2h_{AB}\dd x^A\dd x^B=(\e^{2\gamma}+\e^{2\delta})\dd \theta^2+4\sin\theta\sinh(\gamma-\delta)\dd \theta\dd\phi+\sin^2\theta(\e^{-2\gamma}+\e^{-2\delta})\dd \phi^2.
\end{equation}
Without going into much detail on  the construction in \cite{SeriesVIII}, what was then called the \textit{outgoing radiation condition} (and what is essentially the imposition of analyticity in $1/r$) is most prominently implemented by the assumption that both $\gamma$ and $\delta$ have an expansion towards $\Scrip$ of the form
        \begin{equation}\label{eq:intro:Bondiexp1}
            (\gamma, \delta)=\frac{c(u,\theta,\phi)}{r}+\frac{b(u,\theta,\phi)}{r^2}+\O(r^{-3}),
        \end{equation}
with $b$ \textit{vanishing}! (For, it was argued, if $b$ does not vanish, then logarithmic terms would appear in the expansion of $V$, for instance, violating the analyticity assumption.)

Now, as was shown in \cite{Kroon99}, the asymptotics derived in Theorem~\ref{thm:intro:main}, namely that $\al$ (which corresponds to $\Psi_0$ in the N--P formalism) decays like $r^{-4}+r^{-5}\log^2 r+r^{-5}\log r+\O(r^{-5})$ towards $\Scrip$, mean that both $\gamma$ and $\delta$ have a schematic expansion of the form
\begin{equation}\label{eq:intro:Bondiexp2}
    (\gamma,\delta)=\frac{c}{r}+\frac{b}{r^2}+\frac{b'\log^2 r}{r^3}+\dots,
\end{equation}
with neither $b$ nor $b'$ vanishing (note that this roughly corresponds to the expansion of $\xh$ in \eqref{eq:intro:thm:xh}\footnote{Furthermore, $c(u,\theta,\phi)$ then roughly corresponds to the limit of $r\xh$ in our notation. In particular, in the setting of Theorem~\ref{thm:intro:main}, we would get that $c(u,\cdot)=c_0+c_1|u|^{-1}+\dots$ as $u\to-\infty$, with the coefficient $c_0$ determined by antipodal matching as in \eqref{eq:intro:thm:antipodal}.}).
(Similarly, in the case of parabolic orbits (Thm.~\ref{thm:intro:para}), one would instead find an expansion of the form $c/r+b_{\mathrm{par}}/r^{8/3}$.)

If one now follows through the integrations presented in \cite{SeriesVIII} (see also \cite{MWini16, SeriesXIV}), one will find that the quantity $U$ has a logarithmic term appearing at order $r^{-3}\log r$, the quantity $V$ has a logarithmic term appearing at order $r^{-1}\log r$, and so on. We leave the details of the derivation to the reader.

In other words, the notion of Bondi coordinates is capable of capturing the constructions of the present work if one simply weakens the assumption on the expansions on $\gamma$ and $\delta$ from \eqref{eq:intro:Bondiexp1} to \eqref{eq:intro:Bondiexp2}. That some weakening of this kind should be done was already suggested in \cite{SeriesVIII}, Page 110.
See also \cite{SeriesXIV,Cap21,marcgeiller} for references where such generalisations are studied in the physics literature. (Note, for instance, that in \cite{marcgeiller}, the quantity $b$ in \eqref{eq:intro:Bondiexp1} is not assumed to vanish, but $b'$ in \eqref{eq:intro:Bondiexp2} is). 
\subsubsection{Modifications to asymptotic simplicity}
In the case of asymptotic simplicity, the modification is less straight-forward.
Of course, one can sufficiently weaken the requirement on the regularity with which a conformal boundary can be attached. The results of \cite{HV20}, for instance, suggest that one can attach a boundary with $C^{1,\alpha}$-regularity for any $\alpha\in(0,1)$.
It is entirely unclear to us, however, what the \textit{optimal} regularity is with which a conformal boundary can be attached given the decay rates of Theorem~\ref{thm:intro:main}.
\subsubsection{Asymptotic flatness in the context of the IVP: The no incoming radiation condition}\label{sec:intro:noincoming}
As we have already briefly mentioned in \S\ref{sec:intro:motivation}, it is not so straight-forward to come up with a modification of asymptotic flatness at the Riemannian level that captures the physics described in this paper. 
Indeed, as point \hyperlink{thm2}{(II)} of Thm.~\ref{thm:intro:main} shows, the decay along $t=0$ is such that $\al\sim r^{-3}$ (which loosely corresponds to $\bar{g}-\delta \sim r^{-1}$). This decay is \textit{slower} than what is assumed in most stability works starting from Cauchy data \cite{CK93,KlNi03, LR10, HV20}.  
At the same time, point \hyperlink{thm1}{(I)} of Thm.~\ref{thm:intro:main} shows that $\al\sim r^{-4}$ towards $\Scrip$, which is \textit{faster} than what is shown in most stability works starting from Cauchy data. 
For instance, in the monumental \cite{CK93},\footnote{We will occasionally use bold font to clearly distinguish quantities that concern solutions to the full, nonlinear equations, from the corresponding linearised quantities.} $\bm{\alpha}$ decays like $\bm{r}^{-7/2}$ near $i^0$, and this decay is propagated towards $\Scrip$, i.e.~$\bm{\alpha}$ decays like $\bm{r}^{-7/2}$ towards $\Scrip$, cf.~Fig.~\ref{fig:intro:noincoming}. 
Furthermore, this decay can be shown to be sharp (using, for instance, the methods of this paper). See also \cite{KlNi03b}, where it is shown that if $\bm{\alpha}$ decays faster on $t=0$, then it will decay faster towards $\Scrip$.

The improvement in decay towards $\Scrip$ compared to $i^0$ as showcased in Thm.~\ref{thm:intro:main} is thus something special. In fact, it is precisely related to the no incoming radiation condition. Thus, the appropriate correction to the usual notions of asymptotic flatness at the level of Cauchy data would, on the one hand, involve slower decay than usually assumed, and, on the other hand, involve some extra structure capturing the no incoming radiation condition, leading to faster decay towards $\Scrip$.

We will investigate this further in future work, but we want to stress the slight \textit{shift in paradigm, namely that the scattering problem, as opposed to the Cauchy problem, is the more natural problem for understanding problems concerning spacelike infinity}.
\begin{figure}
\centering
\begin{subfigure}{0.45\textwidth}
    \includegraphics[width=140pt]{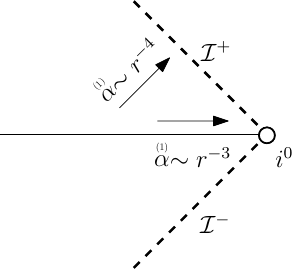}
    \caption{No incoming radiation, improved decay towards $\Scrip$.}
    \label{fig:a}
\end{subfigure}
\hspace{0.05\textwidth}
\begin{subfigure}{0.45\textwidth}
    \includegraphics[width=140pt]{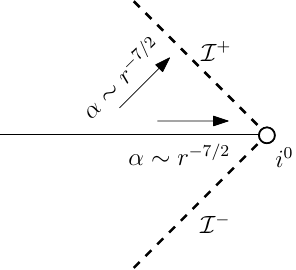}
    \caption{Generic data, no improvement towards $\Scrip$.}
    \label{fig:b}
\end{subfigure}
\caption{In Fig.~\ref{fig:a}, we depicted the results of Thm.\ref{thm:intro:main}, which state that the decay of $\overone{\alpha}$ towards $\Scrip$ is one power better than the decay towards $i^0$. This is in contrast to the generic situation, where you can only propagate the decay from $i^0$ as depicted in Fig.~\ref{fig:b}.}
\label{fig:intro:noincoming}
\end{figure}

Let us finally mention that points \hyperlink{thm3}{(III)} and \hyperlink{thm6}{(VI)} of Theorem~\ref{thm:intro:main} provide clear asymptotics for \textit{characteristic} initial value problems set up near future null infinity, cf.~Holzegel's definition of asymptotically extendable data \cite{Holz16}.

\subsubsection{A word on Christodoulou's argument}\label{sec:intro:CHr}
The starting point for this series was an attempt (cf.~\S1.2 of \cite{I}) to describe a very short sketch of an argument against smooth conformal compactification due to Christodoulou~\cite{Chr02} in the \textit{nonlinear theory}, which suggests that a contradiction to a smooth null infinity can already be reached by just making an (a priori) assumption on the limit of $\bm{r}\hat{\underline{\bm{\chi}}}$ towards $\Scrip$ (denoted $\bm{\Xi}$), namely that it decays like $|\bm{u}|^{-2}$ as $\bm{u}\to-\infty$ (as predicted by the Post-Newtonian theory for a system of $N$ infalling masses following approximately hyperbolic orbits in the infinite past), and by assuming that $\bm{r}\hat{\bm{\chi}}$ vanishes towards $\Scrim$, i.e.~by prohibiting incoming radiation from $\Scrim$.

Now, the scattering theory approach taken in the present work is of course perfectly tailored to \textit{dynamically} capturing the radiation emitted by $N$ infalling masses from $i^-$ as well as the assumption of no incoming radiation from $\Scrim$. 
Moreover, at least within the linear theory around Schwarzschild, it provides a complete understanding of the asymptotics around spacelike infinity, which, in particular, \textit{definitively confirms} the prediction from \cite{Chr02} that~$\be$ decays like $r^{-4}\log r$ (cf.~\eqref{eq:intro:belimit}) and that $\al$ decays like $r^{-4}$ towards $\Scrip$ (cf.~\eqref{eq:intro:al_asy}), and which recovers the prediction that the limit of $r\xhb$ decays like $|u|^{-2}$ (cf.~\eqref{eq:intro:thm:limitofalphabar} and \eqref{eq:intro:xhblimit}). 
In addition, as we said in the previous section, the decay rate towards $i^0$ is actually slower than what is assumed in \cite{CK93}.

In order to clear up some confusion, we nevertheless want to point out that the accounts given in \cite{Daf-Bourbaki} (Thm.~1.4 therein) or in \cite{I} (Section 1.2 therein), which attempt to capture Christodoulou's argument, are incorrect without additional assumptions, at least in the linearised theory around Schwarzschild as considered in the present work.
This can be seen, for instance, by considering \eqref{eq:intro:belimit} and \eqref{eq:intro:xhblimit} combined with \eqref{eq:intro:thm:limitofalphabar} of Thm.~\ref{thm:intro:main}.
Together, these three equations imply that the limits $\lim_{u\to-\infty}|u|^2\lim_{v\to\infty}r\xhb $ and $\lim_{v\to\infty}\frac{r^4}{\log r}\be$ are independent, contradicting eq.~(19) of \cite{Daf-Bourbaki}, which states that these limits are related by a differential operator on $\Stwo$.\footnote{We recall that (19) in Thm.~1.4 \cite{Daf-Bourbaki} simply combines (4), (8), and the three equations below (11) in \cite{Chr02}, and is identical to (1.9) in \cite{I}.}
Indeed, whether or not $\mathscr{A}$ in Thm.~\ref{thm:intro:main} vanishes, the decay rate of $r\xhb$ along $\Scrip$ is $|u|^{-2}$ if $\underline{\mathscr{B}}\neq 0$. 
But if $\mathscr{A}$ vanishes, then $\be=\O(r^{-4})$ according to \eqref{eq:intro:belimit}. 
At the same time, the no incoming radiation condition is, of course, completely independent of $\mathscr{A}$.

Another way of putting this is as follows: There exist solutions to linearised gravity around Schwarzschild  that have
\begin{itemize}
\item No incoming radiation from $\Scrim$.
\item The limit $\lim_{u\to-\infty}u^2\lim_{v\to\infty}r^2\xhb$ is non-vanishing.
\item The decay of the solution restricted to $t=0$ is consistent with the decay assumed in \cite{CK93}.
\item $r^4\be$ takes a finite limit towards $\Scrip$.
\end{itemize}
But in both \cite{Daf-Bourbaki} and \cite{I}, it is stated that if the first three bullet points are satisfied, then $\be=\O(r^{-4}\log r)$.
Thus, Theorem~\ref{thm:intro:main} with $\mathscr{A}=0\neq \underline{\mathscr{B}}$ would form a counter-example to both Theorem~1.4 of \cite{Daf-Bourbaki} and to the argument presented in \S1.2 of \cite{I}.  (The same is true for the account given in \S7.1 of \cite{MTWang}.)
Again, there is the caveat that Thm.~\ref{thm:intro:main} concerns linearised gravity, whereas the references above all concern the nonlinear theory! 
We expect, however, that nonlinear effects will not alter our conclusion, e.g.~because the relation~(19) in \cite{Daf-Bourbaki} is linear. 

Notice that the counter-example above still has a non-smooth future null infinity, because, in view of \eqref{eq:intro:al_asy}, $\frac{r^5}{\log^2 r}\al$ attains a finite limit towards $\Scrip$.
In particular, this counter-example is \textit{not} a counter-example to the idea that the two assumptions on $\Scrim$ (no incoming radiation) and on $\Scrip$ (namely that $\bm{\Xi}$ decays as in the Post-Newtonian prediction) are sufficient to preclude a smooth $\mathcal{I}$.  
We here do not give any speculation on the mathematical realisation of this idea; see, however, Thm.~2.2 of \cite{I}.

\subsubsection{The modified Newman--Penrose charges}\label{sec:intro:NPcharges}
A particular discovery that was made within the Newman--Penrose framework was a set of conserved quantities along null infinity nowadays called the Newman--Penrose charges \cite{NPconstants65,NPconstants68}, the physical interpretation of which having occupied researchers' minds ever since.

Now, the definition of the original N--P charges requires asymptotic simplicity; moreover, it is believed that, in the full, nonlinear theory, only a finite amount of them are conserved (see also~\cite{Long22}).
It turns out, however, that modified versions of the Newman--Penrose charges can still be defined in a much more general context, see also \cite{Kroon00, GK}, and, in the context of linearised gravity around Schwarzschild, \S\ref{sec:cons:NP} of the present paper.
Indeed, the reason why these are conserved (namely the existence of a set of approximate conservation laws that asymptotically become exact conservation laws) plays a crucial role in the proof of Thm.~\ref{thm:intro:main}.
The modified Newman--Penrose charges corresponding to the solution of Thm.~\ref{thm:intro:main} are given in Theorems~\ref{thm:alp:0} and~\ref{thm:alb:0}, respectively, and it turns out that they are equivalent to the conserved limits in \eqref{eq:intro:thm:limitsconserved}.

We here want to highlight three features of these modified Newman--Penrose charges:
Firstly, concerning their physical interpretation: In our context, even though they are conserved along $\Scrip$, the modified N--P charges are still dynamically obtained quantities to which we can ascribe a clear interpretation in terms of the initial seed data, with the modified N--P charge of angular mode $\ell$ loosely being a product of monopole (in this case the Schwarzschild mass $M$) and $\ell$-th multipole of the data.

Secondly, in the linear theory, the modified N--P charges entirely fix the leading-order asymptotics at late times, as was described in detail in \cite{GK}.

Thirdly, we strongly expect all (infinitely many) modified N--P charges to still be conserved even in the full, non-linear theory. Loosely speaking, this is because of the high degree of conformal irregularity that they measure, which we expect to be robustly propagated also in the nonlinear theory.
The robustness of the \textit{modified} N--P charges in the case of low conformal regularity can be showcased already at the linear level: Under the assumption of a smooth null infinity, the associated \textit{classical} Newman--Penrose charges for linearised gravity around Schwarzschild need to be substantially altered compared to the Minkowski case in order to still be conserved. On the other hand, if null infinity is suitably irregular, then the modified Newman--Penrose charges wash away the difference between the Minkowski and the Schwarzschild case.

\subsection{A conjecture on global asymptotics}\label{sec:intro:conjecture}
Based on Theorem~\ref{thm:intro:main2}, the discussion in \S\ref{sec:intro:NPcharges} and the general insights of \cite{GK}, we formulate the following conjecture:
\begin{figure}[htbp]
\centering
\begin{subfigure}{0.4\textwidth}
    \includegraphics[width=140pt]{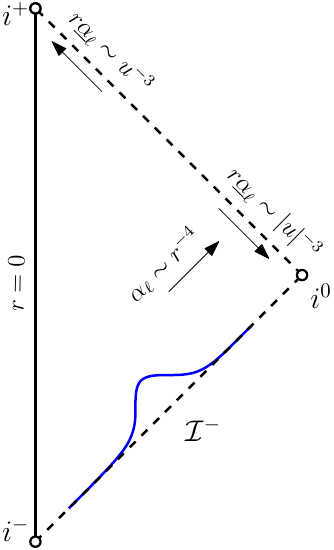}
    \caption{The dispersive case.}
    \label{fig:dispersive}
\end{subfigure}
\hspace{0.05\textwidth}
\begin{subfigure}{0.4\textwidth}
    \includegraphics[width=140pt]{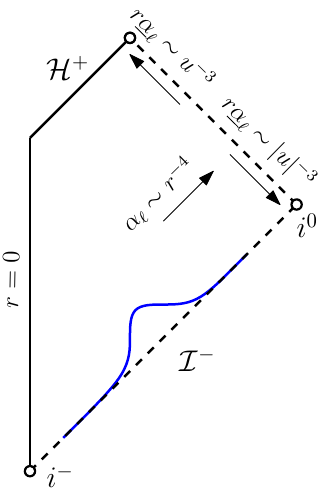}
    \caption{The collapsing case.}
    \label{fig:collapse}
\end{subfigure}
        
\caption{Penrose diagrams corresponding to Conjecture~\ref{Conjecture}}
\label{fig:conjecture}
\end{figure}
\begin{conjecture}\label{Conjecture}
    Consider the scattering problem for the \textbf{nonlinear Einstein vacuum equations \eqref{eq:intro:EVE}} with a Minkowskian $i^-$ and with sufficiently small, compactly supported data along~$\Scrim$.  
    Then the resulting spacetime will exhibit asymptotics as depicted in Fig.~\ref{fig:conjecture}, namely:
    The component of the Weyl curvature $\alpha$ decays like $r^{-4}$ towards $\Scrip$, and the limits $\lim_{v\to\infty} r^4\alpha_{\ell}$ are conserved along $\Scrip$ for all $\ell$.
    Furthermore, each angular mode of the limit of $r\underline{\alpha}$ towards~$\Scrip$ decays like $u^{-3}$ both towards $u=-\infty$ and towards $u=+\infty$.

    We further conjecture that the same rates are valid in the case of black hole formation such as in \cite{Chr09}.
\end{conjecture}

In order to give context for the rates conjectured above, we want to contrast them to the asymptotic behaviour in a few other situations:
\begin{enumerate}
    \item[1.)] \textbf{Asymptotics near $i^0$ for compactly supported data along $\Scrim$ in Minkowski.} If we consider \textit{linearised gravity} with $M=0$, then, as we have already seen in Observation~\ref{Observation}, compactly supported data along $\Scrim$ will lead to a picture where the limit of $r\alb_{\ell}$ towards $\Scrip$ decays like $|u|^{-3-\ell}$ as $u\to-\infty$, and where $r^5\al_{\ell}$ attains a finite limit towards $\Scrip$. Cf.~Fig.~\ref{fig:M=0}.
    \item[2.)] \textbf{Compactly supported data along $t=0$ in Schwarzschild.} If, on the other hand, we consider compactly supported data along a Cauchy surface for \textit{linearised gravity around Schwarzschild}, then the late-time decay behaviour will be, by the results of \cite{MZ21}\footnote{The work \cite{MZ21} only deals with the Teukolsky equations on Schwarzschild, not the full system of linearised gravity.}, $\lim_{v\to\infty}r\alb_{\ell}\sim M |u|^{-4-\ell}$ as $u\to\infty$. (If $M=0$, then the late-time asymptotics are trivial by the strong Huygens' principle.) Cf.~Fig.~\ref{fig:t=0}.
    \item[3.)] \textbf{Compactly supported Cauchy data in the nonlinear theory.} There is, as of now, no rigorous work on the late-time asymptotics for solutions to the nonlinear Einstein vacuum equations \eqref{eq:intro:EVE} arising from compactly supported perturbations of, say, Schwarzschild initial data. However, heuristics \cite{Bizon10,Luktalk21} suggest that these decay rates may be one power slower than in the linear theory. If this is true, then they would still be two powers faster than what is conjectured in Conjecture~\ref{Conjecture} for compactly supported data along $\Scrim$. Cf.~Fig.~\ref{fig:t=0}.
\end{enumerate}
\begin{figure}
\centering
\begin{subfigure}{0.35\textwidth}
    \includegraphics[width=140pt]{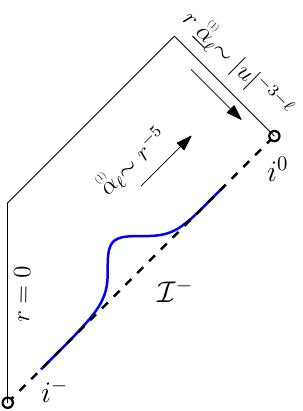}
    \caption{Linearised gravity for $M=0$.}
    \label{fig:M=0}
\end{subfigure}
\hspace{0.05\textwidth}
\begin{subfigure}{0.5\textwidth}
    \includegraphics[width=200pt]{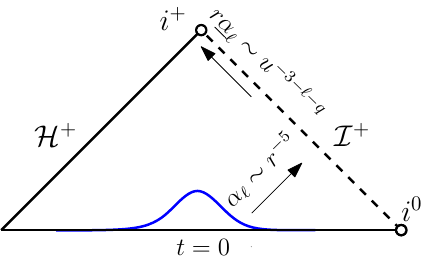}
    \caption{Compactly supported Cauchy data.}
    \label{fig:t=0}
\end{subfigure}
\caption{Fig.~\ref{fig:M=0} depicts the early-time asymptotics for linearised gravity with $M=0$, i.e.~on Minkowski, cf.~Observation~\ref{Observation}. 
Fig.~\ref{fig:t=0} depicts the conjectured late-time asymptotics for compactly supported Cauchy data either for the Einstein vacuum equations (in which case $q=0$), or the late-time asymptotics for linearised gravity around Schwarzschild proved in \cite{MZ21} (in which case $q=1$).
}

\label{fig:context}
\end{figure}

\subsection{Overview of the paper}\label{sec:intro:structure}
We finally give an overview of the paper's structure.
We begin with some preliminaries to make the paper self-contained.
First, in \S\ref{sec:SS}, we introduce the Schwarzschild family of spacetimes and give the geometric foundations for the remainder  of the paper. This section is largely informed by~\cite{Chr09,DHR16,Czimek17}.

In \S\ref{sec:lin} (corresponding to \S3(a) of \cite{IV}), we then introduce the system of linearised gravity around Schwarzschild, following \cite{DHR16,Masaood22b}. As this system is written down using the Christodoulou--Klainerman formalism, and many are more familiar with the Newman--Penrose formalism, we provide a detailed dictionary between the two  formalisms in Appendix~\ref{app:dictionaryCKNP}.

In \S\ref{sec:gauge} (corresponding to \S3(b) of \cite{IV}), we recall the class of pure gauge and linearised Kerr solutions identified in \cite{DHR16}.

Beyond these preliminaries, the paper is divided into three parts, the first one corresponding to the results introduced in \S\ref{sec:intro:results1}, and the last two corresponding to the results introduced in \S\ref{sec:intro:results2}.

\textbf{In Part~\hyperlink{V:partI}{I},} comprised of \S\ref{sec:setup}--\S\ref{sec:physd}, we set up and solve the scattering problem for the system of linearised gravity around Schwarzschild. 

In \S\ref{sec:setup}, we define a notion of seed scattering data; we write down the main theorem (corresponding to Thm.~\ref{thm:intro:LEE Scattering wp}), give an overview of its proof, and we provide an \textit{a priori} analysis of the corresponding solution along $\Cin\cup\Scrim$. 
We also give an account of how to remove the unphysical degrees of freedom from a seed scattering data set (cf.~Def.~\ref{defi:intro:initialgauge}).

In \S\ref{sec:RWscat}, we then develop a \textit{semi-global} scattering theory for the Regge--Wheeler equation~\eqref{eq:intro:RW}. The main result of \S\ref{sec:RWscat} is a statement adapted from \cite{Masaood22, Masaood22b}, but we give an alternative proof and prove a few extra results concerning uniqueness and decay that will prove useful for our schemes.

In \S\ref{sec:construct}, we then infer from this scattering theory for the Regge--Wheeler equation the existence of a unique scattering solution to the entire system of linearised gravity around Schwarzschild.

Finally, in \S\ref{sec:physd}, we give the explicit definition of the no incoming radiation condition (which is gauge-independent and somewhat surprising), and we make various definitions for what it means for a seed scattering data set to describe the exterior of $N$ incoming bodies from the infinite past, cf.~Def.~\ref{defi:intro:N}. We then provide a preliminary analysis of solutions arising from such data. 

To facilitate the comparison with \cite{IV}, we state that \S\ref{sec:setup} corresponds to \S3(c)--\S3(d) of \cite{IV}, \S\ref{sec:RWscat} corresponds to \S4(b), \S\ref{sec:construct} to \S4(c), and \S\ref{sec:physd} to  \S3(e) and \S5(a) of \cite{IV}.

\textbf{In Part~\hyperlink{V:partII}{II},} comprised of \S\ref{sec:cons}--\S\ref{sec:sum}, we find the asymptotics for solutions to the Teukolsky and Regge--Wheeler equations of general spin. This part of the paper can be read independently of the previous part, and it corresponds to \S5(b)--\S5(e) of \cite{IV}.

In \S\ref{sec:cons}, we derive a class of approximate conservation laws enjoyed by fixed angular frequency solutions to the Teukolsky (or Regge--Wheeler) equations, and we define the associated modifed Newman--Penrose constants.

In \S\ref{sec:al}, we then perform the asymptotic analysis for general solutions to the Teukolsky equations of spin $s$ ($s=2$ corresponds to $\al$, $s=-2$ to $\alb$) with no incoming radiation, making, however, a simplifying additional assumption in the case $s<0$. 

In \S\ref{sec:alp}, we then apply these results to the physically motivated data of \S\ref{sec:physd} for $\al$, and then, in \S\ref{sec:Psi}, for the Regge--Wheeler quantities.

Similarly, in \S\ref{sec:alb}, we apply the results of \S\ref{sec:al} to $\alb$, also removing the simplifying assumption made in \S\ref{sec:al}.

Section \S\ref{sec:PM} contains various explicit computations whose results have been used in the previous sections.

Finally, in \S\ref{sec:sum}, we comment on (but don't resolve) the issue of summing the asymptotics obtained for fixed angular frequencies.

\textbf{In Part~\hyperlink{V:partIII}{III},} corresponding to \S6 of \cite{IV}, we finally take the results from the previous part on the asymptotic behaviour of solutions to Teukolsky and Regge--Wheeler equations in order to compute the semi-global asymptotics for the entire system of linearised gravity around Schwarzschild, restricting to solutions arising from scattering data describing the exterior of hyperbolic orbits or compactly supported gravitational perturbations (cf.~Def.~\ref{defi:intro:N} or Defs.~\ref{defi:physd:Nbodyseed},~\ref{defi:physd:graviton}). 

Comparisons of the formalism used in the present work to the  Newman--Penrose formalism, as well relations to spin-weighted functions, operators and spherical harmonics are found in Appendix \ref{app:SpinSphereNPCK}. Various computations of integrals used throughout the work are found in Appendix~\ref{app:integrals}.

\newpage

\section{The Schwarzschild family of spacetimes}\label{sec:SS}
\subsection{The geometry of the Schwarzschild solution}
 We define, for $M\in\mathbb R_{\geq 0}$, the Schwarzschild family of spacetimes as the family of smooth Lorentzian manifolds $(\Schw , g_M)$, where $\Schw =\{(t,r,\theta,\varphi)\in\mathbb R\times (2M,\infty) \times \Stwo\}$, and where the metric $g_M$ in coordinates $(t,r,\theta,\varphi)$ is given by 
\begin{align}
    g_M=-\left(1-\frac{2M}{r}\right)\dd t^2+\left(1-\frac{2M}{r}\right)^{-1}\dd r^2+\slashed{g},
\end{align}
where $\slashed{g}$ is simply the rescaled metric on the unit sphere:
\begin{equation}
\slashed{g}=r^2\mathring{\slashed{g}}=r^2(\dd \theta^2+\sin^2\theta\dd \varphi^2)
\end{equation}

Next, we introduce the Eddington--Finkelstein double null coordinates 
\begin{align}\label{eq:SS:definition of u and v}
    u=\frac{1}{2}(t-r_*),\qquad v=\frac{1}{2}(t+r_*),
\end{align}
where $r_*$ is the familiar radial function defined by
\begin{align}
    r_*=r+2M\log\left(\frac{r}{2M}-1\right).
\end{align}
In this double null coordinate system $(u,v,\theta,\varphi)$, the metric takes the form
\begin{align}\label{eq:metricEF}
g_M:=-4\Omega^2\dd u \dd v+r^2 \mathring{\slashed{g}},&& \Omega=\sqrt{1-\frac{2M}{r}}.
\end{align}

We denote by $\C_\lambda$ the outgoing null hypersurface $\{u=\lambda\}$, while $\Cbar_{\nu}$ will refer to the ingoing null hypersurface $\{v=\nu\}$. The sphere formed by the intersection of $\C_\lambda\cap\Cbar_\nu=\{(\lambda,\nu)\times\Stwo\}$ will be denoted by $\mathcal{S}_{\lambda,\nu}$. 

The Schwarzschild manifold can be extended to a manifold with boundary in various directions.
Using the coordinate system $(v,r,\theta,\varphi)$, we can attach the boundary $  \mathcal{H}^+:=\{(v,r,\theta,\varphi): r=2M\}$, called the future event horizon, thus extending the manifold to $\mathbb{R}\times[2M,\infty)\times \Stwo$. 
We can similarly attach $\mathcal H^-$, the past event horizon, by working in coordinates $(u,r,\theta,\varphi)$.

In a similar fashion, we can attach boundaries at infinity (without viewing them as conformal boundaries). 
Working in the coordinate system $(x,u,\theta,\varphi)$, where $ x(u,v):=\frac{1}{r(u,v)}$,
we attach the hypersurface $\Scrip:=\{x_v,u,\theta,\varphi: x_v=0\}$, called future null infinity.
We can similarly attach~$\Scrim$, called past null infinity, by working in coordinates $(x_u,v,\theta,\varphi)$.
See~\cite{Masaood22,DRSR18} for an extended discussion, in particular Def.~4.2.1 in the latter.

A note on conventions: 
Even though we have now introduced several different coordinate systems; whenever we write $\pu$ or $\pv$, it will always be with respect to the double null coordinate system $(u,v,\theta,\varphi)$.
Furthermore, we will allow ourselves to use the coordinates $u,v$ to  also refer to objects living in the extended manifold $\SchwEx :=\Schw\cup\mathcal H^-\cup \mathcal H^+\cup\mathcal I^-\cup \mathcal I^+$ by permitting $u,v$ to attain the values $\pm\infty$.
Finally, when we talk about limits "as $u\to-\infty$" or "as $v\to\infty$" etc., it will always be understood that the respective other double null coordinates are kept fixed.

\subsection{The double null gauge}
The metric in the form \eqref{eq:metricEF} is an example of spacetime metrics $\bm{g}$ that are cast in a \textit{double null gauge}, i.e.~a coordinate system $(\bm{u},\bm{v},\bm{\theta}^a)$ where $\bm{u},\bm{v}$ are null coordinates with respect to  $\bm{g}$ and $(\bm{\theta}^a)$ refers to an atlas on the topological spheres $\bm{\mathcal{S}}_{\bm{u},\bm{v}}$ that are formed by the intersections of the loci of $\bm{u},\bm{v}$. 
In such a coordinate system, the metric $\bm{g}$ takes the form
\begin{align}\label{eq:SS:DNGmetric}
    \bm{g}=-4\bm{\Omega}^2\dd\bm{u}\dd\bm{v}+\bm{\slashed{g}}_{ab}(\dd\bm{\theta}^a-\bm{b}^a\dd\bm{v})(\dd\bm{\theta}^b-\bm{b}^b\dd\bm{v}).
\end{align}
In the above, for any $(\bm{u},\bm{v})$, $\bm{\Omega}(\bm{u},\bm{v},\bm{\theta}^a)$ is a scalar field on $\bm{\mathcal{S}}_{\bm{u},\bm{v}}$, $\bm{b}(\bm{u},\bm{v},\bm{\theta}^a)$ is a vector field tangent to $\bm{\mathcal{S}}_{\bm{u},\bm{v}}$ and $\bm{\slashed{g}}(\bm{u},\bm{v},\bm{\theta}^a)$ is a positive definite bilinear form on $\bm{\mathcal{S}}_{\bm{u},\bm{v}}$ which has vanishing projection in both null directions. 
Note that $\bm{\slashed{g}}$ is the metric induced by $\bm{g}$ on $\bm{\mathcal{S}}_{\bm{u},\bm{v}}$. 

The metric \eqref{eq:SS:DNGmetric} gives rise to an orthonormal frame $(\bm{e}_3,\bm{e}_4,\bm{e}_A)$ via
\begin{align}\label{eq:SS:framegeneral}
    \bm{e}_3=d\bm{u}^\sharp=\bm{\Omega}^{-1}\partial_{\bm{u}},\qquad \bm{e}_4=d\bm{v}^\sharp=\bm{\Omega}^{-1}\partial_{\bm{v}}-\bm{b}^A\bm{e}_A,
\end{align}
and $\bm{e}_A$ defined such that $(\bm{e}_A|_{\mathcal{S}_{\bm{u},\bm{v}}})$ is an orthonormal frame on $\mathcal{S}_{\bm{u},\bm{v}}$. 
For example, in the double null gauge defined by the Eddington--Finkelstein coordinates on the Schwarzschild exterior, the orthonormal frame derived from the metric \eqref{eq:metricEF} is given by
\begin{equation}\label{ONF}
\left(e_1=\frac{1}{r}\partial_\theta,\, e_2=\frac{1}{r\sin\theta}\partial_\varphi,\, e_3=\frac{1}{\Omega}\pu,\, e_4=\frac{1}\Omega\pv\right).
\end{equation}
An important convention in this paper is that uppercase Latin indices $A,B,C\dots$ will always refer to the indices on the spheres.

\subsection{\texorpdfstring{$\mathcal{S}_{u,v}$}{S}-tangent tensors}
When working in a double null gauge, it is convenient to decompose the metric, connection and curvature components into tensor fields that are everywhere tangent to the spheres $\bm{\mathcal{S}}_{\bm{u},\bm{v}}$. 
A beautiful and thorough presentation of this in generality is given in Chapter 1.2 of \cite{Chr09}. 

Here, we restrict our attention to the Schwarzschild exterior in Eddington--Finkelstein coordinates, where the spheres $\mathcal{S}_{u,v}$ of the Eddington--Finkelstein foliation are round spheres with area-radius function given by  $r=r(u,v)$ and metric given by $\slashed{g}=r^2\mathring{\slashed{g}}$, $\mathring{\slashed{g}}$ being the standard metric of the unit sphere.
A one-form $X$ is then said to be $\mathcal{S}_{u,v}$-tangent if $X(e_3)=0=X(e_4)$ at each point, and a vector field $V$ is said to be $\mathcal S_{u,v}$-tangent if it is tangent to~$\mathcal S_{u,v}$ for all $u,v$.
These definitions can easily be generalised to include covariant, contravariant and mixed-type $\mathcal S_{u,v}$-tensor fields, see~\cite{Chr09}.
Moreover, it is easy to show that the musical isomorphism (i.e.~the raising and lowering of indices via the metric $g$ and its inverse $g^{-1}$) preserves the property of a tensor field to be $\mathcal{S}_{u,v}$-tangent.
We denote the bundle of $\mathcal{S}_{u,v}$-tangent $(p,q)$-tensors on $\Schw$ by $\pqbundletensors{p}{q}{\Schw}$, and similarly for any spherically symmetric subset of $\Schw$.
In what follows, for any tensor bundle $\mathcal{T}$ over $\Schw$, we denote by $\Gamma^\infty(\mathcal{T})$ the space of smooth sections over $\mathcal{T}$.

A special role will be played by the tensor bundle $\bundlestfs{\Schw}$, which consists of all elements of $\pqbundletensors{0}{2}{\Schw}$ that are symmetric and trace-free with respect to $\slashed{g}$.
We will also just call these tensors stf 2-tensors and, given $\alpha\in\Gamma^\infty( \pqbundletensors{0}{2}{\Schw})$, we will denote
\begin{equation}
(\alpha_{stf})_{AB}=\tfrac12(\alpha_{AB}+{\alpha_{BA}})-\slashed g_{AB}\slashed{g}^{CD} \alpha_{CD}=\tfrac12(\alpha_{AB}+{\alpha_{BA}})-\slashed g_{AB}\tr_{\slashed{g}}\alpha.
\end{equation}

Let $\nabla$ be the Levi--Civita connection of the Schwarzschild spacetime $(\mathcal{M}_M,g)$. 
We denote by~$\sll$ the covariant derivative with respect to the metric $\slashed{g}$ on the spheres $\mathcal S_{u,v}$. 
Note that, at each point $(u,v)$, this is related to the Levi--Civita connection $\sl$ on $\Stwo$ simply via $\sl=r(u,v)\cdot \sll$.

For any tensor field $\digamma\in\Gamma^{\infty}(\pqbundletensors{p}{q}{\mathcal{M}_M})$, we denote by $\sll_3\digamma$ the component of $\nabla_3\digamma$ that is everywhere tangent to $\mathcal S_{u,v}$. An analogous definition gives $\sll_4$. 
We will very often work with the weighted covariant derivatives
\begin{align}
\Du:=\Omega\su=\sll_{\pu},&&\Dv:=\Omega\sv=\sll_{\pv}.
\end{align}
Note the relations
\begin{align}
    [\Du,\Dv]=0,&&[\Du,r\sll_A]=0=[\Dv,r\sll_A].
\end{align}

In the computations of this paper, the following simple fact will be used a lot:
 For any $\mathcal{S}_{u,v}$-tangent $(0,q)$-tensor $\Xi$, we have
\begin{nalign}\label{eq:SS:puvsDu}
(\Du \Xi)(e_{A_1},\dots,e_{A_q})=\pu(\Xi(e_{A_1},\dots,e_{A_q}))\\
(\Dv \Xi)(e_{A_1},\dots,e_{A_q})=\pv(\Xi(e_{A_1},\dots,e_{A_q}))
\end{nalign} 
In words, evaluating the transversal covariant derivative $\Du$ or $\Dv$ of an $\mathcal S_{u,v}$-tensor in the orthonormal frame $(e_A,e_B)$ is the same as computing the partial derivative of the components of the tensor with respect to that frame.
In the following, whenever we write down expressions such as $\int_{v_1}^{v_2}\Dv\Xi\dd v$ or similar, we will mean this component-wise. 
\subsection{Norms and the \texorpdfstring{$\O$}{big O}-notation}\label{sec:SS:O}
Given $\mathcal{S}_{u,v}$-tangent $(p,q)$-tensor fields $\digamma$, $\tilde{\digamma}$ on $\Schw$, we define the inner product
\begin{equation}
    \digamma_1\cdot\digamma_2:=(\slashed{g}^{-1})^{A_1 B_1}\cdots(\slashed{g}^{-1})^{A_p B_p}\slashed{g}_{C_1 D_1}\cdots \slashed{g}_{C_q D_q}\digamma_{A_1,\dots,A_p}^{C_1,\dots,C_p}\tilde{\digamma}_{B_1,\dots,B_p}^{D_1,\dots,D_p},
\end{equation}
and we define the pointwise norm via 
\begin{equation}
    |\digamma|^2:=\digamma\cdot\digamma.
\end{equation}
For $f\in\Gamma^\infty(\Schw)$, we then write
\begin{equation}
    \digamma =\O(f)
\end{equation}
if there exists a positive uniform constant $C$ such that $|\digamma|\leq C\cdot f$.
Moreover, we write 
\begin{equation}
    \digamma=\O^u(f)
\end{equation}
if $\digamma=\O(f)$ and $\Dv\digamma=0$ in $(u,v,\theta,\varphi)$-coordinates.

Speaking of constants, $C$ will generally be a constant that is allowed to change from line to line, and instead of saying that $f\leq C\cdot g$, we will often just write $f\lesssim g$.

We will also use the $\O$-notation restricted to subsets of $\Schw$.
We will then also employ the $\O_{\infty}$-notation: 
If $\digamma$ is an $\mathcal{S}_{u,v}$-tangent tensor field on $\Cbar_{v}\cap\{u\leq -1\}$, then
\begin{equation}
    \digamma=\O_{\infty}(f)
\end{equation}
will mean that for all $m,\,n\in\mathbb N$ there exists uniform positive constants $C_{m,n}$ s.t.
\begin{equation}\label{eq:Onotationbs}
    |\Du^m\sl^n\digamma|\leq C_{m,n}|\pu^m f|.
\end{equation}
(In particular, with this convention, $\digamma=\O_{\infty}(1)$ will mean that all $\Du$-derivatives of $\digamma$ vanish identically.)
If we only write $\digamma=\O_{\tilde{m}}({f})$, it will mean that \eqref{eq:Onotationbs} only holds for all $m\leq\tilde m$, $n\in\mathbb N$.

We define the following $L^2$-norm (note that this has no $r$-weight!):
\begin{align}
    \|\digamma(u,v,\cdot)\|^2_{L^2(\mathcal{S}_{u,v})}:=\int_{\mathcal{S}_{u,v}}|\digamma|^2\dw.
\end{align}
 We similarly define
\begin{align}
     \|\digamma(u,\cdot)\|^2_{L^2(\C_u\cap\{v_1\leq v\leq v_2\})}:=\int_{v_1}^{v_2}\int_{\Stwo}|\digamma|^2(u,v,\theta,\varphi)\dw\dd v,\\
     \|\digamma(\cdot,v,\cdot)\|^2_{L^2(\Cbar_v\cap\{u_1\leq u\leq u_2\})}:=\int_{u_1}^{u_2}\int_{\Stwo}|\digamma|^2(u,v,\theta,\varphi)\dw\dd v.
\end{align}
We will often just write $\Stwo$ rather than $\mathcal S_{u,v}$, and $[v_1,v_2]\times\Stwo$ instead of $\C_u\cap\{v_1\leq v\leq v_2\}$. We finally define the Sobolev norm for $H^1(\C_u\cap\{v_1\leq v\leq v_2\})$ via $$\|\digamma\|^2_{H^1(\C_u\cap\{v_1\leq v\leq v_2\})}=\|\digamma\|_{L^2(\C_u\cap\{v_1\leq v\leq v_2\})}+\|\Dv\digamma\|_{L^2(\C_u\cap\{v_1\leq v\leq v_2\})}+\|\sl\digamma\|_{L^2(\C_u\cap\{v_1\leq v\leq v_2\})}$$
and similarly for $\Cbar$ and higher order Sobolev norms.
The corresponding spaces are defined as completions of the space of smooth tensor fields with respect to these norms.

\subsection{Differential operators on the sphere}
We now construct a variety of differential operators from $\sll$.
Since $\sll=r^{-1}\cdot \sl$, we will simply provide the definitions for the corresponding differential operators on the unit-sphere~$\Stwo$; lifting these definitions to differential operators acting on $\mathcal{S}_{u,v}$-tangent tensors over $\Schw$ is straight-forward. (Given a covariant tensor field on $\Stwo$, we canonically identify it with a tensor field on $\mathcal S_{u,v}$ for each $(u,v)$, and we extend it to an $\mathcal S_{u,v}$-tangent tensor field on $\Schw$ by demanding that it vanishes when contracted with $e_3$ or $e_4$.) 

 We denote the volume form associated to the unit-sphere metric $\mathring{\slashed{g}}$ by $\mathring{\slashed{\varepsilon}}$:
\begin{align}
    \mathring{\slashed{\varepsilon}}=\sin\theta\dd\theta\wedge \dd\varphi,
\end{align}
we use ${}^\ast$ to denote the Hodge dual on $(\Stwo,\mathring{\slashed{g}})$, and we denote the unit-sphere Laplacian by $\lap$ (so we have that $\lap=r^2\lapp$).
\begin{defi}
We define for any covariant tensor field $\Xi$ the divergence via 
\begin{equation}
\div \Xi_{A_2\dots A_p}=\mathring{\slashed{g}}^{A_0 A_1}\sl_{A_0}\Xi_{A_1A_2\dots A_p},
\end{equation}
 and the curl via
 \begin{equation}
 \curl \Xi_{A_2\dots A_p}=\mathring{\slashed{\varepsilon}}^{A_0 A_1}\sl_{A_0}\Xi_{A_1A_2\dots A_p}.
 \end{equation}
We then define
\begin{align}
    \D1: \oneforms\to \functions\times \functions, && \beta \mapsto (\div\beta,\curl\beta)\\
    \D2: \stfs \to \oneforms, && \alpha \mapsto \div\beta
\end{align}
We further define the $L^2$-adjoints of these operators:
\begin{align}
    \Ds1:\functions \times \functions\to\oneforms,&& (f,g)\mapsto -\sl f+\sls g\\
    \Ds2:\oneforms\to\stfs,&&\beta\mapsto -(\sl \beta)_{\mathrm{stf}}
\end{align}
We finally define \emph{the magnetic adjoint} of $\D1$ to be
\begin{align}
\overline{\D1}: \oneforms\to \functions\times \functions,&&  \beta \mapsto (\div\beta,-\curl\beta).
\end{align}
\end{defi}
For the relations between these operators and the "$\eth$"-operators that scholars of the Newman--Penrose formalism may be more familiar with, see the Appendix~\ref{app:spinweightedfunctions}.

The following simple lemma will be used very frequently throughout the paper: 
\begin{lemma}\label{lem:SS:commutation}
We have the relationships
\begin{nalign}\label{eq:SS:DD=Laplacian}
-\Ds1\D1=\lap-1,&&-\D1\Ds1&=\lap,\\
-2\Ds2\D2=\lap-2,&&-2\D2\Ds2&=\lap+1.
\end{nalign}
\end{lemma}
\begin{proof}
The results follow from relating the commutator of two covariant derivatives to the Riemann tensor, using that the Riemann tensor in two dimensions is determined by the Gauss curvature, and using that the Gauss curvature of $\Stwo$ is 1.
\end{proof}
\subsection{Definition of tensorial spherical harmonics}
Let $Y_{\ell,m}\in\functions$ be the usual real-valued spherical harmonics with $\ell\in\mathbb N$, $m=-\ell,-\ell+1,\dots,\ell$. 
We now define the corresponding 1-form spherical harmonics and symmetric traceless two-tensor spherical harmonics following \cite{Czimek17}.

For this, we first recall the general fact that any smooth 1-form $\beta$ on $\Stwo$ can be written as 
\begin{equation}\label{eq:SS:oneformrep}
\beta=\Ds1(f,g)
\end{equation}
for a unique pair of smooth functions $(f,g)$ that is supported on $\ell\geq 1$. Moreover, the kernel of the operator $\Ds1$ is spanned by functions supported on $\ell=0$.

Similarly, any smooth symmetric traceless two-tensor $\alpha$ on $\Stwo$ can be uniquely represented as
\begin{equation}\label{eq:SS:stfsrep}
\alpha=\Ds2\Ds1(f,g),
\end{equation}
with $(f,g)$ a pair of smooth functions supported on $\ell\geq 2$. Moreover, the kernel of the operator $\Ds2\Ds1$ is spanned by functions supported on $\ell=0,1$.

We can now already define what it means for a one-form or a symmetric, tracefree two-tensor to be supported on angular frequency $\ell$.
\begin{defi}\label{defi:SS:supportedonl}
We say that a 1-form (or a traceless, symmetric two-tensor) is supported on angular frequency $\ell$ if the functions $(f,g)$ in the unique representation \eqref{eq:SS:oneformrep} (or \eqref{eq:SS:stfsrep}) are supported on angular frequency $\ell$.
By definition, one-forms have no support on $\ell=0$, and stf two-tensors have no support on $\ell=0,1$.
\end{defi}
It will, however, be convenient to also make a choice for an explicit orthonormal decomposition of the space of square-integrable one-forms and stf two-tensors on $\Stwo$. 
\begin{defi}\label{def:SS:sphericalharmonics}
Let $Y_{\ell,m}$ be the usual real-valued spherical harmonics. Then we define
\begin{nalign}
    \YlmE{1}=(\ell(\ell+1))^{-\frac12}\Ds1(Y_{\ell,m},0),&& \YlmH{1}=(\ell(\ell+1))^{-\frac12}\Ds1(0,Y_{\ell,m}),\\
    \YlmE{2}=(\tfrac12(\ell-1)(\ell+2))^{-\frac12}\Ds2\YlmE{1},&& \YlmH{2}=(\tfrac12(\ell-1)(\ell+2))^{-\frac12}\Ds2\YlmH{1}.
\end{nalign}
\end{defi}
The superscripts $\mathrm{E}$ and $\mathrm{H}$, standing for electric and magnetic part, respectively, have been chosen in view of the parity properties of these harmonics.
\begin{defi}\label{def:SS:magneticconjugate}
For any $\alpha=\Ds2\Ds1(f,g)\in \stfs$, we define its \emph{electric part}~$\alpha^{\mathrm{E}}$, its \emph{magnetic part}~$\alpha^{\mathrm{H}}$ as well as its \emph{magnetic conjugate}~$\overline{\alpha}$ via
\begin{align}
\alpha^{\mathrm{E}}:=\Ds2\Ds1(f,0),&&\alpha^{\mathrm{H}}:=\Ds2\Ds1(0,g),&&
\overline{\alpha}:=\Ds2\Ds1(f,-g).
\end{align}
\end{defi}
\begin{rem}
Up to $r$-scaling and factors of $\sqrt{2}$, these spherical harmonics are the same as Thorne's pure spin vector and tensor harmonics \cite{Thorne80}.
We also provide a detailed discussion relating them to the Newman--Penrose spin-weighted spherical harmonics in the Appendix~\ref{app:spinweightedfunctions}.
\end{rem}

It follows from either $[\lap,\Ds1]=\Ds1$ or $[\lap,\Ds2\Ds1]=4\Ds2\Ds1$ that, for $s=1,2$:
\begin{align}
\lap \YlmE{s}=(-\ell(\ell+1)+s^2) \YlmE{s},&&\lap \YlmH{s}=(-\ell(\ell+1)+s^2)) \YlmH{s}.
\end{align}
We note that, in view of Def.~\ref{defi:SS:supportedonl}, the same holds true more generally for 1-forms/traceless symmetric two-tensors supported on angular frequency $\ell$.

We also have the following basic observations, which follow directly from Def.~\ref{def:SS:sphericalharmonics} and Lemma~\ref{lem:SS:commutation}:
\begin{lemma} We have
\begin{nalign}
    \Ds1 (\Ylm,\Ylm)&=\sqrt{\ell(\ell+1)}(\YlmE1+\YlmH1),\\
    \D1(\YlmE1,\YlmH1)&=\sqrt{\ell(\ell+1)}(\Ylm,\Ylm),
\end{nalign}
as well as
\begin{nalign}
    \Ds2 \YlmE1=\sqrt{\tfrac12(\ell-1)(\ell+2)}\YlmE2,&& \Ds2 \YlmH1=\sqrt{\tfrac12(\ell-1)(\ell+2)}\YlmH2,\\
    \D2 \YlmE2=\sqrt{\tfrac12(\ell-1)(\ell+2)}\YlmE1,&&\D2 \YlmH2=\sqrt{\tfrac12(\ell-1)(\ell+2)}\YlmH1.
\end{nalign}
\end{lemma}

Next, we define the projection onto these spherical harmonics in the following way:
\begin{defi}
Let $\beta$ be a smooth 1-form. Then we define
\begin{align}
    \beta_{\ell,m}^E:=\int_{\Stwo} \beta\cdot \YlmE{1}\dw,&&\beta_{\ell,m}^H:=\int_{\Stwo} \beta\cdot \YlmH{1} \dw
\end{align}
as well as
\begin{equation}
    \beta_\ell:=\sum_{m=-\ell}^\ell (\beta_{\ell,m}^E\YlmE{1}+ \beta_{\ell,m}^H\YlmH{1}).
\end{equation}
 Similarly, let $\alpha$ be a smooth symmetric traceless two-tensor. Then we define
 \begin{align}
      \alpha_{\ell,m}^E:=\int_{\Stwo} \alpha\cdot \YlmE{2},&&\alpha_{\ell,m}^H:=\int_{\Stwo} \alpha\cdot \YlmH{2}
 \end{align}
 as well as 
 \begin{equation}
    \alpha_\ell:=\sum_{m=-\ell}^\ell (\alpha_{\ell,m}^E\YlmE{2}+ \alpha_{\ell,m}^H\YlmH{2}).
\end{equation}
In particular, while the $\alpha_{\ell,m}^E$ are scalars, the $\alpha_\ell$ are traceless symmetric two-tensors.
\end{defi}
The following standard proposition is proved in \cite{Czimek17}:
\begin{prop}
The family of electric and magnetic 1-form spherical harmonics together forms a complete, orthonormal basis of $\oneformsLtwo$, and, for any $\beta \in \oneformsLtwo$:
\begin{equation}
    ||\beta ||_{L^2(\Stwo)}=\sum_{\ell=1}^\infty|| \beta_\ell||_{L^2(\Stwo)}=\sum_{\ell=1}^{\infty}\sum_{m=-\ell}^{\ell}(\beta_{\ell,m}^E)^2+(\beta_{\ell,m}^H)^2.
\end{equation}
Similarly, the family of electric and magnetic stf two-tensor spherical harmonics forms a complete, orthonormal basis of $\stfsLtwo$, and, for any $\alpha\in\stfsLtwo$,
\begin{equation}
||\alpha ||_{L^2(\Stwo)}=\sum_{\ell=2}^\infty|| \alpha_\ell||_{L^2(\Stwo)}=\sum_{\ell=2}^{\infty}\sum_{m=-\ell}^{\ell}(\alpha_{\ell,m}^E)^2+(\alpha_{\ell,m}^H)^2.
\end{equation}
\end{prop}

As before, the definitions of these tensor spherical harmonics on $\Stwo$ can easily be extended to tensor spherical harmonics on $\Schw$.

\subsection{The Kerr family}\label{sec:SS:kerr}
The Schwarzschild family smoothly embeds into the two-parameter family of Kerr spacetimes $(\mathcal M_{M,a}, g_{M,a})$ (see, for instance, \cite{DRSR16}). 
This means that when studying
perturbations of the Schwarzschild family, one will, generically, see nearby members of the Kerr family.  
We will return to this point in \S\ref{sec:gauge}, where we will write down the linearised version of the Kerr solution around Schwarzschild.
\newpage

\section[The linearised Einstein vacuum equations in a double null gauge]{The linearised Einstein vacuum equations around Schwarzschild in a double null gauge}\label{sec:lin}
\subsection{The linearisation procedure}
We here give a rough outline of the linearisation procedure, the purpose being that this paper can be read in a self-contained manner. For details, see~\cite{DHR16}.

Consider a one-parameter family of metrics $\bm{g}(\epsilon)$ in double null gauge \eqref{eq:SS:DNGmetric} solving the Einstein vacuum equations on $\mathcal M_M$ such that $\bm{g}(\epsilon=0)$ equals the Schwarzschild metric:
\begin{align}\label{eq:lin:DNGmetric}
    \bm{g}(\epsilon)=-4\bm{\Omega}^2(\epsilon)\dd{u}\dd{v}+\bm{\slashed{g}}_{ab}(\epsilon)(\dd{\theta}^a-\bm{b}^a(\epsilon)\dd{v})(\dd{\theta}^b-\bm{b}^b\dd{v}).
\end{align}
The condition that $\epsilon=0$ give the Schwarzschild metric means that $\bm{\Omega^2}(0)=\Omega^2$, $\bm{b}(0)=0$ and  $\bm{\slashed{g}}(0)=\slashed{g}$. 
Next, we associate to this family of metrics a family of \textit{orthonormal} frames $(\bm{e}_1,\bm{e}_2,\bm{e}_3,\bm{e}_4)$ according to \eqref{eq:SS:framegeneral}, and we introduce the following decomposition of the connection coefficients:

\begin{nalign}
\bm{\chi}_{AB}:=\bm{g}(\bm{\nabla}_A \bm{e}_4,\bm{e}_B),&&\bm{\underline{\chi}}_{AB}:=\bm{g}(\bm{\nabla}_A\bm{e}_3,\bm{e}_B),\\
\bm{\eta}_A:=-\frac12 \bm{g}(\bm{\nabla}_3 \bm{e}_A,\bm{e}_4),&&\bm{\underline\eta}_A:=-\frac12 \bm{g}(\bm{\nabla}_4\bm{e}_A,\bm{e}_3),\\
\bm{\omega}:=\frac12 g(\bm{\nabla}_4\bm{e}_3,\bm{e}_4),&& \bm{\underline{\omega}}:=\frac12 \bm{g}(\bm{\nabla}_{3}\bm{e}_4,\bm{e}_3).
\end{nalign}
Since the metric $g$ is in double null gauge, all other connection coefficients (except for those on the sphere) either vanish or are easily related to the ones above. For instance, $\bm{\xi}_A:=\frac12\bm{g}(\bm{\nabla}_4\bm{e}_4,\bm{e}_A)=0$, and $\bm{\zeta}_A:=\frac12\bm{g}(\bm{\nabla}_A\bm{e}_4,\bm{e}_3)=\bm{\eta}_A-\bm{\slashed{\nabla}}_A\bm{\Omega}$.

Similarly, we decompose the Riemann curvature tensor $\bm{R}$ into the coefficients ($\bm{\slashed{\varepsilon}}$ denoting the induced volume form on $\bm{\mathcal{S}_{\bm{u},\bm{v}}}$): 
\begin{nalign}
    \bm{\alpha}_{AB}:=\bm{R}_{A4B4},&&\bm{\underline{\alpha}}_{AB}:=\bm{R}_{A3B3},\\
    \bm{\beta}_A:=\bm{R}_{A434},&&\bm{\underline{\beta}}_{A}:=\bm{R}_{A343},\\
    \bm{\rho}:=\tfrac{1}{4} \bm{R}_{3434},&& \bm{\sigma}\bm{\slashed{\varepsilon}_{AB}}:=\tfrac{1}{2} \bm{R}_{AB34}.
\end{nalign}
The symmetries of the Riemann tensor allow to recast the last two definitions into the form
\begin{equation}
    -\bm{\rho}\bm{\slashed{g}}_{AB}+\bm{\sigma}\bm{\slashed{\varepsilon}}_{AB}=\bm{R}_{A3B4}.
\end{equation}
Finally, we denote the Gauss curvature on the spheres by $\bm{K}$.
\begin{rem}
These decompositions are the basis of the Christodoulou--Klainerman formalism, which is essentially the same as the original Newman--Penrose formalism. 
In particular, writing $\sqrt{2}m=\bm{e}_1+i\bm{e}_2$, $\sqrt{2}\conj{m}=\bm{e}_1-i\bm{e}_2$\footnote{Note that $\conj{m}$ does \textit{not} denote the magnetic conjugate (cf.~Def.~\ref{def:SS:magneticconjugate}) here, but the complex conjugate. No further confusion should arise in the main body of the paper, as we won't further mention these complex frame vector fields.}, we can define the Newman--Penrose scalars~$\Psi_i$ via:
\begin{align}
    \Psi_0&: =\bm\alpha(m,m),&&  \Psi_4: =\bm{\underline{\alpha}}(\conj{m},\conj m)\\
    \Psi_1&:=\bm\beta(m)&& \Psi_3:=\bm{\underline{\beta}}(\conj{m})\\
    \Psi_2&:=-\bm\rho+i\bm\sigma 
 \end{align}
 For a detailed dictionary between the two, see the Appendix \ref{app:dictionaryCKNP}.
\end{rem}
At this point, we write down the Einstein vacuum equations as well as the Bianchi equations with respect to this decomposition, and linearise each of the resulting equations with respect to the parameter $\varepsilon$ (by writing $\bm{\Omega}=\Omega+\varepsilon\cdot \Om +\O(\varepsilon^2)$ etc.~and then discarding all terms of order $\varepsilon^2$). 
In the following, we will always use the superscript $\overone{{}}$ to denote linearised quantities and bold font to denote nonlinear quantities.
\subsection{The system of equations of linearised gravity around Schwarzschild~\texorpdfstring{\fullsystem}{(3.9)--(3.32)}}
Following the linearisation procedure outlined above gives the system of linearised gravity around Schwarzschild. 
In fact, this system can be formulated without reference to any linearisation by viewing it as a geometric system of partial differential equations on Schwarzschild for a set of unknown quantities:
\begin{defi}\label{def:solution to LEE}
Let $\mathcal D$ be any spherically symmetric open subset of $\Schw$, and let
\begin{align}
    \mathfrak{S}=\( \gsh,\, \trg,\, \Om,\, \b,\,\trx,\, \trxb,\,\xh,\,\xhb,\,\et,\, \etb,\, \om,\, \omb,\, \al,\,\be,\,\rh,\,\sig,\,\beb,\,\alb,\,	\K	\),
\end{align}
where 
\begin{itemize}
\item  $\Om$, $\trg$, $\trx$, $\trxb$, $\om$, $\omb$, $\rh$, $\sig$ and $\K$ $\in C^\infty(\mathcal{D})$, 
\item  $\b$, $\et$, $\etb$, as well as $\be$ and $\beb$ $\in \oneformsfields{\mathcal{D}}$, and
\item $\gsh$, $\xh$, $\xhb$ as well as $\al$ and $\alb$ $\in\stffields{\mathcal{D}}$.
\end{itemize}
 We say $\mathfrak{S}$ is a solution to the linearised Einstein vacuum equations around Schwarzschild if the components of $\mathfrak{S}$ satisfy the equations \fullsystem~everywhere on $\mathcal{D}$. 
\end{defi}
We now write down the equations of linearised gravity around Schwarzschild~\fullsystem.

\newpage
\subsubsection{The equations governing the metric coefficients}
First, we have the equations governing the metric coefficients:

\begin{subequations}\label{eq:lin:trg}
    \noindent\begin{minipage}{0.5\textwidth}
\begin{equation}
 \pu \trg =2\trxb \label{eq:lin:trg3}
\end{equation}
    \end{minipage}%
    \begin{minipage}{0.5\textwidth}
    \begin{minipage}{0.1\textwidth}
    \centering
    \end{minipage}
\begin{equation}
\pv \trg=2\trx-2\div r^{-1}\b \label{eq:lin:trg4}
\end{equation}
    \end{minipage}\vskip1em
\end{subequations}
\begin{subequations}\label{eq:lin:gsh}
    \noindent\begin{minipage}{0.5\textwidth}
\begin{equation}
\Du \gsh =2\Omega\xhb\label{eq:lin:gsh3}
\end{equation}
    \end{minipage}%
     \begin{minipage}{0.1\textwidth}
    \centering
    \end{minipage}
    \begin{minipage}{0.5\textwidth}
\begin{equation}
\Dv \gsh=2\Omega\xh+2\Ds{2}r^{-1}\b \label{eq:lin:gsh4}
\end{equation}
    \end{minipage}\vskip1em
\end{subequations}
\begin{subequations}\label{eq:lin:Omm}
    \noindent\begin{minipage}{0.5\textwidth}
\begin{equation}
\pu\Omm=\omb\label{eq:lin:Omm3}
\end{equation}
    \end{minipage}%
     \begin{minipage}{0.1\textwidth}
    \centering
    \end{minipage}
    \begin{minipage}{0.5\textwidth}
\begin{equation}
\pv\Omm=\om\label{eq:lin:Omm4}
\end{equation}
    \end{minipage}%
    \vskip1em
\end{subequations}
\begin{equation}
2\sl\Omm=r(\et+\etb) \label{eq:lin:OmmA}
\end{equation}
\begin{equation}\label{eq:lin:b3}
\Du (r^{-1}\b)=\frac{2\Omega^2}{r}(\et-\etb)
\end{equation}
\subsubsection{The equations governing the connection coefficients}
Next, we have the equations governing the connection coefficients:
\begin{equation}\label{eq:lin:trx3}
\Du\(r\trx\)=2\Omega^2\(\div \et+r\rh-\frac{4M}{r^2}\Omm\)-\Omega^2\trxb
\end{equation}
 \begin{equation}\label{eq:lin:trxb4}
 \Dv \(r\trxb\)=2\Omega^2\(\div \etb+r\rh-\frac{4M}{r^2}\Omm\)+\Omega^2\trx
 \end{equation}
\begin{subequations}\label{eq:lin:trx+trxb}
    \noindent\begin{minipage}{0.5\textwidth}
\begin{equation}
\Du \(\frac{r^2}{\Omega^2}\trxb\)=-4r\omb\label{eq:lin:trxb3}
\end{equation}
    \end{minipage}%
     \begin{minipage}{0.1\textwidth}
    \centering
    \end{minipage}
    \begin{minipage}{0.5\textwidth}
\begin{equation}
\Dv \(\frac{r^2}{\Omega^2}\trx\)=4r\om\label{eq:lin:trx4}
\end{equation}
    \end{minipage}%
    \vskip1em
\end{subequations}
 \begin{subequations}\label{eq:lin:xh}
    \noindent\begin{minipage}{0.5\textwidth}
\begin{equation}
\Du\(r\Omega\xh\)=-2\Omega^2\Ds2\et-{\Omega^2}\Omega \xhb\label{eq:lin:xh3}
\end{equation}
    \end{minipage}%
     \begin{minipage}{0.1\textwidth}
    \centering
    \end{minipage}
    \begin{minipage}{0.5\textwidth}
\begin{equation}
\Dv\(\frac{r^2\xh}{\Omega}\)=-r^2\al\label{eq:lin:xh4}
\end{equation}
    \end{minipage}%
    \vskip1em
\end{subequations}
  \begin{subequations}\label{eq:lin:xhb}
    \noindent\begin{minipage}{0.5\textwidth}
\begin{equation}
\Du\(\frac{r^2\xhb}{\Omega}\)=-r^2\alb\label{eq:lin:xhb3}
\end{equation}
    \end{minipage}%
     \begin{minipage}{0.1\textwidth}
    \centering
    \end{minipage}
    \begin{minipage}{0.5\textwidth}
\begin{equation}
\Dv\(r\Omega\xhb\)=-2\Omega^2\Ds2\etb+{\Omega^2}\Omega \xh\label{eq:lin:xhb4}
\end{equation}
    \end{minipage}%
    \vskip1em
\end{subequations}
  \begin{subequations}\label{eq:lin:et}
    \noindent\begin{minipage}{0.5\textwidth}
\begin{equation}
\Du(r^2\et)=2r \sl\omb-r^2\Omega\beb\label{eq:lin:et3}
\end{equation}
    \end{minipage}%
     \begin{minipage}{0.1\textwidth}
    \centering
    \end{minipage}
    \begin{minipage}{0.5\textwidth}
\begin{equation}
\Dv(r\et)=-r\Omega\be +\Omega^2\etb\label{eq:lin:et4}
\end{equation}
    \end{minipage}%
    \vskip1em
\end{subequations}
  \begin{subequations}\label{eq:lin:etb}
    \noindent\begin{minipage}{0.5\textwidth}
\begin{equation}
\Du(r\etb)=r\Omega\beb -\Omega^2\et\label{eq:lin:etb3}
\end{equation}
    \end{minipage}%
     \begin{minipage}{0.1\textwidth}
    \centering
    \end{minipage}
    \begin{minipage}{0.5\textwidth}
\begin{equation}
\Dv(r^2\etb)=2r \sl\omega+r^2\Omega\be \label{eq:lin:etb4}
\end{equation}
    \end{minipage}%
    \vskip1em
\end{subequations}
\begin{equation}\label{eq:lin:om3}
\pu\om=-\Omega^2\(\rh-\frac{4M}{r^3}\Omm\)= \pv \omb
\end{equation}
Furthermore, we have the following elliptic equations:
\begin{align}
\div\xh&=-\Omega\etb-r\be+\frac{1}{2\Omega}\sl\trx\label{eq:lin:divxh}\\
\div\xhb&=\Omega\et+r\beb+\frac{1}{2\Omega}\sl\trxb\label{eq:lin:divxhb}\\
&\curl\et=r\sig=-\curl\etb\label{eq:lin:curleta}
\end{align}
We also have the following linearised equation for the Gaussian curvature:
\begin{equation}\label{eq:lin:K}
\K=-\rh+\frac{1}{2r}(\trx-\trxb)-\frac{2\Omega^2}{r^2}\Omm,
\end{equation}
where $\K$ is defined via
\begin{equation}\label{eq:lin:Kdef}
\K:=-\frac1{4r^2}(\lap+2)\trg +\frac12 \divv\divv \gsh.
\end{equation}
\subsubsection{The equations governing the curvature coefficients}
The system of linearised gravity is completed by the linearised Bianchi equations:
\begin{equation}\label{eq:lin:alb4}
\Dv(r\Omega^2\alb)=2\Ds2 \Omega^2\Omega\beb +\frac{6M\Omega^2}{r^2}\Omega\xhb
\end{equation}
  \begin{subequations}\label{eq:lin:beb}
    \noindent\begin{minipage}{0.34\textwidth}
\begin{equation}
\Du\frac{r^4\beb}{\Omega}=-\div r^3\alb\label{eq:lin:beb3}
\end{equation}
    \end{minipage}%
     \begin{minipage}{0.1\textwidth}
    \centering
    \end{minipage}
    \begin{minipage}{0.66\textwidth}
\begin{equation}
\Dv(r^2\Omega\beb)=\Ds1(r\Omega^2\rh,r\Omega^2\sig)+\frac{6M\Omega^2}{r}\etb \label{eq:lin:beb4}
\end{equation}
    \end{minipage}%
    \vskip1em
\end{subequations}
  \begin{subequations}\label{eq:lin:rh}
    \noindent\begin{minipage}{0.5\textwidth}
\begin{equation}
\pu(r^3\rh)=-\div r^2\Omega\beb +3M\trxb\label{eq:lin:rh3}
\end{equation}
    \end{minipage}%
     \begin{minipage}{0.1\textwidth}
    \centering
    \end{minipage}
    \begin{minipage}{0.5\textwidth}
\begin{equation}
\pv(r^3\rh)=\div r^2\Omega\be +3M\trx \label{eq:lin:rh4}
\end{equation}
    \end{minipage}%
    \vskip1em
\end{subequations}
  \begin{subequations}\label{eq:lin:sig}
    \noindent\begin{minipage}{0.5\textwidth}
\begin{equation}
\pu(r^3\sig)=-\curl r^2\Omega\beb\label{eq:lin:sig3}
\end{equation}
    \end{minipage}%
     \begin{minipage}{0.1\textwidth}
    \centering
    \end{minipage}
    \begin{minipage}{0.5\textwidth}
\begin{equation}
\pv(r^3\sig)=-\curl r^2\Omega\be \label{eq:lin:sig4}
\end{equation}
    \end{minipage}%
    \vskip1em
\end{subequations}
  \begin{subequations}\label{eq:lin:be}
    \noindent\begin{minipage}{0.66\textwidth}
\begin{equation}
\Du(r^2\Omega\be)=\Ds1(-r\Omega^2\rh,r\Omega^2\sig)-\frac{6M\Omega^2}{r}\et\label{eq:lin:be3}
\end{equation}
    \end{minipage}%
    \begin{minipage}{0.34\textwidth}
\begin{equation}
\Dv\frac{r^4\be}{\Omega}=\div r^3\al \label{eq:lin:be4}
\end{equation}
    \end{minipage}%
    \vskip1em
\end{subequations}
 \begin{equation}\label{eq:lin:al3}
 \Du(r\Omega^2\al)=-2\Ds2 \Omega^2\Omega\be +\frac{6M\Omega^2}{r^2}\Omega\xh
 \end{equation}

\subsection{The Teukolsky equations and the Regge--Wheeler equations}\label{sec:lin:Teuk+RW}
Having written down the full system of linearised gravity around Schwarzschild, we now derive from \fullsystem~a set of equations that will play a central role in this paper, among them the decoupled Teukolsky and the Regge--Wheeler equations.
\subsubsection{The Teukolsky equations \texorpdfstring{satisfied by $\overone{\alpha}$ and $\overone{\underline{\alpha}}$}{}}
Multiply \eqref{eq:lin:al3} by $\frac{r^4}{\Omega^4}$ and apply $\Dv$ to it using \eqref{eq:lin:be4} and \eqref{eq:lin:xh4} to obtain a decoupled wave equation for $\al$, known as the Teukolsky equation (recall $-2\Ds2\Ds2=\lap-2$ from \eqref{eq:SS:DD=Laplacian}):
\begin{equation}\label{eq:lin:Teukal}\tag{Teuk}
\Dv\(\frac{r^4}{\Omega^{4}}\Du(r\Omega^2\al)\)=r^3(\lap-2)\al-6Mr^2\al.
\end{equation}
Similarly, starting from \eqref{eq:lin:alb4}, multiplying it by $\frac{r^4}{\Omega^4}$ and applying $\Du$ to it using \eqref{eq:lin:beb3} and \eqref{eq:lin:xhb3}, we obtain
\begin{equation}\label{eq:lin:Teukalb}\tag{\underline{Teuk}}
\Du\(\frac{r^4}{\Omega^4}\Dv(r\Omega^2 \alb)\)=r^3(\lap-2)\alb-6Mr^2\alb
\end{equation}
These equations will be discussed in much more detail in \S\ref{sec:cons}.
\subsubsection{The Regge--Wheeler equations \texorpdfstring{satisfied by $\overone{\Psi}$ and $\overone{\underline{\Psi}}$}{}}
In applications, the first order terms in \eqref{eq:lin:Teukal}, \eqref{eq:lin:Teukalb} (which appear if one writes the equations as $\Dv\Du\al=\dots$) give rise to several difficulties. 
Fortunately, these first-order terms can be removed by commutations with suitable vector fields.
First, we note that,
\begin{align}\label{eq:lin:al3_v2}
\(\frac{r^2\Du}{\Omega^2}\)(r\Omega^2\al)&=-2\Ds2 r^2\Omega\be+6M\Omega\xh,\\
\(\frac{r^2\Du}{\Omega^2}\)^2(r\Omega^2\al)&=2\Ds2 \Ds1(r^3\rh,-r^3\sig)+6M(r\Omega\xh-r\Omega\xhb)\label{eq:lin:al33}.
\end{align}
Eq.~\eqref{eq:lin:al3_v2} is just \eqref{eq:lin:al3} multiplied by $\frac{r^2}{\Omega^2}$, and \eqref{eq:lin:al33} then follows from \eqref{eq:lin:al3_v2} by using \eqref{eq:lin:be3} and \eqref{eq:lin:xh3}.

Similarly, using \eqref{eq:lin:alb4} and \eqref{eq:lin:beb4} with \eqref{eq:lin:xhb4}, we get
\begin{align}
\(\frac{r^2\Dv}{\Omega^2}\)(r\Omega^2\alb)&=2\Ds2 r^2\Omega\beb+6M\Omega\xhb,\label{eq:lin:alb4_v2}\\
\(\frac{r^2\Dv}{\Omega^2}\)^2(r\Omega^2\alb)&=2\Ds2 \Ds1(r^3\rh,r^3\sig)+6M(r\Omega\xh-r\Omega\xhb)\label{eq:lin:alb44}
\end{align}

We now define:
\begin{nalign}\label{eq:lin:transformations}
\ps:=\(\frac{r^2\Du}{\Omega^2}\)(r\Omega^2\al),&&\psb:=\(\frac{r^2\Dv}{\Omega^2}\)(r\Omega^2\alb).\\
\Ps:=\(\frac{r^2\Du}{\Omega^2}\)^2(r\Omega^2\al),&&\Psb:=\(\frac{r^2\Dv}{\Omega^2}\)^2(r\Omega^2\alb).
\end{nalign}
An observation of purely algebraic nature (cf.~Lemma~\ref{lem:lin:RW Teuk identity}) is that if $\al$ satisfies \eqref{eq:lin:Teukal} (or if $\alb$ satisfies \eqref{eq:lin:Teukalb}), then the remarkably simple Regge--Wheeler equation,
\begin{equation}\label{eq:lin:RW}\tag{RW}
\Du\Dv\Psi-\frac{\Omega^2}{r^2}(\lap-4)\Psi-\frac{6M\Omega^2}{r^3}\Psi=0,
\end{equation}
is satisfied for $\Psi=\Ps$, $\Psb$. (See \S\ref{sec:cons} for much more details on commutations with $\frac{r^2}{\Omega^2}\Du$, $\frac{r^2}{\Omega^2}\Dv$.)

As a consequence of
\begin{equation}\label{eq:lin:Ps-Psb=sig}
\Ps-\Psb=-4\Ds2\Ds1(0,r^3\sig),
\end{equation} eq.~\eqref{eq:lin:RW}
is therefore also satisfied by $\Ds2\Ds1(0,r^3\sig)$. In fact, in that case, \eqref{eq:lin:RW} is just the tensorialised version of the scalar Regge--Wheeler equation:
\begin{equation}\label{eq:lin:RWsig}\tag{RW-scalar}
\pu\pv(r^3\sig)-\frac{\Omega^2}{r^2}\lap(r^3\sig)-\frac{6M\Omega^2}{r^3}(r^3\sig)=0,
\end{equation}
which can also be verified directly by acting with $\pv$ on \eqref{eq:lin:sig3} and using \eqref{eq:lin:beb4}.

\subsubsection{The Teukolsky--Starobinsky identities and further relations}
We further note the relations (which follow directly from  \eqref{eq:lin:Teukal})
\begin{align}\label{eq:lin:Ps4}
\(\frac{r^2\Dv}{\Omega^2}\)\Ps&=\((\lap-2)-2\(1-\frac{{3}M}{r}\)\) \frac{r^4}{\Omega^4}\Du (r\Omega^2\al)-\frac{6Mr^{{2}}}{\Omega^2}r\Omega^2\al\\
\label{eq:lin:Ps44}
\frac{\Omega^4}{r^4}\(\frac{r^2\Dv}{\Omega^2}\)^2\Ps&=(\lap-2)(\lap-4)r\Omega^2\al-6M(\Du+\Dv)(r\Omega^2\al)
\end{align}
Similarly, we have (as a consequence of \eqref{eq:lin:Teukalb})
\begin{align}\label{eq:lin:Psb3}
\(\frac{r^2\Du}{\Omega^2}\)\Psb&=\((\lap-2)-2\(1-\frac{{3}M}{r}\)\) \frac{r^4}{\Omega^4}\Du (r\Omega^2\alb)+\frac{6Mr^{{2}}}{\Omega^2}r\Omega^2\alb\\
\label{eq:lin:Psb33}
\frac{\Omega^4}{r^4}\(\frac{r^2\Du}{\Omega^2}\)^2\Psb&=(\lap-2)(\lap-4)r\Omega^2\alb+6M(\Du+\Dv)(r\Omega^2\alb)
\end{align}
From the identity \eqref{eq:lin:Ps-Psb=sig}, we can now derive the the celebrated Teukolsky--Starobinsky identities~\cite{Teukolsky74} by either acting with $(\Omega^{-2}r^2\Du)^2$, applying \eqref{eq:lin:Psb33} and \eqref{eq:lin:sig3} and \eqref{eq:lin:beb3}, or by acting with$(\Omega^{-2}r^2\Du)^2$, applying \eqref{eq:lin:Ps44} and \eqref{eq:lin:sig4} and \eqref{eq:lin:be4}:
\begin{multline}\label{eq:lin:TSI+}
\frac{\Omega^4}{r^4}\(\frac{r^2\Du}{\Omega^2}\)^4(r\Omega^2\al)\\
=(\lap-2)(\lap-4)r\Omega^2\alb-4\Ds2\Ds1(0,\curl\div r\Omega^2\alb )+6M(\Du+\Dv)(r\Omega^2\alb),
\end{multline}
\begin{multline}\label{eq:lin:TSI-}
\frac{\Omega^4}{r^4}\(\frac{r^2\Dv}{\Omega^2}\)^4(r\Omega^2\alb)\\
=(\lap-2)(\lap-4)r\Omega^2\al-4\Ds2\Ds1(0,\curl\div r\Omega^2\al)-6M(\Du+\Dv)(r\Omega^2\al).
\end{multline}
\begin{rem}
For any $\digamma\in\Gamma^\infty(\pqbundletensors{p}{q}{\Schw})$, we have  $r(\Omega^{-2}\Du)^4(r^4\digamma)=r^{-4}(\Omega^{-2}r^2\Du)^4(r\digamma)$.
\end{rem}
\begin{rem}
The angular operator on the RHS of \eqref{eq:lin:TSI+} or \eqref{eq:lin:TSI-} can be rewritten:
Writing $\alpha=\Ds2\Ds1(f,g)$,  it follows from Lemma~\ref{lem:SS:commutation} that
\begin{equation}
-4\Ds2\Ds1(0,\curl\D2\Ds2\Ds1(f,g))=-2(\lap-2)(\lap-4)\Ds2\Ds1(0,g).
\end{equation}
Conversely, one can show that for any $\alpha\in\stfs$, ${2}\Ds2\Ds1\D1\D2\alpha=(\lap-2)(\lap-4)\alpha$.
Thus, we can equivalently write the angular operator on the RHS of \eqref{eq:lin:TSI+} or \eqref{eq:lin:TSI-}  as the operator that sends $\Ds2\Ds1(f,g)\mapsto (\lap-2)(\lap-4)\Ds2\Ds1(f,-g)$, or as ${2}\Ds2\Ds1\overline{\D1}\D2$, i.e.~we have
\begin{equation}
\frac{\Omega^4}{r^4}\(\frac{r^2\Dv}{\Omega^2}\)^4(r\Omega^2\alb)
={2}\Ds2\Ds1\overline{\D1}\D2 r\Omega^2\al-6M(\Du+\Dv)(r\Omega^2\al).
\end{equation}

In terms of the spin-weighted "eth"-operators $\eth$, $\eth'$ of the Newman--Penrose formalism, the operator $\Ds2\Ds1\overline{\D1}{\D2}$ corresponds to $\eth'^4$ (or $\eth^4$ when acting on $\alb$), see the appendix for details (Remark~\ref{rem:app:eth4}.) 
\end{rem}
\subsection*{A useful lemma}
For future reference, let us here already collect the following statements concerning the relation between the Regge--Wheeler and Teukolsky operators:
\begin{defi}\label{def:lin:Teuk}
    Given $\digamma\in \stffields\Schw$, define 
    \begin{align}
        \Teuk[\digamma]&:=\frac{\Omega^2}{r^2}\Dv \frac{r^4}{\Omega^4}\Du \digamma-\(\lap-2-\frac{6M}{r}\)\digamma,\\
          \Teukb[\digamma ]&:=\frac{\Omega^2}{r^2}\Du \frac{r^4}{\Omega^4}\Dv \digamma-\(\lap-2-\frac{6M}{r}\)\digamma,\\
          \RW[\digamma ]&:=\frac{r^2}{\Omega^2}\(\Du\Dv-\frac{\Omega^2}{r^2}\left(\lap-4-\frac{6M}{r}\right)\)\digamma .
    \end{align}
\end{defi}

\begin{lemma}\label{lem:lin:RW Teuk identity}
    For any $\digamma \in\stffields\Schw$, we have
    \begin{align}\label{eq:lin:RW Teukal identity}
        \RW\left[\(\frac{r^2}{\Omega^2}\Du\)^2{(\digamma)}\right] &=\(\frac{r^2}{\Omega^2}\Du\)^2\Teuk[\digamma ],\\
    \label{eq:lin:RW Teukalb identity}
        \RW\left[\(\frac{r^2}{\Omega^2}\Dv\)^2{(\digamma)}\right] &=\(\frac{r^2}{\Omega^2}\Dv\)^2\Teukb[\digamma ],
    \end{align}
    as well as 
     \begin{align}\label{lem:lin:RWTEUKSTARO}
        \Teuk\left[\frac{\Omega^2}{r^2}\Dv\frac{r^2}{\Omega^2}\Dv \digamma\right]=\frac{\Omega^2}{r^2}\Dv\frac{r^2}{\Omega^2}\Dv \left\{\mathrm{RW}\left[\digamma\right]\right\}.
    \end{align}
\end{lemma}
\begin{proof} The first two identities are proved in section 3 of \cite{Masaood22}, so we only prove the last.
  We first compute
    \begin{align}
        \left[\frac{\Omega^2}{r^2}\Dv\frac{r^2}{\Omega^2}\Dv\right]\Du=\Du\left[\frac{\Omega^2}{r^2}\Dv\frac{r^2}{\Omega^2}\Dv\right]-\frac{2\Omega^2}{r^2}(3\Omega^2-2)\Dv.
    \end{align}
    From this, we similarly infer that
    \begin{align}
    \begin{split}
        \left[\frac{\Omega^2}{r^2}\Dv\frac{r^2}{\Omega^2}\Dv\right]\left(\frac{r^2}{\Omega^2}\Du\Dv\right)&=\frac{\Omega^2}{r^2}\Dv\frac{r^4}{\Omega^4}\left[\Du\(\frac{\Omega^2}{r^2}\Dv\frac{r^2}{\Omega^2}\Dv\)-\frac{2\Omega^2}{r^2}(3\Omega^2-2)\Dv\right]\\
        &=\left[\frac{\Omega^2}{r^2}\Dv\frac{r^4}{\Omega^4}\Du\right]\left[\frac{\Omega^2}{r^2}\Dv\frac{r^2}{\Omega^2}\Dv\right]-12M\frac{\Omega^2}{r^2}\Dv\\&\;\;\;\;-2(3\Omega^2-2)\left[\frac{\Omega^2}{r^2}\Dv\frac{r^2}{\Omega^2}\Dv\right].
    \end{split}
    \end{align}
    Finally, note that
    \begin{align}
        \left[\frac{\Omega^2}{r^2}\Dv\frac{r^2}{\Omega^2}\Dv\right]\frac{6M}{r}=-12M\frac{\Omega^2}{r^2}\Dv+\frac{6M}{r}\left[\frac{\Omega^2}{r^2}\Dv\frac{r^2}{\Omega^2}\Dv\right].
    \end{align}
    Applying the above relations to the expression
    \begin{align}
        \frac{\Omega^2}{r^2}\Dv\frac{r^2}{\Omega^2}\Dv\left[\frac{r^2}{\Omega^2}\Du\Dv-\lap+\(4-\frac{6M}{r}\)\right]
    \end{align}
    proves the claim.
\end{proof}

\newpage 

\section{Pure gauge and linearised Kerr solutions}\label{sec:gauge}
Central to the understanding of the system of linearised gravity around Schwarzschild is the existence of two classes of explicit solutions to it. 
The first of them arises by linearising either a Schwarzschild solution with nearby mass, or a Kerr solution with mass $M$ near Schwarzschild (cf.~\S\ref{sec:SS:kerr})--written in double null gauge--around Schwarzschild. This class is the class of linearised Kerr solutions.

The other class, the class of pure gauge solutions, arises by considering coordinate transformations (at the nonlinear level) of the variables $\bm{u},\bm{v},\bm{\theta}^A$ such that the double null form of the nonlinear metric \eqref{eq:SS:DNGmetric} is preserved, and then linearising the metric in these new coordinates around the Schwarzschild metric. 
These solutions to \fullsystem, which represent  the gauge ambiguity inherited from \eqref{eq:intro:EVE},  have been derived and classified in Section~6 of~\cite{DHR16}; we here write them down for completeness. 

\subsection{The linearised Kerr solutions}
\begin{prop}\label{prop:gauge:SS}
Let $\mathfrak{m}\in\mathbb{R}$. Then the following solves \fullsystem:
\begin{align}
\begin{split}
    &\trg=-2\mathfrak{m},\qquad \Omm=-\frac{1}{2}\mathfrak{m},\\ 
    &\rh=-\frac{2M}{r^3}\mathfrak{m},\qquad \K=\frac{\mathfrak{m}}{r^2},
\end{split}
\end{align}
with all remaining components vanishing.
Note that this solution is entirely supported on $\ell=0$. We will refer to it as \emph{the linearised nearby Schwarzschild solution with parameter $\mathfrak{m}$}, or by $\mathfrak{S}_{\mathfrak{m}}$.
\end{prop}
\begin{prop}\label{prop:gauge:Kerr}
    Let $\mathfrak{a}=\{\mathfrak{a}_{-1},\mathfrak{a}_{0},\mathfrak{a}_{+1}\}\in\mathbb{R}^3$. Then the following is a solution to \fullsystem: 
    \begin{nalign}
        \b
        =\frac{4M}{r^2}\sqrt{2} \sum_{m=-1}^1 \mathfrak{a}_m Y^{\mathrm{H},1}_{\ell 1},&&
        \et=\frac{3}{4r}\b=-\etb,\\
        \be=\frac{\Omega}{r}\et=-\beb,\qquad\qquad&& \sig=\frac{6M}{r^4} \sum_{m=-1}^1 \mathfrak{a}_m Y_{\ell 1},
    \end{nalign}
with all remaining components vanishing. Note that this solution is entirely supported on $\ell=1$. We will refer to it as \emph{the linearised Kerr solution with parameter $\mathfrak{a}$}, or by $\mathfrak{S}_{\mathfrak{a}}$.
\end{prop}
\subsection{The pure gauge solutions}
\subsubsection{The outgoing gauge solutions}
The solution below arises from an infinitesimal coordinate transformation of $\bm{v}$:
\begin{prop}\label{prop:gauge:out}
Let $f(v,\theta^A)$ be a smooth scalar function on $\mathbb{R}\times \Stwo$. Then the following solves \fullsystem:
\begin{nalign}
   \frac{\Om}{\Omega}&=\frac{1}{2\Omega^2}\pv(f\Omega^2),\quad &\trg&=\frac{4}{r}\left(\lap+\Omega^2\right)f,\quad &\gsh&=\frac{4}{r}\Ds{2}\Ds{1}(f,0), \\
    \b&=2\Ds{1}\(r\pv(\tfrac{f}{r}),0\),\quad &\et&=-\frac{\Omega^2}{r^2}\Ds{1}(f,0),\quad &\etb&=-\Ds{1}\left(\tfrac{1}{\Omega^2}\pv(\tfrac{\Omega^2}{r}f),0\right),\\
    \xhb&=2\frac{\Omega}{r^2}\Ds{2}\Ds{1}(f,0),\quad &\trx&=2\pv(\tfrac{\Omega^2}{r}f),\quad &\trxb&=2\frac{\Omega^2}{r^2}\left(\lap+2\Omega^2-1\right)f,\\
    \rh&=\frac{6M\Omega^2}{r^4}f,\quad &\beb&=-\frac{6M\Omega}{r^4}\Ds{1}(f,0),\quad &\K&=-\frac{\Omega^2}{r^3}\left(\lap+2\right)f,
\end{nalign}
with $\om$, $\omb$ related to $\Om$ via \eqref{eq:lin:Omm}, and with the remaining components $\xh,\,\al,\,\alb,\,\be,\,\sig$ vanishing. We will refer to this solution as \emph{the outgoing gauge solution generated by $f$}, or by $\mathfrak{S}_{f}$.
\end{prop}
In practice, we will require that, as $v\to\infty$, $f(v,\theta^A)=f_0(\theta^A)v +\O_{\infty}(v^{1-\epsilon})$ for some $\epsilon>0$ and some smooth, potentially identically vanishing $f_0(\theta^A)$. 
\subsubsection{The ingoing gauge solutions}
The solution below arises from an infinitesimal coordinate transformation of $\bm{u}$:
\begin{prop}\label{prop:gauge:in}
Let $\underline{f}(u,\theta^A)$ be a smooth scalar function on $\mathbb{R}\times \Stwo$. The following is a solution to \fullsystem:
\begin{nalign}
    \frac\Om\Omega&=\frac{1}{2\Omega^2}\pu(\underline{f}\Omega^2),\quad &\trg&=-\frac{4\Omega^2}{r}\underline{f},\quad &\gsh&=0,\\
    \b&=-2\frac{\Omega^2}{r}\,\Ds{1}(\underline{f},0), \quad &\et&=-\Ds{1}\left(\tfrac{1}{\Omega^2}\pu(\tfrac{\Omega^2}{r}\underline{f}),0\right),\quad &\etb&={+}\frac{\Omega^2}{r^2}\Ds{1}(\underline{f},0),\\
    \xh&=2\frac{\Omega}{r^2}\Ds{2}\Ds{1}\left(\underline{f},0\right),\quad
    &\trx&=2\frac{\Omega^2}{r^2}\left(\lap+2\Omega^2-1\right)\underline{f},\quad &\trxb&=-2\pu(\tfrac{\Omega^2}{r}\underline{f}),\\
    \rh&=-\frac{6M\Omega^2}{r^4}\underline{f},\quad &\be&={+}\frac{6M\Omega}{r^4}\Ds{1}(\underline{f},0),\quad &\K&=\frac{\Omega^2}{r^3}\left(\lap+2\right)\underline{f},
\end{nalign}
with $\om$, $\omb$ related to $\Om$ via \eqref{eq:lin:Omm}, and with the remaining components $\xhb,\,\al,\,\alb,\,\beb,\,\sig$ vanishing. We will refer to this solution as \emph{the ingoing gauge solution generated by $\underline{f}$}, or by $\mathfrak{S}_{\underline{f}}$.
\end{prop}
In practice, we will require that, as $u\to-\infty$, $\underline{f}(u,\theta^A)=\underline{f}_0(\theta^A)u +\O_{\infty}(u^{1-\epsilon})$ for some $\epsilon>0$ and some smooth, potentially identically vanishing $\underline{f}_0(\theta^A)$. 
\begin{rem}\label{rem:gauge:constantgaugesolution}
Notice that the outgoing gauge solution generated by a constant $C$ is identical to the ingoing gauge solution generated by $-C$.
\end{rem}
\subsubsection{The sphere gauge solutions}
The solution below arises from an infinitesimal coordinate transformation of $\{\bm{\theta}^A\}$:
\begin{prop}\label{prop:gauge:sphere}
Let $q_1(v,\theta^A)$, $q_2(v,\theta^A)$ be two smooth scalar functions on $\mathbb{R}\times \Stwo$. Then the following solves \fullsystem:
\begin{align}
    \gsh&=2\Ds{2}\Ds{1}(q_1,q_2),\qquad &\trg&=2\lap(q_1),\qquad &r^{-1}\b&=\Ds{1}(\pv q_1,\pv q_2),
\end{align}
with all other components vanishing. We refer to these solutions as the \emph{sphere gauge solutions} generated by $(q_1,q_2)$, or by $\mathfrak{S}_{(q_1,q_2)}$.
\end{prop}
In practice, $q_1,\, q_2$ and all its derivatives will be required to be bounded.

\begin{defi}
A solution is said to be pure gauge if it is a linear combination of the solutions from Propositions~\ref{prop:gauge:out}--\ref{prop:gauge:sphere}.
\end{defi}

\subsection{Solutions supported on \texorpdfstring{$\ell\leq 1$}{ell<2} are non-dynamical}\label{sec:gauge:ell01} 
The following statement is proved (in slightly modified form) in Theorem~9.2 of \cite{DHR16}:
\begin{prop}\label{prop:gauge:l<2}
If a solution $\mathfrak{S}$ solves \fullsystem~on an open subset of $\SchwEx$ and is supported on  $\ell=0,1$, then, up to suitable addition of a pure gauge solution, $\mathfrak{S}$ can be written as a linearised nearby Schwarzschild solution (Prop.~\ref{prop:gauge:SS}) plus a linearised Kerr solution (Prop.~\ref{prop:gauge:Kerr}).
\end{prop}
\begin{rem}
The proposition cited above embodies the fact that the $\ell=0,1$ components of a solution to \fullsystem~are non-dynamical, and the dynamics of the system \fullsystem~are contained in the $\ell\geq2$. \textbf{In light of this fact, any physical discussion of \fullsystem~is mostly a discussion of the $\ell\geq2$-modes.}

\end{rem}
\subsection{Solutions supported on \texorpdfstring{$\ell\geq2$}{ell>2} with \texorpdfstring{$\overone{\alpha}=0=\overone{\underline{\alpha}}$}{alpha=0=alphabar} are pure gauge}
The following is proved in Theorem B.1 of \cite{DHR16}:
\begin{prop}\label{prop:gauge:vanishing of al alb}
If a solution $\mathfrak{S}$ solves \fullsystem~on an open subset of $\SchwEx$, is supported on $\ell\geq 2$ and satisfies $\al=\alb=0$, then $\mathfrak{S}$ is a pure gauge solution.
\end{prop}
\newpage

\section*{Part I:\\ Setting up and solving the scattering problem\hypertarget{V:partI}{}}
\addcontentsline{toc}{section}{{\textbf{Part I}: Setting up and solving the scattering problem}}
Having given a thorough recap of the system of linearised gravity around Schwarzschild, we now have all ingredients to formulate and solve the scattering problem for this system. 
As opposed to the paper~\cite{Masaood22b}, which discusses the global scattering problem with data on $\mathcal H^-$ and~$\mathcal I^-$, the present paper discusses the semi-global scattering problem where  scattering data are posed on $\Scrim$ and an ingoing null hypersurface $\Cin$ at finite $v=v_1$, restricted to negative values of $u$ (see already Fig.~\ref{fig:setup:Penrosediagram}).  
Here, we opt to give an entirely self-contained presentation of it with a somewhat different approach that, for instance, does not use a gauge fixing procedure and that proves a few additional statements.

We now give an overview of the next few sections:

In \S\ref{sec:setup:basics}, we define a notion of \textit{seed scattering data} for \fullsystem, define what it means for a solution to \fullsystem~to be a scattering solution realising these seed data, and we write down Theorem~\ref{thm:setup:LEE Scattering wp}, which expresses the solvability of the scattering problem under fairly general assumptions on the seed data.
We then give an overview of the proof of this theorem in \S\ref{sec:setup:overview}, we construct out of the seed data the remaining scattering data along $\Cin$ in~\S\ref{sec:setup:2.i}, and we discuss certain gauge considerations such as Bondi normalisation in~\S\ref{sec:setup:Bondinor}.

The key to obtaining the scattering theory for \fullsystem~is the scattering theory for the Regge--Wheeler equation~\eqref{eq:lin:RW} developed already in \cite{Masaood22}. We will give a self-contained account of it in \S\ref{sec:RWscat}, also proving a few extra statements.

In \S\ref{sec:construct}, we then present the full details of the proof of Theorem~\ref{thm:setup:LEE Scattering wp}.

With this result having been proved under fairly general assumptions and clarifying the role of seed data in the scattering problem, we will then, in \S\ref{sec:physd}, write down a few sets of explicit assumptions on a seed scattering data set such that it can be said to describe the exterior of a system of $N$ infalling bodies with no incoming radiation from $\Scrim$. These physical seed data will be at the heart of the entire paper.

\section{The general scattering data setup for \texorpdfstring{\fullsystem}{linearised gravity}}\label{sec:setup}

We begin by fixing our notation:
We denote by $\Cin$ the incoming null hypersurface given by
\begin{equation*}
\Cin:=\{(u,v,\theta,\varphi)\in\SchwEx \,|\, v=v_1, u\leq u_0<0\}.
\end{equation*}
We further denote $\Sone:=\{(u,v,\theta,\varphi)\in\Schw \,|\, v=v_1, u= u_0\}$, and we similarly denote the limiting sphere $\Sinfty:=\{(u,v,\theta,\varphi)\in{\SchwEx} \,|\, v=v_1, u= -\infty\}$.
Finally, we denote the part of~$\Scrim$ that lies to the future of $\Sinfty$ by $\Scrim_{v\geq v_1}$. See Fig.~\ref{fig:setup:Penrosediagram}.

We will denote the future domain of dependence $D^+(\Cin\cup\Scrimv)$ of $\Cin\cup\Scrimv$ by $\DoD$. 
\begin{figure}[htpb]
\includegraphics[width=180pt]{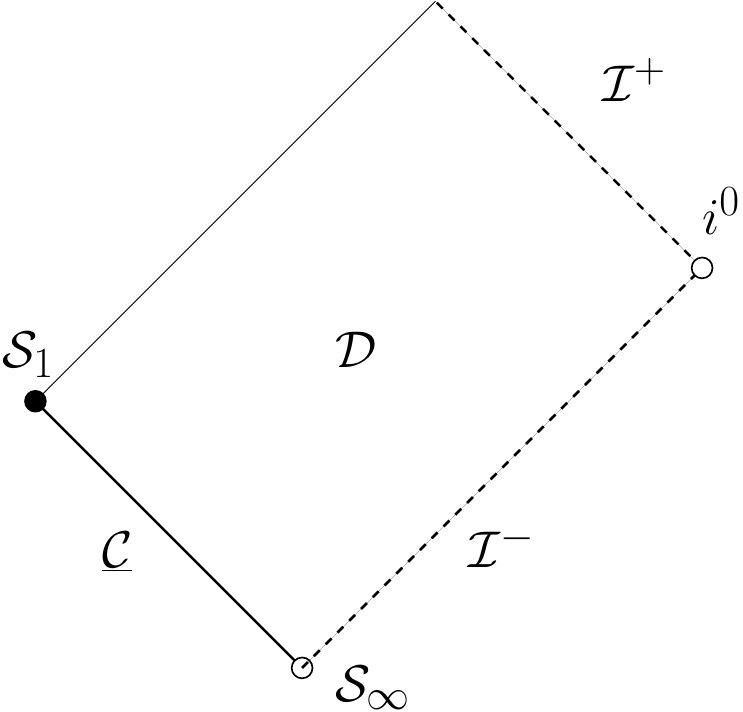} 
\caption{We will pose our scattering data on $\Cin$ and the part of $\Scrim$ that lies to the future of $\Cin$ (i.e.~$\Scrim_{v\geq v_1}$).}\label{fig:setup:Penrosediagram}
\end{figure}

\subsection{Seed scattering data and the main theorem (Thm.~\ref{thm:setup:LEE Scattering wp})}\label{sec:setup:basics}

\begin{defi}\label{def:setup:scattering data}
    A \emph{smooth seed scattering data set} $\mathfrak{D}$ is a {nonuple}
    \begin{align}
        \left(\,\rad{\Om}{\Scrim},\,\,\rad{\xh}{\Scrim},\,\,\rad{\b}{\Scrim},\,\,\rad{\gsh}{\Cin}, \,\,\rad{\omb}{\Cin},\,\,\rad{\trx}{\Sone},\,\,\rad{\trxb}{\Sone},\,\,\rad{\beb}{\Sone}, \,\,\rad{\trg}{\Sone}\, \right),
    \end{align}
    where $\rad{\Om}{\Scrim}$, $\rad{\b}{\Scrim}$ and $\rad{\xh}{\Scrim}$ are specified on $\Scrimv$:
    \begin{align*}
      \rad{\Om}{\Scrim}\in C^{\infty}\left(\Scrimv\right),\qquad \rad{\b}{\Scrim}\in \customfieldoneforms{\Scrimv},\qquad \rad{\xh}{\Scrim}\in\customfieldstf{\Scrimv},
    \end{align*}
    where $ \rad{\omb}{\Cin}$, $\rad{\gsh}{\Cin}$ are specified along $\Cin$:
    \begin{align*}
        \rad{\omb}{\Cin}\in C^\infty(\Cin),\qquad\rad{\gsh}{\Cin}\in\customfieldstf{\Cin},\qquad 
    \end{align*}
    and, lastly, where $\rad{\trx}{\Sone}$, $ \rad{\trxb}{\Sone}$, $\rad{\trg}{\Sone}$ as well as $\rad{\beb}{\Sone}$ are specified on $\Sone$:
    \begin{align*}
        \rad{\trx}{\Sone},\,\, \rad{\trxb}{\Sone},\,\, \rad{\trg}{\Sone}\in C^\infty(\Stwo) ,\qquad \rad{\beb}{\Sone}\in \Gamma^\infty\left(T^{(0,1)}\Stwo\right).  
    \end{align*}
\end{defi}
\begin{defi}\label{def:setup:scattering solution}
    Given a smooth seed scattering data set $\mathfrak{D}$, we call a solution $\mathfrak{S}$ to \fullsystem~on $\DoD=D^+(\Cin\cup\Scrim)$ \emph{the scattering solution realising $\mathfrak{D}$} if 
    \begin{itemize}
    \item $\Om$, $r^{-1}\b$ and $r\xh$ realise $\radi\Om$, $\radi\b$ and $\radi\xh$ as their pointwise limits at $\Scrim$.
    \item $\omb$, $\gsh$ restrict to $\rad{\omb}{\Cin}$, $\rad{\gsh}{\Cin}$ on $\Cin$, 
    \item  $\trx$, $\trxb$, $\trg$ and $\be$ restrict to $\rad{\trx}{\Sone}$, $\rad{\trxb}{\Sone}$, $\rad{\trg}{\Sone}$ and $\rad{\be}{\Sone}$ on $\Sone$.
    \end{itemize}
\end{defi}

Over the course of the next  few pages, we will prove the following theorem:

\begin{thm}\label{thm:setup:LEE Scattering wp}
    Let $\mathfrak{D}$ be a smooth scattering data set on $\Cin\cup\Scrim$ as in Def.~\ref{def:setup:scattering data}.
    Suppose there exist positive numbers $\epsilon,\delta\in\mathbb R_{>0}$ as well as a smooth stf two-tensor $\rad{\gsh}{\Sinfty}$ such that, as $u\to-\infty$, the components of $\mathfrak{D}$ along $\Cin$ satisfy
\begin{align}\label{eq:setup:general decay along C}
\lim_{u\to-\infty}\rad{\gsh}{\Cin}=\rad{\gsh}{\Sinfty},&&  \rad{\gsh}{\Cin}-\rad{\gsh}{\Sinfty}=\O_{\infty}\left(r^{-\frac{1}{2}-\delta}\right),&&
    \rad{\omb}{\Cin}=\O_{\infty}\left(r^{-1-\epsilon}\right).
\end{align}
Then there exists a unique scattering solution $\mathfrak{S}$  realising $\mathfrak{D}$ in the sense of Def.~\ref{def:setup:scattering solution}, the uniqueness being understood w.r.t.~the class of solutions with finite Regge--Wheeler energy (cf.~Prop.~\ref{prop:construct:LEE uniqueness}).
\end{thm}
\begin{rem}
We already point out that a large part of the seed scattering data does not carry physical information. For instance, by addition of pure gauge solutions (cf.~\S\ref{sec:gauge}), the quantities $\rads\trg$, $\rads\trxb$,  $\rad{\gsh}{\Sinfty}$, $\rad{\omb}{\Cin}$, $\rad{\Om}{\Scrim}$ as well as $\rad{\b}{\Scrim}$ can all simultaneously be set to 0.
This will be discussed in detail in \S\ref{sec:setup:Bondinor}.
\end{rem}
\begin{rem}
The assumed decay rate on $\radc\gsh$ ensures that the seed scattering data set induces finite Regge--Wheeler energy along $\Cin$ and cannot be improved without  using different methods. This decay rate is also what corresponds to the decay assumed in \cite{bieri}. 

The assumed decay rate on $\radc\omb$, on the other hand, can be weakened due to the linearity of the system: 
Since this is not relevant for applications, we will stick with assuming $\omb$ to be integrable.
\end{rem}

\subsection{Overview over the proof of Thm.~\ref{thm:setup:LEE Scattering wp}}\label{sec:setup:overview}
Our proof of Theorem~\ref{thm:setup:LEE Scattering wp} will be divided into the following steps. 
\begin{enumerate}[label=(\roman*)]
    \item Let $\mathfrak{D}$ be as in Thm.~\ref{thm:setup:LEE Scattering wp}. 
    From the components of $\mathfrak{D}$ on $\Cin$, $\Sone$ and $\Scrim$, we uniquely construct data for all the remaining components of the system \fullsystem~($\rad{\xhb}{\Cin}$, $\rad{\alb}{\Cin}$, $\radc{\beb}$, $\radc{\Ps}$ etc.) such that any scattering solution realising $\mathfrak{D}$ must restrict to these constructed components along $\Cin$. This is done in \S\ref{sec:setup:2.i} (Prop.~\ref{prop:setup:uniqueness of data on Cin}).\label{step:setup:overviewC}
    \item We use a uniqueness clause for the Teukolsky equation (Cor.~\ref{cor:RWscat:uniqueness alpha}) to deduce the uniqueness  of a solution realising $\mathfrak D$ in \S\ref{sec:construct:unique} (Prop.~\ref{prop:construct:LEE uniqueness}).
    This uniqueness clause for the Teukolsky equation, in turn, is derived in detail from uniqueness  for the Regge--Wheeler equation in \S\ref{sec:RWscat}. 
    \label{step:setup:overviewUnique}
    \item Not having to worry about uniqueness anymore, we  initiate the construction of a solution in \S\ref{sec:construct:existl>1}. First, we define data at $\Scrim$ for certain components of a solution to \fullsystem. In particular, we define data for $\Dv\Ps$ at $\Scrim$. Combining this with \ref{step:setup:overviewC}, we now have scattering data for the Regge--Wheeler equation \eqref{eq:lin:RW}.
    We can thus construct a unique solution $\Ps$ to \eqref{eq:lin:RW} realising these scattering data using the results of \cite{Masaood22}  (see Thm.~\ref{thm:RWscat:mas20 RW}), of which we give a fully self-contained presentation in \S\ref{sec:RWscat}.
\label{step:setup:ScatRW}
    \item We construct out of this solution $\Ps$ and the already defined data along $\Cin$ and $\Scrim$ first $\al$ by integrating \eqref{eq:lin:transformations}, and then the entire $\ell\geq2$-part of a solution $\mathfrak{S}$ to \fullsystem~by suitably integrating the remaining equations of \fullsystem~to obtain the remaining components of a solution $\mathfrak{S}$. 
    We finish by explicitly writing down the $\ell=0,1$-part of $\mathfrak{S}$ (which is independent of $\Ps$).
\label{step:setup:Final}
\end{enumerate}
\begin{rem}
The step to construct the entire $\ell\geq2$-part of $\mathfrak{S}$ starting from $\Ps$ is a combination of two ingredients: (1) Deriving local existence for \fullsystem~starting from local existence for $\Psb$. (2) A certain level of a priori quantitative decay analysis as the data are posed at infinity. For instance, when defining $\al$ by integrating $\Ps$ twice from $\Scrim$ using \eqref{eq:lin:transformations}, we need to first prove that $\Ps$ decays sufficiently fast so that \eqref{eq:lin:transformations} is integrable. 

Since the local existence proof is quite lengthy, we also give a short overview of it in \S\ref{sec:construct:existl>1}. 
\end{rem}

\subsection{Defining the full scattering data along \texorpdfstring{$\Cin$}{C}}\label{sec:setup:2.i}

Given a seed scattering data set according to Definition \ref{def:setup:scattering data}, we may uniquely prescribe data for the full system on $\Cin$ by suitably solving the constraint equations on each of $\Cin$:

\begin{prop}\label{prop:setup:uniqueness of data on Cin}
Let $\mathfrak{D}$ be a seed data set satisfying the assumption of Thm.~\ref{thm:setup:LEE Scattering wp}.
Then $\mathfrak{D}$ defines a unique tuple of functions, $\mathcal S_{u,v}$-tangent one-forms and stf two-tensor fields along $\Cin$
\begin{align}
    \mathfrak{S_{\Cin}}=\( \rad{\gsh}{\Cin},\, \rad{\trg}{\Cin},\, \rad{\Om}{\Cin},\, \rad{\b}{\Cin},\,\rad{\trx}{\Cin},\, \rad{\trxb}{\Cin},\,\rad{\xh}{\Cin},\,\rad{\xhb}{\Cin},\,\rad{\et}{\Cin},\, \rad{\etb}{\Cin},\, \rad{\omb}{\Cin},\,\rad{\be}{\Cin},\,\rad{\rh}{\Cin},\,\rad{\sig}{\Cin},\,\rad{\beb}{\Cin},\,\rad{\alb}{\Cin},\,\rad{\K}{\Cin}\)
\end{align}
such that if $\mathfrak{S}$ is  scattering solution realising $\mathfrak{D}$ as its seed scattering data set according to Definition \ref{def:setup:scattering solution}, the restriction of $\mathfrak{S}$ to $\Cin$ gives $\mathfrak{S_{\Cin}}$. 
In particular, all $\Du$-equations of \fullsystem~except for \eqref{eq:lin:om3} and \eqref{eq:lin:al3}, as well as the elliptic equations \eqref{eq:lin:OmmA}, \eqref{eq:lin:divxh}--\eqref{eq:lin:K} are satisfied along $\Cin$ by  $\mathfrak{S_{\Cin}}$. 

If $\mathfrak{S}$ is such that $r\al$ and $\om$ converge to $-\Dv\radi\xh$, $\pv\radi\Om$ as $u\to-\infty$, respectively, then there additionally exist unique $\radc\al$ and $\radc\om$ such that the restriction of $\mathfrak{S}$ to $\Cin$ also gives $\radc\al$ and $\radc\om$. In particular, the equations \eqref{eq:lin:om3} and \eqref{eq:lin:al3} are satisfied by $\mathfrak{S}_{\Cin}\cup \( \radc \om,\radc \al\)$  along~$\Cin$.
\end{prop}

\begin{proof}
The proof proceeds by systematically defining the elements of  $\mathfrak{S_{\Cin}}$ from $\mathfrak{D}$ as solutions to a subset of the relevant equations of \fullsystem~along $\Cin$, and by a posteriori proving that the remaining of the $\Du$- and elliptic equations of \fullsystem~are automatically satisfied as well.
The uniqueness clause is addressed at the end.
\begin{enumerate}[leftmargin=*,label=\arabic*)]
\item  We define $\rad{\Om}{\Cin}$ along $\Cin$ via $\rad{\omb}\Cin$ by integrating  \eqref{eq:lin:Omm3} from $\Sinfty$ with $\rad{\Om}{\Scrim}|_{\Sinfty}$ as data. (Condition~\eqref{eq:setup:general decay along C} ensures that this is integrable.)
\item We define $\rad{\trxb}{\Cin}$ along $\Cin$ via $\rad{\omb}{\Cin}$ by integrating \eqref{eq:lin:trxb3} from $\Sone$ with $\rad{\trxb}{\Sone}$ as data. 
\item We then define $\rad{\trg}{\Cin}$ along $\Cin$ by integrating \eqref{eq:lin:trg3} with data $\rad{\trg}{\Sone}$ on $\Sone$. 
\item Next, we define $\rad{\xhb}{\Cin}$ along $\Cin$ from $\rad{\gsh}{\Cin}$ via \eqref{eq:lin:gsh3}, and we similarly define $\rad{\alb}{\Cin}$  from $\rad{\xhb}{\Cin}$ via~\eqref{eq:lin:xhb3}.
\item Given $\rad{\alb}{\Cin}$, we define $\rad{\beb}{\Cin}$ along $\Cin$ by integrating \eqref{eq:lin:beb3} with data $\rad{\beb}{\Sone}$ at $\Sone$. 
\item Given $\rad{\xhb}{\Cin}$, $\rad{\beb}{\Cin}$ as well as $\rad{\trxb}{\Cin}$, we now define $\rad{\et}{\Cin}$ along $\Cin$ as solution to the Codazzi equation \eqref{eq:lin:divxhb}. By multiplying \eqref{eq:lin:divxhb} with $\frac{r^2}{\Omega}$ and acting with $\Du$, we directly verify that~\eqref{eq:lin:et3} holds along $\Cin$. Indeed, using all our previous definitions:
\begin{equation}
\Du(r^2\rad{\et}{\Cin})=\div\Du(\tfrac{r^2}{\Omega}\rad{\xhb}{\Cin})-\frac{1}{r}\Du(\tfrac{r^4}{\Omega}\rad{\beb}{\Cin})-r^2\Omega \rad{\beb}{\Cin}-\sl\Du(\tfrac{r^2}{2\Omega^2}\rad{\trxb}{\Cin})
=-r^2\Omega \rad{\beb}{\Cin}+2r\sl\rad{\omb}{\Cin}.
\end{equation}
\item We now define $\rad{\etb}{\Cin}$ along $\Cin$ via \eqref{eq:lin:OmmA} and our definitions of $\rad{\Om}{\Cin}$ and $\rad{\et}{\Cin}$. This directly implies that $\curl(\rad{\et}{\Cin}+\rad{\etb}{\Cin})=0$.
Moreover,  inserting this definition into \eqref{eq:lin:et3} further shows that~\eqref{eq:lin:etb3} holds along $\Cin$. 
\item Next, we define $\rad{\sig}{\Cin}$ along $\Cin$  as solution to~\eqref{eq:lin:curleta}: $r\rad{\sig}{\Cin}:=\curl\rad{\et}{\Cin}(=-\curl\rad{\et}{\Cin})$, where the second equality follows by construction.
Acting on this equation with $\Du(r^2\cdot)$ and using \eqref{eq:lin:et4} proves that \eqref{eq:lin:sig3} holds along $\Cin$ (note that $\curl\sl f=0$ for any smooth function~$f$). 
\item We define the Gaussian curvature $\rad{\K}{\Cin}$ via \eqref{eq:lin:Kdef}.
\item We now need to define $\rad{\rh}{\Cin}$: First, define $\rad{\rh}{\Sone}$ on $\Sone$ as solution to the Gauss equation \eqref{eq:lin:K}. Then, define $\rad{\rh}{\Cin}$ along $\Cin$ by integrating~\eqref{eq:lin:rh3} from $\Sone$ with data~$\rad{\rh}{\Sone}$. 
\item By construction, \eqref{eq:lin:K} holds on $\Sone$. We now define $\rad{\trx}{\Cin}$ along $\Cin$ as solution to \eqref{eq:lin:K} along $\Cin$.
Multiplying \eqref{eq:lin:K} by $r^2$ and acting with $\Du$ on it then shows, utilizing all the previous definitions, that \eqref{eq:lin:trx3} holds along $\Cin$.

At this point, it is left to define $\rad{\b}{\Cin}$, $\rad{\xh}{\Cin}$, $\rad{\be}{\Cin}$ (as well as $\rad{\al}{\Cin}$ and $\radc\om$).
In order to define these quantities, we need to integrate their relevant equations from $\Scrim$---we thus require that these equations are integrable. Let us therefore list a few of the decay rates as $u\to-\infty$ of the  quantities obtained so far:
\begin{itemize}[leftmargin=*]
   \item[$\diamond$] As $\rad{\omb}{\Cin}=\O_{\infty}(r^{-1-\epsilon})$ according to \eqref{eq:setup:general decay along C}, we find from \eqref{eq:lin:trxb3} that $r^2\radc{\trxb}=\O_{\infty}(r^{1-\epsilon}(1+\delta_{\epsilon,1}\log r)+1)$. 
   In particular, this decay rate together with \eqref{eq:lin:trg3} implies that $\radc{\trg}$ converges as $u\to-\infty$, and thus, by definition and by \eqref{eq:setup:general decay along C}, $r^2\radc\K$ also converges. We write:
   \begin{align}\label{eq:setup:rad trg at Sinfty}
\begin{split}
    \rad{\trg}{\Sinfty}:=&\lim_{u\to-\infty}\radc{\trg}=\rad{\trg}{\Sone}{-}\frac{{2}r}{\Omega^2}\Big|_{\Sone}\rad{\trxb}{\Sone}{+\lim_{u\to-\infty}}\left(\frac{8}{r}\int_{u}^{u_0} r\rad{\omb}{\Cin}\dd u'\right){-}8\int_{-\infty}^{u_0} \rad{\omb}{\Cin} \dd u'.
\end{split}\\
    \Ks_{\Sinfty}:=&\lim_{u\to-\infty} r^2\K(u,v_1,\theta^A)=\frac{1}{4}\left(\lap+2\right)\rad{\trg}{\Sinfty}+\frac{1}{2}\div\div\rad{\gsh}{\Sinfty}. \label{eq:setup:rad K at Sinfty}
\end{align}

\item[$\diamond$] The assumption \eqref{eq:setup:general decay along C} gives  $\radc{\xhb}=\O_{\infty}(r^{-\frac{3}{2}-\delta})$ via \eqref{eq:lin:gsh3} and thus $\radc{\alb}=\O_{\infty}(r^{-\frac{5}{2}-\delta})$ via~\eqref{eq:lin:xhb3}.
  \item[$\diamond$]We estimate $\radc{\beb}$ and $\radc{\sig}$ via \eqref{eq:lin:beb3} and \eqref{eq:lin:sig3} to get $r^4\radc{\beb}=\O_{\infty}\(r^{\frac32-\delta}(1+\delta_{\frac32,\delta}\log r)+1\)$, $\radc{\sig}=\O_{\infty}\(r^{\frac1{2}-\delta}(1+\delta_{\frac32,\delta}\log r+\delta_{\frac12,\delta}\log r)+1\)$.
\item[$\diamond$] Inserting the rates above into the linearised Codazzi equation \eqref{eq:lin:divxhb} implies the estimate $r^2\radc{\et}=\O_{\infty}\(r^{1-\min\left(\delta+\frac{1}{2},\epsilon\right)}(1+\log r(\delta_{\frac32,\delta}+\delta_{1,\epsilon}))+1\)$.
\item[$\diamond$] Eq.~\eqref{eq:lin:OmmA} now gives $\radc{\etb}=\frac{(2\sl\,\rad{\Om}{\Scrim}|_{\Sinfty})}{r}+\O_{\infty}\(r^{-1-\min\left(\delta+\frac{1}{2},\epsilon\right)}(1+\log r(\delta_{\frac32,\delta}+\delta_{1,\epsilon}))+r^{-2}\)$. We write
\begin{equation}\label{eq:setup:rad etb at Sinfty}
\radsinf{\etb}:=\lim_{u\to-\infty}r\radc{\etb}=2\sl\radi{\Om}|_{\Sinfty}
\end{equation}
\item[$\diamond$] Equation \eqref{eq:lin:rh3} implies $r^3\radc{\rh}=\O\(r^{\frac12-\min(\delta,\frac{1}{2}+\epsilon)}(1+\log r(\delta_{\frac12,\delta}+\delta_{\frac32,\delta}+\delta_{\epsilon,1}))+1\)$. It follows via \eqref{eq:lin:K} that $r\radc{\trx}$ attains a limit as $u\to-\infty$:
\begin{equation}\label{eq:setup:rad trx at Sinfty}
\radsinf{\trx}:=\lim_{u\to-\infty}r\radc{\trx}=2\radsinf{\K}+4\radi{\Om}|_{\Sinfty}.
\end{equation}
\end{itemize}

We now define the remaining quantities along $\Cin$:

\item By the decay rates above, the RHS of \eqref{eq:lin:xh3} is integrable. We thus define $\radc{\xh}$ along $\Cin$ via integration of \eqref{eq:lin:xh3} with $\rad{\xh}{\Scrim}|_{\Sinfty}$ as data for $\Omega r\cdot \radc{\xh}$ at $\Sinfty$. By construction, angular derivatives of $r\radc\xh$ will then converge to angular derivatives of $\rad{\xh}{\Scrim}|_{\Sinfty}$.
\item Similarly, we  define $\radc{\b}$ along $\Cin$ via integration of \eqref{eq:lin:b3} with data  $\rad{\b}{\Scrim}|_{\Sinfty}$ for $r^{-1}\radc{\b}$.
\item Next, we define $\rad{\be}{\Cin}$ along $\Cin$ as the solution to the Codazzi equation \eqref{eq:lin:divxh}. It follows that $r^2\radc{\be}$ attains a limit as $u\to-\infty$:
\begin{equation}\label{eq:setup:rad be at Sinfty}
    \rad{\be}{\Sinfty}:=\lim_{u\to-\infty}r^2\radc{\be}=-\div \rad{\xh}{\Scrim}|_{\Sinfty}-\radsinf{\etb}+\frac12\sl\radsinf{\trx}=-\div \rad{\xh}{\Scrim}|_{\Sinfty}+\sl \rad{\K}{\Sinfty}.
\end{equation}
Multiplying \eqref{eq:lin:divxh} with $r\Omega$ and acting with $\Du$, we then deduce that \eqref{eq:lin:be3} holds along all of $\Cin$.
\item We define $\radc{\om}$ along $\Cin$ via integration of \eqref{eq:lin:om3} with $\pv\rad{\Om}{\Scrim}|_{\Sinfty}$ as data for $\radc{\om}$ at $\Sinfty$.
\item Finally, we define $\radc{\al}$ along $\Cin$ by integrating \eqref{eq:lin:al3} with data $-\Dv\radi{\xh}|_{\Sinfty}$ for $r\Omega^2\radc{\al}$ at~$\Sinfty$.
\end{enumerate}

This concludes the construction. The uniqueness part (for all quantities except $\radc\al$ and $\radc\om$) goes as follows: If there are two different scattering solutions realising $\mathfrak{D}$, then their difference restricts to another solution with trivial seed scattering data. By demanding that this solution satisfies the equations \fullsystem, we can then repeat the procedure above but for trivial seed scattering data to deduce that all other quantities along $\Cin$ must necessarily vanish identically.

The uniqueness of $\radc\al$ and $\radc\om$  follows similarly under the additional convergence assumption of the proposition.\end{proof}

\begin{rem}
We see from the construction of the other quantities that, instead of specifying $r\trx$ and $\trg$ on $\mathcal S_1$, we can equivalently specify these quantities on $\mathcal S_{\infty}$.
\end{rem}
\begin{cor}\label{cor:setup:asym on Cin towards Scrim}
The components of $\mathfrak{S}_\Cin$ constructed in Proposition \ref{prop:setup:uniqueness of data on Cin} satisfy, in particular, 
\begin{nalign}\label{eq:setup:asymptoticsalongCINforxhbbebetc}
\radc{\xhb}&=\O_{\infty}(r^{-\frac32-\delta}),\qquad\qquad\qquad\qquad\qquad\radc{\alb}=\O_{\infty}(r^{-\frac52-\delta}),\\
r^4\radc\beb&=\O_{\infty}(r^{\frac32-\delta}(1+\log r\delta_{3/2,\delta})+1),\qquad r^3\radc\sig=\O_{\infty}(r^{\frac12-\delta}(1+\log r(\delta_{3/2,\delta}+\delta_{1/2,\delta}))+1),\\
r^3\radc{\rh}&=\O_{\infty}(r^{\frac12-\min(\delta,\epsilon+\frac12)}(1+\log r(\delta_{3/2,\delta}+\delta_{1/2,\delta}+\delta_{\epsilon,1}))+1).&&
\end{nalign}
\end{cor}

Proposition \ref{prop:setup:uniqueness of data on Cin} allows us to also define transversal derivatives along $\Cin$ to any order. In particular, we have:

\begin{cor}\label{cor:setup:psialongC}
Let $\mathfrak{D}$ be a seed data set satisfying the assumption of Thm.~\ref{thm:setup:LEE Scattering wp}, and let $\mathfrak{S}_{\Cin}$ be as in Prop.~\ref{prop:setup:uniqueness of data on Cin}.
Define the following tuple of $\mathcal S_{u,v}$-tangent stf-tensors along $\Cin$:
\begin{nalign}
\radc{\psb}&:=2\Ds2 r^2\Omega \radc{\beb}+6M\Omega\radc{\xhb},\\\radc{\ps}&:=-2\Ds2 r^2\Omega\radc{\be}+6M\Omega\radc{\xh}\\
\radc{\Psb}&:=2\Ds2\Ds1(r^3\radc{\rh},r^3\radc{\sig})+6M(r\Omega\radc{\xh}-r\Omega\radc{\xhb}),\\ \radc{\Ps}&:=2\Ds2\Ds1(r^3\radc{\rh},-r^3\radc{\sig})+6M(r\Omega\radc{\xh}-r\Omega\radc{\xhb})
\end{nalign}
Then any smooth scattering solution realising $\mathfrak{D}$ will have $\tfrac{r^2}{\Omega^2}\Dv{r\Omega^2\al}$, $(\tfrac{r^2}{\Omega^2}\Dv)^2{r\Omega^2\alb}$, $\tfrac{r^2}{\Omega^2}\Du(r\Omega^2\alb)$ and $(\tfrac{r^2}{\Omega^2}\Du)^2(r\Omega^2\al)$ restricting to $\radc{\psb}$, $\radc{\Psb}$, $\radc{\ps}$ and $\radc{\Ps}$, respectively.

Moreover, we have that
    \begin{align}\label{eq:setup:asym Ps Psb on Cin towards Scrim}
        \rad{\Ps}{\Cin}=\O_{\infty} (r^{\frac{1}{2}-\delta}(1+\delta_{1/2,\delta}\log r)+1)= \rad{\Psb}{\Cin},&&r^2\radc\psb=\O_{\infty}(r^{\frac32-\delta}(1+\delta_{3/2,\delta}\log r)+1).
    \end{align}
\end{cor}
\begin{proof}
The first part of the statement follows directly from \eqref{eq:lin:al3_v2}--\eqref{eq:lin:alb44} and Proposition~\ref{prop:setup:uniqueness of data on Cin}.

For \eqref{eq:setup:asym Ps Psb on Cin towards Scrim}, however, the rates derived in Proposition~\ref{prop:setup:uniqueness of data on Cin} would naively only give (ignoring $\log$-terms) 
  \begin{align}
       \rad{\Ps}{\Cin}=\O  \(r^{\frac{1}{2}-\min\(\delta,\frac{1}{2}+\epsilon\)}+1\)= \rad{\Psb}{\Cin}.
    \end{align}
To show that these rates are in fact independent of $\epsilon$, one first confirms that \eqref{eq:lin:Psb3} is satisfied along $\Cin$, and then integrates this equation from $\Sone$ together with the knowledge that $\radc\alb=\O_{\infty}(r^{-\frac52-\delta})$. 
This shows \eqref{eq:setup:asym Ps Psb on Cin towards Scrim} for $\radc\Psb$. The result for $\radc\Ps$ then follows from \eqref{eq:lin:Ps-Psb=sig} and the rate~\eqref{eq:setup:asymptoticsalongCINforxhbbebetc} for $\radc{\sig}$.
\end{proof}
Finally, we deduce a corollary pertaining to lower angular modes. This result will be used for the construction of the $\ell=0,1$-part of the solution in \S\ref{sec:construct:existl<2}.
\begin{cor}\label{cor:setup:l<2}
Let $\mathfrak{D}_{\ell=0,1}$ be a smooth seed scattering data set satisfying the assumptions of Thm.~\ref{thm:setup:LEE Scattering wp} that is supported only on $\ell=0,1$. Then, in addition to the limit $\radsinf\K$ defined in~\eqref{eq:setup:rad K at Sinfty}, the following limits exist as well:
\begin{align}
\radsinf\rh=\lim_{u\to-\infty}r^3\radc\rh,&&\radsinf\beb=\lim_{u\to-\infty}r^4\radc\beb,&&\lim_{u\to-\infty}r^4\radc\sig=-\curl\radsinf\beb.
\end{align}
Moreover, specifying  $\mathfrak{D}_{\ell=0,1}=\{\radi\Om,\radi\b, \radi\omb, \rads\trx, \rads\trxb, \rads\beb, \rads\trg\}$ is equivalent to specifying the corresponding $\mathfrak{D}'_{\ell=0,1}=\{\radi\Om,\radi\b, \radi\omb, \radsinf\rh, \rads\trxb, \radsinf\beb, \radsinf\trg\}$. (Recall that $\gsh$ and $\xh$ are supported on $\ell\geq 2$ since they are stf two-tensors.)
\end{cor}
\begin{proof}
Since $\mathfrak{D}_{\ell=0,1}$ is supported on $\ell=0,1$, we have that $\Du\frac{r^4\radc\beb}{\Omega}=0$, and so $r^4\radc\beb$ takes a limit. The claim for $\radc\rh$ is similar. Finally, in order to compute the limit of $r^4\radc\sig$, we write 
\begin{equation}
\lim_{u\to-\infty} r^4\radc\sig=\lim_{u\to-\infty}r^3\curl \radc\et=-\lim_{u\to-\infty}r^4\curl \radc\beb,
\end{equation}
where we used \eqref{eq:lin:curleta} and \eqref{eq:lin:divxhb}, together with the fact that $\curl\sl \radc\trxb=0=(\radc\xhb)_{\ell<2}$.

In order to see the equivalence of $\mathfrak{D}$ and $\mathfrak{D}'$, we note that by $\Du\frac{r^4\radc\beb}{\Omega}=0$, specifying $\radsinf\beb$ is equivalent to specifying $\rads\beb$. Similarly, \eqref{eq:lin:trg3} implies that specifying $\rads\trg$ is equivalent to specifying $\radsinf\trg$.
Finally, via \eqref{eq:lin:rh3}, specifying $\radsinf\rh$ is equivalent to specifying $\radc\rh|_{\Sone}$, which, in turn, is equivalent to specifying $\rads\trx$ by the Gauss equation \eqref{eq:lin:K}.
\end{proof}
\subsection{Bondi normalisation of seed scattering data}\label{sec:setup:Bondinor}
We now describe how to \textit{Bondi normalise} a seed scattering data set. This serves the purpose of identifying the physical degrees of freedom contained in a seed scattering data set.
\begin{defi}\label{def:setup:Bondi}
A smooth seed scattering data set $\mathfrak{D}$ is said to be \emph{Bondi-normalised} if the following conditions are satisfied:
\begin{itemize}
\item The lapse and the shift vector vanish at $\Scrim$: $\radi\Om=0=\radi\b.$
\item The limit $\radsinf\gsh$ of $\radc\gsh$, as well as the induced limit $\radsinf\K$ vanishes at $\Sinfty$ (cf.~\eqref{eq:setup:rad K at Sinfty}). (This automatically forces $\radsinf\trg$ and $\radsinf\trx$ to vanishes as well, cf.~\eqref{eq:setup:rad trg at Sinfty}, \eqref{eq:setup:rad trx at Sinfty}).
\end{itemize}
\end{defi}
Note that the vanishing of $\radsinf\trg$ and $\radsinf\trx$ is automatically enforced by the vanishing of the other quantities.
We have the following proposition:
\begin{prop}\label{prop:setup:Bondi}
For any smooth seed scattering data set $\mathfrak{D}$, there exist smooth functions $f(v,\theta^A)$, $\underline{f}(u,\theta^A)$ and $(q_1(v,\theta^A),q_2(v,\theta^A))$ such that the $\ell\geq 2$-part of $\mathfrak{D}-\mathfrak{D}_f-\mathfrak{D}_{\underline{f}}-\mathfrak{D}_{(q_1,q_2)}$ is Bondi-normalised, where $\mathfrak{D}_{f}$, $\mathfrak{D}_{\underline{f}}$ and $\mathfrak{D}_{(q_1,q_2)}$ denote the seed scattering data belonging to the pure gauge solutions $\mathfrak{S}_f$, $\mathfrak{S}_{\underline{f}}$ and $\mathfrak{S}_{(q_1,q_2)}$.
\end{prop}
An analogous statement is true for $\ell<2$. (In this case, one also has to subtract a linearised Schwarzschild and a linearised Kerr solution.) In fact, given seed scattering data supported on $\ell<2$, the explicit scattering solution to these data is written down in \S\ref{sec:construct:existl<2}.
\begin{proof}
We first find  $\underline{f}$ supported on $\ell\geq2$ by demanding that
\begin{equation}
\lim_{u\to-\infty} r^{-3}(\lap+2)\underline f =(\radsinf\K)_{\ell\geq 2}.
\end{equation}
(Recall that $\lap+2$ is invertible on $\ell\geq2$). This still leaves us with considerable freedom for $\underline f$.

Next, we find $f$ such that
\begin{equation}
\pv f=(2\radi\Om)_{\ell\geq 2}+(\lap+2)^{-1}(\radsinf\K)_{\ell\geq 2}.
\end{equation}
Notice that $\mathfrak{D}_{\ell\geq 2}-\mathfrak{D}_f-\mathfrak{D}_{\underline{f}}$ now has vanishing Gaussian curvature~$\K$ at~$\Sinfty$ and vanishing lapse~$\Om$ at~$\Scrim$. It follows that the corresponding $\radsinf\trx$ vanishes as well (by \eqref{eq:setup:rad trx at Sinfty}).

Finally, we fix $(q_1,q_2)$ such that
\begin{align}
\radsinf\gsh=\left.2\Ds2\Ds1(q_1,q_2)\right|_{v=v_1},&& \Ds2\radi\b=\Ds2\Ds1(\pv q_1,\pv q_2).
\end{align}
It follows that $\mathfrak{D}_{\ell\geq 2}-\mathfrak{D}_f-\mathfrak{D}_{\underline{f}}-\mathfrak{D}_{(q_1,q_2)}$ has vanishing limits for the corresponding~$\radsinf\gsh$ and~$\radi\b$. 
The vanishing of $\radsinf\trg$ follows by \eqref{eq:setup:rad K at Sinfty}.
\end{proof}
\begin{rem}\label{rem:setup:bondi}
The proof above shows that even after Bondi-normalising the seed scattering data, we still have a considerable amount of gauge freedom left: Changing $f$ by a function independent of $v$ or changing $\underline{f}$ by a function that grows slower than $u$ does not affect the Bondi--normalisation of a seed scattering data set.\footnote{Contained in this freedom is the ($\ell\geq2$-part of the) BMS group at $\Scrim$, see \cite{Masaood22b} for a more detailed discussion.}
\end{rem}
In particular, we can use this remaining gauge freedom to infer from the proof of Prop.~\ref{prop:setup:Bondi}:
\begin{cor}\label{cor:setup:Bondiplus}
For any smooth seed scattering data set $\mathfrak{D}$,  there exist smooth functions $f(v,\theta^A)$, $\underline{f}(u,\theta^A)$ and $(q_1(v,\theta^A),q_2(v,\theta^A))$ such that the $\ell\geq 2$-part of $\mathfrak{D}-\mathfrak{D}_f-\mathfrak{D}_{\underline{f}}-\mathfrak{D}_{(q_1,q_2)}$ is Bondi-normalised, and, in addition $(\radc{\omb})_{\ell\geq 2}=0$ and $\(\rads\trxb\)_{\ell\geq 2}=0$.
\end{cor}
Again, the analogous statement for $\ell<2$ is shown in \S\ref{sec:construct:existl<2}.

\newpage

\section{Scattering for the Regge--Wheeler equation and for the Teukolsky equation}\label{sec:RWscat}
The main ingredient for our construction of a scattering solution to \fullsystem~is the scattering theory for the Regge--Wheeler and Teukolsky equations developed in \cite{Masaood22}. 
This, in turn, relies heavily on the fact that the Regge--Wheeler equation admits a conserved energy. 

\subsection{Energy conservation for the Regge--Wheeler equation}

\begin{defi}\label{defi:RWSCAT:energy}
    For any $u, U_1, U_2, v, V_1, V_2$ with $U_1\leq U_2$, $V_1\leq V_2$, define
    \begin{align*}
         E_{u}[\Psi](V_1,V_2)&:=\int_{\C_{u}\cap\{v\in[V_1,V_2]\}}\dd\bar{v}\dw\,\left[|\Dv\Psi|^2+\frac{\Omega^2}{r^2}\(|\sl\Psi|^2+(3\Omega^2+1)|\Psi|^2\)\right],\\
         \underline{E}_{v}[\Psi](U_1,U_2)&:=\int_{\underline{\C}_{v}\cap\{u\in[U_1,U_2]\}}\dd\bar{u}\dw\,\left[|\Du\Psi|^2+\frac{\Omega^2}{r^2}\(|\sl\Psi|^2+(3\Omega^2+1)|\Psi|^2\)\right].
    \end{align*}
    Define furthermore 
    \begin{align*}
     E_{u}[\sl\Psi](V_1,V_2)&:=\int_{\C_{u}\cap\{v\in[V_1,V_2]\}}\dd\bar{v}\dw\,\left[|\Dv\sl\Psi|^2+\frac{\Omega^2}{r^2}\(|\lap\Psi|^2+(3\Omega^2+1)|\sl\Psi|^2\)\right],\\
         \underline{E}_{v}[\sl\Psi](U_1,U_2)&:=\int_{\underline{\C}_{v}\cap\{u\in[U_1,U_2]\}}\dd\bar{u}\dw\,\left[|\Du\sl\Psi|^2+\frac{\Omega^2}{r^2}\(|\lap\Psi|^2+(3\Omega^2+1)|\sl\Psi|^2\)\right].
    \end{align*}

   Finally, define, for $n\in\mathbb{N}$,
    \begin{align*}
        E^{(n)}_{u}[\Psi](V_1,V_2):=\begin{cases}
        E_u[\lap^{\frac{n}{2}}\Psi](V_1,V_2),& \text{if  } \frac{n}{2}\in\mathbb N,\\
         E_u[\sl(\lap^{\frac{n-1}{2}}\Psi)](V_1,V_2),& \text{if  } \frac{n-1}{2}\in\mathbb N,
        \end{cases}
    \end{align*}
 as well as
     \begin{align*}
        \underline{E}^{(n)}_{v}[\Psi](U_1,U_2):=\begin{cases}
         \underline{E}^{(n)}_{v}[\lap^{\frac{n}{2}}\Psi](U_1,U_2),& \text{if  } \frac{n}{2}\in\mathbb N,\\
         \underline{E}^{(n)}_{v}[\sl(\lap^{\frac{n-1}{2}}\Psi)](U_1,U_2),& \text{if  } \frac{n-1}{2}\in\mathbb N.
        \end{cases}
    \end{align*}
\end{defi}

\begin{lemma}\label{lem:RWscat:energyidentity}
    Let $\Psi$ be a smooth solution to the Regge--Wheeler equation~\eqref{eq:lin:RW}, and let $n\in\mathbb N$. Then $\Psi$ satisfies
    \begin{align}
        E^{(n)}_{U_1}[\Psi](V_1,V_2)+\underline{E}^{(n)}_{V_1}[\Psi](U_1,U_2)=E^{(n)}_{U_2}[\Psi](V_1,V_2)+\underline{E}^{(n)}_{V_2}[\Psi](U_1,U_2).
    \end{align}
\end{lemma}
Later on, we will also make use of the following 
\begin{lemma}\label{lem:RWscat:uweighted}
    Let $\Psi$ be a smooth solution to \eqref{eq:lin:RW}, and define, for any $\delta'>0$:
    \begin{equation*}
        \underline{E}^{\delta'}_{v}[\Psi](U_1,U_2):=\int_{\Cbar_{v}\cap\{u\in[U_1,U_2]\}}\dd u\dd\omega\,|u|^{\delta'}\left[|\Du\Psi|^2+\frac{\Omega^2}{r^2}\(|\sl\Psi|^2+(3\Omega^2+1)|\Psi|^2\)\right]
    \end{equation*}
    as well as ${E}^{\delta'}_u[\Psi](V_1,V_2)=|u|^{\delta'}{E}_u[\Psi](V_1,V_2)$.
Then, if $U_2<0<\delta'$, 
   \begin{align}\label{eq:RWscat:uweightedestimate}
        E^{\delta'}_{U_1}[\Psi](V_1,V_2)+\underline{E}^{\delta'}_{V_1}[\Psi](U_1,U_2)\leq E^{\delta'}_{U_2}[\Psi](V_1,V_2)+\underline{E}^{\delta'}_{V_2}[\Psi](U_1,U_2).
    \end{align}
\end{lemma}
\begin{proof}
    Multiply \eqref{eq:lin:RW} by $|u|^{\delta'}\sll_T$; the arising bulk term will have a good sign in view of $U_2<0<\delta'$.
\end{proof}
\subsection{Scattering for the Regge--Wheeler equation}
We now give a compact, fully self-contained presentation of the scattering problem for the Regge--Wheeler equation adapted to our setting. Theorem~\ref{thm:RWscat:mas20 RW} is a statement adapted from \cite{Masaood22}, but it's proved in an alternative way; all other statements in \S\ref{sec:RWscat} are new.
\newcommand{\PsiCin}{\bm{\Uppsi}_{\Cin}}
\newcommand{\PP}{\mathtt{P}_{\Scrim}}
\begin{thm}[Adapted from~\cite{Masaood22}]\label{thm:RWscat:mas20 RW}
For $\PP\in L^2(\Scrim\cap\{v\in[v_1,v_2]\})$ and $\PsiCin$ with $E_{v_1}[\PsiCin](-\infty,u_0)<\infty$, there exists a unique solution $\Psi$ to the Regge--Wheeler equation~\eqref{eq:lin:RW} such that for any $v\in[v_1,v_2]$ we have
\begin{align}
    &\lim_{u\to-\infty}  \left\|\Dv \Psi(u,\cdot,\cdot)-\PP\right\|_{L^2([v_1,v]\times \Stwo)}^2=0,  \qquad \Psi|_{\mathcal{C}}=\PsiCin.
\end{align}
This solution satisfies
\begin{equation}\label{eq:RWscat:thm:limitofenergy}
\lim_{u\to-\infty}E_u[\Psi](v_1,v)=  \|\PP\|_{L^2([v_1,v]\times \Stwo)}^2.
\end{equation}

In addition, if $\PsiCin$ and $\PP$ are smooth, and if $E^{(n)}_{v_1}[\PsiCin](-\infty,u_0)<\infty$ for any $n\in\mathbb N$, then~$\Psi$ is smooth as well and
\begin{equation}\label{eq:RWscat:energyconvergencecommuted}
\lim_{u\to-\infty} E_u^{(n)}[\Psi](v_1,v)= \|\sl^s\lap^{\frac{n-s}{2}}\PP\|_{L^2([v_1,v]\times \Stwo)}^2,
\end{equation}
where $s=1$ if $n$ is odd and $s=0$ if $n$ is even.
\end{thm}

We first specify and prove the uniqueness clause of Thm.~\ref{thm:RWscat:mas20 RW}:

\begin{prop}\label{prop:RWscat:RW uniqueness local energy}
    Assume that $\Psi$ is a solution to the Regge--Wheeler equation \eqref{eq:lin:RW} such that 
    \begin{align}
        &\lim_{u\to-\infty}  \left\|\Dv \Psi(u,v,\cdot)\right\|_{L^2([v_1,v_2]\times \Stwo)}^2=0,\qquad \Psi|_{\mathcal{C}}=0.
    \end{align}
    Assume furthermore that $\Psi|_{\C_{u_0}\cap\{v\in[v_1,v_2]\}}\in H^1_{loc}(\C_{u_0})$, $\Psi|_{\underline{\C}_{v_2}\cap\{u\leq u_0\}}\in H^1_{loc}(\underline{\C}_{v_2}\cap\{u\leq u_0\})$. Then~$\Psi=0$.
\end{prop}
\begin{proof}
    Let $\widetilde{\Psi}\in H^2(\mathcal S_{u,v})$ such that $\lap\widetilde{\Psi}=\Psi$. Then $\widetilde{\Psi}$ also satisfies the Regge--Wheeler equation with $\widetilde{\Psi}|_{\Cin}=0$.
     Applying Hardy's inequality on $\C_{u}\cap\{v\in[v_1,v_2]\}$ for any $v_2\geq v_1$, $u\leq u_0$, together with Poincar\'e's inequality on $\Stwo$ gives
    \begin{align*}
        \int_{v_1}^{v_2}\int_{\Stwo}\dd \bar{v}\dw\, \frac{\Omega^2}{r^2}|\widetilde{\Psi}|^2\leq \int_{v_1}^{v_2}\int_{\Stwo}\dd \bar{v}\dw\, |\Dv\widetilde{\Psi}|^2\leq \int_{v_1}^{v_2}\int_{\Stwo}\dd \bar{v}\dw\, |\Dv{\Psi}|^2,
    \end{align*}
    \begin{align*}
        \int_{v_1}^{v_2}\int_{\Stwo}\dd \bar{v}\dw\, \frac{\Omega^2}{r^2}|\sl\widetilde{\Psi}|^2&\leq \int_{v_1}^{v_2}\int_{\Stwo}\dd \bar{v}\dw\, |\Dv\sl\widetilde{\Psi}|^2\\
        &\leq \int_{v_1}^{v_2}\int_{\Stwo}\dd \bar{v}\dw\, |\Dv{\Psi}|^2.
    \end{align*}
    Therefore, for any $v_2\geq v_1$ we have
    \begin{align*}
        \lim_{u\to-\infty}\int_{v_1}^{v_2}\int_{\Stwo}\dd\bar{v}\dw\,\left[|\Dv\widetilde{\Psi}|^2+\frac{\Omega^2}{r^2}|\widetilde{\Psi}|^2+\frac{\Omega^2}{r^2}|\sl\widetilde{\Psi}|^2\right]\Big|_{\C_{u}}=0.
    \end{align*}
    Energy conservation for Regge--Wheeler implies $\widetilde{\Psi}=0$, thus $\Psi=0$. (Here, we used that $\lap$ is invertible on stf $\mathcal S_{u,v}$-tensors, but the argument clearly also works for scalar functions by dealing with the spherically symmetric part in the same manner.)
\end{proof}
\begin{rem}
    In view of H\"ormander's propagation of singularities theorem applied to the Regge--Wheeler equation when expressed in the variables $(v,r^{-1})$, we expect that the above uniqueness statement extends beyond the class of solutions to Regge--Wheeler with finite global energy, e.g.~to the class of locally square-integrable solutions.
\end{rem}
\begin{proof}[Proof of the existence clause of Thm.~\ref{thm:RWscat:mas20 RW}]
    We provide an argument that avoids the use of the Arzel\`a--Ascoli theorem and instead directly constructs the solution.
    
    We begin by assuming that $\PP$, $\PsiCin$ are smooth and that $\PsiCin$ is compactly supported. 
    Let $\{u_n\}_{(n\in\mathbb N)}\subset \mathbb{R}$ be a dyadic sequence with $u_n\to-\infty$ as $n\to\infty$, and with $\{u\leq u_1\}$ beyond the support of $\PsiCin$. 
    We define a sequence $\{{\Psi_n}\}_{(n\in\mathbb N)}$ of smooth solutions to the Regge--Wheeler equation \eqref{eq:lin:RW}, where ${\Psi_n}$ is defined on $[u_n,u_0]\times[v_1,v_2]\times \Stwo$, by taking $\Psi_{n}$ to arise from characteristic data $\PsiCin(u,\theta^A)$ at $\Cin\cap\{u\in [u_n,u_0]\}$ and $\int_{v_1}^{v}\dd\bar{v}\,\PP(\bar{v},\theta^A)$ on $\C_{u_n}\cap\{v\in[v_1,v_2]\}$. (As part of the construction, we will also extend each $\Psi_n$ to $u\leq u_n$ later on.)
    Note that the energy norm of $\Psi_n$ on the initial outgoing cone $\C_{u_n}\cap\{v\in[v_1,v_2]\}$ is then bounded by:
    \begin{align}\label{eq:RWscat:energy estimate of data of Psi n at u n}
        E_{u_n}[\Psi_n](v_1,v_2)\leq \(1+\frac{(v_2-v_1)^2}{r(u_n,v_1)^2}\)\|\PP\|_{L^2([v_1,v_2]\times\Stwo)}^2+\frac{(v_2-v_1)^2}{r(u_n,v_1)^2}\|\sl\PP\|_{L^2([v_1,v_2]\times\Stwo)}^2.
    \end{align}
   
    The strategy now is to show that the restrictions of $\Psi_n$ to the null cones $\C_{u_0}\cap\{v\in[v_1,v_2]\}$ and $\Cbar_{v_2}\cap\{u\leq u_0\}$ converge as well, so that we can construct the limiting solution as the backwards solution arising from the limiting data $\Psi|_{\C_{u_0}\cap\{v\in[v_1,v_2]\}}$, $\Psi|_{\Cbar_{v_2}\cap\{u\leq u_0\}}$.
    
   First, we use the energy estimate to bound
    \begin{align}\label{eq:RWscat:energy of difference between u n and u 0}
      E_{u_0}[\Psi_n-\Psi_{n+1}](v_1,v_2)+ \underline{E}_{v_2}[\Psi_{n}-\Psi_{n+1}](u_n,u_0)\leq E_{u_n}[\Psi_{n}-\Psi_{n+1}](v_1,v_2).
    \end{align}
    In order to estimate the RHS, we first write
    \begin{align}\label{eq:RWscat:pointwise diff at n}
    \begin{split}
        \Psi_{n}(u_n,v,\theta^A)-\Psi_{n+1}(u_n,v,\theta^A)&=\int_{v_1}^{v}\PP(\bar{v},\theta^A) \dd\bar{v}-\Psi_{n+1}(u_n,v,\theta^A)\\
        &=\Psi_{n+1}(u_{n+1},v,\theta^A)-\Psi_{n+1}(u_n,v,\theta^A),
    \end{split}
    \end{align}
   and we then estimate, using the equation \eqref{eq:lin:RW} :
    \begin{align}\label{eq:RWscat:energy diff at n}
    \begin{split}
       & \int_{v_1}^{v_2}\int_{\Stwo}\dd\bar{v}\dw\,|\Dv\Psi_{n+1}(u_n,\bar{v},\theta^A)-\Dv\Psi_{n+1}(u_{n+1},\bar{v},\theta^A)|^2\\
       =&\int_{v_1}^{v_2}\int_{\Stwo}\dd\bar{v}\dw\,\left|\int_{u_{n+1}}^{u_n}\dd\bar{u}\frac{\Omega^2}{r^2}(\bar{u},\bar{v})\(\lap-(3\Omega^2+1)\)\Psi_{n+1}(\bar{u},\bar{v},\theta^A)\right|^2\\
       \lesssim& \frac{1}{r(u_n,v_1)}\int_{v_1}^{v_2}\dd\bar{v}\int_{u_{n+1}}^{u_n}\dd\bar{u}\int_{\Stwo}\dw\,\frac{\Omega^2}{r^2}|\lap\Psi_{n+1}|^2\\
       \lesssim&\frac{(v_2-v_1)}{r(u_n,v_1)}\left[\|\sl\PP\|_{L^2([v_1,v_2]\times\Stwo)}^2+\frac{(v_2-v_1)^2}{r(u_n,v_1)^2}\|\lap\PP\|_{L^2([v_1,v_2]\times\Stwo)}^2\right],
    \end{split}
    \end{align}
    where we used the energy estimate and~\eqref{eq:RWscat:energy estimate of data of Psi n at u n} in the last line.
    This gives an estimate for the $\Dv$-derivative part of the energy $E_{u_n}[\Psi_{n}-\Psi_{n+1}](v_1,v_2)$. The other terms can be estimated similarly (without having to invoke \eqref{eq:lin:RW}), and thus, by \eqref{eq:RWscat:energy diff at n} and \eqref{eq:RWscat:pointwise diff at n}, we obtain (we hide $(v_2-v_1)$-weights inside $\lesssim$)
    \begin{align}\label{eq:RWscat:energy of difference at u n}
        E_{u_n}[\Psi_{n}-\Psi_{n+1}](v_1,v_2)\lesssim_{v_2-v_1} \frac{1}{r(u_n,v_1)}\sum_{|\gamma|\leq2}\|\sl^\gamma\PP\|_{L^2([v_1,v_2]\times\Stwo)}^2.
    \end{align}

    Next, we bound
    \begin{align}
        \underline{E}_{v_2}[\Psi_{n+1}](u_{n+1},u_n)=
        E_{u_{n+1}}[\Psi_{n+1}](v_1,v_2)-E_{u_{n}}[\Psi_{n+1}](v_1,v_2)
    \end{align}
    by first estimating (notice that the estimate below is still valid for $u_n$ replaced by any $u'\in[u_{n+1},u_n]$)
    \begin{nalign}\label{eq:RWscat: estimate in u n u n plus one}
        &\int_{v_1}^{v_2}\int_{\Stwo}\dd\bar{v}\dw\, |\Dv\Psi_{n+1}(u_{n+1},\bar{v},\theta^A)|^2-|\Dv\Psi_{n+1}(u_n,\bar{v},\theta^A)|^2\\
        =&\int_{v_1}^{v_2}\int_{\Stwo}\dd\bar{v}\dw\, \Big[|\Dv\Psi_{n+1}(u_{n+1},\bar{v},\theta^A)-\Dv\Psi_{n+1}(u_n,\bar{v},\theta^A)|\\
        &\qquad \cdot|\Dv\Psi_{n+1}(u_{n+1},\bar{v},\theta^A)+\Dv\Psi_{n+1}(u_n,\bar{v},\theta^A)|\Big]\\
        \leq& \sqrt{\int_{v_1}^{v_2}\int_{\Stwo}\dd\bar{v}\dw\, |\Dv\Psi_{n+1}(u_{n+1},\bar{v},\theta^A)-\Dv\Psi_{n+1}(u_n,\bar{v},\theta^A)|^2}\\
        &\qquad \cdot \sqrt{\int_{v_1}^{v_2}\int_{\Stwo}\dd\bar{v}\dw\, |\Dv\Psi_{n+1}(u_{n+1},\bar{v},\theta^A)+\Dv\Psi_{n+1}(u_n,\bar{v},\theta^A)|^2}\\
        \lesssim_{v_2-v_1}& \frac{1}{\sqrt{r(u_n,v_1)}}\left[\|\sl\PP\|_{L^2([v_1,v_2]\times\Stwo)}^2+\frac{1}{r(u_n,v_1)^2}\|\lap\PP\|_{L^2([v_1,v_2]\times\Stwo)}^2\right],
    \end{nalign}
    where we used the estimates \eqref{eq:RWscat:energy estimate of data of Psi n at u n} and \eqref{eq:RWscat:energy diff at n}. 
    We deal with the remaining terms in the difference $ E_{u_{n+1}}[\Psi_{n+1}](v_1,v_2)-E_{u_{n}}[\Psi_{n+1}](v_1,v_2)$ similarly, hence we  arrive at:
    \begin{align}\label{eq:RWscat:energy of n+1 on u n+1 to u n}
        \underline{E}_{v_2}[\Psi_{n+1}](u_{n+1},u_n)\lesssim\frac{1}{\sqrt{r(u_n,v_1)}}\sum_{|\gamma|\leq2}\|\sl^\gamma\PP\|_{L^2([v_1,v_2]\times\Stwo)}^2.
    \end{align}
    
We now extend $\Psi_n$ to the region $(-\infty,u_n]\times[v_1,v_2]\times\Stwo$. 
We could to this by defining~$\Psi_n$ to be the backwards solution to~\eqref{eq:lin:RW} with data $\Psi_n|_{\C_{u_n}\cap\{v\in[v_1,v_2]\}}$ on $\C_{u_n}\cap\{v\in[v_1,v_2]\}$ and constant data $\int_{v_1}^{v_2}\PP \dd v$ along $\Cbar_{v_2}\cap\{u\leq u_n\}$, but then the extension would only be continuous. 

Instead, we provide a \textit{smooth} extension of $\Psi_n|_{\Cbar_{v_2}}$ along $\Cbar_{v_2}$: 
Let $h(u)$ be a smooth cutoff function which cuts off to 0 on $u\geq1$ and is equal to $1$ on $u\leq0$, and let for any $n>0$ $h_n:=h\left(\frac{u-u_n}{u_{n-1}-u_n}\right)$, so that $h_n$ cuts off on $u\geq u_{n-1}$ and is equal to 1 on $u\leq u_n$.
We apply Seeley's extension theorem \cite{Seeley} to extend $\Psi_n|_{\Cbar_{v_2}\cap\{u\geq u_n\}}\cdot h_n$ to the region $u\leq u_n$. Denote this extension by $\mathtt{E}(\Psi_n|_{\Cbar_{v_2}\cap\{u\geq u_n\}}\cdot h_n)$, in short just $\mathtt{E}$. Note that the extension satisfies \cite{Seeley}
    \begin{nalign}\label{eq:RWscat:extension continuity in Hk}
        \|\mathtt{E}(\Psi_n|_{\Cbar_{v_2}\cap\{u\geq u_n\}}\cdot h_n)\|_{H^k((-\infty,u_n]\times\Stwo)}&\lesssim \|\Psi_n|_{\Cbar_{v_2}\cap\{u\geq u_n\}}\cdot h_n\|_{H^k([u_n,\infty)\times\Stwo)}\\
        &\lesssim \|\Psi_n|_{\Cbar_{v_2}\cap\{u\geq u_n\}}\|_{H^k([u_n,u_{n-1}]\times\Stwo)}
    \end{nalign}
    for any $k\geq0$. Let $\tilde{h}_n$ be the cutoff function given by $\tilde{h}_n:=h(u_n-u)$, which cuts off on $u\leq u_n-1$ and is equal to $1$ on $u\geq u_n$. We have by \eqref{eq:RWscat:extension continuity in Hk}
    \begin{align}\label{eq:RWscat:extension bddness in energy 1}
    \begin{split}
       \|r^{-1}\tilde{h}_n\cdot \mathtt{E}\|^2_{L^2([u_n-1,u_n]\times\Stwo)}
        \leq& \frac{1}{r(u_n,v_2)^2} \|\tilde{h}_n\cdot \mathtt{E}\|^2_{L^2([u_n-1,u_n]\times\Stwo)}
        \\\lesssim& \frac{1}{r(u_n,v_2)^2}\|\Psi_n|_{\Cbar_{v_2}\cap\{u\geq u_n\}}\|^2_{L^2([u_{n},u_{n-1}]\times\Stwo)}
        \\\lesssim& \|r^{-1}\Psi_n|_{\Cbar_{v_2}\cap\{u\geq u_n\}}\|^2_{L^2([u_{n},u_{n-1}]\times\Stwo)}.
    \end{split}
    \end{align}
    In the same way, we have 
    \begin{nalign}\label{eq:RWscat:extension bddness in energy 2}
        \|r^{-1}\tilde{h}_n\cdot \sl\mathtt{E}\|^2_{L^2((u_n-1,u_n)\times\Stwo)}&\lesssim \|r^{-1}\sl\Psi_n|_{\Cbar_{v_2}\cap\{u\geq u_n\}}\|^2_{L^2([u_{n},u_{n-1}]\times\Stwo)},\\
          \|\Du(\tilde{h}_n\cdot \mathtt{E})\|^2_{L^2([u_n-1,u_n]\times\Stwo)}&\lesssim \|\Du\Psi_n|_{\Cbar_{v_2}\cap\{u\geq u_n\}}\|^2_{L^2([u_{n},u_{n-1}]\times\Stwo)}
    \end{nalign}
    We can thus define data for $\Psi_n$ on $\underline{\C}_{v_2}\cap\{u\leq u_n\}$ via
    \begin{align}
        \Psi_{n}|_{\underline{\C}_{v_2}\cap\{u\leq u_n\}}:=\tilde{h}_n\cdot \mathtt{E}(\Psi_n|_{v_2}\cdot h_n).
    \end{align}
     Combining, \eqref{eq:RWscat:extension bddness in energy 1}, \eqref{eq:RWscat:extension bddness in energy 2} and using \eqref{eq:RWscat: estimate in u n u n plus one} shows 
    \begin{align}\label{eq:RWscat:energy estimate on extension near scrim}
        \underline{E}_{v_2}[\Psi_{n}|_{\Cbar_{v_2}}](-\infty,u_n)\lesssim_{v_2-v_1} \frac{1}{\sqrt{r(u_n,v_2)}}\(\|\sl\PP\|_{L^2([v_1,v_2]\times\Stwo)}^2+\|\lap\PP\|_{L^2([v_1,v_2]\times\Stwo)}^2\).
    \end{align}
    The estimates \eqref{eq:RWscat:energy of difference between u n and u 0}, \eqref{eq:RWscat:energy of difference at u n}, \eqref{eq:RWscat:energy of n+1 on u n+1 to u n}, and \eqref{eq:RWscat:energy estimate on extension near scrim} lead to
    \begin{align}
        \underline{E}_{v_2}[\Psi_{n+1}-\Psi_n](-\infty,u_0)+E_{u_0}[\Psi_{n+1}-\Psi_n](v_1,v_2)\to0
    \end{align}
    as $n\to\infty$, and thus the restrictions $\Psi_n|_{\C_{u_0}\cap\{v\in[v_1,v_2]\}}$, $\Psi_n|_{\Cbar_{v_2}\cap\{u\leq u_0\}}$ converge with respect to the energy norms~$E_{u_0}[\cdot](v_1,v_2)$, $E_{v_2}[\cdot](-\infty,u_0)$ to some limits $\Psi|_{\C_{u_0}\cap\{v\in[v_1,v_2]\}}$, $\Psi|_{\Cbar_{v_2}\cap\{u\leq u_0\}}$.
    (Notice that the above estimates carry over also to the commuted energies $  \underline{E}^{(n)}_{v_2}[\Psi_{n+1}-\Psi_n](-\infty,u_0)+E^{(n)}_{u_0}[\Psi_{n+1}-\Psi_n](v_1,v_2)$.)

    Let now $\Psi$ be the backwards solution to Regge--Wheeler \eqref{eq:lin:RW} arising from the limiting data $\Psi|_{\C_{u_0}\cap\{v\in[v_1,v_2]\}}$, $\Psi|_{\Cbar_{v_2}\cap\{u\leq u_0\}}$, then $\Psi$ satisfies
    \begin{align}
        \underline{E}_{v_2}[\Psi](-\infty,u_0)+ E_{u_0}[\Psi](v_1,v_2)<\infty ,
    \end{align}
    and we have that $\Psi|_{\Cin}=\PsiCin$. We also have
    \begin{align}\label{eq:RWscat:energy convergence at Scrim of limit}
        \lim_{u\to-\infty}\int_{v_1}^{v_2}\int_{\Stwo}\dd\bar{v}\dw\, |\Dv\Psi(u,\bar{v},\theta^A)-\PP(\bar{v},\theta^A)|^2=0
    \end{align}
    by the following argument: for a given $u$, let $n$ be such that $u\in[u_{n},u_{n-1})$.  Then we have
    \begin{nalign}\label{eq:RWscat:argument of convergence in energy}
     &   \int_{v_1}^{v_2}\int_{\Stwo}\dd\bar{v}\dw\,  |\Dv\Psi(u,\bar{v},\theta^A)-\PP(\bar{v},\theta^A)|^2
        \\\leq&\int_{v_1}^{v_2}\int_{\Stwo}\dd\bar{v}\dw\, |\Dv\Psi_n(u,\bar{v},\theta^A)-\PP(\bar{v},\theta^A)|^2\\
      &\quad\quad\quad  +\int_{v_1}^{v_2}\int_{\Stwo}\dd\bar{v}\dw\,  |\Dv\Psi(u,\bar{v},\theta^A)-\Dv\Psi_n(u,\bar{v},\theta^A)|^2\\
          \lesssim_{v_2-v_1}& \frac{1}{\sqrt{r(u_n,v_1)}}\left[\|\sl\PP\|_{L^2([v_1,v_2]\times\Stwo)}^2+\frac{1}{r(u_n,v_1)^2}\|\lap\PP\|_{L^2([v_1,v_2]\times\Stwo)}^2\right]\\
       &\quad\quad\quad +\int_{v_1}^{v_2}\int_{\Stwo}\dd\bar{v}\dw\,  |\Dv\Psi(u,\bar{v},\theta^A)-\Dv\Psi_n(u,\bar{v},\theta^A)|^2. 
    \end{nalign}
    In the last line, we used \eqref{eq:RWscat: estimate in u n u n plus one}. The convergence~\eqref{eq:RWscat:energy convergence at Scrim of limit} is now implied by the convergence of~$\Psi_n$ to~$\Psi$ with respect to the energy norms.
    The statement \eqref{eq:RWscat:thm:limitofenergy} is proved similarly.

    Now, assume that $\PP\in L^2([v_1,v_2]\times\Stwo)$, $\PsiCin$ such that $\underline{E}_{v_1}[\PsiCin](-\infty,u_0)<\infty$. 
    There exists a sequence $\{\mathtt{P}_{\Scrim,n}\}_{(n\in\mathbb N)}$ of smooth data on $[v_1,v_2]\times\Stwo$ that approximates~$\PP$ in $L^2([v_1,v_2]\times\Stwo)$ and a sequence $\{{\PsiCin}_{,n}\}_{(n\in\mathbb N)}$ of smooth, compactly supported data on $\Cin$ that approximates $\PsiCin$ in the norm given by $E_{v_1}[\,\cdot\,](-\infty,u_0)$. 
    Let $\Psi_{n}$ be the sequence of smooth solutions to Regge--Wheeler \eqref{eq:lin:RW} arising from data $\mathtt{P}_{\Scrim,n}$ on $\Scrimv$,  ${\PsiCin}_{,n}$ on $\Cin$. 
    Then~$\Psi_n|_{\underline{\C}_{v_2}}$ converges to $\Psi_{\underline{\C}_{v_2}}$ with $\underline{E}_{v_2}[\Psi_{\underline{\C}_{v_2}}](-\infty,u_0)<\infty$ and $\Psi_n|_{{\C}_{u_0}}$ converges to $\Psi_{{\C}_{u_0}}$ with $E_{u_0}[\Psi_{{\C}_{u_0}}](v_1,v_2)<\infty$ by energy conservation. 
    Let $\Psi$ be the backwards solution to \eqref{eq:lin:RW} arising from characteristic data $\Psi|_{{\C}_{u_0}}$, $\Psi_{\underline{\C}_{v_2}}$. Then we have for any $u$
    \begin{multline}
        \int_{v_1}^{v_2}\int_{\Stwo}\dd\bar{v}\dw\,|\Dv\Psi(u,\bar{v},\theta^A)-\PP(\bar{v},\theta^A)|^2\\\leq \int_{v_1}^{v_2}\int_{\Stwo}\dd\bar{v}\dw\,|\Dv\Psi(u,\bar{v},\theta^A)-\Dv\Psi_n(u,\bar{v},\theta^A)|^2\\+\int_{v_1}^{v_2}\int_{\Stwo}\dd\bar{v}\dw\,|\Dv\Psi_n(u,\bar{v},\theta^A)-\mathtt{P}_{\Scrim,n}(\bar{v},\theta^A)|^2\\+\int_{v_1}^{v_2}\int_{\Stwo}\dd\bar{v}\dw\,|\PP(\bar{v},\theta^A)-\mathtt{P}_{\Scrim,n}(\bar{v},\theta^A)|^2.
    \end{multline}
    Choosing $n$ large enough ensures that the first and third terms above are small. Choosing $u$ large enough and using \eqref{eq:RWscat:energy convergence at Scrim of limit} then implies
    \begin{align}\label{eq:RWscat:energy convergence at Scrim of limit two}
        \lim_{u\to-\infty}\int_{v_1}^{v_2}\int_{\Stwo}\dd\bar{v}\dw\, |\Dv\Psi(u,\bar{v},\theta^A)-\PP(\bar{v},\theta^A)|^2=0,
    \end{align}
    and the remaining statement that \eqref{eq:RWscat:thm:limitofenergy} holds can be proved similarly.
  
  The final claim that the solution is smooth if the data are smooth and if the energies $E_{v_1}^{(n)}[\PsiCin](-\infty,u_0)$ are finite for any $n\in\mathbb N$ follows by commuting with angular derivatives and Sobolev embedding (which automatically also proves \eqref{eq:RWscat:energyconvergencecommuted}).
\end{proof}
\subsection{Three additional results concerning improved decay and uniqueness}
We can directly infer from Thm.~\ref{thm:RWscat:mas20 RW}  the following 
\begin{prop}\label{prop:RWscat:blabla}
Assume that we are in the second clause of Thm.~\ref{thm:RWscat:mas20 RW} (with smooth data).
Then we have that, for any $n\in\mathbb N$, 
\begin{equation}\label{eq:RWscat:apriori}
|\sl^n\Psi| (u,v)\lesssim r^{\frac12}\cdot \sqrt{\sum_{i\leq n+2}E^{(i)}_{v_1}[\PsiCin](-\infty,u)+\|\sl^s\lap^{\frac{i-s}{2}}\PP\|_{L^2([v_1,v]\times \Stwo)}^2},
\end{equation}
where $s=1$ if $i$ is odd and $0$ if $i$ is even. 
Furthermore, for any $m\in\mathbb N_{\geq1}$,
\begin{equation}\label{eq:RWscat:uniformconvergence}
\Dv^m\sl^n\Psi\to \Dv^{m-1}\sl^n\P_{\Scrim},\quad \text{as  } u\to-\infty.
\end{equation}
uniformly on compact $v$-intervals.
\end{prop}
\begin{rem}
The convergence \eqref{eq:RWscat:uniformconvergence} would follow in the same way as \eqref{eq:RWscat:energyconvergencecommuted} if we also assumed the energy to remain finite for an arbitrary amount of $\sll_T=\Du+\Dv$-commutations. Here, we avoid the use of $\sll_T$-commutations.
\end{rem}
\begin{proof}

The first statement is proved by first estimating, via the fundamental theorem of calculus:
  \begin{align}
    \left|\Psi(u,v,\theta^A)-\PsiCin(u,\theta^A)\right|&\leq \int_{v_1}^v|\Dv \Ps(u,\bar{v},\theta^A)|\dd\bar{v},
    \end{align}
from which it follows via Cauchy--Schwarz that
\begin{multline}
\left\|\Psi (u,v,\cdot)-\PsiCin(u,\cdot)\right\|_{L^2(\Stwo)}^2 \leq \int_{v_1}^v\dd\bar{v}\cdot  \int_{v_1}^v\int_{\Stwo} |\Dv\Psi|^2 \sin\theta\dd \theta\dd \varphi \dd \bar v \\
\leq (v-v_1)\cdot E_{u}[\Psi](v_1,v)\leq (v-v_1)\cdot\left[E_{v_1}[\Psi](-\infty,u)+\lim_{u\to-\infty} E_{u}[\Psi](v_1,v)\right],
\end{multline}
where we also used the energy estimate of Lemma~\ref{lem:RWscat:energyidentity} in the last estimate. 
We similarly estimate the initial data term $\left\|\PsiCin\right\|_{L^2(\Stwo)}^2$ against its energy by considering the fundamental theorem of calculus for $r^{-1}\Psi$:
\begin{equation}
\|r^{-1}\PsiCin\|_{L^2(\Stwo)}^2\lesssim r^{-1}\int_{-\infty}^u\int_{\Stwo}|r\Du(r^{-1}\PsiCin)|^2\dw\dd u'\lesssim		E_{v_1}[\PsiCin](-\infty,u).
\end{equation}
(The boundary term at $u=-\infty$ can be seen to vanish in the same manner.)
Estimate \eqref{eq:RWscat:apriori} now follows by commuting and applying the standard Sobolev inequality on the sphere.

Finally, in order to prove the uniform convergence of $\Dv^m\sl^n\Psi$ to $\Dv^{m-1}\sl^n\PP$ along compact intervals in $v$, we invoke the equation \eqref{eq:lin:RW} itself:
Inserting the pointwise estimate~\eqref{eq:RWscat:apriori} for $n=0$ and integrating~\eqref{eq:lin:RW} in $u$, it follows that $\Dv\Psi$ converges uniformly (on compact $v$-intervals)  to a limit as $u\to-\infty$. By \eqref{eq:RWscat:energyconvergencecommuted}, this limit has to agree with $\PP$. The commuted result follows inductively. 
\end{proof}
Using the $u$-weighted energy estimate of Lemma~\ref{lem:RWscat:uweighted}, we can further show:
\begin{prop}\label{prop:RWscat:uweighted}
   Suppose that $\PP=0$. Then, provided that $\underline{E}^{\delta'}_{v_1}[\PsiCin](-\infty,u_0)<\infty$ for some $\delta'>0$, we have for any $v\geq v_1$, $u\leq u_0$ and for any $0\leq\delta\leq\delta'$:
    \begin{nalign}
\underline{E}^\delta_v[\Psi](-\infty,u)\leq \underline{E}^\delta_{v_1}[\PsiCin](-\infty,u)\\
\lim_{u\to-\infty}E_u^\delta[\Psi](v_1,v)=0.
    \end{nalign}
\end{prop}
\begin{proof}
  In view of \eqref{eq:RWscat:uweightedestimate}, in order to prove the result, it suffices to show that
    \begin{equation*}
        \lim_{u\to-\infty}E^{\delta}_{u}[\Psi](v_1,v) 
    \end{equation*}
    vanishes. To show this, let $\Psi_N$ be the solution to \eqref{eq:lin:RW} arising from data $\PsiCin\cdot \chi_N$ along $\Cin$ and trivial data at $\Scrimv$ ($\PP=0$). Here, the $\chi_N$ are translates of cut-offs such that $\chi_N=0$ for $u\leq -N$. 
    Apply now \eqref{eq:RWscat:uweightedestimate} in the region bounded by $\Scrim$, $u=u_{-\infty}$, $v=v_1$, and $v=v_2$, for $N>-u_{-\infty}$:
    \begin{align*}
        E^{\delta}_{u_{-\infty}}[\Psi_{N}](v_1,v_2)&\leq \underline E^{\delta}_{v_1}[\Psi_N](-\infty,u_{-\infty})\lesssim \underline E^{\delta}_{v_1}[\Psi](-\infty,u_{-\infty}).
    \end{align*}
    Since $E^{\delta}_{u_{-\infty}}[\Psi_{N}](v_1,v_2)=|u_{-\infty}|^{\delta}E_{u_{-\infty}}[\Psi_{N}](v_1,v_2)$, and since $\PsiCin\cdot \chi_N$ clearly converges to $\PsiCin$ in energy, have
    \begin{equation*}
              E^{\delta}_{u_{-\infty}}[\Psi](v_1,v_2)\lesssim \underline E^{\delta}_{v_1}[\Psi](-\infty,u_{-\infty}).
    \end{equation*}
   Sending $u_{-\infty}\to-\infty$ then completes the proof.
\end{proof}


We also have the following enhanced uniqueness statement for solutions to \eqref{eq:lin:RW} of finite energy:

\begin{cor}\label{cor:RWscat:pointwise uniqueness in finite energy}
    Assume $\Psi$ is a solution to \eqref{eq:lin:RW} on $[v_1,v_2]\times(-\infty,u_0)\times \Stwo$ for $v_2>v_1$ such that $\underline{E}_{v_2}[\Psi](-\infty,u_0)<\infty$ for some $v_2> v_1$ and $\|\Psi\|_{H^1(\C_{u_0}\cap\{v\in[v_1,v_2]\})}<\infty$. Assume further that $\Psi|_{\Cin}=0$ and that $\Dv\Psi(u,v,\theta^A)\to0$ as $u\to-\infty$ for $v\in[v_1,v_2]$. Then $\Psi=0$.
\end{cor}

\begin{proof}
    Since $\underline{E}_{v_2}[\Psi](-\infty,u_0)<\infty$, we have that $\Dv\Psi$ converges in $L^2([v_1,v_2]\times\Stwo)$ as $u\to-\infty$ by a slight modification of the argument of \eqref{eq:RWscat:argument of convergence in energy} in Theorem~\ref{thm:RWscat:mas20 RW}, and so the family $\{\Dv\Psi(u,v,\theta^A)\}_{u\in(-\infty,u_0)}$ is uniformly integrable in $L^2([v_1,v_2]\times \Stwo)$. As $\lim_{u\to-\infty}\Dv\Psi(u,v,\theta^A)=0$ by assumption, Vitali's convergence theorem says that
    \begin{align}
        \lim_{u\to-\infty}\int_{v_1}^{v_2}\int_{\Stwo}\dd\bar{v}\dw\,|\Dv\Psi|^2=0.
    \end{align}
    Theorem \ref{thm:RWscat:mas20 RW} now implies $\Psi=0$.
\end{proof}
\subsection{Two uniqueness results for the Teukolsky equation \texorpdfstring{\eqref{eq:lin:Teukal}}{}}\label{sec:RWscat:teukunique}
We now study uniqueness for \eqref{eq:lin:Teukal}.

\begin{prop}\label{prop:RWscat:uniqueness alpha}
    Let $\alpha$ be a solution to \eqref{eq:lin:Teukal} on $(u,v,\theta^A)\in(-\infty,u_0]\times [v_1,v_2]\times \Stwo$ such that $\underline{E}_{v_2}[\Psi](-\infty,u_0)<\infty$, where $\Psi=\(r^2\Omega^{-2}\Du\)^2(r\Omega^2\alpha)$. If $\alpha$ satisfies 
    \begin{align}\label{eq:RWscat:Teukal uniqueness hypothesis}
        \lim_{u\to-\infty}r\alpha(u,v,\theta^A)=0,\qquad \Du( r\Omega^2\alpha)|_{\mathcal{C}}=0,
    \end{align}
    then $\alpha=0$.
\end{prop}

\begin{proof}
    Recall that $\Psi$ defined out of $\alpha$ via $\Psi=\(r^2\Omega^{-2}\Du\)^2(r\Omega^2\alpha)$ satisfies the Regge--Wheeler equation \eqref{eq:lin:RW} in view of Lemma~\ref{lem:lin:RW Teuk identity}. 
    Let $\psi=\frac{r^2}{\Omega^2}\Du (r\Omega^2\alpha)$.
    We import the following estimates for arbitrary $\epsilon>0$ from Propositions 12.1.1 and 12.1.2 of \cite{DHR16} (adapted to the region near $\Scrim$):\footnote{This estimate derives solely from energy conservation for $\Psi$ and the relations \eqref{eq:lin:transformations} between $\alpha$, $\psi$ and $\Psi$. One takes the relation $r^{-2}\Omega^2\Psi=\Du \psi$ and multiplies it by $r^{-\epsilon}\psi$, integrates by parts in the region $[u,u_0]\times[v_1,v_2]\times\Stwo$, and the spacetime estimate for $\psi$ follows by integrated local energy decay for $\Psi$. The spacetime estimate can then be used to estimate the flux term by a similar procedure. An identical argument applies to $\alpha$ once the estimate for $\psi$ is obtained.}
    \begin{multline}\label{eq:RWscat:DHR estimate Teukal}
        \int_{v_1}^{v_2}\int_{\Stwo}\dd\bar{v}\dw\,\left[|r\Omega^2\alpha(u,\bar{v},\theta^A)|^2+|\psi(u,\bar{v},\theta^A)|^2\right]\\
        +\epsilon\int_{u}^{u_0}\int_{v_1}^{v_2}\int_{\Stwo}\dd\bar{u}\dd\bar{v}\dw\,\frac{\Omega^2}{r^{1+\epsilon}}\left[|r\Omega^2\alpha(u,\bar{v},\theta^A)|^2+|\psi(u,\bar{v},\theta^A)|^2\right]\\
        \lesssim \int_{v_1}^{v_2}\int_{\Stwo}\dd\bar{v}\dw\,\left[|r\Omega^2\alpha(u_0,\bar{v},\theta^A)|^2+|\psi(u_0,\bar{v},\theta^A)|^2\right]\\+E_{u_0}[\Psi](v_1,v_2)+\underline{E}_{v_2}[\Psi](-\infty,u_0).
    \end{multline}
    The above estimate remains valid with $\sl r\Omega^2\alpha$, $\sl \psi$ on both sides of the estimates \textit{without requiring a higher order energy for $\Psi$}. The plan now is to use \eqref{eq:RWscat:DHR estimate Teukal} to show that the following quantity vanishes: Define $\widetilde{\Psi}$ via
    \begin{align}\label{eq:RWscat:reverse RW}
        \frac{\Omega^2}{r^2}\slashednabla_v \frac{r^2}{\Omega^2}\slashednabla_v \widetilde{\Psi}=r\Omega^2\alpha,
    \end{align}
     with $\widetilde{\Psi}$, $\slashednabla_v\widetilde{\Psi}$ vanishing at $v_1$. It is clear that 
    \begin{align}
        \mathrm{RW}[\widetilde{\Psi}]|_{\Cin}=0.
    \end{align}
    Note that 
    \begin{align*}
        \Dv \mathrm{RW}[\widetilde{\Psi}]=\psi+\frac{r^2}{\Omega^2}(5-3\Omega^2-\lap)\Dv\widetilde{\Psi}+6M\widetilde{\Psi}.
    \end{align*}
    Since $\Du (r\Omega^2\alpha)=0$ at $\Cin$, we also have
    \begin{align*}
        \Dv\mathrm{RW}[\widetilde{\Psi}]|_{\Cin}=0.
    \end{align*}
    Therefore, we have by \eqref{lem:lin:RWTEUKSTARO} from Lemma~\ref{lem:lin:RW Teuk identity} that $\widetilde\Psi$ satisfies the Regge--Wheeler equation~\eqref{eq:lin:RW}. 
    
    Now, by definition, we have
    \begin{align}
        \Dv\widetilde{\Psi}=\frac{\Omega^2}{r^2}(u,v)\int_{v_1}^v\dd \bar{v} r(u,\bar{v})^3\alpha(u,\bar{v},\theta^A).
    \end{align}
    We make use of the fact that $r\alpha$ decays pointwise towards $\Scrim$ to estimate
    \begin{align}
    \begin{split}
        \int_{\Stwo}\dw\,\left|\frac{\Omega^2}{r^2}\int_{v_1}^v r^3\alpha\dd\bar v\right|^2&=\frac{\Omega^4}{r^4}\int_{\Stwo}\dw\,\left[\int_{v_1}^{v}\dd\bar{v}\,\frac{r^2}{\Omega^2}\int_{-\infty}^{u_0}\dd\bar{u}\,\frac{\Omega^2}{r^2}\psi\right]^2\\
        &\lesssim \int_{\Stwo}\dw\,\left[\int_{v_1}^{v}\dd\bar{v}\,\int_{-\infty}^{u_0}\dd\bar{u}\,\frac{\Omega^2}{r^2}|\psi|\right]^2\\
        &\lesssim \frac{(v_2-v_1)}{r(u,v_1)}\int_{v_1}^{v}\dd\bar{v}\,\int_{-\infty}^{u_0}\dd\bar{u}\int_{\Stwo}\dw\,\frac{\Omega^2}{r^2}|\psi|^2
        \\&\lesssim \frac{(v_2-v_1)}{r(u,v_1)}\times\Big(\int_{v_1}^{v_2}\int_{\Stwo}\dd\bar{v}\dw\,\left[|r\Omega^2\alpha|_{u_0}|^2+|\psi|_{u_0}|^2\right]\\&\qquad\qquad\qquad\qquad+E_{u_0}[\Psi](v_1,v_2)+\underline{E}_{v_2}[\Psi](-\infty,u_0)\Big),
    \end{split}
    \end{align}
    where we used \eqref{eq:RWscat:DHR estimate Teukal} in the last step.
    Thus $\|\Dv\widetilde{\Psi}(u,\cdot,\cdot)\|_{L^2([v_1,v_2]\times\Stwo)}\to0$ as $u\to-\infty$. Combining this with $\widetilde{\Psi}|_{\Cin}=0$, Proposition~\ref{prop:RWscat:RW uniqueness local energy} implies $\widetilde{\Psi}=0$. Thus $\alpha=0$.
\end{proof}

The argument above allows us to conclude the following result, which will later on allow us to conclude uniqueness for the full system \fullsystem.

\begin{cor}\label{cor:RWscat:uniqueness alpha}
    Let $\alpha$ be a solution to \eqref{eq:lin:Teukal} on $(u,v,\theta^A)\in(-\infty,u_0]\times [v_1,v_2]\times \Stwo$ such that $\underline{E}_{v_2}[\Psi](-\infty,u_0)<\infty$, where $\Psi=\(r^2\Omega^{-2}\Du\)^2r\Omega^2\alpha$. If $\alpha$ satisfies 
    \begin{align}\label{eq:RWscat:Teukal integral uniqueness hypothesis}
        \lim_{u\to-\infty}\frac{1}{r(u,v_1)}\int_{v_1}^v\dd\bar{v}\, r(u,\bar{v})^2\alpha(u,\bar{v},\theta^A)=0,\qquad \Du (r\Omega^2\alpha)|_{\mathcal{C}}=0,
    \end{align}
    then $\alpha=0$.
\end{cor}
\begin{proof}
    Let $\widetilde{\Psi}$ be as in the proof of Proposition \ref{prop:RWscat:uniqueness alpha}. We compute $E_{u}[\widetilde{\Psi}](v_1,v_2)$ by directly estimating
    
    \begin{multline}\label{eq:RWscat:estimate Psi tilde 1}
        \|\Dv\widetilde{\Psi}(u,\cdot,\cdot)\|_{L^2([v_1,v_2]\times\Stwo)}^2
        \lesssim \int_{v_1}^{v_2}\dd v\int_{\Stwo}\dw\left(\int_{v_1}^v\dd v' r^3\alpha\right)^2\\
        \lesssim (v_2-v_1)^2 \int_{v_1}^{v_2}\dd v\int_{\Stwo}\dw\,\,(r\alpha)
^2        \lesssim (v_2-v_1)^2\|r\alpha(u,\cdot,\cdot)\|_{L^2([v_1,v_2]\times\Stwo)}^2,
    \end{multline}
    and by similarly estimating
    \begin{align}\label{eq:RWscat:estimate Psi tilde 2}
        \|r^{-1}\widetilde{\Psi}(u,\cdot,\cdot)\|_{L^2([v_1,v_2]\times\Stwo)}&\lesssim \frac{(v_2-v_1)^\frac{3}{2}}{r(u,v_1)}\|r\alpha(u,\cdot,\cdot)\|_{L^2([v_1,v_2]\times\Stwo)},
   \\
   \label{eq:RWscat:estimate Psi tilde 3}
        \|r^{-1}\sl\widetilde{\Psi}(u,\cdot,\cdot)\|_{L^2([v_1,v_2]\times\Stwo)}&\lesssim \frac{(v_2-v_1)^\frac{3}{2}}{r(u,v_1)}\|\sl r\alpha(u,\cdot,\cdot)\|_{L^2([v_1,v_2]\times\Stwo)},
    \end{align}
    Combining \eqref{eq:RWscat:DHR estimate Teukal} with \eqref{eq:RWscat:estimate Psi tilde 2}, we see that $\|r^{-1}\widetilde{\Psi}(u,\cdot,\cdot)\|_{L^2([v_1,v_2]\times\Stwo)}$ decays to 0 as $u\to-\infty$. 
    Thus, $r^{-1}\Psi$ decays in measure as $u\to-\infty$ by Vitali's convergence theorem. Furthermore, the assumption that $\Psi$ is of finite energy at $v_2$ means the same applies to $\widetilde{\Psi}$ by combining~\eqref{eq:RWscat:DHR estimate Teukal} with \eqref{eq:RWscat:estimate Psi tilde 1}, \eqref{eq:RWscat:estimate Psi tilde 2}, \eqref{eq:RWscat:estimate Psi tilde 3}. 
    A computation shows 
    \begin{align}
        \frac{1}{r(u,v_1)}\int_{v_1}^v\dd\bar{v}\, r(u,\bar{v})^2\alpha(u,\bar{v},\theta^A)=\left[\frac{1}{\Omega^2}\Dv \widetilde{\Psi}+\frac{1}{r}\widetilde{\Psi}\right]\Bigg|_{(u,v,\theta^A)},
    \end{align}
    hence $\Dv\widetilde{\Psi}$, too, decays in measure as $u\to-\infty$. By Corollary~\ref{cor:RWscat:pointwise uniqueness in finite energy}, we then deduce that $\widetilde{\Psi}=0$. It follows that $\alpha=0$.
\end{proof}
\newpage

\section{Construction of the (unique) scattering solution to \texorpdfstring{\fullsystem}{the full system}}\label{sec:construct}
In this section, we present the full proof of Theorem~\ref{thm:setup:LEE Scattering wp}. First, we prove uniqueness of scattering solutions in \S\ref{sec:construct:unique}.
We then present the construction of the $\ell\geq2$-part of the scattering solution in \S\ref{sec:construct:existl>1} and of the $\ell<2$-part in \S\ref{sec:construct:existl<2}. The latter will be essentially trivial, all the difficulty is contained in $\ell\geq2$.
\subsection{Uniqueness of solutions to the scattering problem}\label{sec:construct:unique}

We now prove the uniqueness clause of Theorem \ref{thm:setup:LEE Scattering wp}:
\begin{prop}\label{prop:construct:LEE uniqueness}
    Assume that $\mathfrak{S}$ is a solution to \fullsystem~realising a vanishing seed scattering data set $\mathfrak{D}$. Assume moreover that 
    \begin{align}\label{eq:construct:finite energy condition}
        \underline{E}_{v}[\Ps](-\infty,u_0)<\infty
    \end{align}
    for all $v\geq v_1$. Then $\mathfrak{S}$ vanishes identically.
\end{prop}
\begin{proof}
   Let $\mathfrak{S}$ be a scattering solution realising $\mathfrak{D}$. 
Since $\mathfrak{D}$ is trivial, we deduce from Prop.~\ref{prop:setup:uniqueness of data on Cin} that, in particular, the restriction to $\Cin$ of $\trg$, $\xh$, $\xhb$, $\be$, $\sig$, $K$, $\rh$, $\beb$, $\alb$ belonging to $\mathfrak{S}$ vanishes as well.
By \eqref{eq:lin:al3}, we infer that $\Du(r\Omega^2\al)$ vanishes along $\Cin$ as well. The condition \eqref{eq:construct:finite energy condition}, together with
\begin{align}
    &\lim_{u\to-\infty}\frac{1}{r(u,v)}\int_{v_1}^v\dd\bar{v}\,r^2\al(u,\bar{v},\theta^A)\\
    =-&\lim_{u\to-\infty}
    \left(\frac{1}{r(u,v)}\(\frac{r^2\xh}{\Omega}(u,v,\theta^A)-\frac{r^2\xh}{\Omega}(u,v_1,\theta^A)\)\right)=0
\end{align}
implies by Corollary \ref{cor:RWscat:uniqueness alpha} that $\al\equiv0.$  It follows that $\Ps\equiv0$.
 
  Consecutively integrating \eqref{eq:lin:xh4}, \eqref{eq:lin:be4} as well as \eqref{eq:lin:sig4}  from $\Cin$ then yields that $\be\equiv0\equiv \xh$ as well as $\sig\equiv 0$. It follows from \eqref{eq:lin:Ps-Psb=sig} that $\Psb\equiv0$, and thus $\alb\equiv0$ (cf.~Cor.~\ref{cor:setup:psialongC}).
  
With $\al$ and $\alb$ both vanishing identically, we now apply  Proposition~\ref{prop:gauge:vanishing of al alb} to deduce that the $\ell\geq2$ component of $\mathfrak{S}$ is pure gauge, i.e.~$\mathfrak{S}_{\ell\geq 2}=\mathfrak{S}_{f_{\ell\geq 2}}+\mathfrak{S}_{\underline{f}_{\ell\geq 2}}+\mathfrak{S}_{(q_1,q_2)_{\ell\geq 2}}$ for some functions $f_{\ell\geq2}$, $\underline f_{\ell\geq2}$ and $(q_1,q_2)_{\ell\geq2}$. 

We now show that these functions vanish:
Firstly, the vanishing of $\xh$ forces $\underline f_{\ell\geq2}=0$. (Recall that the kernel of $\Ds2\Ds1$ is spanned by the $\ell=0,1$-modes.)
 Secondly, the vanishing of $\lim_{u\to-\infty}\Om$ implies that the outgoing gauge solution must have $\pv f_{\ell\geq2}=0$, and the vanishing of, say, $\rh|_{\Cin}$, implies that $f_{\ell\geq2}=0$.
 Thirdly, the vanishing of $\lim_{u\to-\infty}r^{-1}\b$ implies that $\pv {q_1}_{\ell\geq2}$, $\pv {q_2}_{\ell\geq2}$ vanish, and the fact that $\gsh$ is trivial on $\Cin$ leads to $ {q_1}_{\ell\geq2}= {q_2}_{\ell\geq2}=0$.
 
 At this point, we know that $\mathfrak{S}$ is supported only on $\ell=0,1$, so it is a linear combination of a pure gauge solution, a linearised nearby Schwarzschild solution, and a linearised Kerr solution by Prop.~\ref{prop:gauge:l<2}.

From the fact that $\sig|_{\Cin}=0$, we can infer that the linearised Kerr solution must vanish.

In order to also show that the linearised nearby Schwarzschild solution vanishes, we assume that it doesn't:
Then this solution generates nonvanishing $\rh_{\mathfrak{m}}=-\frac{2M\mathfrak{m}}{r^3}=-\frac2r\K_{\mathfrak{m}}$ along $\Cin$, supported on $\ell=0$, by Prop.~\ref{prop:gauge:SS}.
But since $\mathfrak{S}$ has $\rh|_{\Cin}=0=\K|_{\Cin}=0$, it must now be possible to find pure gauge solutions supported on $\ell=0$ to kill off $\rh_{\mathfrak{m}}$ and $\K_{\mathfrak{m}}$ along $\Cin$. 
An inspection of the expressions given in Propositions~\ref{prop:gauge:out}--\ref{prop:gauge:in} shows that such solutions do not exist: 
Since the outgoing gauge solution decays too fast near $\Scrim$ ($\rh_{f}\sim r^{-4}$, $\K\sim r^{-3}$), the only candidate to kill off the leading order behaviour of $\rh_{\mathfrak{m}}$ and $\K_{\mathfrak{m}}$ is the ingoing gauge solution, but this has $\rh_{\underline{f}}=-\tfrac3r \K_{\underline{f}}$.
Thus the linearised nearby Schwarzschild solution also vanishes.

We now know that $\mathfrak{S}$ is a pure gauge solution supported on $\ell=0,1$, i.e.~$\mathfrak{S}=\mathfrak{S}_{f_{\ell\leq 1}}+\mathfrak{S}_{\underline{f}_{\ell\leq1}}+\mathfrak{S}_{(q_1,q_2)_{\ell\leq 1}}$.
The condition that $\rh|_{\Cin}=0$ implies that $\underline{f}_{\ell\leq 1}(u,\theta^A)=f_{\ell\leq 1}(v=v_1,\theta^A)$, so $\underline{f}_{\ell\leq 1}$ is independent of $u$.
The condition that $\Om\to0$ as $u\to-\infty$ then implies that $\pv f_{\ell\leq 1}=0$, so we have that $\underline f_{\ell\leq 1}(u,\theta^A)=\tilde{f}_{\ell\leq 1}(\theta^A)=f_{\ell\leq 1}(v,\theta^A)$ for some $\tilde f$.
Comparing now the expressions for $\trg$ generated by each of these solutions along $\Cin$, we deduce that $\tilde{f}_{\ell=1} =0$. (The expression for $\trg$ generated by $(q_1,q_2)$ comes with a different $r$-weight.)
We thus have that $\underline f_{\ell\leq 1}(u,\theta^A)=C=f_{\ell\leq 1}(v,\theta^A)$ for some constant $C$, and hence, in view of Remark~\ref{rem:gauge:constantgaugesolution}, $\mathfrak{S}_{f_{\ell\leq 1}}+\mathfrak{S}_{\underline{f}_{\ell\leq1}}$=0.

Thus, we have that $\mathfrak{S}=\mathfrak{S}_{(q_1,q_2)_{\ell\leq 1}}$. 
Since $r^{-1}\b\to0$ as $u\to-\infty$, and since $\trg$ vanishes along $\Cin$, it finally follows that $\mathfrak{S}=0$.
\end{proof}

\subsection{Construction of the \texorpdfstring{$\ell\geq2$}{L>1}-part of the solution to the scattering problem}
\label{sec:construct:existl>1}

We now present an explicit construction of a scattering solution $\mathfrak{S}$ to \fullsystem~realising a given smooth seed scattering data set $\mathfrak{D}$ such that $\mathfrak{S}$ and $\mathfrak{D}$ are related via Definition~\ref{def:setup:scattering solution}. 
\begin{prop}\label{prop:construct:ellgeq2}
Given a smooth seed scattering data set $\mathfrak{D}_{\ell\geq2}$ supported on $\ell\geq 2$ and satisfying the assumptions of Theorem~\ref{thm:setup:LEE Scattering wp}, there exists a scattering solution $\mathfrak{S}_{\ell\geq2}$ realising~$\mathfrak{D}_{\ell\geq2}$. By the previous Prop.~\ref{prop:construct:LEE uniqueness}, this is the unique scattering solution realising $\mathfrak{D}_{\ell\geq2}$. 
\end{prop}

Since the proof of Prop.~\ref{prop:construct:ellgeq2} is quite long, we first give a detailed overview. 
\subsubsection{Overview of the construction}\label{sec:construct:overview}

The construction is organised into the following steps:
\begin{enumerate}[leftmargin=*,label=\textbf{(\Roman*)}]
    \item Out of $\rad{\xh}{\Scrim}$ and data on $\Cin$ (from Prop.~\ref{prop:setup:uniqueness of data on Cin}), we define various tensor fields such as $\radi{\be}$, $\radi{\al}$ and $\P_{\Scrim}$ at $\Scrim$. (These are to be thought of as the pointwise limits of $r^2\be$, $r\al$ and $\Dv\Ps$ of the eventual solution.) \label{steps:const:dataScrim}
    \item We now construct the unique scattering solution $\Ps$ to \eqref{eq:lin:RW} such that $\Ps$ restricts to $\radc{\Ps}$ and $\Dv\Ps$ tends to $\P_{\Scrim}$ at $\Scrim$ via Theorem~\ref{thm:RWscat:mas20 RW}, where $\radc\Ps$ has been constructed in Cor.~\ref{cor:setup:psialongC}.
    Using \eqref{eq:RWscat:apriori} then gives us the bound $\Ps\lesssim r^{1/2}$
    \label{steps:const:Psi}
    \item We then define $\al$ by integrating $\Ps=(\tfrac{r^2}{\Omega^2}\Du)^2(r\Omega^2\al)$ (cf.~\eqref{eq:lin:transformations}) twice from $\Scrim$, defining the boundary terms at $\Scrim$ by using the definitions of $\radi{\be}$ and $\radi{\al}$ from step~\ref{steps:const:Psi} and demanding \eqref{eq:lin:al3} to hold at $\Scrim$. We then show that $\al$ satisfies \eqref{eq:lin:Teukal} by virtue of Lemma~\ref{lem:lin:RW Teuk identity} and $\Ps$ satisfying \eqref{eq:lin:RW}.
    \label{steps:const:alpha}
   \item Having constructed $\al$ in the previous step, we now define $\xh$ and $\be$ by integrating \eqref{eq:lin:xh4} and \eqref{eq:lin:be4} from $\Cin$, respectively, with the boundary terms $\radc{\xh}$  and $\radc{\be}$ defined in Prop.~\ref{prop:setup:uniqueness of data on Cin}.
    The fact that $\al$ satisfies \eqref{eq:lin:Teukal} implies that \eqref{eq:lin:al3} is satisfied. 
    We then show that $r\xh$ converges to $\rad{\xh}{\Scrim}$ as $u\to-\infty$, and similarly for $\be$.
    \label{steps:const:xhbe}
   \item We define $\sig$ by integrating \eqref{eq:lin:sig4} from $\Cin$, with boundary term $\radc{\sig}$ as in Prop.~\ref{prop:setup:uniqueness of data on Cin}.
   We then prove that $\sig$ satisfies the scalar Regge--Wheeler equation \eqref{eq:lin:RWsig}. This is done as follows:
    \begin{itemize}
        \item We show that the $\curl$ of \eqref{eq:lin:divxh}, replacing $\curl\et$ with $r\sig$, is satisfied:
        \begin{align}\label{eq:construct:sigmatrick1}
            \curl\div  r^2\Omega^{-1}\xh=r^3\sig-\curl \,r^3\Omega^{-1}\be.
        \end{align}
        This  is proved by showing that $\Dv \eqref{eq:construct:sigmatrick1}$ is satisfied using~\eqref{eq:lin:xh4}, \eqref{eq:lin:be4} and~\eqref{eq:lin:sig4} and using that \eqref{eq:construct:sigmatrick1} holds along $\Cin$ by Prop.~\ref{prop:setup:uniqueness of data on Cin}.
        \item We then show that the $\curl$ of \eqref{eq:lin:beb4} is satisfied:
        \begin{align}\label{eq:construct:sigmatrick2}
            \Du \,\curl r^2\Omega\be=-\lap r\Omega^2\sig-6M\Omega^2\,\curl \sig.
        \end{align}
        This again follows by virtue of \eqref{eq:construct:sigmatrick2} being satisfied along $\Cin$ and by showing that $\Dv \frac{r^2}{\Omega^2}\eqref{eq:construct:sigmatrick2}$ is satisfied using \eqref{eq:lin:be4}, \eqref{eq:lin:sig4}, \eqref{eq:lin:xh4} as well as \eqref{eq:construct:sigmatrick1}.
    \end{itemize}
  Equations \eqref{eq:lin:sig4} and \eqref{eq:construct:sigmatrick2} together imply that $\sig$ solves \eqref{eq:lin:RWsig}.
  \label{steps:const:sigma}
   \item Next, we define $\Psb=\Ps+4\Ds{2}\Ds{1}(0,r^3\sig)$. Then $\Psb$ satisfies the Regge--Wheeler equation~\eqref{eq:lin:RW}. From here one, we can to some extent mirror the previous approach, with a few extra difficulties that we will highlight below:
    We define $\alb$ by integrating $(\tfrac{r^2}{\Omega^2}\Dv)^2(r\Omega^2\alb)$ twice from $\Cin$ with data terms coming from Prop.~\ref{prop:setup:uniqueness of data on Cin}. The fact that~$\Psb$ satisfies~\eqref{eq:lin:RW} is then used to show that $\alb$ satisfies \eqref{eq:lin:Teukalb} via Lemma~\ref{lem:lin:RW Teuk identity}.
    \label{steps:const:Psbalb}
   \item We would like to define $\xhb$ by integrating \eqref{eq:lin:xhb3} from $\Scrim$. Now, in general, we will have decay no better than $\alb=\O(|u|^{-\frac52-})$ (cf.~\eqref{eq:setup:asymptoticsalongCINforxhbbebetc}), \eqref{eq:lin:xhb3}, so $\Du(\Omega^{-1}r^2\xhb)=-r^2\alb$ will not be integrable. 
   This problem is resolved by commuting \eqref{eq:lin:xhb3} with $\Dv$, using that $\Dv\alb$ decays one power faster than $\alb$: We define $\xhb$ by integrating
   \begin{align}\label{eq:construct:overview:DUDVXHB}
       \Du\Dv \frac{r^2\xhb}{\Omega}=-\(2-\frac{1}{\Omega^2}\)r\Omega^2\alb-\frac{r}{\Omega^2}\Dv(r\Omega^2\alb)
   \end{align}
    first in $u$ from $\Scrim$ and then in $v$ from $\Cin$. 
    We can then define  $\beb$  via \eqref{eq:lin:alb4} and deduce that \eqref{eq:lin:beb3} holds by virtue of $\alb$ satisfying \eqref{eq:lin:Teukalb}.
    \label{steps:const:xhbbeb}
 \item The difficulty at this point is that we have no way of directly verifying  \eqref{eq:lin:sig3} ($\pu(r^3\sig)=-\curl r^2\Omega\beb$). The work-around to this problem is to define a different~$\sig'$ that satisfies~\eqref{eq:lin:sig3}, and to then show that $\sig'=\sig$ as follows:
     \begin{itemize}
        \item Define $r^3\sig'$ as the solution to $\curl\eqref{eq:lin:divxhb}$ (cf.~\eqref{eq:construct:sigmatrick1}):
        \begin{align}\label{eq:construct:sigmatrick3}
            r^3\sig':=\curl\div \frac{r^2\xhb}{\Omega}-\curl \frac{r^3\beb}{\Omega}.
        \end{align}
        We directly deduce from the definition that $\sig'$ satisfies \eqref{eq:lin:sig3}. 
        \item As in step \ref{steps:const:sigma}, we can then prove that $\curl \eqref{eq:lin:beb4}$ holds with $\sig$ replaced by $\sig'$, from which we can infer that $r^3\sig'$ satisfies \eqref{eq:lin:RWsig}.
        We then show that $\sig$ and $\sig'$ attain the same data on $\Cin\cup\Scrimv$ and appeal to the uniqueness clause of Theorem~\ref{thm:RWscat:mas20 RW} to deduce $\sig=\sig'$.
    \end{itemize}
    \label{steps:const:sigma'}
    \item We define $\et$, $\etb$ via \eqref{eq:lin:xh3} and \eqref{eq:lin:xhb4}, respectively. 
    It is then a simple computation to confirm that \eqref{eq:lin:et4} and \eqref{eq:lin:etb3} hold as well.
    Then, using, in particular, equations~\eqref{eq:lin:sig3} and \eqref{eq:lin:sig4}, we deduce \eqref{eq:lin:curleta}. 
    Since thus $\curl(\et+\etb)=0$, we can define $\Om$ as solution to \eqref{eq:lin:OmmA}. 
    From $\Om$, we define $\om$ and $\omb$ via \eqref{eq:lin:Omm}, and \eqref{eq:lin:etb4}, \eqref{eq:lin:et3} immediately follow from \eqref{eq:lin:etb3}, \eqref{eq:lin:et4}.\label{steps:const:etetb}
   \item Similarly, since we have already proved that the $\curl$-parts of \eqref{eq:lin:divxhb} and \eqref{eq:lin:divxh} vanish in \eqref{eq:construct:sigmatrick1}, \eqref{eq:construct:sigmatrick3}, we can define $\trx$ via \eqref{eq:lin:divxh} and $\trxb$ via \eqref{eq:lin:divxhb}. 
   Directly from this definition, we infer \eqref{eq:lin:trx+trxb} by appropriately differentiating \eqref{eq:lin:divxh} in $v$ and \eqref{eq:lin:divxhb} in $u$.\label{steps:const:trxtrxb}

   \item As the penultimate step, we define $\rh$ and prove that all remaining equations of \fullsystem~featuring $\rh$ are satisfied (except the Gauss equation, since we have not yet defined the metric components $\gsh$ and $\trg$).
   The easiest way to do this is as follows:
   \begin{itemize}
   \item Prove that $\curl\div\eqref{eq:lin:al33}$ is satisfied. This allows to define $\Ds2\Ds1(\rh,0)$, and thus $\rh$, via \eqref{eq:lin:al33}, i.e.~as solution to 
   \begin{equation}\label{eq:constr:rhotrick}
   \Ps=2\Ds2\Ds1(r^3\rh,-r^3\sig)+6M(r\Omega\xh-r\Omega\xhb).
   \end{equation}
   \item By acting with $\Du$ or $\Dv$ on \eqref{eq:constr:rhotrick} and using all the previous equations, we can deduce that \eqref{eq:lin:rh3}, \eqref{eq:lin:rh4} are satisfied. We can then deduce the remaining Bianchi equations \eqref{eq:lin:be3}, \eqref{eq:lin:beb4}.
   \item We deduce \eqref{eq:lin:trx3} by multiplying \eqref{eq:lin:divxh} by $r\Omega$ and then acting with $\Du$. We similarly prove \eqref{eq:lin:trxb4}.
   \item We finally prove \eqref{eq:lin:om3} by acting on the definition \eqref{eq:lin:OmmA} of $\Om$ with $\Du\Dv$ and using the already established equations.
\end{itemize} \label{steps:const:rho}
  \item  We define the remaining metric coefficients $\b$, $\gsh$ and $\trg$ via integration of \eqref{eq:lin:b3} from~$\Scrim$ with $\radi\b$ as data,  and integration of  \eqref{eq:lin:gsh4} and \eqref{eq:lin:trg4} from $\Cin$ with $\radc\gsh$ and $\radc\trg$ (defined in Prop.~\ref{prop:setup:uniqueness of data on Cin}) as data.
  By further taking the $\Dv$-derivative of \eqref{eq:lin:trg3} and \eqref{eq:lin:gsh3}, we can establish that \eqref{eq:lin:trg4} and \eqref{eq:lin:gsh4} hold everywhere since they hold along $\Cin$.\label{steps:const:metric}  
The final thing left to do is to prove the Gauss equation \eqref{eq:lin:K}. For this, we prove that $\Dv(r^2\eqref{eq:lin:K})$ holds and that, by construction, \eqref{eq:lin:K} holds along $\Cin$.
We have now constructed a solution to \fullsystem. By construction and Prop.~\ref{prop:setup:uniqueness of data on Cin}, this solution realises the prescribed seed data.
   This completes the proof.

\end{enumerate}

\subsubsection{The full details of the construction}\label{sec:construct:2.ii}
\begin{proof}[Proof of Proposition~\ref{prop:construct:ellgeq2}]
We now present the full details of the construction. The reader already convinced by the overview should feel free to skip this section.
\subsubsection*{Step \ref{steps:const:dataScrim}: Defining scattering data at $\Scrim$}
We first define a number of fields at $\Scrim$ that will play the role of data at $\Scrim$ in the construction carried out in the subsequent sections:
\begin{defi}\label{def:construction:dataatSCRI}
    Given a smooth seed scattering data set $\mathfrak{D}$, we define the following fields at~$\Scrim$:
    \begin{align}
        \rad{\be}{\Scrim}=\div \rad{\xh}{\Scrim}-\sl\rad{\K}{\Sinfty},&&   \rad{\al}{\Scrim}=-\Dv\radi{\xh},&&\radi{\ps}=-2\Ds2\radi\be,
    \end{align}
    where $\rad{\K}{\Sinfty}$ has been defined in \eqref{eq:setup:rad K at Sinfty}, 
  as well as
  \begin{align}
 \radi{\Sb}=-4\Ds{2}\sl \rad{\Om}{\Scrim}+\rad{\xh}{\Scrim}, && \P_{\Scrim}=-2(\lap-4)\Ds2 \be_{\Scrim}-6M\rad{\al}{\Scrim}.
  \end{align}
\end{defi}
\begin{rem}
The above definitions will be used as data at $\Scrim$ for $r^2\be$, $r\al$, $\ps$, $\Dv(r^2\Omega^{-1}\xhb)$, and $\Dv\Ps$, respectively.
\end{rem}
The definition above will be  used explicitly to define the scattering solution over the next few pages. 
On the other hand, we will now define a number of fields at $\Scrim$ for which it will later follow that they are attained as limits of the constructed solution. 
\begin{defi}\label{def:construction:dataatSCRI2}
Given a smooth seed scattering data set $\mathfrak{D}$, we define the following fields at~$\Scrim$: 
\begin{nalign}
    \rad{\gsh}{\Scrim}(v)=\rad{\gsh}{\Sinfty}+\int_{v_1}^v 2\Ds2\rad{\b}{\Scrim}(\bar{v}) \dd\bar{v},&& \rad{\trg}{\Scrim}(v)=\radsinf{\trg}-\int_{v_1}^v\div \rad{\b}{\Scrim}(\bar{v}) \dd\bar{v},\\
\radi{\om}=\pv\radi{\Om},&&\radi{\etb}=2\sl\radi{\Om},\\
\radi{\K}=\rad{\K}{\Sinfty},&&\radi{\trx}=2\radi{\K}+4\radi{\Om}.
\end{nalign}
\end{defi}
\begin{rem}
We will show in our construction that the fields above are attained as limits of the quantities $\gsh$, $\trg$, $\om$, $r\etb$, $r^2\K$  and $r\trx$.
On the other hand, it will follow from our construction that the limits of the following quantities vanish as $u\to-\infty$: $\omb$, $r\trxb$, $r\et$, $r\xhb$, $r^2\rh$, $r^2\sig$, $r^2\beb$, $r^2\alb$.
\end{rem}
\subsubsection*{Step \ref{steps:const:Psi}: Constructing $\Ps$ using scattering theory for \eqref{eq:lin:RW}}
We now construct $\Ps$ via the scattering theory for the Regge--Wheeler equation \eqref{eq:lin:RW} given by Theorem \ref{thm:RWscat:mas20 RW}.
\begin{prop}\label{prop:construct:Ps}
    For a smooth seed scattering data $\mathfrak{D}$, there exists a unique smooth finite energy solution $\Ps$ to \eqref{eq:lin:RW} such that, for any $v\geq v_1$:
    \begin{align}
        \lim_{u\to-\infty}\|\Dv\Ps-\P_{\Scrim} \|_{L^2([v_1,v]\times \Stwo)}=0,\qquad \Ps|_{\Cin}=\rad{\Ps}{\Cin},
    \end{align}
    where $\rad{\Ps}{\Cin}$ is  defined in Cor.~\ref{cor:setup:psialongC} and $\P_{\Scrim}$ is defined in Def.~\ref{def:construction:dataatSCRI}.
    Moreover, we have the estimate
    \begin{equation}\label{eq:construction:aprioriPsi}
    |\sl ^n \Ps|\leq C_n(v) r^{1/2},\qquad \forall n\in\mathbb N
    \end{equation}
    for some $C_n(v)$ depending continuously on $v$, and for any $n\in\mathbb N$, $m\in\mathbb N_{\geq 1}$, we have
    \begin{equation}\label{eq:construction:uniformPsi}
    \lim_{u\to-\infty}\Dv^m \sl^n \Ps=\Dv^{m-1}\sl^n\P_{\Scrim}
    \end{equation}
    as $u\to-\infty$, uniformly on $[v_1,v]$ for any $v\geq v_1$.
\end{prop}
\begin{proof}
By the decay rate in Cor.~\ref{cor:setup:psialongC}, we deduce that $\radc{\Ps}$ has finite energy along $\Cin$. 
The result thus follows from the scattering theory for \eqref{eq:lin:RW} (Thm.~\ref{thm:RWscat:mas20 RW}) as well as \eqref{eq:RWscat:apriori} and \eqref{eq:RWscat:uniformconvergence}.
\end{proof}
\subsubsection*{Step \ref{steps:const:alpha}: Constructing $\al$, $\be$ and $\xh$ }
Next, we construct $\al$ by integrating \eqref{eq:lin:al33} from $\Scrim$:
\begin{defi}\label{defi:construct:al}
    For a smooth seed scattering data $\mathfrak{D}$ and $\Ps$ arising via Proposition \ref{prop:construct:Ps}, define~$\ps$ by integrating 
    \begin{align}
     \frac{\Omega^2}{r^2}   \Ps=\Du \ps
    \end{align}
    in $u$ from $\Scrim$ with data $-2\Ds{2}\rad{\be}{\Scrim}=\radi\ps$.
    Similarly, we define $\al$ by integrating
    \begin{align}\label{eq:construct:alb to pblin}
       \frac{\Omega^2}{r^2}\ps=\Du r\Omega^2\al
    \end{align}
    with data $\rad{\al}{\Scrim}$ at $\Scrim$ (defined in Def.~\ref{def:construction:dataatSCRI}). By construction, $\al|_{\Cin}=\radc{\al}$.
\end{defi}
Note that these definitions are well-defined in view of estimate \eqref{eq:construction:aprioriPsi}.

We now show that $\Dv, \sl$-derivatives of $r^3\Omega\plin$ converge to $\Dv, \sl$-derivatives of $-2\Ds{2}\rad{\be}{\Scrim}$ as $u\to-\infty$, and a similar statement for $\al$:
\begin{lemma}\label{lem:construct:convergence of derivatives of pblin al}
    For a smooth seed scattering data $\mathfrak{D}$ and $\pblin, \al$ arising via Definition \ref{defi:construct:al}, we have, for any $n,m\in \mathbb N$,
    \begin{align}
        &\sl^n\Dv^m \ps(u,v,\theta^A)\to\sl^n\Dv^m\(-2\Ds{2}\rad{\be}{\Scrim}\)(v,\theta^A),\\
        &\sl^n\Dv^m r\Omega^2\al(u,v,\theta^A)\to-\sl^n\Dv^{m+1}\rad{\xh}{\Scrim}(v,\theta^A)
    \end{align}
    as $u\to-\infty$. For any $v\geq v_1$, this convergence is uniform on $[v_1,v]$.
\end{lemma}
\begin{proof}
We compute that 
    \begin{align}
       \frac{r^2}{\Omega^2}\Du\Dv \ps=  \frac{r^2}{\Omega^2}\Dv\Du \ps= \Dv\Ps-\frac{3\Omega^2-1}{r}\Ps.
    \end{align}
    Integrating in $u$ and using \eqref{eq:construction:aprioriPsi} as well as \eqref{eq:construction:uniformPsi}, we see that $\Dv\ps$ converges to a limit as $u\to-\infty$ and that this convergence is uniform on $[v_1,v]$ for any $v\geq v_1$.
     Since $\ps\to-\Ds{2}\rad{\be}{\Scrim}$, it follows that $\Dv\ps\to-\Dv\Ds{2}\rad{\be}{\Scrim}$.
The result for $\Dv$- and $\sl$-derivatives follows by straight-forward commutation.     

The result for $\al$ follows analogously.
\end{proof}
Next, we want to show that $\al$ solves the Teukolsky equation \eqref{eq:lin:Teukal}. For this, we first show:
\begin{lemma}\label{lem:construct:asymptotics of Teuk and Du Teuk at Scrim}
    For a smooth seed scattering data $\mathfrak{D}$ and $\plin, \al$ arising via Definition \ref{defi:construct:al}, we have that
    \begin{align}
        \lim_{u\to-\infty}\Teuk[r\Omega^2\al]=\lim_{u\to-\infty}\frac{r^2}{\Omega^2}\Du\Teuk[r\Omega^2\al]=0
    \end{align}
    for any $v\geq v_1$, where the operator $\Teuk$ was defined in Def.~\ref{def:lin:Teuk}.
\end{lemma}
\begin{proof}
We can re-write the operators above as follows: 
    \begin{align}
        \Teuk[r\Omega^2\al]&=\frac{3\Omega^2-1}{r} \ps+\Dv \ps+\(-\lap+2+\frac{6M}{r}\)r\Omega^2\al,\\
        \frac{r^2}{\Omega^2}\Du\Teuk[r\Omega^2\al]&=\Dv\Ps-\(\lap-3\Omega^2-1\)\ps+6M r\Omega^2\al.
    \end{align}
    Therefore, using the previous Lemma~\ref{lem:construct:convergence of derivatives of pblin al}:
    \begin{align}
        &\lim_{u\to-\infty}\Teuk[r\Omega^2\al]=-2\Ds{2}\Dv\rad{\be}{\Scrim}-\(\lap-2\)\rad{\al}{\Scrim}.\\
        &\lim_{u\to-\infty}\frac{r^2}{\Omega^2}\Du\Teuk[r\Omega^2\al]=\P_{\Scrim}+2\(\lap-4\)\Ds{2}\rad{\be}{\Scrim}-6M\Dv\rad{\xh}{\Scrim}.
    \end{align}
    The right hand sides above both vanish by construction (cf.~Def.~\ref{def:construction:dataatSCRI}).
\end{proof}
We may now infer:
\begin{cor}\label{cor:construct:al satisfies Teukal}
     For a smooth seed scattering data $\mathfrak{D}$ and $\al$ defined in Definition \ref{defi:construct:al}, we have that $\al$ satisfies the Teukolsky equation \eqref{eq:lin:Teukal}.
\end{cor}
\begin{proof}
    Using Lemma \ref{lem:lin:RW Teuk identity}, we have that
    \begin{align}
        \(\frac{r^2}{\Omega^2}\Du\)^2\Teuk[r\Omega^2\al]=\mathrm{RW}[\Ps]=0
    \end{align}
    since $\Ps$ satisfies the Regge--Wheeler equation \eqref{eq:lin:RW}. We integrate this equation twice from~$\Scrim$, where the boundary terms vanish by Lemma~\ref{lem:construct:asymptotics of Teuk and Du Teuk at Scrim}.
\end{proof}
\subsubsection*{Step~\ref{steps:const:xhbe}: Constructing of $\xh$ and $\be$}
We now construct $\xh$, $\be$ and show that \eqref{eq:lin:al3} is satisfied:
\begin{defi}\label{defi:construct: xh be}
    With $\al$ defined in Def.~\ref{defi:construct:al}, define $\xh$, $\be$ to be the unique solutions to \eqref{eq:lin:xh4}, \eqref{eq:lin:be4} with data $\frac{r^2}{\Omega}\rad{\xh}{\Cin}$, $\frac{r^4}{\Omega}\rad{\be}{\Cin}$ respectively.
\end{defi}
\begin{lemma}
    The quantities $\xh$, $\be$ and $\al$ satisfy eq.~\eqref{eq:lin:al3}
\end{lemma}
\begin{proof}
    Def. \ref{defi:construct: xh be}, together with the Teukolsky equation \eqref{eq:lin:Teukal} satisfied by $\al$, implies that
    \begin{align}
        \Dv\left[\frac{r^4}{\Omega^4}\Du r\Omega^2\al+2\Ds{2}\frac{r^4\be}{\Omega}-6M\frac{r^2\xh}{\Omega}\right]=0.
    \end{align}
  Integrating this in $v$ from $\Cin$, and using that $\al_{\Cin}$ is related to  $\rad{\be}{\Cin}$, $\xh_{\Cin}$ via \eqref{eq:lin:al3} by Proposition~\ref{prop:setup:uniqueness of data on Cin} proves the lemma.
\end{proof}
For later purposes, we will also need to verify that $r\xh$ and  $r^2\be$ realise $\radi{\xh}$ and $\radi{\be}$ (which we have used to define $\ps$ and $\al$) as their limits at $\Scrim$:
\begin{cor}\label{cor:construction:betaconvergence}
    For $\xh$, $\be$ constructed in Definition \ref{defi:construct: xh be}, we have that for any $m,n\in\mathbb N$,
    \begin{align}
        \lim_{u\to-\infty} \Dv^m\sl^m r\xh=\Dv^m \sl^m \rad{\xh}{\Scrim},&& \lim_{u\to-\infty}\Dv^m \sl^m r^2\be=\Dv^m \sl^m \rad{\be}{\Scrim},
    \end{align}
    the convergence being uniform in $v$ in compact $v$-intervals.
\end{cor}
\begin{proof}
    We will prove the statement for $\xh$; the proof for $\be$ being similar:
    First of all, we have, by definition, 
    \begin{equation}
   \frac1r\left( \frac{r^2}{\Omega}\xh(u,v)-\frac{r^2}{\Omega}\radc{\xh}(u)\right)=-\frac1r \int_{v_1}^v  r^2\al(u,\bar v) \dd \bar v.
    \end{equation}
     Since $r\al(u,v)$ converges uniformly to $-\radi\xh$ on $[v_1,v]$ for any fixed $v\geq v_1$ by Lemma~\ref{lem:construct:convergence of derivatives of pblin al}, and since the ratio $\frac{r(u,\bar{v})}{r(u,v)}$ tends to 1 uniformly in on $[v_1,v]$ for any fixed $v$,  we have that 
    \begin{align}
      -  \frac{1}{r(u,v)}\int_{v_1}^v r(u,\bar{v})^2\al(u,\bar{v}) \dd\bar{v}\to -\int_{v_1}^v\rad{\al}{\Scrim}(\bar v) \dd\bar{v}=\rad{\xh}{\Scrim}(v)-\rad{\xh}{\Scrim}(v_1)
    \end{align}
    as $u\to-\infty$ for any $v\geq v_1$. Using that $r\radc{\xh}\to \radi{\xh}(v=v_1)$ by construction (Prop.~\ref{prop:setup:uniqueness of data on Cin}), it follows that $r\xh\to \rad{\xh}{\Scrim}$ as claimed. 
    The proof for higher derivatives follows from Lemma~\ref{lem:construct:convergence of derivatives of pblin al} and commuting.
\end{proof}

\subsubsection*{Step~\ref{steps:const:sigma}: Constructing $\sig$}
\begin{defi}\label{defi:construct:sig}
    Define $\sig$ to be the unique solution to \eqref{eq:lin:sig4} with data $r^3\rad{\sig}{\Cin}$ on $\Cin$.
\end{defi}

\begin{cor}\label{cor:construction:sigmaconvergence}
We have for any $n\in\mathbb N$ and some $C_n(v)$ depending continuously on $v$ that
\begin{equation}\label{eq:construction:apriorisigma}
|\sl^n(r^3\sig)|\leq C_n(v) r^{\frac12}.
\end{equation}
In addition, the following convergence is uniform in $v$ on any compact $v$-interval:
\begin{equation}\label{eq:construction:uniformsigma}
\lim_{u\to-\infty}\Dv^m\sl ^n r^3\sig=-\Dv^{m-1}\sl^n\curl\radi\be,\qquad \text{for all }m\in\mathbb{N}_{\geq1}, \,n\in\mathbb N.
\end{equation}
\end{cor}
\begin{proof}
The convergence \eqref{eq:construction:uniformsigma} follows from the definition of $\sig$ and Cor.~\ref{cor:construction:betaconvergence}.
The upper bound~\eqref{eq:construction:apriorisigma} follows from the bound \eqref{eq:construction:aprioriPsi}, the definition of $\sig$ and the rates \eqref{eq:setup:asymptoticsalongCINforxhbbebetc} for~$\radc\sig$.
\end{proof}
\begin{lemma}
    The quantity $\sig$ defined in Definition \ref{defi:construct:sig} satisfies $\curl\eqref{eq:lin:divxh}$ as well as $\curl\eqref{eq:lin:be3}$:
    \begin{align}
        &\curl \div  r\Omega\xh=r^2\Omega^2\sig-\curl  r^2\Omega\be,\label{eq:construct:curldivxh}\\
        &\pu\, \curl  r^2\Omega\be=-\lap\,r\Omega^2\sig-6M\Omega^2\sig\label{eq:construct:curl be3}.
    \end{align}
\end{lemma}
\begin{proof}
    Equations \eqref{eq:lin:be4}, \eqref{eq:lin:sig4} imply
    \begin{align}
        \pv\,\(r^3\sig-\curl r^3\Omega^{-1}\be\)=-\curl r^2\al=\pv\,\curl \div \frac{r^2\xh}{\Omega}
    \end{align}
    and, since $\xh$, $\sig$, $\be$ satisfy \eqref{eq:construct:curldivxh} at $\Cin$, equation \eqref{eq:construct:curldivxh} propagates for $v\geq v_1$. 
    
    The proof for \eqref{eq:construct:curl be3} is similar, but involves more computations: We first multiply the LHS of \eqref{eq:construct:curl be3} and differentiate in $v$:
    \begin{nalign}
        \pv\left(\frac{r^2}{\Omega^2}\,\pu\,\curl r^2\Omega\be\right)&=\pv\left(\frac{r^2}{\Omega^2}\right)\pu\curl r^2\Omega\be+\pu\curl \Dv(r^2\Om\be)\\
        &=\pu\,\curl \div r^2\al+\frac{3\Omega^2-1}{r}\curl \div r^3\al-2(3\Omega^2-2)\,\curl r^2\Omega\be\\
        &=\,\curl \div \left[-2\Ds{2}r^2\Omega\be+6M\Omega\xh\right]-2(3\Omega^2-2)\,\curl r^2\Omega\be,
    \end{nalign}
    where we used \eqref{eq:lin:be4} in the second line and \eqref{eq:lin:al3} in the third line.
    We re-write the angular operator (recall the notation $\div=\D2$) using eq.~\eqref{eq:SS:DD=Laplacian}:
    \begin{align*}
        \D{1}\left[-2\D{2}\Ds{2}\right]=\D{1}\left[\lap+1\right]=\D{1}\left[2-\Ds{1}\D{1}\right]=\left[2-\D{1}\Ds{1}\right]\D{1}=\left[\lap+2\right]\D{1},
    \end{align*} 
    so we have, in particular, that
    \begin{equation*}
        \curl \div \(-2\Ds{2}\,r^2\Omega\be\)=\(\lap+2\)\curl r^2\Omega\be,
    \end{equation*}
    and we can write 
    \begin{align*}
        \pv\left(\,\frac{r^2}{\Omega^2}\,\pu\,\curl r^2\Omega\be\right)=\(\lap+\frac{12M}{r}\)\,\curl r^2\Omega\be+6M\,\curl \div \Omega\xh.
    \end{align*}
    Finally, we apply \eqref{eq:construct:curldivxh} to this:
    \begin{equation}\label{eq:construct:proofaux1}
        \pv\left(\,\frac{r^2}{\Omega^2}\,\pu\,\curl r^2\Omega\be\right)=\(\lap+\frac{12M}{r}\)\,\curl r^2\Omega\be+6M\(\Omega^2\sig-\curl r\Omega\be\).
    \end{equation}

    Next, we apply the operator $\pv\frac{r^2}{\Omega^2}$ to the right hand side of \eqref{eq:construct:curl be3}:
    \begin{equation}\label{eq:construct:proofaux2}
        \pv\left(\lap r^3\sig+6M r^2\sig\right)=-\lap\,\curl r^2\Omega\be-6Mr\Omega^2\sig-6M\,\curl r\Omega\be.
    \end{equation}
    Combining \eqref{eq:construct:proofaux1} and \eqref{eq:construct:proofaux2}, we deduce that
    \begin{align}\label{eq:constr:proofaux5}
        \pv\left(\frac{r^2}{\Omega^2}\pu\,\curl r^2\Omega\be+\lap r^3\sig+6M r^2\sig\right)=0.
    \end{align}
    But since \eqref{eq:construct:curl be3} holds along $\Cin$, we can now integrate the above from $\Cin$ to deduce that it holds everywhere.
\end{proof}
\begin{cor}
    The quantity $\sig$ constructed in Definition \ref{defi:construct:sig} satisfies the scalar Regge--Wheeler equation \eqref{eq:lin:RWsig}.
\end{cor}
\begin{proof}
    This follows from \eqref{eq:lin:sig4} and \eqref{eq:construct:curl be3}.
\end{proof}
\subsubsection*{Step~\ref{steps:const:Psbalb}: Constructing $\Psb$ and $\alb$.}
We are now in position to define $\Psb$ and thus $\alb$:
\begin{defi}\label{defi:construct:alb}
    Define $\Psb$ via
    \begin{align}
        \Psb:=\Ps+2\Ds{2}\Ds{1}(0,r^3\sig).
    \end{align}
    Define $\psb$ to be the unique solution to 
    \begin{align}
        \frac{\Omega^2}{r^2}\Psb=\Dv\psb,
    \end{align}
    with data $\radc\ps=2\Ds{2}r^2\Omega\rad{\beb}{\Cin}+6M\Omega\rad{\xhb}{\Cin}$ at $\Cin$. 
    Define $\alb$ to be the unique solution to 
    \begin{align}
        \frac{\Omega^2}{r^2}\psb=\Dv (r\Omega^2\alb),
    \end{align}
    with data $r\Omega^2\rad{\alb}{\Cin}$ at $\Cin$.
\end{defi}
\begin{cor}
Since $\Ps$ and $\Ds2\Ds1(0,r^3\sig)$ satisfy the Regge--Wheeler equation \eqref{eq:lin:RW}, $\Psb$ also satisfies \eqref{eq:lin:RW}. Moreover, we have for any $n\in\mathbb N$ and for some $C_n(v)$ depending continuously on $v$ the estimates
\begin{align}\label{eq:construct:estimatesforPSIbpsibalphab}
|\sl^n\Psb|\leq C_n(v) r^{\frac12},&&|\sl^n\psb|\leq C_n(v) r^{-\frac12},&&|\sl^n\alb|\leq C_n(v) r^{-5/2},
\end{align}
as well as the uniform convergence on compact $v$-intervals of
\begin{equation}\label{eq:construct:uniformconvergencePSIb}
\lim_{u\to-\infty}\Dv^m\sl^n\Psb=\Dv^{m-1}\sl^n (\radi\Ps -\Ds2\Ds1(0,\curl \radi\be)),\qquad \text{for all  }m\in\mathbb N_{\geq1}, \, n\in\mathbb N.
\end{equation}
\end{cor}
\begin{proof}
The first of \eqref{eq:construct:estimatesforPSIbpsibalphab} follows from the estimates \eqref{eq:construction:aprioriPsi}, \eqref{eq:construction:apriorisigma}. Then using the definition of $\psb$ and $\alb$ together with the rates \eqref{eq:setup:asymptoticsalongCINforxhbbebetc} and \eqref{eq:setup:asym Ps Psb on Cin towards Scrim}, the other two rates follow as well.

The convergence \eqref{eq:construct:uniformconvergencePSIb} follows directly from \eqref{eq:construction:uniformPsi} and \eqref{eq:construction:uniformsigma}.
\end{proof}
As we did for $\al$, we now show that $\alb$ satisfies \eqref{eq:lin:Teukalb} by appealing to equation \eqref{eq:lin:RW Teukalb identity} from Lemma~\ref{lem:lin:RW Teuk identity}.

\begin{cor}
    The quantity $\alb$ constructed in Def.~\ref{defi:construct:alb} satisfies the Teukolsky equation~\eqref{eq:lin:Teukalb}.
\end{cor}
\begin{proof}
    Since $\Psb$ satisfies the Regge--Wheeler equation \eqref{eq:lin:RW}, we have by Lemma~\ref{lem:lin:RW Teuk identity} that 
    \begin{align}
        \(\frac{r^2}{\Omega^2}\Dv\)^2\Teukb[r\Omega^2\alb]=0.
    \end{align}
    A computation analogous to that in the proof of Lemma~\ref{lem:construct:asymptotics of Teuk and Du Teuk at Scrim} shows that $\Teukb[r\Omega^2\alb]$, $\Dv\Teukb[r\Omega^2\alb]$ vanish at $\Cin$.
\end{proof}
\subsubsection*{Step~\ref{steps:const:xhbbeb}: Constructing $\xhb$ and $\beb$}
We now construct $\xhb$, $\beb$. Note that since $|\alb|\lesssim C(v) r^{-\frac52}$, we cannot integrate \eqref{eq:lin:xhb3} in  $u$ from $\Scrim$.
 However, since $|\Dv\alb|\lesssim C(v) r^{-7/2}$ (by \eqref{eq:construct:estimatesforPSIbpsibalphab}), we can instead integrate the $\Dv$-commuted equation:
\begin{defi}\label{defi:construct:xhb}
    Define $\Sb$ to be the unique solution to 
    \begin{equation}
        \Du \Sb=-r^{-1}\psb-\(2-\frac{1}{\Omega^2}\)r\Omega^2\alb,
    \end{equation}
    with data $\radi{\Sb}$ at $\Scrim$ defined in Def.~\ref{def:construction:dataatSCRI}.
    Define $\xhb$ to be the unique solution to 
    \begin{equation}\label{eq:proofaux4}
        \Dv \frac{r^2\xhb}{\Omega}=\Sb,
    \end{equation}
    with data $r^2\Omega^{-1}\rad{\xhb}{\Cin}$ at $\Cin$.
\end{defi}

\begin{lemma}
    The quantity $\xhb$ constructed in Definition \ref{defi:construct:xhb} satisfies
    \begin{equation}\label{eq:const:proofaux3}
        \Du\frac{r^2\xhb}{\Omega}=-r^2\alb.
    \end{equation}
  Moreover, we have the following uniform convergence along compact $v$-intervals:
  \begin{equation}\label{eq:construct:xhbconvergence}
  \lim_{u\to-\infty}\Dv^m \sl^n r^2\xhb=\Dv^{m-1}\sl^n \radi\Sb,\qquad \text{for all  } m\in\mathbb N_{\geq1},\, n\in\mathbb N.
  \end{equation}
\end{lemma}
\begin{proof}
    The first statement follows by integrating the identity
    \begin{align}
        \Dv (r^2\alb)=r^2\Omega\pblin+\(2-\frac{1}{\Omega^2}\)r\Omega^2\alb=-\Du\Sb=-\Dv\Du\frac{r^2\xhb}{\Omega}
    \end{align}
    in $v$ from $\Cin$, using that \eqref{eq:const:proofaux3} is satisfied along $\Cin$ by $\radc{\alb}$ and $\radc\xhb$ as a consequence of Prop.~\ref{prop:setup:uniqueness of data on Cin}.
    
    The proof of \eqref{eq:construct:xhbconvergence} proceeds exactly as the proof of Lemma~\ref{lem:construct:convergence of derivatives of pblin al}, using the rates \eqref{eq:construct:estimatesforPSIbpsibalphab}.
\end{proof}

\begin{defi}\label{defi:construct:beb}
    Define $\beb$ to be the unique solution to \eqref{eq:lin:alb4}, with $\xhb$ defined as in Definition~\ref{defi:construct:xhb} and $\alb$ defined as in Definition \ref{defi:construct:alb}. By construction and Prop.~\ref{prop:setup:uniqueness of data on Cin}, $\beb|_{\Cin}=\radc\beb$.
\end{defi}
\begin{cor}
    The quantity $\beb$ constructed in Definition \ref{defi:construct:beb} satisfies \eqref{eq:lin:beb3}.
\end{cor}
\begin{proof}
    This follows by computing $\Du(\tfrac{r^4}{\Omega^4}\eqref{eq:lin:alb4})$, using \eqref{eq:lin:xhb3} for $ \xhb$ as well as the Teukolsky equation \eqref{eq:lin:Teukalb} for $\alb$. (Recall that the operator $\Ds2$ is invertible on $\ell\geq 2$.)
\end{proof}

\subsubsection*{Step~\ref{steps:const:sigma'}: Constructing $\sig'$ and proving that $\sig'=\sig$}
Recall the difficulties addressed in the overview of Step~\ref{steps:const:sigma'} in \S\ref{sec:construct:overview}: Since we have no direct way of inferring \eqref{eq:lin:sig4} at this point, we instead define a different $\sig'$:
\begin{defi}\label{defi:construct: chiral twin sig}
    Define $\sig'$ via the $\curl$ of \eqref{eq:lin:divxhb}:
    \begin{align}\label{eq:construct: chiral twin sig}
        r^3\sig':=\curl \div r^2\Omega^{-1}\xhb-\curl r^3\Omega^{-1}\beb.
    \end{align}
    By construction and Prop.~\ref{prop:setup:uniqueness of data on Cin}, $\sig'$ restricts to $\radc{\sig}$ on $\C$.
\end{defi}
\begin{lemma}\label{lem:construct:sig'3}
    The quantity $\sig'$ satisfies
    \begin{align}\label{eq:construct:lemma:sig'3}
        \pu (r^3\sig')=\curl r^2\Omega\beb.
    \end{align}
\end{lemma}
\begin{proof}
    This follows by taking the $\pu$-derivative of the LHS of \eqref{eq:construct: chiral twin sig} using the already established equations \eqref{eq:lin:xhb3}, \eqref{eq:lin:beb3}.
\end{proof}

In order to infer that $\sig'=\sig$, we first show that their radiation fields at $\Scrim$ are identical:

\begin{lemma}\label{lem:construct:rad dvsig'}
For $\sig'$, we have 
    \begin{align}
        \lim_{u\to-\infty} \pv (r^3\sig')=-\curl \rad{\be}{\Scrim}=\lim_{u\to-\infty}\pv(r^3\sig).
    \end{align}
    This convergence is uniform in $v$ on compact $v$-intervals.
\end{lemma}
\begin{proof}
Note that we have already shown this for $\sig$.
For $\sig'$, we first take the $\pv$-derivative of \eqref{eq:construct: chiral twin sig} and use \eqref{eq:construct:xhbconvergence} to find
\begin{equation}
\lim_{u\to-\infty}\pv(r^3\sig')=\curl\div\radi\Sb-\lim_{u\to-\infty}\curl\Dv(\tfrac{r^3}{\Omega}\beb).
\end{equation} 
In order to evaluate the limit on the RHS, we compute, using \eqref{eq:lin:alb4}: 
   \begin{align}
     \lim_{u\to-\infty}  \Ds{2}\,\Dv (r^3\Omega^{-1}\beb)=\lim_{u\to-\infty} (\Dv (r\psb)-6M\Dv (r\Omega\xhb)).
   \end{align}
   But 
   \begin{equation*}
       \Dv (r\psb)=r^{-1}\Omega^2\Psb+\Omega^2\psb
   \end{equation*}
   tends to 0 as $u\to-\infty$ by \eqref{eq:construct:estimatesforPSIbpsibalphab} and \eqref{eq:construct:uniformconvergencePSIb}, and so does $\Dv(r\Omega\xhb)$ by \eqref{eq:construct:xhbconvergence}. 
   It thus follows that $\lim_{u\to-\infty}\Ds2\Dv r^3\Omega^{-1}\be=0$.
   In order to deduce that $\lim_{u\to-\infty}\curl r^3\Omega^{-1}\beb=0$ as well, we simply apply $\curl\D2$: To be precise, we use that for any  1-form $\beta$:
\begin{equation}\label{eq:constr:anothercommutationformula}
2\curl\D2\Ds2\beta=-(\lap+2)\curl\beta.
\end{equation}   
It follows that 
\begin{equation}
\lim_{u\to-\infty}-(\lap+2)\pv(r^3\sig')=-(\lap+2)\curl\div\radi\Sb-\lim_{u\to-\infty}2\curl\D2\Ds2\Dv(\tfrac{r^3}{\Omega}\beb),
\end{equation} 
The term on the RHS now vanishes by the same argument as above, and, by the invertibility of $\lap+2$  on $\ell\geq 2$, {a Poincaré inequality and Sobolev embedding}:
\begin{equation}
\lim_{u\to-\infty} \pv( r^3\sig')=\curl\div\radi{\Sb}=\curl\div\radi{\xh}=\curl\radi\be,
\end{equation}
where the last two equalities follow from definitions of $\radi\Sb$ and $\radi\be$ in Def.~\ref{def:construction:dataatSCRI} (and the fact that $\curl\div\Ds2\sl f=0=\curl\sl f$ for any $f\in\functions$).
\end{proof}

We can now show that $\sig'$ also satisfies \eqref{eq:lin:RWsig}:
\begin{lemma}\label{lem:construct:curlbeb4=sig'}
    The quantities $\beb$, $\sig'$ satisfy the $\curl$ of \eqref{eq:lin:beb3}, i.e.:
    \begin{align}\label{eq:construct:lemma:curlbeb4}
        \pv\,\curl r^2\Omega\beb=-\lap r\Omega^2\sig'-6M\Omega^2\sig'.
    \end{align}
    In particular, in view of Lemma~\ref{lem:construct:sig'3}, $\sig'$ satisfies \eqref{eq:lin:RWsig}.
\end{lemma}
\begin{proof}
    This is similar to the proof of \eqref{eq:construct:curl be3}. In the same way in which we proved \eqref{eq:constr:proofaux5}, we first show that
    \begin{align}
        \pu\left[\frac{r^2}{\Omega^2}\pv\,\curl r^2\Omega\beb-\lap r^3\sig'-6M r^2\sig'\right]=0.
    \end{align}
   Notice that we cannot integrate this from $\Scrim$, because the object on which $\pu$ is acting does not converge near $\Scrim$. We again resolve this problem by considering the $\pv$-derivative:

   \begin{align}\label{eq:const:aux10}
        \pu\pv\left[\frac{r^2}{\Omega^2}\pv\,\curl r^2\Omega\beb-\lap r^3\sig'-6M r^2\sig'\right]=0.
    \end{align}
  We integrate this in $u$, using that the boundary term at $\Scrim$ vanishes. Indeed:
  \begin{subclaim}\label{lem:construct:dvr2dvr2be}
    \begin{align}\label{eq:construct:subclaim}
        \lim_{u\to-\infty} \pv \left(\frac{r^2}{\Omega^2}\pv\,\curl (r^2\Omega\beb)\right)=\lap\,\curl \rad{\be}{\Scrim}=\lim_{u\to-\infty}\lap r^3\sig'.
    \end{align}
\end{subclaim}
\begin{proof}[Proof of the sublemma.]
    As in the proof of Lemma~\ref{lem:construct:rad dvsig'}, we use \eqref{eq:lin:alb4} to compute
    \begin{align}
        2\Ds{2}\Dv\left(\frac{r^2}{\Omega^2}\Dv( r^2\Omega\beb)\right)=\Dv\Psb-6M\Dv^2\frac{r^2\xhb}{\Omega}+6M\frac{(3\Omega^2-1)}{r}\Dv\frac{r^2\xhb}{\Omega}-12M(3\Omega^2-2)\Omega\xhb;
    \end{align}
  so we have 
  \begin{equation}
  \lim_{u\to-\infty} 2\Ds2\Dv\left(\frac{r^2}{\Omega^2}\Dv (r^2\Omega\beb)\right)=\lim_{u\to-\infty}\Dv\Psb-6M\Dv\radi\Sb.
  \end{equation}
In view of \eqref{eq:construct:uniformconvergencePSIb}, we have
\begin{nalign}
\lim_{u\to-\infty}\Dv\Psb=\P_{\Scrim}-4\Ds2\Ds1(0,\curl\radi\be)&=-2(\lap-4)\Ds2\radi\be-4\Ds2\Ds1(0,\curl\radi\be)-6M\radi\al
\\&=2\Ds{2}\Ds{1}\(\div \rad{\be}{\Scrim},-\curl \rad{\be}{\Scrim}\)+6M\Dv\radi\xh,
\end{nalign}
where we also used \eqref{eq:SS:DD=Laplacian}.
Acting now on this with $\curl\D2$, we obtain 
\begin{nalign}
 \lim_{u\to-\infty} 2\curl\D2\Ds2\Dv\left(\frac{r^2}{\Omega^2}\Dv (r^2\Omega\beb)\right)=-(\lap+2)(\lap)\curl\radi\be.
\end{nalign}
The result then follows using \eqref{eq:constr:anothercommutationformula} (and that we are supported on $\ell\geq2$).
\end{proof}
Using this sublemma, we can now integrate \eqref{eq:const:aux10} from $\Scrim$ to deduce that
    \begin{align}
        \pv\left[\frac{r^2}{\Omega^2}\pv\,\curl r^2\Omega\beb-\lap r^3\sig'-6M r^2\sig'\right]=0.
    \end{align}
    In turn, integrating this from $\Cin$, and verifying that $\pv\,\curl r^2\Omega\beb-\lap r\Omega^2\sig'-6M \Omega^2\sig'=0$ at~$\Cin$ by computing the $\Dv$ derivative of $\beb$ in spacetime via \eqref{eq:lin:alb4}, restricting to $\Cin$ and using Proposition \ref{prop:setup:uniqueness of data on Cin}, we conclude the proof of \eqref{eq:construct:lemma:curlbeb4}. 
    
    Acting now with \eqref{eq:construct:lemma:curlbeb4} on \eqref{eq:construct:lemma:sig'3}, we infer that $\sig'$ satisfies \eqref{eq:lin:RWsig}.
\end{proof}

\begin{cor}\label{cor:construct:sig=sig'}
With $\sig'$ defined in Def.~\ref{defi:construct: chiral twin sig} and $\sig$ defined in Def.~\ref{defi:construct:sig}, we have
    $\sig'=\sig$.
\end{cor}
\begin{proof}
Since $\sig'=\sig$ on $\Cin$ and $\lim_{u\to-\infty}\pv r^3\sig=\lim_{u\to-\infty}\pv r^3\sig'$ uniformly in $v$, the uniqueness clause of Theorem \ref{thm:RWscat:mas20 RW} implies $\sig'=\sig$ by virtue of both $\sig$ and $\sig'$ satisfying \eqref{eq:lin:RWsig}.
\end{proof}
\subsubsection*{Step \ref{steps:const:etetb}: Constructing $\et$, $\etb$ as well as $\Om$, $\om$ and $\omb$}
\begin{defi}\label{defi:construct:et etb}
  Recall that the kernel of $\Ds2$ is spanned by $\ell=0,1$.   Define $\et$, $\etb$ to be the unique solutions to \eqref{eq:lin:xh3} and \eqref{eq:lin:xhb4}.
    \begin{align}\label{eq:construct:et etb}
        -2\Ds{2}\Omega^2\et=\Du (r\Omega\xh)+\Omega^3\xhb,\qquad -2\Ds{2}\Omega^2\etb=\Dv (r\Omega\xhb)-\Omega^3\xh.
    \end{align}
\end{defi}

\begin{cor}
   The one-forms $\et$ and $\etb$ satisfy \eqref{eq:lin:etb3}, \eqref{eq:lin:et4} and \eqref{eq:lin:curleta}. Moreover, we have $\et|_{\Cin}=\radc\et$, $\etb|_{\Cin}=\radc\etb$ as well as 
   \begin{align}
   \lim_{u\to-\infty}r\et=0, &&\lim_{u\to-\infty}r\etb=\radi\etb,
   \end{align}
where $\radi\etb$ is defined in Def.~\ref{def:construction:dataatSCRI2}. This convergence is uniform on compact $v$-intervals and commutes with $\sl$ and $\Dv$-derivatives.
\end{cor}
\begin{proof}
We prove \eqref{eq:lin:et4} by multiplying the first of \eqref{eq:construct:et etb} (i.e.~\eqref{eq:lin:xh3}) with $\Omega^{-2}r$ and then acting with $\Dv$, using \eqref{eq:lin:xh4}, \eqref{eq:lin:al3} as well as the second of \eqref{eq:construct:et etb} (i.e.~\eqref{eq:lin:xhb4}).

The proof of \eqref{eq:lin:etb3} is analogous. 

In order to prove \eqref{eq:lin:curleta}, we take the $\curl \div$ of \eqref{eq:lin:xh3}. From there, using \eqref{eq:construct:curldivxh}, \eqref{eq:construct: chiral twin sig} and Corollary \ref{cor:construct:sig=sig'}, we get
    \begin{align}
        \pu\(r^2\Omega^2\sig-\curl  r^2\Omega\be\)=\(\lap+2\)\,\curl \Omega^2\et-\frac{\Omega^2}{r}\(r^2\Omega\sig+\curl r^2\Omega\beb\).
    \end{align}
    We then utilise \eqref{eq:lin:sig3} and \eqref{eq:construct:curl be3} to obtain
    \begin{align}
        \(\lap+2\)\curl  r\et=\(\lap+2\)r^2\sig,
    \end{align}
    which implies $\curl\et=r\sig$ since $\lap+2$ is invertible on $\ell\geq 2$.
     
     The claim for $\curl \etb$ follows in the same way, using equations \eqref{eq:lin:xhb4}, \eqref{eq:construct: chiral twin sig}, \eqref{eq:lin:sig4}, Lemma~\ref{lem:construct:curlbeb4=sig'} and Corollary~\ref{cor:construct:sig=sig'}.
     
     The claim that $\et$ restricts to $\radc{\et}$ follows from the definition, the analogous statements for $\xh$ and $\xhb$ as well the fact that the first of  \eqref{eq:construct:et etb} is satisfied along $\Cin$ by Prop.~\ref{prop:setup:uniqueness of data on Cin}.
     Furthermore, it is easy to see that $r\et$ tends to 0 as $u\to-\infty$. 
     
     The uniform convergence of $r\etb$ follows from its definition and the uniform convergence of $\Sb$ and $\xh$ to $\radi\Sb$ and $\radi\xh$, respectively. By integrating \eqref{eq:lin:etb3}, we then infer that $\et|_{\Cin}=\radc\et$.
\end{proof}
In particular, since we have now shown that $\curl(\et+\etb)=0$, the following definition is well-defined:
\begin{defi}
Define $\Om$ via $2\sl\Omm=r(\et+\etb)$, and define $\om=\pv\Omm$, $\omb=\pu\Omm$. 
\end{defi}
\begin{cor}
Equations \eqref{eq:lin:et3} and \eqref{eq:lin:etb4} are satisfied. 
Moreover, $\lim_{u\to-\infty}\Om=\radi\Om$ (this convergence being uniform on compact $v$-intervals and commuting with $\sl^n$ and $\Dv^m$), and $\omb|_{\Cin}=\radc\omb$.
\end{cor}

\begin{proof}
Equation \eqref{eq:lin:et3} follows directly from \eqref{eq:lin:etb3}, and the definition of $\Om$ and $\omb$. Equation~\eqref{eq:lin:etb4} follows similarly.
The convergence of $\sl\Om$ to $\sl\radi{\Om}$ follows from the limiting behaviour of $\et$ and $\etb$, and the convergence of $\Om$ to $\radi\Om$ then follows straight-forwardly (using that we are supported on $\ell\geq2$).
The restriction of $\omb$ to $\Cin$ follows from $\Om$ restricting to $\radc\Om$ by definition and Prop.~\ref{prop:setup:uniqueness of data on Cin}.
\end{proof}
\subsubsection*{Step \ref{steps:const:trxtrxb}: Constructing $\trx$ and $\trxb$.}
\begin{defi}
We define $\trx$ and $\trxb$ according to \eqref{eq:lin:divxh} and \eqref{eq:lin:divxhb}, i.e.
\begin{align}\label{eq:construct:definition of trx}
\frac{1}{2\Omega}\sl\trx=\div\xh+\Omega\etb+r\be,&&\frac{1}{2\Omega}\sl\trx=\div\xhb-\Omega\et+r\be.
\end{align}
By construction and Prop.~\ref{prop:setup:uniqueness of data on Cin}, $\trx$ and $\trxb$ restrict to $\radc\trx$ and $\radc\trxb$ along $\Cin$, respectively.
\end{defi}
\begin{rem}
These definitions are well-defined since we already know that the $\curl$ of the respective RHS's vanish by \eqref{eq:construct:curldivxh}, \eqref{eq:construct: chiral twin sig} and $\curl \et=r\sig=r\sig'=-\curl \etb$.
\end{rem}
\begin{cor}
The quantities $\trx$ and $\trxb$ defined above satisfy equations \eqref{eq:lin:trx+trxb}, and $\lim_{u\to-\infty }r\trxb=0$, $\lim_{u\to-\infty}r\trx=\radi\trx$, these convergences being uniform on compact $v$-intervals and commuting with $\Dv$- and $\sl$-derivatives.
\end{cor}
\begin{proof}
Equation \eqref{eq:lin:trxb3} follows by acting on the first of \eqref{eq:construct:definition of trx} with $\Du(\frac{r^2}{\Omega}\cdot)$ and using \eqref{eq:lin:xh3}, \eqref{eq:lin:etb3} as well as \eqref{eq:lin:be3}. 
Equation \eqref{eq:lin:trx4} follows analogously. 

The limiting behaviour follows from the definitions and the previous estimates.
\end{proof}

\subsubsection*{Step \ref{steps:const:rho}: Constructing $\rh$ and proving all equations featuring $\rh$}
Before we can define $\rh$, we need to prove the following
\begin{lemma}
We have
\begin{equation}\label{eq:construct:pedanticsrho}
\curl\div \Ps=2\curl\div \Ds2\Ds1(0,-r^3\sig)+6M\curl\div(r\Omega\xh-r\Omega\xhb).
\end{equation}
\end{lemma}
\begin{proof}
First, note that this equation holds along $\Cin$ by construction (cf.~Proposition~\ref{prop:setup:uniqueness of data on Cin} and Corollary~\ref{cor:setup:psialongC}).
Next, we confirm that $\Dv\eqref{eq:construct:pedanticsrho}$ holds by computing  the $\Dv$-derivative of the LHS \eqref{eq:construct:pedanticsrho} using the definition of $\Ps$ and the Teukolsky equation \eqref{eq:lin:Teukal} (the result is \eqref{eq:lin:Ps4}), and using \eqref{eq:lin:sig4}, \eqref{eq:lin:xh4} and \eqref{eq:lin:xhb4} for the RHS. 
\end{proof}
The above lemma enables the following definition:
\begin{defi}\label{defi:construct:rho}
    Define $\rh$ to be the unique solution to 
    \begin{align}\label{eq:constr:rhodefintion1}
        \Ps=2\Ds{2}\Ds{1}\(r^3\rh,-r^3\sig\)+6M\(r\Omega\xh-r\Omega\xhb\).
    \end{align}
    By construction and by Prop.~\ref{prop:setup:uniqueness of data on Cin}, $\rh|_{\Cin}=\radc\rh$.
\end{defi}
\begin{rem}
    Definition \ref{defi:construct:rho} implies, via \eqref{eq:lin:Ps-Psb=sig}:
    \begin{align}\label{eq:constr:rhodefintion2}
        \Psb=2\Ds{2}\Ds{1}\(r^3\rh,r^3\sig\)+6M\(r\Omega\xh-r\Omega\xhb\).
    \end{align}
\end{rem}
\begin{lemma}
 The quantity $\rh$ defined above satisfies equations \eqref{eq:lin:be3}, \eqref{eq:lin:beb4} as well as~\eqref{eq:lin:rh}. 
\end{lemma}
\begin{proof}
We first prove \eqref{eq:lin:be3}: We take the definition \eqref{eq:constr:rhodefintion1}, and rewrite its LHS as $\Ps=(\frac{r^2}{\Omega^2}\Du)^2(r\Omega^2\al)$. We now rewrite this as the $\tfrac{r^2}{\Omega^2}\Du$-derivative of the RHS of \eqref{eq:lin:al3_v2}. Equation~$\Ds2\eqref{eq:lin:be3}$ then follows by taking into account the already established equation~\eqref{eq:lin:xh3}, and the result follows from the invertibility of $\Ds2$.

Equation \eqref{eq:lin:beb4} follows similarly, starting from \eqref{eq:constr:rhodefintion2}.

Next,  we prove \eqref{eq:lin:rh4}. 
Using \eqref{eq:lin:be4} and \eqref{eq:lin:al3}, we compute:
    \begin{align}
        \Dv\left(\frac{r^2}{\Omega^2}\Du r^2\Omega\be\right)=-2\D{2}\Ds{2}\,(r^2\Omega\be)+6M\div \Omega\xh-2(3\Omega^2-2)r^2\Omega\be.
    \end{align}
    Therefore, by \eqref{eq:lin:be3} and \eqref{eq:lin:et4} we have
    \begin{align}
        \Ds{1}\pv\(-r^3\rh, r^3\sig\)=\(\lap+1\)r^2\Omega\be+6M\,\div \Omega\xh-2(3\Omega^2-2)\,r^2\Omega\be+6M\(r\Omega\be-\Omega^2\etb\).
    \end{align}
    Taking the divergence of both sides and using $\div\(\lap+1\)\be=\(\lap+2\)\div \be$ and equation~\eqref{eq:lin:divxh}, we get
    \begin{align}
        \lap\pv\,(r^3\rh)=\lap\,\div r^2\Omega\be+3M\lap\trx.
    \end{align}
    This shows that \eqref{eq:lin:rh4} is satisfied by the invertibility of $\lap$ on $\ell\geq2$.
    
    The argument to prove \eqref{eq:lin:rh3} is analogous.
\end{proof}

\begin{lemma}
The equations \eqref{eq:lin:trx3} and \eqref{eq:lin:trxb4} are satisfied by $\trx$ and $\trxb$.
\end{lemma}
\begin{proof}
Let's first prove \eqref{eq:lin:trx3}:
We differentiate the definition of $\trx$, \eqref{eq:lin:divxh}:
\begin{multline}
\frac12\sl\Du(r\trx)=\Du(\div r\Omega\xh)+\Du(r\Omega^2\etb)+\Du(r^2\Omega\be)\\
=-2\Omega^2\D2\Ds2\et-\Omega^3\div\xhb+\Du(\Omega^2)r\etb+r\Omega^3\beb-\Omega^4\et
+\Ds1(-r\Omega^2\rh,r\Omega^2\sig)-\frac{6M\Omega^2}{r}\et,
\end{multline}
where we used \eqref{eq:lin:xh3}, \eqref{eq:lin:etb3} and \eqref{eq:lin:be3}.
We now use the fact that 
\begin{equation*}
-2\Omega^2\D2\Ds2\et=-\Omega^2\Ds1\D1\et+2\Omega^2\et=\Omega^2(\sl\div\et-\sls\curl\et+2\et)=\Omega^2(\sl\div\et-\sls r\sig+2\et),
\end{equation*}
as well as equations \eqref{eq:lin:OmmA} and \eqref{eq:lin:divxhb} to conclude that $\sl\eqref{eq:lin:trx3}$ holds. 
Eq.~\eqref{eq:lin:trxb4} follows similarly.
\end{proof}

\begin{lemma}
The equations \eqref{eq:lin:om3} are satisfied by $\om$ and $\omb$, respectively. 
\end{lemma}
\begin{proof}
We compute $\Du\sl\om=\Du\Dv\sl\Omm=\Dv\sl\omb$ by using \eqref{eq:lin:OmmA}, \eqref{eq:lin:etb3}, \eqref{eq:lin:et4} to compute $\sl\pv\Omm$, and then using \eqref{eq:lin:beb4},      \eqref{eq:lin:be3}, \eqref{eq:lin:etb3} and \eqref{eq:lin:et4} to compute $\Du\Dv\sl\Omm$.
\end{proof}

\subsubsection*{Step \ref{steps:const:metric}: Constructing the remaining metric components $\gsh$, $\trg$ and $\b$}
\begin{defi}
Define $\b$ to be the solution to $\Du(r^{-1}\b)=2r^{-1}\Omega^2(\et-\etb)$ (i.e.~\eqref{eq:lin:b3}) with data $\radi\b$ for $r^{-1}\b$ at $\Scrim$.
\end{defi}
\begin{defi}
Define $\trg$ and $\gsh$ as solutions to \eqref{eq:lin:trg4} and \eqref{eq:lin:gsh4} with $\radc\trg$ and $\radc\gsh$ as data along $\Cin$.
\end{defi}
\begin{lemma}
The equations \eqref{eq:lin:trg3} and \eqref{eq:lin:gsh3} are satisfied by $\trg$ and $\gsh$.
\end{lemma}
\begin{proof}
We present the proof for \eqref{eq:lin:trg3}; the proof for \eqref{eq:lin:gsh3} is similar.
Since \eqref{eq:lin:trg3} is satisfied along $\Cin$, it suffices to prove that $\pv\eqref{eq:lin:trg3}$ holds. We prove the latter by computing $\pv\pu\trg=\pu\pv\trg$ via \eqref{eq:lin:trx3} and \eqref{eq:lin:b3}, and by computing $\pv(2\trxb)$ via \eqref{eq:lin:trxb4}. The result then follows.
\end{proof}
\begin{cor}
The following convergence is uniform in $v$ on compact $v$-intervals:
\begin{align}
\lim_{u\to-\infty}\trg =\radi\trg,&&\lim_{u\to-\infty}\gsh=\radi\gsh,&&\lim_{u\to-\infty}r^2\K=\radi \K.
\end{align}
These limits commute with $\Dv$-and $\sl$-derivatives.
\end{cor}
From the definition of $\b$ and the decay rates of $\et$ and $\etb$, it is easy to see that $r^{-1}\b$ converges uniformly in $v$ to $\b$ as $u\to-\infty$, and it thus follows that $\gsh$, $\trg$ as well as $r^2\K$ tend to $\radi\gsh$, $\radi\trg$ and $\radi\K$, respectively. 

The final equation of \fullsystem~that we have to prove is the Gauss equation:
\begin{lemma}
The Gauss equation \eqref{eq:lin:K} is satisfied, with $\K$ defined in \eqref{eq:lin:Kdef}.
\end{lemma}
\begin{proof}
We compute the $\pv$-derivative of the Gauss equation using \eqref{eq:lin:trg4}, \eqref{eq:lin:gsh4}, \eqref{eq:lin:rh4}, \eqref{eq:lin:trx4}, \eqref{eq:lin:trxb4} and \eqref{eq:lin:Omm4}. This shows that $\pv\eqref{eq:lin:K}$ is satisfied. The result now follows, as~\eqref{eq:lin:K} is satisfied along $\Cin$ by Prop~\ref{prop:setup:uniqueness of data on Cin}.
\end{proof}
We have now constructed a solution to \fullsystem, and we have shown that it satisfies Definition~\ref{def:setup:scattering solution}. This concludes the proof of Prop.~\ref{prop:construct:ellgeq2}.
\end{proof}

\subsection{Construction of the \texorpdfstring{$\ell<2$}{L<2}-part of the solution to the scattering problem}\label{sec:construct:existl<2}
We now present the construction of the $\ell<2$-part of the solution:
This construction will be entirely explicit:
\begin{prop}
Given a smooth seed scattering data set $\mathfrak{D}_{\ell=0}$ supported on $\ell=0$, let $\radc\rh$, $\radsinf{\K}$ be as in Prop.~\ref{prop:setup:uniqueness of data on Cin} (see \eqref{eq:setup:rad K at Sinfty}), and let $\radsinf\rh$ be as in Cor.~\ref{cor:setup:l<2}.
Then $\mathfrak{S}_{\mathfrak{m}}+\mathfrak{S}_{f}+\mathfrak{S}_{\underline{f}}$ is a scattering solution realising $\mathfrak{D}_{\ell=0}$, with
\begin{nalign}
\mathfrak{m}&=\frac{\radsinf\rh}{M}+3\radsinf\K\\
\underline{f}&=\frac{r}{6M\Omega^2}(-r^3\radc\rh-2\radsinf\rh+6M\radsinf\K)\\
f&=\int_{v_1}^v2(\radi\Om+\radsinf\K+\frac{\radsinf\rh}{4M})\dd v'.
\end{nalign}
The decomposition $\mathfrak{S}_{\mathfrak{m}}+\mathfrak{S}_{f}+\mathfrak{S}_{\underline{f}}$ is unique up to the ambiguity addressed in Remark~\ref{rem:gauge:constantgaugesolution}.
\end{prop}
\begin{proof}
It is easy to directly verify that this is a scattering solution realising $\mathfrak{D}_{\ell=0}$.
Still, we find it insightful to provide the construction:
We  begin by realising that only the ingoing gauge solution $\mathfrak{S}_{\underline{f}}$ and the nearby Schwarzschild solution $\mathfrak{S}_{\mathfrak{m}}$ have non-trivial limits for $r^3\rh|_{\Cin}$ and $r^2\K|_{\Cin}$. If we write $\lim_{u\to-\infty}r^{-1}\underline{f}=\underline{f}_0$, we then find the system of equations
\begin{align}
\radsinf\rh=-6M\underline{f}_0-2M\mathfrak{m},&&\radsinf\K=2\underline{f}_0+\mathfrak{m},
\end{align}
which is solved by $\mathfrak{m}=\frac{\radsinf\rh}{M}+3\radsinf\K$, $\underline{f}_0=-\radsinf\K-\frac{\radsinf\rh}{2M}$.
We have thus already determined~$\mathfrak{S}_{\mathfrak{m}}$.

Next, we fully fix $\underline{f}$ by demanding that
\begin{equation}
\radc\rh=-\frac{6M\Omega^2}{r^4}\underline{f}-\frac{2M\mathfrak{m}}{r^3}=-\frac{6M\Omega^2}{r^4}\underline{f}-\frac{2\radsinf\rh}{r^3}-\frac{6M\radsinf\K}{r^3}.
\end{equation}
Finally, we fix $f$ by noticing that both $\mathfrak{m}$ and $\underline{f}$ generate a limit of $\Om$ at $\Scrim$, by demanding that $f|_{\Cin}=0$ and by
\begin{equation}
\radi\Om=\tfrac{1}{2}\pv f+\tfrac12\left(\radsinf\K+\frac{\radsinf\rh}{2M}\right)-\frac{\radsinf\rh}{2M}-\frac{3\radsinf\K}{2}.
\end{equation}

It is left to show that this is indeed a scattering solution realising $\mathfrak{D}_{\ell=0}$:
By Corollary \ref{cor:setup:l<2}, it suffices to show that the constructed solution realises $\radsinf\rh$, $\radsinf\trg$, $\radi\Om$, $\radc\omb$ and $\rads\trx$:
Now, by construction, the solution $\mathfrak{S}_{\mathfrak{m}}+\mathfrak{S}_{f}+\mathfrak{S}_{\underline{f}}$ realises $\radsinf\rh$,  $\radsinf\K=-\frac12\radsinf\trg$ as well as~$\radi\Om$. Moreover, since it realises $\radc\rh$, we can deduce from \eqref{eq:lin:rh3} that it also realises $\radc\trxb$ (and thus $\rads\trxb$) by Prop.~\ref{prop:setup:uniqueness of data on Cin}. 
Finally, we deduce from \eqref{eq:lin:trxb3} that $\mathfrak{S}_{\mathfrak{m}}+\mathfrak{S}_{f}+\mathfrak{S}_{\underline{f}}$ also realises~$\radc\omb$.
\end{proof}
\begin{prop}
Given a smooth seed scattering data set $\mathfrak{D}_{\ell=1}$ supported on $\ell=1$, let $\radc\rh$, $\radsinf{\K}$ be as in Prop.~\ref{prop:setup:uniqueness of data on Cin}, and let $\radsinf\rh$, $\radsinf\beb$  be as in Cor.~\ref{cor:setup:l<2}.
Then $\mathfrak{S}_{\mathfrak{a}}+\mathfrak{S}_{f}+\mathfrak{S}_{\underline{f}}+\mathfrak{S}_{(q_1,q_2)}$ is a scattering solution realising $\mathfrak{D}_{\ell=1}$ for
\begin{nalign}
\mathfrak{a}_m&=-3M\sqrt{2}(\radsinf{\beb})^{\mathrm{H},1}_{\ell=1,m},\\
(q_1,q_2)_{\ell=1,m}&=\left(\frac{1}{4}\left(\frac{3}{2M}\radsinf\rh-\radsinf\trg\right),0\right))_{\ell=1,m}+\int_{v_1}^v \sqrt2((\radi\b)^{\mathrm{E},1}_{1,m},(\radi\b)^{\mathrm{H},1}_{1,m})\dd v',\\
f_{\ell=1,m}&=-\frac{1}{6M\sqrt{2}}(\radsinf{\beb})^{\mathrm{E},1}_{\ell=1,m}+\int_{v_1}^{v}2\left(\radi\Om-\frac{\radsinf\rh}{6M}\right)_{\ell=1,m}\dd v',\\
\underline{f}_{\ell=1,m}&=-\frac{r^4}{6M\Omega^2}(\radsinf\rh)_{\ell=1,m}+\frac{1}{\sqrt{2}\cdot 6M}(\radsinf\beb)^{\mathrm{E},1}_{\ell=1,m}.
\end{nalign}
\end{prop}
\begin{proof}
We uniquely determine $\mathfrak{a}$ directly from Prop.~\ref{prop:gauge:Kerr} and $\radsinf\beb$. (Note that none of the pure gauge solutions affect the magnetic part of $\beb$.)

Now, we choose a preliminary definition of $\underline{f}_p$ by demanding that $\mathfrak{S}_{\underline{f}_p}$ restricts to $\radc\rh$:
\begin{equation}
\underline{f}_p=-\frac{r^4}{6M\Omega^2}\radc\rh+\tilde{\underline{f}}(\theta^A).
\end{equation}
We will determine $\tilde{\underline{f}}(\theta^A)$ later. 
This choice for the ingoing gauge function induces a limit for $\trg$ at $\Scrim$. This limit is independent of $\tilde{\underline{f}}(\theta^A)$. Since the outgoing gauge function always induces a vanishing limit for $\trg$ at $\Scrim$, we can thus fix $q_1(v=v_1,\theta^A)$ by demanding that $\mathfrak{S}_{\underline{f}_p}+\mathfrak{S}_{(q_1,q_2)}$ restrict to $\radsinf\trg$:
\begin{equation}
\radsinf\trg=\frac{3}{2M}\radsinf\rh-4 q_1(v=v_1,\theta^A).
\end{equation}
We now fully fix $(q_1,q_2)$ by demanding that $q_2(v=v_1,\theta_A)=0$ and that $\Ds1(\pv q_1,\pv q_2)=\radi\b$.

Next, we fix $\pv f$ by demanding that the limit $\radi\Om$ is met:
\begin{equation}
\frac{\pv f}{2}=\radi\Om-\frac{\radsinf\rh}{12M},
\end{equation}
where the second term comes from $\underline{f}_p$.
We fully fix $f$ by demanding that the limit $\radsinf\beb$ is attained, i.e.:
\begin{equation}
f_{\ell=1,m}(v=v_1)=-\frac{1}{6M\sqrt{2}}(\radsinf{\beb})^{\mathrm{E},1}_{\ell=1,m}.
\end{equation}
Finally, we fix $\tilde{\underline{f}}$ by demanding that $\radc\rh$ is attained.

It is left to show that the constructed solution realises $\mathfrak{D}_{\ell=1}$. By construction, $\mathfrak{S}_{\mathfrak{a}}+\mathfrak{S}_{f}+\mathfrak{S}_{\underline{f}}+\mathfrak{S}_{(q_1,q_2)}$ realises $\radi\Om$  and $\radi\b$ as limits at $\Scrim$, and it realises $\radsinf\beb$, $\radsinf\trg$ as well as~$\radc\rh$. Moreover, in view of $\Du\frac{r^4\beb_{\ell=1}}{\Omega}=0$, it follows that $\radc\omb$ and $\rads\trxb$ are also attained by~\eqref{eq:lin:rh3} and \eqref{eq:lin:trxb3}. The result now follows from Corollary \ref{cor:setup:l<2} (as in the case $\ell=0$).
\end{proof}
\newpage

\section{The physical scattering data}\label{sec:physd}

The previous sections featured a general discussion of how to set up and solve the scattering problem for a general class of seed scattering data. 

In the present section, we specify this discussion to physically motivated seed data. 
We will first define the condition of no incoming radiation, and then give and physically justify various definition of seed scattering data describing the far-region of a system of $N$ infalling bodies, and finally collect a few basic consequences arising from these definitions that will be useful for the analysis in the later parts of this paper.
\subsection{The condition of no incoming radiation}\label{sec:physd:nir}
We first define the condition of no incoming radiation from $\Scrim$ (this definition is stated without any prior gauge fixing!):
\begin{defi}\label{def:noincomingradiation}
A set of smooth seed scattering data $\mathfrak{D}$ is said to satisfy the \emph{no incoming radiation condition} if the following condition is satisfied:
\begin{equation}\label{eq:physd:noincomingradiation}
\div\radi\xh=\sl\radsinf\K,
\end{equation}
where $\radsinf\K$ is defined in \eqref{eq:setup:rad K at Sinfty}. This is equivalent to demanding that
\begin{align}
\rad{\xh}{\Scrim}=\Ds2\Ds1(f_x,0),&& (\lap+2)f_x=2\radsinf\K
\end{align}
for some $f_x\in\Gamma^\infty(\Scrimv)$ with $\pv f_x=0$.
\end{defi}

The property of a seed scattering data set to have no incoming radiation cannot be changed by addition of a pure gauge solution (Prop.~\ref{prop:gauge:out}--Prop.~\ref{prop:gauge:sphere}). We are therefore justified in calling it a \textbf{gauge invariant condition}.

Conversely, we  recall from \S\ref{sec:setup:Bondinor} that the other quantities of $\mathfrak{D}$ at $\Scrim$, $\radi\b$ and $\radi\Om$, can be set to 0 by addition of pure gauge solutions.
Since $\radsinf\K$ can also be set to zero, the no incoming radiation condition therefore\textbf{ completely eliminates the radiative, physical degrees of freedom at $\Scrim$}.

Note that, in particular, if $\mathfrak{D}$ is as in Thm.~\ref{thm:setup:LEE Scattering wp} and satisfies the no incoming radiation condition, then all the radiation fields defined in Def.~\ref{def:construction:dataatSCRI} vanish. For later reference, we highlight that this in particular implies that $r\al$, $\Dv\Ps$ and $\Dv\Psb=\Dv(\frac{r^2}{\Omega^2}\Dv)^2(r\Omega^2\alb)$ all attain vanishing limits at $\Scrim$. Note that the no incoming radiation condition is strictly stronger than demanding that the limit of $\Dv\Ps$ vanishes.

\begin{rem}
The definition of no incoming radiation here is written down in such a way as to be gauge invariant. 
Most readers will be more familiar with a version of the no incoming radiation condition in the nonlinear theory where a gauge has already been fixed: 
provided one works in a gauge where the spheres at~$\Scrim$ are round (i.e. $\bm{r}^2\bm{K}\to1$ as $\bm{u}\to-\infty$), the no incoming radiation condition is captured by the requirement  that the Bondi mass along~$\Scrim$ be constant. In view of the mass growth formula  (this is simply the time reversal of the Bondi mass loss formula \cite{SeriesVII,SeriesVIII,CK93}),
\begin{equation}
\frac{\dd \bm{M}_{\mathrm{Bondi}}}{\dd \bm{v}}=\lim_{\bm{u}\to-\infty}\frac{1}{{16\pi}}\int_{\Stwo} |\bm{r}\hat{\bm{\chi}}|^2\dd \bm{\Omega},
\end{equation}
this clearly requires $\bm{r}\hat{\bm{\chi}}=0$ at $\Scrim$. 
This is exactly what our condition restricts to in the case of Bondi-normalised data (which have $\radsinf\K=0$).
\end{rem}
\begin{rem}\label{rem:physd:noincomingDUvanishes}
Clearly, it would not be enough to replace the condition of no incoming radiation \eqref{eq:physd:noincomingradiation} with a weaker condition such as demanding that $\radi\al=0$ along $\Scrim$, as this would be consistent with a constant $\P_{\Scrim}$ at past null infinity, which would correspond to a constant growth in energy. 
However, it is interesting to note that condition \eqref{eq:physd:noincomingradiation} is, in fact, equivalent to demanding that $\radi\al=0$ and $\lim_{u\to-\infty}\frac{r^2}{\Omega^2}\Du(r\Omega^2\radc\al)=0$, as can be seen from \eqref{eq:lin:al3} and Definition~\ref{def:construction:dataatSCRI}. In other words, the no incoming radiation condition forces $\radc\al=o(r^{-2})$ along~$\Cin$.
The additional requirement of finite Regge--Wheeler energy then means that we must have $\radc\al=o(r^{-5/2})$.
\end{rem}

\subsection{Seed data describing the far-field region of the hyperbolic \texorpdfstring{$N$}{N}-body problem}
The link to the physical scenario of $N$ infalling masses (with negligible internal structure, i.e.~point masses) moving along approximately hyperbolic orbits in the infinite past is made by introducing Post-Newtonian approximations to the Post-Minkowskian multipolar expansion framework; this has been described in detail in \S2 of \cite{IV}. 
In short, the argument in the Post-Minkowskian setting consists of writing down a general expression of an outgoing vacuum solution to the Einstein equations, which takes the following form for the curvature components of $\al$ and $\alb$:\footnote{Note that this is exactly analogous to the statement that the general fixed-angular frequency solution $\phi_\ell$ to the Minkowskian wave equation $\Box_\eta \phi=0$ takes the form
\begin{equation}
\phi_\ell=r^{\ell}\sum_{m=-\ell}^\ell \pu^{\ell}\left(\frac{f_{\ell,m}(u)\Ylm}{r^{\ell+1}}\right)+\pv^{\ell}\left(\frac{\bar{f}_{\ell,m}(v)\Ylm}{r^{\ell+1}}\right),
\end{equation} with the $\pv$-part of the solution vanishing because it is demanded that the solution be of purely outgoing type. }
\begin{align}\label{eq:physd:quadral}
\al_{\ell}&=\frac{(\ell+2)!}{(\ell-2)!}r^{\ell-2}\sum_{m=-\ell}^\ell \Du^{\ell-2}\left(\frac{I_{\ell,m}(u) \YlmE{2}-S_{\ell,m}(u)\YlmH{2}}{r^{3+\ell}}\right),\\
\alb_{\ell}&=r^{\ell-2}\sum_{m=-\ell}^\ell \Du^{\ell+2}\left(\frac{I_{\ell,m}(u) \YlmE{2}+S_{\ell,m}(u)\YlmH{2}}{r^{\ell-1}}\right).\label{eq:physd:quadralb}
\end{align}
for some real-valued functions $I_{\ell,m}$, $S_{\ell,m}$ depending only on $u$.
Notice that \eqref{eq:physd:quadral} and \eqref{eq:physd:quadralb} are exactly related via the Teukolsky--Starobinsky identity~\eqref{eq:lin:TSI+} with $M=0$.

These expressions, which are valid in the vacuum region of spacetime, are then matched to expressions for the matter region of spacetime (i.e.~the region of the $N$ bodies), which, in turn, are derived under a number of approximations and matching to the Newtonian theory. 
The result is that the $I_{\ell,m}$ acquire the interpretation of the $\ell$-th Newtonian mass multipole moment, and the $S_{\ell,m}$ acquire the interpretation of the $\ell$-th Newtonian current multipole moment. 
Finally, these moments are computed using the Newtonian theory (or perturbations thereof). Since, schematically, the $\ell$-th mass multipole moment goes like $ m\cdot r^\ell$, where~$m$ and~$r$ denote mass and size of the system, and since, for a system of masses following hyperbolic orbits, the size grows linearly in time, we thus find that $I_{\ell,m}(u)\sim |u|^{\ell}$ as $u\to-\infty$. 
The same rate is found for $S_{\ell,m}$, and it can thus be seen that $\al_{\ell}$, within this framework, is predicted to decay like $|u|^2/r^5\sim r^{-3}$ near past null infinity.
Similarly, one finds that $\alb_{\ell}$ decays like $r^{-4}$ near past null infinity, and that $\xhb_{\ell}$ decays like $r^{-2}$ near past null infinity.

On the other hand, for \textit{parabolic orbits}, for which the size grows like $r^{2/3}$, we find that $I_{\ell,m}(u)\sim |u|^{2\ell/3}$. Thus, in the case of parabolic orbits, we have that both $\al_{\ell=2}^{\mathrm{E}}$ and $\alb_{\ell=2}^{\mathrm{E}}$ decay like $|u|^{-11/3}$ towards $\Scrim$, with all other angular modes decaying faster. We refer the reader to \S2 of \cite{IV} for details.

We will now assume that this general information obtained from a perturbative framework around Minkowski is valid up until some finite advanced time, and we implement it in the context of linearised gravity around Schwarzschild in the following way, focusing first on the more important case of hyperbolic orbits.
\begin{defi}\label{defi:physd:Nbodyseed}
A seed scattering data set $\mathfrak{D}$ satisfying the no incoming radiation condition is said to \emph{describe the far-field region of a system of $N$ infalling masses following approximately hyperbolic Keplerian orbits} if the following conditions are satisfied:
\begin{enumerate}[label=(\Roman*)]
\item There exists $\delta>0$ and some $\albdata\in\Gamma^\infty(\bundlestfs(\Cin))$ with $\Du\albdata=0$ (and supported on all angular modes) such that 
\begin{equation}
\radc\alb=-6\albdata r^{-4}+\O_{\infty}(r^{-4-\delta}).
\end{equation}\label{item:defNBodyalb}
\item The limit $\lim_{u\to-\infty} r^2\radc\xhb=\radsinf\xhb$ exists and is non-vanishing (and is supported on all angular modes).\label{item:defNBodyxhb}
\item The limit $\lim_{u\to-\infty} r^3\radc\al=\aldata$ exists and is non-vanishing (and is supported on all angular modes).\label{item:defNBodyal}
\end{enumerate} 
Here, the quantities $\radc\alb$, $\radc\al$ and $\radc\xhb$ are as defined in Prop.~\ref{prop:setup:uniqueness of data on Cin}.
\end{defi}
\begin{rem}
The factor $-6$ is introduced for later notational convenience.
Condition \ref{item:defNBodyal} already partially follows from the previous conditions: The no incoming radiation condition implies that $\lim_{u\to-\infty} r\radc\al=0=\lim_{u\to-\infty} r^2\radc\al$, and condition \ref{item:defNBodyalb} implies that the limits $\lim_{u\to-\infty} r^3(\radc\rh,\radc\sig)=(\radsinf\rh, \radsinf\sig)$ exist. By condition \ref{item:defNBodyxhb}, and equations \eqref{eq:lin:divxhb} and \eqref{eq:lin:curleta}, we moreover have that $\radsinf\sig=\curl\div \radsinf\xhb\neq 0$.
 It then follows by construction (\eqref{eq:lin:al3} and \eqref{eq:lin:be3}) that $\lim_{u\to-\infty}r^3\radc\al=\Ds2\Ds1(\radsinf\rh,-\radsinf\sig)$.
 Thus, condition \ref{item:defNBodyal} only adds the additional requirement that $\radsinf\rh\neq 0$.
\end{rem}
\subsection{Seed data describing the far-field region of the parabolic \texorpdfstring{$N$}{N}-body problem}
Even though the seed data of Def.~\ref{defi:physd:Nbodyseed}, corresponding to hyperbolic orbits, will be the protagonist of this work, we still, for the sake of completeness, provide an analogous definition for the case of parabolic orbits:
\begin{defi}\label{defi:physd:parabolic}
   A seed scattering data set $\mathfrak{D}$ satisfying the no incoming radiation condition is said to \emph{describe the far-field region of a system of $N$ infalling masses following approximately parabolic Keplerian orbits} if the following conditions are satisfied:
\begin{enumerate}[label=(\Roman*)]
\item There exists  some $\albdata_{\mathrm{par}}\in\Gamma^\infty(\bundlestfs(\Cin))$ with $\Du\albdata_{\mathrm{par}}=0$ (and supported only on electric $\ell=2$ angular modes) such that 
\begin{equation}
\radc\alb=-\albdata_{\mathrm{par}}r^{-11/3}+\O_{\infty}(r^{-4}).
\end{equation}\label{item:defNBodyalbpara}
\item The limits $\lim_{u\to-\infty} r^2\radc\xhb$ and  $\lim_{u\to-\infty} r^3\radc\al$ vanish.
\end{enumerate} 
Here, the quantities $\radc\alb$, $\radc\al$ and $\radc\xhb$ are as defined in Prop.~\ref{prop:setup:uniqueness of data on Cin}. 
\end{defi}

\subsection{Seed data describing the massless \texorpdfstring{$N$}{N}-body problem}
Finally, we turn our attention to the case of massless "graviton" scattering. We can model this by simply posing compactly supported data  for the linearised gravitational field on $\Scrim$. The reader may want to think of these data as being localised around $N$ different points on along $\Scrim$. 
In our context, we can simply model compactly supported radiation along $\Scrim$ by posing trivial data along $\Cin$, and then posing compactly supported data for the incoming radiation measured by
\begin{equation}
    \div\rads\xh-\sl\radsinf\K
\end{equation}
along $\Scrim$, cf. Def.~\ref{def:noincomingradiation}.
For simplicity, we  assume our data to be Bondi-normalised, so that $\radsinf{\K}=0$ and incoming radiation is measured solely by $\rads\xh$.
Using the results of the previous section, it is then straight-forward to show the following
\begin{prop}\label{prop:physd:graviton}
Let $\mathfrak{D}$ be a  Bondi-normalised smooth seed scattering data set with trivial data on $\Cin$ and compactly supported data for $\radi\xh$ along $\Scrim$ such that there exists some $v_2$ such that $\radi\xh(v)=0$ for all $v\geq v_2$.
Then the corresponding solution satisfies along $\Cbar_{v_2}$ the following:
\begin{align}
    \lim_{u\to-\infty}r^2\xhb&=\int_{v_1}^{v_2}\radi\xh\dd v,\\
    \lim_{v\to-\infty}r^5\alb&=\int_{v_1}^{v_2}\int_{v_1}^{v}\int_{v_1}^{v'}(\lap-4)(\lap-2)\conj{\radi\xh}\dd v'' \dd v' \dd v,\\
    \lim_{v\to-\infty}r^{4}\Dv(r\Omega^2\alb)&=\int_{v_1}^{v_2}\int_{v_1}^{v}(\lap-4)(\lap-2)\conj{\radi\xh} \dd v' \dd v,\\
    \lim_{v\to-\infty}r^3\al&=\int_{v_1}^{v_2}(\lap-4)(\lap-2)\radi\xh \dd v.\label{eq:physd:graviton:al}
\end{align}
Furthermore, $r^5\al$ admits an asymptotic expansion in powers of $1/r$ towards $\Scrim$ along $\Cbar_{v_2}$.
\end{prop}
We thus make the following definition:
\begin{defi}\label{defi:physd:graviton}
    A Bondi-normalised seed scattering data set $\mathfrak{D}$ satisfying the no incoming radiation condition is said to \textit{describe the exterior of a compactly supported gravitational perturbation along $\Scrim$} if it arises from compactly supported $\radi\xh$ and trivial data on a null cone to the past of the support of $\radi\xh$ as in Proposition~\ref{prop:physd:graviton}. In particular, the rates of  Prop.~\ref{prop:physd:graviton} hold for such data.
\end{defi}
\begin{rem}
    In particular, a seed scattering data set as in Def.~\ref{defi:physd:graviton} is like a seed scattering data set as in Def.~\ref{defi:physd:Nbodyseed} for which $\albdata$ vanishes!
    As we will see later on (and as was already discussed in the introduction), this implies that solutions arising from either definition  will exhibit essentially the same asymptotic behaviour towards and along~$\Scrip$.
\end{rem}
\begin{rem}\label{rem:physd:moments}
    We also highlight that upon comparing the expression \eqref{eq:physd:graviton:al} with the Post-Newtonian prediction for infalling masses \eqref{eq:physd:quadral}, we could view the integral on the RHS of \eqref{eq:physd:graviton:al}, when restricted to fixed $\ell$-modes, as the massless and relativistic analogue of the Newtonian multipole moments.
\end{rem}

\subsection{A preliminary description of solutions arising from seed data describing the far-field region of the hyperbolic \texorpdfstring{$N$}{N}-body problem}
We now offer a preliminary description of solutions arising from seed data as in the previous definitions. We will give the most detailed description for the case of hyperbolic orbits, as this is more general than parabolic orbits and also includes the case of Def.~\ref{defi:physd:graviton}.

\begin{prop}\label{prop:physd}
Let $\mathfrak{D}$ be a smooth seed scattering data set as in Def.~\ref{defi:physd:Nbodyseed}. By Prop.~\ref{prop:setup:Bondi}, we can without loss of generality assume it to be Bondi-normalised in the sense of Def.~\ref{def:setup:Bondi}. Moreover, by Cor.~\ref{cor:setup:Bondiplus}, we can additionally assume that $\radc\omb=0$ and that $\rads\trxb=0$.

By Theorem~\ref{thm:setup:LEE Scattering wp}, there exists a unique scattering solution $\mathfrak{S}$ realising $\mathfrak{D}$. 
The data along~$\Cin$ induced by $\mathfrak{S}$ satisfy:
\begin{align}\label{eq:physd:alb}
\radc\alb&=-6\albdata r^{-4}+\O_{\infty}(r^{-4-\delta}),\\
\radc\beb&=-6\div\albdata r^{-4}\log r+\mathscr{C}_1 r^{-4}+\O_{\infty}(r^{-4-\delta}),\label{eq:physd:beb}\\
\radc\sig&=\radsinf\sig r^{-3}+6\curl\div \albdata r^{-4}\log r+\mathscr{C}_2 r^{-4}+\O_{\infty}(r^{-4-\delta}),\label{eq:physd:sig}\\
\radc\rh&=\radsinf\rh r^{-3}+6\div\div \albdata r^{-4}\log r+\mathscr{C}_3 r^{-4}+\O_{\infty}(r^{-4-\delta}),\label{eq:physd:rh}\\
\radc\be&=\Ds1(-\radsinf\rh,\radsinf\sig) r^{-3}-3\Ds1\overline{\D1}\D2\albdata r^{-4}\log r+\mathscr{C}_4  r^{-4}+\O_{\infty}(r^{-4-\delta}),\label{eq:physd:be}\\
\radc\al &=\aldata r^{-3}+2\Ds2\Ds1\overline{\D1}\D2\albdata r^{-4}\log r+\mathscr{C}_5  r^{-4}+\O_{\infty}(r^{-4-\delta}),\label{eq:physd:al}\\
\radc\Ps&=2\aldata +12\Ds2\Ds1\overline{\D1}\D2\albdata r^{-1}\log r+\mathscr{C}_6  r^{-1}+\O_{\infty}(r^{-1-\delta}),\label{eq:physd:Ps}\\
\radc\Psb&=2\overline{\aldata} +12\Ds2\Ds1{\D1}\D2\albdata r^{-1}\log r+\mathscr{C}_7  r^{-1}+\O_{\infty}(r^{-1-\delta}),\label{eq:physd:Psb}\\
\radc\psb&=-12\Ds2\D2\albdata r^{-2}\log r+\mathscr{C}_8  r^{-2}+\O_{\infty}(r^{-2-\delta}),\label{eq:physd:psb}
\end{align}
where we recall the notation $\overline{\aldata}$ to denote the magnetic conjugate (cf.~Def.~\ref{def:SS:magneticconjugate}), where $\Du\mathscr{C}_n=0$ for $n=1,\dots,11$---the precise value of these constants will not play a role in this paper---and where 
\begin{equation}
\aldata=\Ds2\Ds1(\radsinf\rh, -\radsinf\sig)=\Ds2\Ds1(\radsinf\rh, -\curl\div\radsinf\xhb).\label{eq:physd:aldata=rhosigma at infinity}
\end{equation}
Similarly, we have for the connection coefficients:
\begin{align}
\radc\trxb&=0=\radc\trg=\radc\omb=\radc\Om,\label{eq:physd:trxb}\\
\radc\xhb&=\radsinf\xhb r^{-2}+6\albdata r^{-3}+\O_{\infty}(r^{-3-\delta}),\label{eq:physd:xhb}\\
\radc\et&=-\radc\etb=\div \radsinf\xhb r^{-2}+6\div\albdata r^{-3}\log r+\mathscr{C}_8 r^{-3}+\O_{\infty}(r^{-3-\delta}),\label{eq:physd:et}\\
\radc\xh&=(-2\Ds2\D2-1)\radsinf\xhb r^{-2}-6\Ds2\D2\albdata r^{-3}\log r+\mathscr{C}_9 r^{-3}+\O_{\infty}(r^{-3-\delta}),\label{eq:physd:xh}\\
\radc\trx&=(2\div\div\radsinf\xhb+2\radsinf\rh)r^{-2}+12\div\div\albdata r^{-3}\log r+\mathscr{C}_{10} r^{-3}+\O_{\infty}(r^{-3-\delta}),\label{eq:physd:trx}\\
\radc \om&=-\radsinf\rh r^{-2}-3\div\div\aldata r^{-3}\log r+\mathscr{C}_{11} r^{-3}+\O_{\infty}(r^{-3-\delta}).\label{eq:physd:om}
\end{align}

On the other hand, all radiation fields defined in Def.~\ref{def:construction:dataatSCRI},~\ref{def:construction:dataatSCRI2} vanish, and we additionally have the following limits at $\Scrim$:
\begin{align}\label{eq:physd:limitsof:rho,sig,xhb}
\lim_{u\to-\infty} r^3\sig=\radi\sig=\radsinf\sig,&&\lim_{u\to-\infty} r^3\rh=\radi\rh=\radsinf\rh,&&\lim_{u\to-\infty} r^2\xhb=\radi\xhb=\radsinf\xhb.
\end{align}
\end{prop}
\begin{proof}
Equations \eqref{eq:physd:alb}--\eqref{eq:physd:al} follow from Prop.~\ref{prop:setup:uniqueness of data on Cin} and by consecutively integrating \eqref{eq:lin:beb3}--\eqref{eq:lin:al3}.
Equations \eqref{eq:physd:Ps}--\eqref{eq:physd:psb} follow from Cor.~\ref{cor:setup:psialongC}.

The last equality in \eqref{eq:physd:aldata=rhosigma at infinity} follows from \eqref{eq:lin:curleta} and \eqref{eq:lin:divxhb}.

Equations \eqref{eq:physd:trxb}--\eqref{eq:physd:om} also follow from Prop.~\ref{prop:setup:uniqueness of data on Cin} and by consecutive integration: Equation \eqref{eq:physd:trxb} follows from \eqref{eq:lin:trxb3} and \eqref{eq:lin:trg3}, equation \eqref{eq:physd:et} follows from \eqref{eq:lin:divxh}, equation \eqref{eq:physd:xh} follows from \eqref{eq:lin:xh3},  and so on.

Finally, \eqref{eq:physd:limitsof:rho,sig,xhb} follows from inspection of the proof of Theorem~\ref{thm:setup:LEE Scattering wp}. (Consider the $\Dv$-derivatives of the $\rh$, $\sig$ and $\xhb$ via \eqref{eq:lin:rh4}, \eqref{eq:lin:sig4} and \eqref{eq:lin:xhb4}. It follows from the proof of the theorem that these converge uniformly to some limit at $\Scrim$, which, by virtue of all limits in Definitions~\ref{def:construction:dataatSCRI} and~\ref{def:construction:dataatSCRI2} vanishing, vanish as well.)
\end{proof}

We can also prove a result analogous to that of Prop.~\ref{prop:physd} for the case of parabolic orbits. We leave the details to the reader and content ourselves with only stating the following:
\begin{prop}
    Let $\mathfrak{D}$ be a smooth scattering data set as in Def.~\ref{defi:physd:parabolic}. Then the corresponding $\radc\al$ satisfies:
    \begin{equation}\label{eq:physd:al:par}
        \radc\al=\frac{81}{80}\Ds2\Ds1\conj{\D1}\D2\albdata_{\mathrm{par}}r^{-11/3}+\O(r^{-4}).
    \end{equation}
\end{prop}

Similarly, we evidently have that Prop.~\ref{prop:physd} applies to the data of Def.~\ref{defi:physd:graviton} when setting $\albdata=0$.

\newpage

\newpage

\section*{Part II:\\ Asymptotic analysis of the Teukolsky equations and the Regge--Wheeler equations\hypertarget{V:partII}{}}
\addcontentsline{toc}{section}{{\textbf{Part II}: Asymptotic Analysis of the Teukolsky equations and the Regge--Wheeler equations}}
In the previous part of this paper, we carefully set up the scattering problem for linearised gravity around Schwarzschild, with scattering data posed on an ingoing null cone $\Cin$ and $\Scrim$, see~Fig.~\ref{fig:setup:Penrosediagram}, and we showed in~\S\S\ref{sec:setup}--\ref{sec:construct} how to construct the corresponding scattering solution in the domain of dependence $\DoD:=D^+(\Cin\cup\Scrimv)$.
Finally, in \S\ref{sec:physd}, we introduced three notions of scattering data describing the exterior of a system of $N$ infalling bodies---we will henceforth refer to these data as "the physical data"---and gave a preliminary characterisation of the corresponding scattering solutions. 

We will now set the foundations for finding the asymptotic properties in all of $\DoD$ for general scattering solutions, and we will apply this in particular to the scattering solutions arising from the physical data of \S\ref{sec:physd}. 
The key to this is an asymptotic analysis of the Teukolsky equations and the Regge--Wheeler equation (introduced in \S\ref{sec:lin:Teuk+RW})---this is the content of the present part of the paper.
Since both the Teukolsky and the Regge--Wheeler equations decouple from the rest of the system of equations \fullsystem, this part of the paper will have to make very little reference to the previous part of the paper, and it can essentially be viewed as independent work on the Teukolsky equations on Schwarzschild.
In fact, we will present the results of this part of the paper so as to be valid for the Teukolsky equations of any spin.
\paragraph{Overview of the structure of this part of the paper:}

In \S\ref{sec:cons}, we take the Teukolsky equations \eqref{eq:lin:Teukal}, \eqref{eq:lin:Teukalb}, re-write them in a unified way as equations for Teukolsky quantities~$\alphas$ of general spin $s$ ($\al$ corresponding to $s=+2$, $\alb$ corresponding to $s=-2$), and derive from them an infinite set of approximate/asymptotic conservation laws satisfied by angular modes of $\alphas$. We deduce similar conservation laws for the Regge--Wheeler equation and provide definitions of the generalised Newman--Penrose charges $\NPpq[\alphas]$ (cf.~\cite{MZ21,GK}). 

In the remaining sections, we then consider a scattering setup for the Teukolsky equation for $\alphas$ with scattering data on $\Cin\cup\Scrimv$, and we use the aforementioned conservation laws to derive asymptotic expressions for angular modes of the resulting scattering solutions:

First, in \S\ref{sec:al}, we prescribe $\alphas$ to decay with some general polynomial decay along $\Cin$, and we prescribe a condition along $\Scrim$ that, in the case $s\geq0$ is the no incoming radiation condition, but, in the case $s<0$, is slightly stronger than the no incoming radiation condition of Def.~\ref{def:noincomingradiation}. 
This will make it easier for us to later highlight the somewhat different behaviour of $\alphas$ for negative $s$. Under this setup, we prove a result for the general asymptotic behaviour throughout $\DoD$ of fixed angular modes of $\alphas$.

In \S\ref{sec:alp}, we then apply these results to $\al$ ($s=+2$)  with the data that we have physically motivated in~\S\ref{sec:physd}, e.g.~with $\al$ along $\Cin$ decaying as in~\eqref{eq:physd:al}. Due to the presence of the logarithmic term in the decay~\eqref{eq:physd:al}, this requires a very minor modification of the result of~\S\ref{sec:al}.

We can similarly apply the results of \S\ref{sec:al} to deduce the asymptotic behaviour of solutions to the Regge--Wheeler equation arising from physical initial data (\eqref{eq:physd:Ps} and \eqref{eq:physd:Psb}). This is done in \S\ref{sec:Psi}. 

In \S\ref{sec:alb}, we then apply the results of \S\ref{sec:al} in the case $s=-2$ to find the asymptotics of $\alb$ for the physical data of \S\ref{sec:physd} (eq.~\eqref{eq:physd:alb}). Again, a minor modification is necessary to account for the fact that our assumption in \S\ref{sec:al} is slightly stronger than the no incoming radiation condition. 
Furthermore, we will find that, for the special data of Def.~\ref{defi:physd:Nbodyseed}, a cancellation occurs, and the leading-order asymptotic results of \S\ref{sec:al} are no longer sufficient: To be precise, the leading-order results of \S\ref{sec:al} only imply that $\lim_{v\to\infty}r\alb=C\cdot |u|^{-2}+...$, but this constant $C$ is found to be vanishing. 

We find the next-to leading-order decay of $\lim_{v\to\infty}r\alb$ by introducing a decomposition of the Teukolsky equation that is similar to a Post--Minkowskian expansion and which allows us to find the next-to leading-order asymptotics of $\alb$ by solving a few recurrence relations. This is the content of \S\ref{sec:PM}.

Finally, with all the results of the previous sections having been obtained for angular modes of $\alphas$, we comment (but do not yet resolve) in \S\ref{sec:sum} on the issue of summing these fixed angular frequency estimates obtained in the previous sections. The resolution of this problem is the subject of upcoming work \cite{X}.
\section{The approximate conservation laws}\label{sec:cons}
In this section, we will derive an infinite set of \textit{approximate} conservation laws (which imply exact \textit{asymptotic} conservation laws, cf. \cite{IV}) for $\al$ and $\alb$ starting from the Teukolsky equations~\eqref{eq:lin:Teukal}, \eqref{eq:lin:Teukalb}.
Let us first fix some notation: Since $\Omega$ only appears with even integer powers in this part of the paper, we denote
\begin{equation}
D:=\Omega^2=1-\frac{2M}{r}.
\end{equation}
Recall that $\pv r=D=-\pu r$. 

\subsection{The Teukolsky equations for general spin \texorpdfstring{$s$}{s} satisfied by \texorpdfstring{$\alphas$}{alpha[s]}}
We begin by massaging the equation \eqref{eq:lin:Teukal} into a form that is more convenient for our purposes:
First, we rewrite
\begin{align*}
&\Dv(r^4D^{-2}\Du(rD\al))\\
=&\Dv(\Du(r^5\al D^{-1}))-\Dv(D^2r^{-4}\pu(r^4D^{-2})(r^5\al D^{-1}))\\
=&\Du(\Dv(r^5\al D^{-1}))-\pv\pu\log(r^4D^{-2})(r^5\al D^{-1})-\pu\log(D^{-2}r^4)\Dv(r^5\al D^{-1})\\
=&\frac{r^4}{D^2}\Du(\frac{D^2}{r^4}\Dv(\frac{r^5\al}{D}))-\pv\pu\log(r^4D^{-2})\frac{r^5\al}{D}
\end{align*}
We further have
\begin{equation*}
\pv\pu\log(D^{-2}r^4)=\frac{4D}{r^2}-\frac{24MD}{r^3}.
\end{equation*}
Using the two equations above, we rewrite the Teukolsky equation \eqref{eq:lin:Teukal} as follows:
\begin{equation}\label{eq:cons:teuk+}
\Du(r^{-4}D^2\Dv\frac{r^5D\al}{D^2})=\frac{D}{r^{2+4}}(\mathring{\slashed\Delta}-2+4)(r^5D\al)-\frac{30MD}{r^{3+4}}r^5D\al
\end{equation}

Next, we suggestively write \eqref{eq:lin:Teukalb} as
\begin{equation}\label{eq:cons:teuk-}
\Du(r^4D^{-2}\Dv(D^2\frac{r\alb}{D}))=\frac{D}{r^{2-4}}(\lap-2)\frac{r\alb}{D}-\frac{6MD}{r^{3-4}}\frac{r\alb}{D},
\end{equation}
and introduce the notation
\begin{align}\label{eq:cons:definitions of s-quantities}
\laps:=\lap+s,\quad 
\alphas:=\begin{cases}
r^5D^{-1}\al,&\quad s=2,\\
rD\alb,&\quad s=-2
\end{cases}
\end{align}
The two equations \eqref{eq:cons:teuk+} and \eqref{eq:cons:teuk-} can then be formulated in a unified way as
\begin{equation}\tag{Teuk:s}\label{eq:cons:teuk}
\Du(r^{-2s}D^s\Dv \alphas)=\frac{D^{s+1}}{r^{2+2s}}\laps\alphas-\frac{2MD^{s+1}}{r^{3+2s}}(1+s)(1+2s)\alphas, \quad s=\pm 2.
\end{equation}
We note that the operator $\laps$ has eigenvalues 
\begin{equation}\label{eq:cons:Lambda}
\lambdas:=-\ell(\ell+1)+s(s+1)=-(\ell-s)(\ell+s+1), \quad \ell\geq 2.
\end{equation}

Similarly, we have
\begin{align}\nonumber
&\Dv\left(\left(\frac{r^2}{D}\right)^s\Du(r^{-2s}D^s\alphas)\right)\\\nonumber
=&\left(\frac{r^2}{D}\right)^s\Du\left(\left(\frac{D}{r^2}\right)^s\Dv\alphas\right)-s\left(\pv\pu\log\frac{r^2}{D}\right){\alphas}\\\nonumber
=&\frac{r^{2s}}{D^{s}}\left(\frac{D}{r^{2}}\laps (r^{-2s}D^s\alphas)-\frac{2MD}{r^{3}}((1+s)(1+2s))(r^{-2s}D^s\alphas)\right)
-s\left(\frac{2D}{r^2}-\frac{12MD}{r^3}\right){\alphas}\\
=&\frac{r^{2s}}{D^{s}}\left(\frac{D}{r^{2}}\laps[-s](r^{-2s}D^s\alphas)-\frac{2MD}{r^{3}}((1-s)(1-2s))(r^{-2s}D^s\alphas)\right),
\label{eq:cons:Teuk'}\tag{Teuk':s}
\end{align}
where we used \eqref{eq:cons:teuk} in the third line as well as the relations $\laps-2s=\laps[-s]$ and $(1+s)(1+2s)-6s=(1-s)(1-2s)$ in the fourth.

\begin{rem}
Note that if a scalar function $\phi$ solves $\Box_g\phi=0$, then $\alphas[0]=r\phi$ satisfies \eqref{eq:cons:teuk} for $s=0$. 
Similarly, the Teukolsky equation for $s=\pm 1$ has the interpretation of describing the dynamics of  electromagnetic perturbations on Schwarzschild, $\alphas[\pm1]$ then being a 1-form. 
We can also make sense of \eqref{eq:cons:teuk} for $|s|\geq 2$ by regarding $\alphas$ as a {symmetric, trace-free} $\mathcal S_{u,v}$-tangent $(0,|s|)$-tensor, though the physical interpretation of the equation for $s>2$ is not clear. 
 
All  computations presented in the remainder of this and the next section are valid for arbitrary integer values of $s$. Nevertheless, the main interest will lie in the cases $s=\pm 2$.
\end{rem}

\subsection{The commuted Teukolsky equations}
We now prove a simple commutation formula, from which we will later derive an approximate conservation law.
\begin{prop}\label{prop:cons:commute}
Let $N\in\mathbb N$, and let $\alphas$ be a smooth solution to \eqref{eq:cons:teuk}. Then
\begin{multline}\label{eq:cons:commute:dv}
\Du\left(\left(\frac{D}{r^2}\right)^{N+s}\Dv\left(\rDv\right)^N\alphas\right)\\
=\left(\frac{D}{r^2}\right)^{N+s+1}\sum_{j=0}^1\left(\a{s}{N}{j}+\bb{s}{N}{j}\laps-\c{s}{N}{j}\frac{2M}{r}\right)(2M)^j\left(\rDv\right)^{N-j}\alphas.
\end{multline}

Similarly, we have
\begin{multline}\label{eq:cons:commute:du}
\Dv\left(\left(\frac{D}{r^2}\right)^{N-s}\Du\left(\rDu\right)^N(r^{-2s}D^s\alphas)\right)\\
=\left(\frac{D}{r^2}\right)^{N-s+1}\sum_{j=0}^1(-1)^j\left(\a{-s}{N}{j}+\bb{-s}{N}{j}\laps[-s]-\c{-s}{N}{j}\frac{2M}{r}\right)(2M)^j\left(\rDu\right)^{N-j}(r^{-2s}D^s\alphas).
\end{multline}

The constants $\a sNj,\bb sNj,\c sNj$ are given explicitly by
\begin{equation}
\a sN0=N(N+1+2s), \quad \bb sN0=1,\quad \c sN0=(1+s)(1+2s)+3N(N+1+2s)
\end{equation}
and 
\begin{align}
\a sN1=N(N+s)(N+2s),&& \bb sN1=0=\c sN1.
\end{align}
\end{prop}
\begin{proof}[Proof of Proposition~\ref{prop:cons:commute}]
We first prove \eqref{eq:cons:commute:dv}. Note that for $N=0$, it reduces to \eqref{eq:cons:teuk}. We now establish an inductive relation to prove it for all $N\in\mathbb N$. Observe:
\begin{align*}
&\Du\left(\left(\frac{D}{r^2}\right)^{N+s+1}\Dv\left(\rDv\right)^{N+1}\alphas\right)\\
=&\Du\Dv\left(\left(\left(\frac{D}{r^2}\right)^{N+s+1}\left(\rDv\right)^{N+1}\alphas\right)\right)-\Du\left(\pv\left(\left(\frac{D}{r^2}\right)^{N+1+s}\right)\left(\rDv\right)^{N+1}\alphas\right)\\
=&\left(\frac{D}{r^2}\right)^{N+s+1}\Dv\left(\left(\frac{r^2}{D}\right)^{N+s+1}\Du\left(\left(\frac{D}{r^2}\right)^{N+s+1}\left(\rDv\right)^{N+1}\alphas\right)\right)\\
&-\pv\log\left(\left(\frac{r^2}{D}\right)^{N+s+1}\right)\Du\left(\left(\frac{D}{r^2}\right)^{N+s+1}\left(\rDv\right)^{N+1}\alphas\right)\\
&-\Du\left(\pv\log\left(\left(\frac{r^2}{D}\right)^{-N-s-1}\right)\left(\frac{D}{r^2}\right)^{N+s+1}\left(\rDv\right)^{N+1}\alphas\right)\\
=&\left(\frac{D}{r^2}\right)^{N+s+1}\Dv\left(\left(\frac{r^2}{D}\right)^{N+s+1}\Du\left(\left(\frac{D}{r^2}\right)^{N+s+1}\left(\rDv\right)^{N+1}\alphas\right)\right)\\
&+\pu\pv\log\left(\left(\frac{r^2}{D}\right)^{N+s+1}\right)\cdot \left(\frac{D}{r^2}\right)^{N+s+1}\left(\rDv\right)^{N+1}\alphas.
\end{align*}
Since $\pu\pv\log r=\frac{D}{r^2}-\frac{4MD}{r^3}$, and since $\pu\pv\log D=\frac{4MD}{r^3}$, we have
\begin{align*}
\pu\pv \log\left(\left(\frac{r^2}{D}\right)^{N+s+1}\right)=(N+s+1)\left(\frac{2D}{r^2}-\frac{12MD}{r^3}\right).
\end{align*}
In conclusion, we thus obtain 
\begin{multline}\label{eq:cons:proof:inductive1}
\Du\left(\left(\frac{D}{r^2}\right)^{N+s+2}\left(\rDv\right)^{N+2}\alphas\right)\\
=\left(\frac{D}{r^2}\right)^{N+s+1}\Dv\left(\left(\frac{r^2}{D}\right)^{N+s+1}\Du\left(\left(\frac{D}{r^2}\right)^{N+s+1}\left(\rDv\right)^{N+1}\alphas\right)\right)\\
+\left(\frac{D}{r^2}\right)^{N+s+2}\left(\rDv\right)^{N+1}\alphas \left(2(N+1+s)-6\cdot\frac{2M}{r}(N+1+s)\right).
\end{multline}

We now assume \eqref{eq:cons:commute:dv} to hold for some fixed $N\in\mathbb N$. 
Using the inductive relation \eqref{eq:cons:proof:inductive1}, we then arrive at \eqref{eq:cons:commute:dv} with $N$ replaced by $N+1$, where
\begin{align*}
\a{s}{N+1}0=\a sN0+2(s+N+1),\quad \bb {s}{N+1}0=\bb sN0, \quad \c s{N+1}0=\c sN0+6(N+1+s)
\end{align*}
and 
\begin{align*}
\a{s}{N+1}1=\a sN1+\c{s}{N}{0}.
\end{align*}
It is left to solve these recurrence relations, whose initial values are provided by \eqref{eq:cons:teuk}: $\a s00=0=\a s01$, $\bb s00=1$, $\c s00=(1+s)(1+2s)$, the result being
\begin{nalign}
\a sN0=N(N+1+2s), \quad \bb sN0=1,\quad \c sN0=(1+s)(1+2s)+3N(N+1+2s)
\end{nalign}
 and
\begin{nalign}
\a{s}{N}1&=\sum_{i=0}^{N-1}\c{s}i0=\sum_{i=0}^{N-1}(1+2s)(1+s)+6si+3i(i+1)\\
&=(N-1)N(N+1)+3s(N-1)N+(1+s)(1+2s)N\\
&=N(N^2-1+3sN-3s+1+3s+2s^2)=N(N^2+3sN+2s^2)
=N(N+s)(N+2s).
\end{nalign}
This concludes the proof of \eqref{eq:cons:commute:dv}. 

Let's now move on to prove \eqref{eq:cons:commute:du}. The case $N=0$ is given by \eqref{eq:cons:Teuk'}.
We then obtain the inductive relationship:
\begin{multline}
\Dv\left(\left(\frac{D}{r^2}\right)^{N-s+2}\left(\rDu\right)^{N+2}(r^{-2s}D^s\alphas)\right)\\
=\left(\frac{D}{r^2}\right)^{N-s+1}\Du\left(\left(\frac{r^2}{D}\right)^{N-s+1}\Dv\left(r^{2s}\left(\frac{D}{r^2}\right)^{N+1}\left(\rDu\right)^{N+1}(r^{-2s}D^s\alphas)\right)\right)\\
+\left(\frac{D}{r^2}\right)^{N-s+2}\left(\rDu\right)^{N+1}(r^{-2s}D^s\alphas) (N-s+1)\left(2-6\cdot\frac{2M}{r}\right).
\end{multline}
Equation \eqref{eq:cons:commute:du} then follows in the same way as equation \eqref{eq:cons:commute:dv} did.
\end{proof}

\subsection{The approximate conservation laws satisfied by \texorpdfstring{$\Asl$}{A[s]ell}}
Consider a solution $\alphas$ to \eqref{eq:cons:teuk} supported only on angular frequency $\ell$. For such solutions, we have $(\a s{\ell-s}0+\bb s{\ell-s}0\laps)\alphas_\ell=0$, so there is a cancellation in equation \eqref{eq:cons:commute:dv} for $N=\ell-s$. 
The remaining term with a bad $r$-weight, the $\a s{\ell-s}1$-term, can then be removed by iteratively subtracting suitable multiples of \eqref{eq:cons:commute:dv} with $N=\ell-s-i$, $i\in\{1,\dots,\ell-s\}$. This leads to the following class of approximate conservation laws:
\begin{cor}
Let $\alphas_\ell$ be a smooth solution to \eqref{eq:cons:teuk} that is supported on some fixed $\ell\in\mathbb N_{\geq |s|}$.
Define
\begin{equation}
\at{s}{N}0:=\a{s}N0-\bb sN0\lambdas=(N+s-\ell)(N+s+\ell+1).
\end{equation}
If $\x{s}{i}{\ell}$ is any collection of constants with $\x si\ell=0$ if $i> \ell-s$, then, as a direct corollary of Proposition~\ref{prop:cons:commute}, 
\begin{multline}\label{eq:cons:cons0}
\Du\left(\left(\frac{D}{r^2}\right)^{\ell}\Dv\left(\sum_{i=0}^{\ell-s} \x s i\ell(2M)^i\left(\rDv\right)^{\ell-s-i}\alphas_\ell\right)\right)\\
=\sum_{i=0}^{\ell-s}\left(2(i+1)\x{s}{i+1}{\ell}\frac{2M}{r}\left(1-\frac{3M}{r}\right)+\x si\ell\cdot \at{s}{\ell-s-i}0+\x{s}{i-1}{\ell}\a{s}{\ell-s-i+1}1-\x si\ell\frac{2M}{r}\cdot \c{s}{\ell-s-i}{0}\right)\\
\cdot \left(\frac{D}{r^2}\right)^{\ell+1}(2M)^i\left(\rDv\right)^{\ell-s-i}\alphas_\ell.
\end{multline}

We now fix the constants $\x{s}i\ell$ via
\begin{equation}
\x si\ell:=(-1)^i\prod_{j=0}^{i-1}\frac{\a{s}{\ell-s-j}1}{\at{s}{\ell-s-j-1}0}=\frac{1}{i!}\frac{(2\ell-i)!}{(2\ell)!}\frac{(\ell-s)!\ell !(\ell+s)!}{(\ell-s-i)!(\ell-i)!(\ell+s-i)!}
\end{equation} 
such that the $a$-terms in \eqref{eq:cons:cons0} vanish. (Note that $\x si\ell=\x{-s}i\ell$, so, in particular $\x si\ell=0$ for $i\geq \ell-|s|$.) Finally, define
\begin{equation}\label{eq:cons:defofAsl}
\Asl :=\sum_{i=0}^{\ell-|s|} \x s i\ell(2M)^i\left(\rDv\right)^{\ell-s-i}\alphas_\ell.
\end{equation}
This quantity then satisfies
\begin{multline}\label{eq:cons:cons1}
\Du\left(\left(\frac{D}{r^2}\right)^{\ell}\Dv\Asl \right)\\
=\sum_{k=0}^{\ell-|s|}\left(2(k+1)\x{s}{k+1}{\ell}\frac{2M}{r}\left(1-\frac{3M}{r}\right)-\x sk\ell\frac{2M}{r}\cdot \c{s}{\ell-s-k}{0}\right)\\
\cdot \left(\frac{D}{r^2}\right)^{\ell+1}(2M)^k\left(\frac{r^2\Dv}{D}\right)^{\ell-s-k}\alphas_\ell.
\end{multline}
\end{cor}
\begin{rem}
An analogous result holds with $u$ and $v$ interchanged.
\end{rem}
\subsection{The modified Newman--Penrose charges}\label{sec:cons:NP}
The approximate conservation laws \eqref{eq:cons:cons1} of the previous subsection imply the conservation of weighted limits of $\Asl$. The following corollary mainly serves the purpose of providing us with some nomenclature for these limits. We therefore state some fairly strong assumptions that guarantee the existence and the conservation of these limits.

\begin{cor}
Let $\alphas$ be a smooth solution to \eqref{eq:cons:teuk} arising from scattering data on $\Cin\cup\Scrim$, and let $p\in\mathbb R$, $q\in\mathbb N_{\geq 0}$.
Suppose that the following bound holds in the domain of dependence $\DoD$ of $\Cin\cup\Scrim$:
\begin{equation}\label{eq:cons:NPdefinition:assumption}
    \left|\Dv(\rDv)^{\ell-s}\alphas_\ell\right|\leq C(u) \frac{\log^q r}{r^p}
\end{equation}
for some constant that is allowed to depend on $u$. If there exists a $u'\leq u_0$ such that the limit $\lim_{v\to\infty}\frac{r^p}{\log ^q r}\Dv(\rDv)^{\ell-s}\alphas_\ell(u',v)$ exists, then the limit
\begin{equation}
    \NPpq [\alphas](u)    := \lim_{v\to\infty}\frac{r^p}{\log^q r}\Dv\Asl (u,v)
\end{equation}
exists for all $u\leq u_0$ and is independent of $u$, $\Du \NP pq[\alphas]=0$. We call this limit the $(p,q)$-modified Newman--Penrose charge of $\alphas$.

Furthermore, if we only assume \eqref{eq:cons:NPdefinition:assumption} for $s=+2$, and assume that $\left| \Dv^2(\rDv)^{\ell-2}\alphas[+2]   \right|=o\left( \frac{\log^q r}{r^p}\right)$ and that $\alphas[+2]$ and $\alphas[-2]$ are related via \eqref{eq:lin:TSI-}, then
\begin{equation}
    \NPpq [\alphas[-2]]=2\Ds2\Ds1\overline{\D1}\D2\NPpq [\alphas[+2]].
\end{equation}
\end{cor}

\subsection{The Regge--Wheeler equations of general spin \texorpdfstring{$s$}{s} satisfied by \texorpdfstring{$\Psis$}{Psi[s]} and their conservation laws}
\label{sec:cons:derivationofRW}
We observe that if $s< 0$, then the commuted equation \eqref{eq:cons:commute:dv} for $N=|s|=-s$ gives the Regge--Wheeler equation (cf.~\eqref{eq:lin:RW}),  and if $s\geq 0$, then the same holds for \eqref{eq:cons:commute:du} with $N=|s|=s$. 
To be precise, we define
\begin{align}\label{eq:cons:RWdef}
\Psis:=\begin{cases}
(\rDv)^{|s|}\alphas,& \,\,\,\text{if}\,\,\, s\leq 0,\\
(\rDu)^s (r^{-2s}D^s\alphas),& \,\,\,\text{if}\,\,\, s> 0,
\end{cases}
\end{align}
and this quantity satisfies
\begin{equation}\label{eq:cons:RW}\tag{RW:s}
\Du\Dv\Psis=\frac{D}{r^2}\left((\laps-\Lambda_0^{[s]})-\cRW{s}{0}{0}\frac{2M}{r}\right)\Psis,
\end{equation}
where $\Lambda_0^{[s]}=-s(s+1)$ is the lowest eigenvalue of $\laps$ (eq.~\eqref{eq:cons:Lambda}), and where $\cRW{s}{0}{0}:=\c{-|s|}{|s|}{0}=1-s^2$.

Starting from \eqref{eq:cons:RW}, we can then, just as before, derive its commuted versions along with the resulting conservation laws:
\begin{prop}
Let $N\in\mathbb N$ and let $\Psis$ be a smooth solution to the Regge--Wheeler equation \eqref{eq:cons:RW}. Then
\begin{multline}\label{eq:cons:RWcomm}
\Du\left(\left(\frac{D}{r^2}\right)^{N}\Dv\left(\rDv\right)^N\Psis\right)\\
=\left(\frac{D}{r^2}\right)^{N+1}\sum_{j=0}^1\left(\aRW{s}{N}{j}+\bRW{s}{N}{j}(\laps-\Lambda_0^{[s]})-\cRW{s}{N}{j}\frac{2M}{r}\right)(2M)^j\left(\rDv\right)^{N-j}\Psis,
\end{multline}
with $\aRW{s}{N}{0}=\a{0}{N}{0}=N(N+1)$, $\bRW{s}{N}{0}=\bb{0}{N}{0}=1$, $\cRW{s}{N}{0}=3N(N+1)+(1-s^2)$ and $\aRW{s}{N}{1}=(N-s)N(N+s)$, and $\bRW sN1=0=\cRW sN1$.

Furthermore, for $\xRW{s}{i}{\ell}:=\x{s}{i}{\ell}$, the quantity 
\begin{equation}
\mathrm{\mathbf{\Psi}}^{[s]}_\ell:=\sum_{i=0}^{\ell-|s|}\xRW{s}{i}{\ell}(2M)^i\left(\rDv\right)^{\ell-i}\Psi_\ell
\end{equation}
satisfies (note that $\cRW{s}{N}0=\c{s}{N-s}0$ for $N\geq s$)
\begin{multline}\label{eq:cons:RWcons}
\Du\left(\left(\frac{D}{r^2}\right)^\ell\Dv\mathrm{\mathbf{\Psi}}^{[s]}_\ell\right)\\
=\sum_{k=0}^{\ell-|s|}\left(2(k+1)\xRW{s}{k+1}{\ell}\frac{2M}{r}\left(1-\frac{3M}{r}\right)-\xRW sk\ell\frac{2M}{r}\cdot \cRW{s}{\ell-k}{0}\right)\\
\cdot \left(\frac{D}{r^2}\right)^{\ell+1}(2M)^k\left(\frac{r^2\Dv}{D}\right)^{\ell-k}\Psis_\ell.
\end{multline}

Finally, provided the limit exists, we denote
\begin{equation}
    \NPpq [\Psi](u):=\lim_{v\to\infty}\frac{r^p}{\log^q r}\Dv \mathrm{\mathbf{\Psi}}^{[s]}_\ell(u,v)
\end{equation}
and call it the modified Newman--Penrose charge for $\Psis$.
\end{prop}
\begin{rem}
If $s\leq 0$, we see immediately that $\mathrm{\mathbf{\Psi}}^{[s]}_\ell=\Asl$.
For $s>0$, a similar statement is valid only to leading order, up to application of some angular operator of order $2s$. 
For instance, for $s=2$, we have (cf.~eq.~\eqref{eq:lin:Ps44}) $$\mathrm{\mathbf{\Psi}}^{[s]}_\ell=(\lap-2)(\lap-4)\Asl-6M(\Du+\Dv)\Asl.$$
 In particular, $\NPpq [\Psis[2]](u)=(\ell-1)\ell(\ell+1)(\ell+2)\NPpq[\alphas[2]](u)$ since the limit of $\frac{r^p}{\log^q r}\Dv(-6M)(\Du+\Dv)\Asl$ will vanish (for the same reason why $\NPpq[\alphas[2]](u)$ is conserved). 
\end{rem}
\newpage


\section{Asymptotic analysis of \texorpdfstring{$\alpha^{[s]}$}{alpha[s]} arising from general scattering data}\label{sec:al}

Having derived a class of approximate conservation laws (eq.~\eqref{eq:cons:cons1}) satisfied by angular modes $\alphas_\ell$ of  solutions $\alphas$ to the Teukolsky equation \eqref{eq:cons:teuk} with arbitrary integer spin~$s$, we now use these to derive precise asymptotic expressions for fixed angular modes $\alphas_\ell$ of scattering solutions~$\alphas$ arising from scattering data prescribed on $\Cin\cup\Scrimv$. As the following quantity will frequently appear throughout the remaining sections, we now introduce the notation
\begin{equation}
r_0(u):=r(u,v=v_1).
\end{equation}  
\paragraph{A note on conventions}
\newcommand{\E}{\mathscr{E}}
In the asymptotic computations to follow, we will often derive asymptotic decay/growth rates for expressions without computing the precise coefficients. For these situations, we will always just denote the coefficients by $\E$.

\subsection{Definition of scattering solutions and the no incoming radiation condition}
As we will study scattering solutions to \eqref{eq:cons:teuk} in this part of the paper, let us first define what constitutes seed scattering data, i.e.~data that determine a scattering solution provided they are sufficiently regular. 
These definitions are only relevant for the present section for a self-contained treatment of the Teukolsky equation of general spin. On the other hand, in the sections to follow, we will always apply our results directly to the scattering solutions constructed in the previous part of the paper (cf.~\S\ref{sec:physd}).

\begin{defi}\label{defi:al:data}
Seed scattering data for \eqref{eq:cons:teuk} consist of specifying $\radc\alphas$ along $\Cin$ and $\mathtt{P}^{[s]}_{\mathcal{I^-}}$ along $\Scrimv$.
If $s<0$, they additionally consist of the specifications of the $|s|$ tensor fields $\Dv^i\rad{\alphas}{\Sone}$ along $\Sone$, with $i=1,\dots,|s|$. 

We say that $\alphas$ is a scattering solution corresponding to these seed data if it solves \eqref{eq:cons:teuk}, if $\alphas|_{\Cin}=\radc\alphas$, if in the case $s<0$ the restrictions of the transversal derivatives to $\Sone$ satisfy $\Dv^i\alphas|_{\Sone}=\Dv^i\rad{\alphas}{\Sone}$ for $i=1,\dots,|s|$, and if $\lim_{u\to-\infty}  \left\|\Dv \Psis-\mathtt{P}^{[s]}\right\|_{L^2([v_1,v]\times \Stwo)}^2=0$ for any $v\geq v_1$, where $\Psis$ is related to $\alphas$ via \eqref{eq:cons:RWdef}.

Note that the transversal derivatives may alternatively be specified on any other sphere along $\Cin$.
\end{defi}
Throughout this entire part of the paper, we shall only study scattering solutions to \eqref{eq:cons:teuk} that satisfy the no incoming radiation condition: In view of the discussion in \S\ref{sec:physd:nir} (cf.~Rem.~\ref{rem:physd:noincomingDUvanishes}), we make the following
\begin{defi}\label{defi:al:nir}
Seed scattering data to \eqref{eq:cons:teuk} satisfy the no incoming radiation condition if $\mathtt{P}_{\Scrim}^{[s]}=0$, i.e.~if
\begin{equation}\label{eq:al:noincomingradiation}
\lim_{u\to-\infty}\|\Dv\Psis\|_{L^2([v_1,v]\times\Stwo)}=0.
\end{equation}
If $s>0$, we moreover require that:
\begin{equation}\label{eq:al:noincomingradiation2}
\lim_{u\to-\infty} \|(D^{-1}r^2\Du)^{i-1}(r^{-2s}D^s{\radc\alphas})\|_{L^2{(\mathbb S^2)}}=0,\quad i=1,\dots s.
\end{equation}
\end{defi}

\subsection{The main theorem (Thm.~\ref{thm:al:0})}
The seed scattering data for which we will prove the main result of this section are as follows: We take $\mathtt{P}^{[s]}_{\Scrim}=0$. Along $\Cin$, we specify  $\radc\alphas$ smooth via 
\begin{equation}\label{eq:al:decayalongCin}
\radc\alphas=\A^{[s]} r_0^{-p+s}+\O_{\max(s+1,0)}(r_0^{-p-\epsilon}),
\end{equation}
for some $\epsilon>0$, $\A^{[s]}\in\Gamma^\infty(\bundlestfss\Cin)$ with $\Du\A^{[s]}=0$, and $p$ satisfying
\begin{equation}\label{eq:al:prange}
p>\max\left(-s-1,-\tfrac{1}{2}\right).
\end{equation}
The range for $p$ guarantees that the seed data satisfy the no incoming radiation condition.
From now on, we will drop the superscript ${}^{[s]}$ from $\A^{[s]}$.	

For $s<0$, we recall that it is additionally necessary to specify all transversal derivatives up to order $-s$ along $\Sone$. 
For the present section, we will simplify our lives by assuming that along $\Sinfty$:\footnote{The restriction on the decay rate \eqref{eq:al:prange} is necessary to ensure that this determines the transversal derivatives along all of $\Cin$, cf.~Prop.~\ref{prop:al:transversal_derivatives}.}
\begin{equation}\label{eq:al:transderivativesdata}
r^{-2s}\Dv^i\radsinf\alphas=0,\quad \text{for  } i=1,\dots,-s,\quad \text{if  }s<0.
\end{equation}
Recall that if $s<0$, the no incoming radiation condition only states that a weighted derivative of order $-s+1$ vanishes along $\Scrim$ (cf.~\eqref{eq:al:noincomingradiation}). Condition \eqref{eq:al:transderivativesdata} has the effect that in addition all lower-order derivatives (up to order $-s$) vanish as well.
This is a simplifying assumption that is independent of the condition of no incoming radiation. We will remove it in \S\ref{sec:alb}, where the case $s<0$ is discussed in the context of the physical data of \S\ref{sec:physd}. (Notice that the "physical" decay rate of $\alphas[-2]$ along $\Cin$ is given by $r_0^{-3}=r_0^{-1-2}$ (cf.~\eqref{eq:physd:alb}), which {marginally} violates the condition \eqref{eq:al:decayalongCin} that $p>-s-1=1$.)

The content of \S\ref{sec:al} is the proof of the following

\begin{thm}\label{thm:al:0}
Let $s\in\mathbb Z$.
Then there exists a unique smooth scattering solution $\alphas$ to~\eqref{eq:cons:teuk} that satisfies the no incoming radiation condition \eqref{eq:al:noincomingradiation}--\eqref{eq:al:noincomingradiation2}, restricts along~$\Cin$ to $\radc\alphas$ as specified in \eqref{eq:al:decayalongCin}, and, if $s<0$,  satisfies \eqref{eq:al:transderivativesdata} in the sense of Definition~\ref{defi:al:data}.
Moreover, for $\ell\geq|s|$, each angular mode~$\alphas_\ell$ of this solution has the following asymptotic expression throughout $\DoD=D^+(\Cin\cup\Scrimv)$:
\begin{align}\label{eq:al:thm}
\begin{split}
    \alphas_{\ell}=&\A_{\ell} r_0^{s-p}\left(\sum_{n=0}^{\max(\ell-s,\lceil p-s \rceil)} \left(\frac{r_0}{	r}\right)^n (S_{\ell,p,\ell-s,n,s}+\O^u(r_0^{-\epsilon})+\O^u(Mr_0^{-1}))\right)\\
    &-(-1)^{\ell-s} M\A_\ell    \frac{(\ell-s)!(s+p)!}{(\ell+s)!}(\ell+p+1)\cdot r^{s-1-p}\cdot F_{\ell,p,s}(r,r_0),
\end{split}
\end{align}
where the constants $S_{\ell,p,\ell-s,n,s}$ are defined in \eqref{eq:al:SLPJN} for $n\leq \ell-s$ and are defined to  vanish for $n>\ell-s.$ If $\ell-p\geq0$, then $F_{\ell,p,s}(r,r_0)$ is given by
\begin{align}\label{eq:al:thm_coef}
    \begin{cases}  (s-p-2)!\cdot  \{(\ell-p)&
\\ +\O((\tfrac{r_0}{r})^{\lceil p \rceil -p}+\tfrac{r_0}{r}(1+\delta_{s-2-p,0}\log r/r_0) +r^{-\epsilon}(1+\delta_{s-p-1,\epsilon}\log r/r_0))\},&1+p-s \notin \mathbb N_{\geq 0},\\
\frac{(-1)^{p+1-s}}{(p+1-s)!}((\ell-p)\log r/r_0+\O(\frac{r_0}{r}\log r/r_0+r^{-\epsilon}(1+\delta_{\epsilon,1}\log r/r_0)),&1+p-s \in \mathbb N_{\geq 0}.
\end{cases}
\end{align}
If, on the other hand, $\ell-p<0$, then $F_{\ell,p,s}(r,r_0)$ is given by
\begin{align}\label{eq:al:thm_coef2}
    \begin{cases}  \mathscr{E}_{\ell,p,s}+\O((\tfrac{r_0}{r})^{\lceil p \rceil -p}+\tfrac{r_0}{r}(1+\delta_{s-2-p,0}\log r/r_0) +r^{-\epsilon}(1+\delta_{s-p-1,\epsilon}\log r/r_0))\},&1+p-s \notin \mathbb N_{\geq 0},\\
\mathscr{E}'_{\ell,p,s}\log r/r_0+\O(\frac{r_0}{r}\log r/r_0+r^{-\epsilon}(1+\delta_{\epsilon,1}\log r/r_0)),&1+p-s \in \mathbb N_{\geq 0}.
\end{cases}
\end{align}
for some constants $\mathscr{E}_{\ell,p,s}$, $\mathscr{E}'_{\ell,p,s}$ that we do not determine.

Moreover, if $\ell-p\geq -1$, then the solution has finite, non-zero and conserved $\ell$-th Newman--Penrose charge $\NP{p'}{q}[\alphas]$ for $p'=2-\ell+p$, $q=\delta_{\ell-p,-1}$; its value can be read off from Proposition~\ref{prop:al:pain:Asl}.
\end{thm}
\begin{rem}\label{rem:al:uniqueness}
The uniqueness asserted in Theorem~\ref{thm:al:0} is with respect to the class of \textit{smooth} solutions $\alphas$ with finite Regge--Wheeler energy, cf.~\S\ref{sec:RWscat}. 
The requirement that the solution be smooth (or sufficiently regular) can be dropped if $s<0$ or if $s\in\{0,2\}$. For the other cases, one can also drop the requirement by simply generalising the arguments of \S\ref{sec:RWscat:teukunique} to general $s>0$. Since the cases $s=\pm 2$ are the main focus of the paper,  we leave this to the interested reader.
\end{rem}
\begin{rem}
It should be straight-forward to prove that \eqref{eq:al:thm_coef} is valid for the entire range of~$p$. 
We leave this as an exercise to the interested reader. (This amounts to explicitly computing the coefficients $\mathscr{E}$ in \eqref{eq:al:prop:painrDvN:l-p<0:limit}--\eqref{eq:al:prop:painrDvN:l-p<0:2}.)
\end{rem}
\begin{rem}\label{rem:elldependence}
The $\ell$-dependent constants hiding in the $\O$ terms in \eqref{eq:al:thm} grow too fast to directly extract an asymptotic expression for the entire solution $\alphas=\sum_{\ell}\alphas_\ell$. 
We return to this point in~\S\ref{sec:sum}. 
In general, we expect the estimates of Thm.~\ref{thm:al:0} to be summable provided that the initial data have sufficient angular regularity. But since we are not making any statements on the summability of the estimates here, the regularity of the error term specified in \eqref{eq:al:decayalongCin} suffices.
\end{rem}
The theorem allows to also directly infer statements about solutions to the Regge--Wheeler equation \eqref{eq:cons:RW}:
\begin{cor}\label{cor:al:0}
Prescribe scattering data for $\Psis$ satisfying $\Psis_{\Cin}=\A r_0^{-p}+\O_1(r_0^{-p-\epsilon})$ together with the no incoming radiation condition \eqref{eq:al:noincomingradiation}. Then the resulting scattering solution satisfies the estimate~\eqref{eq:al:thm} with $\alphas$ replaced by $\Psis$ and with $s$ on the RHS of \eqref{eq:al:thm} and in \eqref{eq:al:thm_coef}, \eqref{eq:al:thm_coef2} replaced by~0.
\end{cor}
We will prove this corollary at the end of this section.

\subsection{Overview of the proof}
The proof consists of the following steps:
\begin{enumerate}
\item First, in \S\ref{sec:al:transversal}, we compute sufficiently many transversal derivatives $(\rDv)^n\alphas$ along~$\Cin$ (in an a priori fashion) by inductively integrating equation \eqref{eq:cons:commute:dv} from $\Sinfty$.
\item We then establish the existence of a scattering solution and prove a preliminary pointwise decay estimate (this decay is sharp near $\Scrim$ but fails to be sharp near $\Scrip$) for $\alphas$ and its derivatives using the Regge--Wheeler energy estimate (Lemma~\ref{lem:RWscat:energyidentity}) in \S\ref{sec:al:energy}. 
Note that control over the Regge--Wheeler energy along $\Cin$ is obtained from tangential $\Du$-derivatives along $\Cin$ in the case $s\geq 0$, whereas, in the case of $s<0$, it is obtained from \textit{transversal} $\Dv$-derivatives. Cf.~Section~\ref{sec:cons:derivationofRW}.
\item Starting from this preliminary decay estimate, we next use the approximate conservation law \eqref{eq:cons:cons1} to derive an asymptotic expression for the quantity $\Dv \Asl$ in \S\ref{sec:al:asyconservationlaws}.
\item The asymptotic estimate for $\Dv\Asl$ allows us to obtain an asymptotic estimate for $\Dv\left(\rDv\right)^{\ell-s}\alphas$. 
We then integrate this latter estimate up to $\ell-s$ terms from $\Cin$, using the estimates obtained in step a, to obtain asymptotic estimates for all lower order derivatives, including $\alphas$ itself. This final step is carried out in \S\ref{sec:al:final}.
\end{enumerate}
In principle, we could then use all the asymptotic estimates for the lower-order derivatives obtained in step d above to arrive at a more precise asymptotic estimate for $\Dv\Asl$ and iterate, but we will not do this here.
\begin{rem}
In the following, all the estimates arising from integrating along characteristics (e.g.\ $\int (\Dv \alphas) \dd v$) will be comparing the components of $\alphas$ in the parallelly propagated frame $(e_1,e_2)$; we will be using that, in view of \eqref{eq:SS:puvsDu},
\begin{equation}
\Du(D^\ell r^{-2\ell}\Dv(D^{-1}r^2\Dv)^{\ell-s}\alphas)_{A\dots B}=\pu(D^\ell r^{-2\ell}\pv(D^{-1}r^2\pv)^{\ell-s}\alphas_{A\dots B}).
\end{equation}
For fixed angular modes $\alphas_\ell$, another way of viewing this is to write $\alphas_\ell=\sum_{m=-\ell}^\ell (\alpha^{[s],\text E}_{\ell,m}\YlmE{|s|}+\alpha^{[s],\text H}_{\ell,m}\YlmH{|s|})$,
then we are just comparing the scalar functions $\alpha^{[s],\text E}_{\ell,m}(u,v)$, $\alpha^{[s],\text H}_{\ell,m}(u,v)$ at different points in spacetime.
\end{rem}

\subsection{Computing transversal derivatives along \texorpdfstring{$\Cin$}{C-in}}\label{sec:al:transversal}
The following is to be understood as an a priori estimate:
\begin{prop}\label{prop:al:transversal_derivatives}
Assume that a scattering solution $\alphas$ attaining the data of Theorem~\ref{thm:al:0} exists, and assume additionally that $\lim_{u\to-\infty}r^{-2s}\Dv^n \alphas=0$ for any $n> \max(0,-s)+1$. Then we have the following expressions for transversal derivatives along $\Cin$ if $N\leq \ell-s$:
\begin{equation}\label{eq:al:transversal_derivatives}
\left.(D^{-1}r^2\Dv)^N \alphas_\ell\right|_{\Cin}=\A_{\ell} r_0^{N-p+s}\underbrace{\frac{(s+p)!}{(N+s+p)!}\prod_{i=0}^{N-1}\at{s}{i}{0}}_{:=\ctr{N}}+\mathcal O(r^{N-p+s-\epsilon}+Mr^{N-p+s-1}).
\end{equation}
The product on the RHS can be computed as
\begin{align}\label{eq:al:transversal_derivatives_coefficients}
\prod_{i=0}^{N-1}\at{s}{i}0 =(-1)^N\frac{(\ell-s)!(N+s+\ell)!}{(\ell-s-N)!(\ell+s)!}.
\end{align}
In particular, 
\begin{equation}
\ctr{\ell-s}=(-1)^{\ell-s}\frac{(s+p)!(\ell-s)!(2\ell)!}{(\ell+p)!(\ell+s)!}.
\end{equation}
On the other hand, if $N>\ell-s$, then we have that
\begin{equation}\label{eq:al:transversal_derivatives_beyond}
\left.(\rDv)^N \alphas\right|_{\Cin}=M\cdot \E r_0^{N-p+s-1}+M\mathcal O(r^{N-p+s-1-\epsilon}+Mr^{N-p+s-2})
\end{equation}
for some constants $\E$ whose precise values won't matter.
\end{prop}

\begin{proof}
The proof proceeds by inductively integrating \eqref{eq:cons:commute:dv} from $\Sinfty$. The estimate~\eqref{eq:al:transversal_derivatives} holds true for $N=0$. Assume now that $N<\ell-s$ is fixed, and assume \eqref{eq:al:transversal_derivatives} to hold for all $n\leq N$. Then we can integrate \eqref{eq:cons:commute:dv} from $\Sinfty$, where we use the assumption that $\lim_{u\to-\infty}r^{-2s}\Dv^i\alphas_{\Scrim}=0$ for any $i\geq 0$, to obtain:
\begin{multline*}
\left(\frac{D}{r^2}\right)^{N+s}\Dv(D^{-1}r^2\Dv)^N\alphas\\
=\int \left(\frac{D}{r^2}\right)^{N+s+1}\at sN0 \cdot  \A_{\ell} r_0^{N+s-p}\frac{(s+p)!}{(N+s+p)!}\prod_{i=0}^{N-1}\at{s}{i}{0}\,\dd u+\mathcal O(r^{N+1+s-p-\epsilon}+Mr^{N+s-p}).
\end{multline*} 
This proves \eqref{eq:al:transversal_derivatives}. Equation \eqref{eq:al:transversal_derivatives_coefficients} is a direct computation.
The last claim follows in a similar fashion, using now that $\at{s}{\ell-s}0=0$.
\end{proof}
\subsection{A robust preliminary decay estimate based on energy conservation}\label{sec:al:energy}
\begin{prop}\label{prop:al:prelim_decay}
There exists a unique smooth scattering solution in $\DoD$ attaining the data described in Theorem~\ref{thm:al:0}. This solution satisfies the assumptions of Prop.~\ref{prop:al:transversal_derivatives} and moreover satisfies the following pointwise bounds throughout $\DoD$:
\begin{align}\label{eq:al:prelim:prop}
    |\alphas_\ell|\leq \ell^{1+\max{(-s,0)}} C r_0^{-p-\frac12}\max(r^{s+\frac12},r_0^{s+\frac12})
\end{align}
for some constant $C$ independent of $\ell$.
If $p> 0$, we have the following sharper estimate for any $N\geq 0$:
\begin{align}\label{eq:al:prelim:propN}
\left|\alphas_\ell\right|\leq \ell^{1+\max{(-s,0)}} C\max\left(r_0^{s-p},r^{s-p},\delta_{p,s}\cdot(\log\tfrac{r}{r_0}+1)\right)
\end{align}
for another constant $C$ independent of $\ell$.
\end{prop}
\begin{proof}
This proof mostly just recollects the ideas from \S\ref{sec:RWscat}:

Recall that, given a smooth solution $\alphas$ to the Teukolsky equation \eqref{eq:cons:teuk}, we know that the derived quantity $\Psis$, defined in \eqref{eq:cons:RWdef}, solves the Regge--Wheeler equation~\eqref{eq:cons:RW}. 
By assumption~\eqref{eq:al:decayalongCin} in the case $s\geq0$, and by assumption~\eqref{eq:al:transderivativesdata} and Prop.~\ref{prop:al:transversal_derivatives} in the case $s<0$, we then know that $\Psis|_{\Cin}=\O_{1}(r^{1/2-\epsilon'})$ for some $\epsilon'>0$; in particular, it has finite Regge--Wheeler energy~$E_{v_1}[\Psis|_{\Cin}](-\infty,u_0)$ as defined in Def.~\ref{defi:RWSCAT:energy}. 
This means that our scattering theory for the Regge--Wheeler equation from \S\ref{sec:RWscat} applies and that it provides us a with a unique scattering solution~$\Psis$ (cf.~Thm~\ref{thm:RWscat:mas20 RW}).\footnote{Note that the statements of \S\ref{sec:RWscat} also apply to the Regge--Wheeler equation of general spin \eqref{eq:cons:RW}.} 
We can then reconstruct out of $\Psis$ a unique solution $\alphas$ by appropriately integrating the definition \eqref{eq:cons:RWdef} from $\Cin$ in the case $s<0$.
On the other hand, if $s>0$, then we must specify in addition the values of $X^{(i)}:=\lim_{u\to-\infty}(D^{-1}r^2\Du)^{i}(r^{-2s}D^s\alphas))$ at $\Scrim$ for $i=0,\dots,s-1$. It turns out that, as a consequence of these limits vanishing at $\Cin$ (cf.~\eqref{eq:al:noincomingradiation2}), the unique choice for these limits is $X^{(i)}\equiv0$; we will defer the proof of this statement to the end of the proof.

It is moreover standard to show that all derivatives $r^{-2s}\Dv^N\alphas\to 0$ as $u\to-\infty$ for any $N\geq 0$ (cf.~\eqref{eq:RWscat:uniformconvergence} of Prop.~\ref{prop:RWscat:blabla}). 

We now move on to prove the estimate \eqref{eq:al:prelim:prop}:
By the Regge--Wheeler $T$-energy identity (cf.~Lemma~\ref{lem:RWscat:energyidentity}), we have that for any $v\geq v_1$:
\begin{align}
    \underline{E}_{v}[\Psis](-\infty,u)\leq  \underline{E}_{v_1}[\Psis](-\infty,u)+\lim_{u'\to-\infty}{E}_{u'}[\Psis](v_1,v).
\end{align}
The limit on the RHS vanishes by assumption \eqref{eq:al:noincomingradiation}.

We now focus on fixed angular modes: Either directly by the initial data assumption in the case $s\geq 0$, or by Prop.~\ref{prop:al:transversal_derivatives} if $s<0$, it follows that, along $\Cin$,
\begin{nalign}\label{eq:al:energyinitial}
\underline{E}_{v_1}[\Psis_\ell](-\infty,u)&=\int_{v=v_1, u'\in(-\infty,u)}|\Du\Psis_\ell |^2+\frac{D\ell(\ell+1)|\Psis_\ell |^2}{r^2}+\frac{6MD|\Psis_\ell |^2}{r^3}\dd u' \\
&\leq \ell^{2+\max(-2s,0)} C r_0^{-2p-1}
\end{nalign}
for some constant $C$ that is allowed to change from line to line and doesn't depend on $\ell$.
We can exploit the positive potential term in \eqref{eq:al:energyinitial} to further estimate for any $v\geq v_1$:
\begin{equation}
\int_{v'=v, u'\in(-\infty,u)} r^2|\Du(r^{-1}\Psis_\ell) |^2\,\dd u'\lesssim   \underline{E}_{v}[\Psis](-\infty,u).
\end{equation}

Now, applying a simple fundamental theorem of calculus argument together with Cauchy--Schwarz and the previous three estimates, we obtain (this is essentially the estimate \eqref{eq:RWscat:apriori}):
\begin{equation}\label{eq:al:prelim:poo}
|r^{-1}\Psis_\ell|\lesssim r^{-1/2} \sqrt{\underline{E}_{v}[\Psis_\ell](-\infty,u)}\leq \ell^{1+\max(-s,0)} C r_0^{-p-1/2}r^{-1/2}.
\end{equation}
Integrating this $|s|$ times either from $\Cin$ or from $\Scrim$ (in which case we also use the vanishing of the corresponding limits), we then obtain the lazy estimate 
\begin{equation}\label{eq:al:prelim:pointwiseproof0}
|\alphas_\ell|\leq \ell^{1+\max{(-s,0)}} C \max(r^{s+1/2},r_0^{s+1/2})   r_0^{-1/2-p}.
\end{equation}
This proves \eqref{eq:al:prelim:prop}. 

 Let us now assume that $p>0$. In order to then remove the factor $\sqrt{r/r_0}$ in \eqref{eq:al:prelim:poo}, we need to consider the $u$-weighted version of the energy estimate given in Lemma~\ref{lem:RWscat:uweighted}. Recall from Prop.~\ref{prop:RWscat:uweighted} that the quantity
 \begin{multline}
 \underline{E}^{1+\delta}_{v}[\Psis_\ell](-\infty,u)\\
 :=\int_{v=const, u'\in(-\infty,u)}|u'|^{1+\delta}\left( |\Du\Psis_\ell |^2+\frac{D\ell(\ell+1)|\Psis_\ell |^2}{r^2}+\frac{6MD|\Psis_\ell |^2}{r^3}\right)\dd u'
 \end{multline}
 decays in $v$ for any $\delta>-1$. 
Since for any $p>0$, there exists a $\delta>0$ such that $\underline{E}^{1+\delta}_{v_1}[\Psis_\ell](-\infty,u) $ is finite at $\Cin$, we can then repeat the previous argument, replacing \eqref{eq:al:prelim:poo} with 
\begin{equation}
|\Psis_\ell| \leq |u|^{-\delta/2}\sqrt{\underline{E}^{1+\delta}_{v_1}[\Psis_\ell](-\infty,u) }\leq \ell^{1+\max(-s,0)} C r_0^{-p}.
\end{equation} 

Estimate~\eqref{eq:al:prelim:propN} then again follows by appropriate integration of \eqref{eq:cons:RWdef}, this time keeping precise track of the weights, using also the integral formulae from Lemma~\ref{lem:appB:A1}.

Let's now finally return to the issue (arising only for $s>0$) of showing that the limits $X^{(i)}:=\lim_{u\to-\infty}(D^{-1}r^2\Du)(r^{-2s}D^s\alphas)$ must vanish. Clearly, taking them to vanish produces a smooth solution restricting correctly to the data, so the question reduces to a uniqueness statement, i.e.~to showing that the limits $X^{(i)}$ vanish for trivial scattering data. Here, we content ourselves with proving uniqueness within the class of smooth (finite-energy) solutions (cf.~Remark~\ref{rem:al:uniqueness}).
For trivial scattering data, the resulting solution will have $\Psis=0$, and so, in general, we have (by integrating $\Psis=0$)
 \begin{align}
        r^{-2s}D^{s}\alphas=\sum_{i=0}^{s-1}\frac{X^{(i)}(v,\theta^A)}{i!r^{i}}.
    \end{align}
We now first apply to this expression the formula \eqref{eq:cons:commute:du} with $N=s-1$---this is simply the formula for $\Dv\Psis$. By the no incoming radiation condition, we thus get
\begin{equation}\label{eq:al:DVPSI}
    (\a{-s}{s-1}{0}+\laps[-s])X^{(s-1)}=2M\a{-s}{s-1}1 X^{(s-2)}.
\end{equation}
On the other hand, applying \eqref{eq:cons:commute:du} with $m-1$ for $m<s$, we obtain the identity
\begin{equation}\label{eq:al:DVXm}
    \Dv X^{(m)}=(\laps[-s]+\a{-s}{m-1}0)X^{(m-1)}-2M\a{-s}{m-1}1X^{(m-2)},
\end{equation}
of which a simple iterate gives 
\begin{equation}\label{eq:al:iterate}
    \Dv^m X^{(m)}=P^{(m)}X^{(0)}
\end{equation}
for some $m$-th order constant coefficient differential operator, where both $\laps[-s]$ and $\Dv$ attain order at most $\lceil m/2\rceil$.
We now act with $\Dv^{s-1}$ on \eqref{eq:al:DVPSI} using \eqref{eq:al:iterate} to obtain
\begin{equation}
    \( (\a{-s}{s-1}{0}+\laps[-s])P^{(s-1)}-2M\a{-s}{s-1}1 \Dv P^{(s-2)}\)X^{(0)}=0.
\end{equation}
For each fixed angular frequency, this is a homogeneous ODE of order $\lceil \frac{s-1}{2}\rceil$, and we can easily see from \eqref{eq:cons:Teuk'} that the initial data at $v=v_1$ vanish, so $X^{(0)}$ vanishes everywhere (for all $v\geq v_1$). We  similarly deduce that all other $X^{(i)}$ vanish for $i<s$.
\end{proof}

\subsection{Deriving asymptotic expressions using the approximate conservation laws}\label{sec:al:asyconservationlaws}
Equipped with the pointwise estimate \eqref{eq:al:prelim:prop} derived from the Regge--Wheeler energy estimate, we shall now use the approximate conservation law \eqref{eq:cons:cons1} to derive asymptotic expressions for $\alphas_\ell$.
Appealing to \eqref{eq:cons:cons1}, while providing us with a simple and intuitive understanding for these asymptotic expressions, means that we cannot directly make statements uniform in $\ell$, we will therefore stop keeping track of $\ell$-dependencies of constants in estimates. See already \S\ref{sec:sum}.
\subsubsection{A preliminary estimate on \texorpdfstring{$(\rDv)^{\ell-s+1}\alphas$}{l-s+1 transversal derivatives of alpha[s]}}
We start by proving the following estimate:
\begin{prop}\label{prop:al:BS}
The solution from Theorem~\ref{thm:al:0} satisfies throughout $\DoD$:
\begin{equation}\label{eq:al:bsprop}
    \left|(\rDv)^{\ell-s+1}\alphas_{\ell}\right|\leq M \cdot C \max(r^{\ell-p},r_0^{\ell-p}).
\end{equation}
\end{prop}
\begin{proof}
\newcommand{\Cbs}{C_{\mathrm{bs}}}
As a first step, notice that, repeating the computations of Prop.~\ref{prop:al:transversal_derivatives}, estimate \eqref{eq:al:prelim:prop} directly implies 
\begin{equation*}
    \left|(\rDv)^N \alphas_{\ell}\right|\leq C_N r_0^{-1/2-p}r^{s+N+1/2},\quad \forall N\in\mathbb N,
\end{equation*}
for some constants $C_N$.
We now write this as
\begin{equation}\label{eq:al:bs0}
    \left|(\rDv)^N \alphas_{\ell}\right|\leq f_N(v) \max(r^{N+s-p},r_0^{N+s-p}),\quad \forall N\in\mathbb N,
\end{equation}
for $f_N(v)$ continuously depending on $v$, and we will show show that there exists a sufficiently large global constant $\Cbs$ (depending also on $M$), to be determined later, such that
\begin{equation}\label{eq:al:bs1}
   \left |(\rDv)^{\ell+1-s} \alphas_{\ell}\right|\leq \Cbs  \max(r^{\ell+1-p},r_0^{\ell+1-p}).
\end{equation}
Consider the set $X:=\{v\in[v_1,\infty) \text{ s.t. \eqref{eq:al:bs1} holds } \forall v'\leq v, u\leq u_0 \}$.
If $\Cbs $ is sufficiently large, then either $\sup_X=\infty$, or $\sup_X=v_2>v_1$ is finite. If the first case holds, we are done, so let us assume the second case. 

In view of \eqref{eq:al:transversal_derivatives}, integrating \eqref{eq:al:bs1} immediately implies 
\begin{equation}
\left|(\rDv)^{\ell-s+1-i} \alphas_{\ell}\right|\leq C\cdot \Cbs \max(r^{\ell+1-i-p},r_0^{\ell+1-i-p})
\end{equation} 
for all $i\in\{1,\dots, \ell-s+1\}$. 
Inserting these estimates for $i=1,2$ into the approximate conservation law \eqref{eq:cons:cons1}, we then obtain that 
\begin{nalign}
\left|r^{2}D^\ell \Dv\Asl\right|\lesssim r^{2\ell+2} \int\frac{1}{r^{3+2\ell}}\Cbs \max(r^{\ell-p},r_0^{\ell-p})\lesssim\Cbs \max(r^{\ell-p},r_0^{\ell-p}),
\end{nalign}
where we used that the boundary term at $\Scrim$ vanishes.
Hence, recalling the definition of~$\Asl$,~\eqref{eq:cons:defofAsl} and using the estimates \eqref{eq:al:bs1}, we deduce that
\begin{nalign}\label{eq:al:bs2}
\left| (\rDv)^{\ell-s+1}\alphas_{\ell}\right| 
\lesssim \Cbs  \max(r^{\ell-p},r_0^{\ell-p}).
\end{nalign}
Estimate \eqref{eq:al:bs2} improves \eqref{eq:al:bs1} within $X$, provided that $\Cbs$ (and thus $v_2$) was chosen sufficiently large. We can now integrate \eqref{eq:al:bs0} with $N=\ell+2-s$ from $v=v_2$ to $v=v_2+\delta$ for a sufficiently small distance $\delta>0$ to show that $(v_2+\delta)\in X$, a contradiction. Thus, $\sup_X=\infty$, and~\eqref{eq:al:bs2} holds throughout~$\DoD$.
\end{proof}
\subsubsection{An asymptotic estimate for \texorpdfstring{$\Dv\Asl$}{Dv(A[s],ell}}
Next, we upgrade \eqref{eq:al:bsprop} to an asymptotic expression for $\Dv\Asl$:
\begin{prop}\label{prop:al:pain:Asl}
The solution from Theorem~\ref{thm:al:0} satisfies the following estimates throughout~$\DoD$:
If $\ell-p>-1$, then
\begin{multline}\label{eq:al:prop:painAsl:l-p>-1}
\rDv\Asl=2M\A_{\ell}\ctr{\ell-s}(2\x s1{\ell}-\c s{\ell-s}0) \cdot \frac{(\ell-p)!(\ell+p+1)!}{(2\ell+2)!}\cdot r^{\ell-p}\\
+M\cdot \O\left(\frac{r_0+M(1+\delta_{\ell-p,0}\log r/r_0)}{	r^{1-(\ell-p)}}+\frac{r_0^{\ell-p+1}}{r}+\frac{\delta_{\ell-p-\epsilon,-1}\log(r/r_0)+1}{r^{\epsilon-(\ell-p)}}\right).
\end{multline}

If $\ell-p\in\mathbb N_{\geq 0}$, then we have more precisely:
\begin{multline}\label{eq:al:prop:painAsl:l-p>0 precise}
\rDv\Asl=2M\A_{\ell}\ctr{\ell-s}(2\x s1{\ell}-\c s{\ell-s}0) \cdot \frac{(\ell-p)!(\ell+p+1)!}{(2\ell+2)!}\cdot r^{\ell-p}\\
\cdot \left(1+\sum_{i=1}^{\ell-p} \frac{1 }{i!}\(\frac{r_0}{r}\)^i\prod_{j=1}^i (\ell+1+p+j)\right)     +M\cdot \O\left(\frac{M(1+\delta_{\ell-p,0}\log r/r_0)}{	r^{1-(\ell-p)}}+\frac{\delta_{\ell-p-\epsilon,-1}\log(r/r_0)+1}{r^{\epsilon-(\ell-p)}}\right).
\end{multline}
Note that $\frac{(2\x s1{\ell}-\c s{\ell-s}0)}{(2\ell+2)!}=-\frac{1}{2(2\ell)!}$.

If $\ell-p=-1$, then
\begin{equation}\label{eq:al:prop:painAsl:l-p=-1}
\rDv\Asl=2M\A_{\ell}\ctr{\ell-s}(2\x s1{\ell}-\c s{\ell-s}0) \cdot\frac{\log(r/r_0)}{r}+M\cdot \O(r^{-1}).
\end{equation}

If $\ell-p<-1$, then
\begin{equation}\label{eq:al:prop:painAsl:l-p<-1}
\left|\rDv\Asl\right|\leq M C\frac{r_0^{\ell-p+1}}{r}.
\end{equation}

Finally, we have that
\begin{multline}
        \rDv\Asl=(\rDv)^{\ell-s+1}\alphas_{\ell} +2M\x s1\ell \A_{\ell}\ctr{\ell-s}\cdot r_0^{\ell-p}\\
        +M\O(r_0^{\ell-p-\epsilon}
        +\max(r^{\ell-p-1},r_0^{\ell-p-1})(1+\delta_{l-p,1}\log r/r_0))
\end{multline}
\end{prop}
\begin{proof}
Integrating the estimate \eqref{eq:al:bsprop} from $\Cin$, we obtain that
\begin{nalign}\label{eq:al:asyproof0}
\left|(\rDv)^{\ell-s}\alphas_{\ell}-\left.(\rDv)^{\ell-s}\alphas_{\ell}\right|_{\Cin}\right|
\lesssim M\cdot  C \max(r^{\ell-p-1},r_0^{\ell-p-1})(1+\delta_{\ell-p,1}(\log r/r_0)),
\end{nalign}
and integrating \eqref{eq:al:asyproof0} from $\Cin$, using also the estimates \eqref{eq:al:transversal_derivatives} to estimate terms on $\Cin$,  gives
\begin{equation}\label{eq:al:asyproof1}
 \left|(\rDv)^{\ell-s-i}\alphas_{\ell}\right|\leq C r_0^{\ell-p-i}+M\cdot C\max(r^{\ell-p-i-1},r_0^{\ell-p-i-1})(1+\delta_{\ell-p-i,1}\log r/r_0).
\end{equation}
(Here, we used Lemma~\ref{lem:appB:hardxxx} and the fact that $\frac{1}{r}\log \frac{r}{r_0}\lesssim \frac{1}{r_0}$.)
We can now prove the proposition by re-inserting the improved estimates \eqref{eq:al:asyproof0}, \eqref{eq:al:asyproof1} into \eqref{eq:cons:cons1}:
\begin{multline*}
\Du(r^{-2\ell}D^\ell\Dv\Asl)=\frac{2M}{r^{2\ell+3}}(2\x s1\ell -\c s{\ell-s}0)(\rDv)^{\ell-s}\alphas_{\ell}|_{\Cin} \\
+\frac{M^2}{r^{2\ell+3}}\O(\max(r^{\ell-p-1},r_0^{\ell-p-1})(1+\delta_{\ell-p,1}\log r/r_0)).
\end{multline*}
Integrating in $u$, we thus obtain that
\begin{multline}\label{eq:al:painAslproof:integral}
        \left|r^{-2\ell} \Dv\Asl-\A_{\ell}\ctr{\ell-s}\int_{-\infty}^u \frac{2M r_0^{\ell-p}}{r^{2\ell+3}}(2\x s1\ell -\c s{\ell-s}0) \dd u' \right|\\
    \leq \frac{M^2 C(1+\delta_{\ell-p,0}\log r/r_0)}{r^{\ell+3+p}}+\frac{M C(1+\delta_{\ell-p-\epsilon,-1}\log \tfrac{r}{r_0})}{r^{\ell+2+p+\epsilon}},
\end{multline} 
where we used Lemmata \ref{lem:appB:A1}--\ref{lem:appB:funnylog} to estimate the RHS (note, in particular, that $\int r^{-N}\log r/r_0\dd u\leq r^{-N+1}$ for any $N\geq 1$ by Lemma~\ref{lem:appB:funnylog}).

We now need to evaluate the integral on the LHS of \eqref{eq:al:painAslproof:integral}.
If $\ell-p<-1$, then \eqref{eq:al:prop:painAsl:l-p<-1} follows immediately.
If $\ell-p=-1$, then \eqref{eq:al:prop:painAsl:l-p=-1} follows from Lemma~\ref{lem:appB:careerlog}.
If $\ell-p>-1$, then~\eqref{eq:al:prop:painAsl:l-p>-1} follows from Lemma~\ref{lem:appB:A1}, \eqref{eq:appB:lemA1:1}.

The refined statement \eqref{eq:al:prop:painAsl:l-p>0 precise} for $\ell-p\in\mathbb N$ follows from \eqref{eq:appB:lemA1:2}, which implies:
\begin{equation*}
    \int \frac{r_0^{\ell-p}}{r^{2\ell+3}}\dd u=\frac{(\ell-p)!(\ell+1+p)!}{(2\ell+2)!}\frac{1}{r^{\ell+2+p}}\left(1+\sum_{i=1}^{\ell-p}\frac{1}{i!}\left(\frac{r_0}{r}\right)^i\frac{(\ell+1+p+i)!}{(\ell+1+p)!}+\O(M/r)\right).
\end{equation*}
\end{proof}

\subsubsection{Asymptotics for \texorpdfstring{$\Dv(\rDv)^{\ell-s}\alphas_{\ell}$ and higher derivatives}{derivatives of order ell-s+1 and higher}}
We now compute the asymptotics for $\Dv(\rDv)^{\ell-s}\alphas_{\ell}$. Even though we could just deduce them from the asymptotics of $\Asl$, we find it easier to compute them directly.
\begin{prop}\label{prop:al:pain:rDvN}
The solution from Theorem~\ref{thm:al:0} satisfies the following estimates throughout~$\DoD$:
If $\ell-p> 0$, then 
\begin{multline}\label{eq:al:prop:painrDvN:l-p>0}
(\rDv)^{\ell-s+1}\alphas_{\ell} =-M\A_{\ell}\ctr{\ell-s}\frac{(\ell-p)!(\ell+p+1)!}{(2\ell)!}r^{\ell-p}+M\O(r_0^{\ell-p}+(r_0+M) r^{\ell-p-1} )\\
+M^2 \O(\max(r^{\ell-p-1},r_0^{\ell-p-1}))+M\O (\max (r^{\ell-p-\epsilon},r_0^{\ell-p-\epsilon})(1+\delta_{\ell-p,\epsilon}\log r/r_0)).
\end{multline}
If $\ell-p\in\mathbb N_{\geq 0}$, then we have more precisely:
\begin{multline}\label{eq:al:prop:painrDvN:l-p>0 precise}
(\rDv)^{\ell-s+1}\alphas_{\ell} =-M\A_{\ell}\ctr{\ell-s}\frac{\c s{\ell-s}0+\a s{\ell-s}1}{\ell+1}r_0^{\ell-p}-M\A_{\ell}\ctr{\ell-s}\frac{(\ell-p)!(\ell+p+1)!}{(2\ell)!}r^{\ell-p}\\
\cdot \left(1-\delta_{\ell-p,0}+\sum_{i=1}^{\ell-p} \frac{1 }{i!}\left(\frac{r_0}{r}\right)^i\prod_{j=1}^i (\ell+1+p+j)\right)     +M\cdot \O\left(\frac{M}{	r^{1-(\ell-p)}}+\frac{\delta_{\ell-p,\epsilon}\log(r/r_0)+1}{r^{\epsilon-(\ell-p)}}\right).
\end{multline}
If $\ell-p<0$, then, for all $1\leq i<\lceil p-\ell \rceil+1$, the limit of $(\rDv)^{\ell-s+i}\alphas_{\ell}$ as $v\to\infty$ exists, and
\begin{equation}\label{eq:al:prop:painrDvN:l-p<0:limit}
\lim_{v\to\infty}(\rDv)^{\ell-s+i}\alphas_{\ell}=\E\cdot M r_0^{\ell-p+i-1}+M\O(r_0^{\ell-p+i-2}+r_0^{\ell-p+i-1-\epsilon}),
\end{equation} 
for some constants $\E$ whose values we don't determine. 
For $1\leq i\leq \lceil p-\ell\rceil$, we have
\begin{equation}\label{eq:al:prop:painrDvN:l-p<0:1}
(\rDv)^{\ell-s+i}\alphas_{\ell}=\lim_{v\to\infty}(\rDv)^{\ell-s+i}\alphas_{\ell}+M\O(r^{-1}r_0^{\ell-p+i}), 
\end{equation}
and, for $i=\lceil p-\ell \rceil+1$, we have
\begin{multline}\label{eq:al:prop:painrDvN:l-p<0:2}
(\rDv)^{\ell-s+i}\alphas_{\ell}=\lim_{v\to\infty}(\rDv)^{\ell-s+i}\alphas_{\ell}\\
+M\cdot \E r^{\ell-p-1+\lceil p-\ell\rceil }(1+\delta_{p-\ell,\lceil p-\ell \rceil}\log r/r_0)+M\O(r^{-1}r_0^{\ell-p+1}).
\end{multline}
\end{prop}
\begin{proof}
The first part of the proof is very similar to the proof of Proposition~\ref{prop:al:pain:Asl}, with the difference that now we want to integrate \eqref{eq:cons:commute:dv} rather than \eqref{eq:cons:cons1}. 
Recall the estimate \eqref{eq:al:asyproof0}: Integrating it, we get
\begin{multline}
(\rDv)^{\ell-s-1}\alphas_{\ell} =\restr{(\rDv)^{\ell-s-1}\alphas_{\ell}}{\Cin}+\left(\frac{1}{r_0}-\frac1r\right)(\rDv)^{\ell-s}\alphas_{\ell}|_{\Cin}\\+M\O(\max(r^{\ell-p-2},r_0^{\ell-p-2})(1+\delta_{\ell-p,2}\log r/r_0)).
\end{multline}
We insert these estimates into \eqref{eq:cons:commute:dv}, with $N=\ell-s$, to obtain
\begin{nalign}\label{eq:al:asyproof2.1}
&\Du\left(\frac{D^\ell}{r^{2\ell}}\Dv(\rDv)^{\ell-s}\alphas_{\ell}\right)\\
=&-\frac{2M}{r}\left(\frac{D}{r^2}\right)^{\ell+1} (\c s{\ell-s} 0+\a s{\ell-s}1 )\restr{(\rDv)^{\ell-s}\alphas_{\ell}}{\Cin}\\
&+2M\a s{\ell-s}1\left(\frac{D}{r^2}\right)^{\ell+1} \left(\restr{(\rDv)^{\ell-s-1} \alphas_{\ell}}{\Cin}+\frac{1}{r_0}\restr{(\rDv)^{\ell-s}\alphas_{\ell}}{\Cin}\right)\\
&+M^2 r^{-2\ell-2}\O(\max(r^{\ell-p-2},r_0^{\ell-p-2})(1+\delta_{\ell-p,2}\log r/r_0)\\
=& -2M (\c s{\ell-s} 0+\a s{\ell-s}1 )\ctr{\ell-s} \frac{r_0^{\ell-p}}{r^{2\ell+3}}+\a s{\ell-s}1 (\ctr{\ell-s-1}+\ctr{\ell-s})\frac{r_0^{\ell-p-1}}{r^{2\ell+2}}\\
&+M^2r^{-2\ell-2}\O(\max(r^{\ell-p-2},r_0^{\ell-p-2})(1+\delta_{\ell-p,2}\log r/r_0)+M\O(\frac{r_0^{\ell-p-1-\epsilon}}{r^{2\ell+2}}).
\end{nalign}
Now, since 
\begin{align*}
\frac{r_0^{\ell-p}}{r^{2\ell+3}}=\pu\left(\frac{1}{2\ell+2}\frac{r_0^{\ell-p}}{r^{2\ell+2}}\right)+\frac{\ell-p}{2\ell+2}\frac{r_0^{\ell-p-1}}{r^{2\ell+2}}+M\O(\tfrac{r_0^{\ell-p-2}}{r^{2\ell+2}}),
\end{align*}
and since $\ctr{\ell-s-1}=-\frac{\ell+p}{2\ell}\ctr {\ell-s}$, integrating \eqref{eq:al:asyproof2.1} gives
\begin{multline}\label{eq:al:asyproof2.2}
(\rDv)^{\ell-s+1}\alphas_{\ell}=-2M\frac{\c s{\ell-s}0+\a s{\ell-s}1}{2(\ell+1)}\ctr{\ell-s}r_0^{\ell-p}\\
-2M\ctr{\ell-s}\left(\frac{\ell-p}{2(\ell+1)}(\c s{\ell-s}0 +\a s{\ell-s}1)-\a s{\ell-s}1 (1-\frac{\ell+p}{2\ell})\right)r^{2\ell+2}\int_{-\infty}^u \frac{r_0^{\ell-p-1}}{r^{2\ell+2}}\dd u'\\
+M^2 \O(\max(r^{\ell-p-1},r_0^{\ell-p-1}))+M\O (\max (r^{\ell-p-\epsilon},r_0^{\ell-p-\epsilon})(1+\delta_{\ell-p,\epsilon}\log r/r_0)).
\end{multline}
Note that the second line in the estimate above vanishes if $\ell-p=0$. (The integral on its own would produce a logarithmic leading order term in that case!) Indeed, we have the following simple expression:
\begin{multline}
\left(\frac{\ell-p}{2(\ell+1)}(\c s{\ell-s}0 +\a s{\ell-s}1)-\a s{\ell-s}1 (1-\frac{\ell+p}{2\ell})\right)=\frac{\ell-p}{2(\ell+1)}(\c s{\ell-s}0 +\a s{\ell-s}1)-\frac{\ell+1}{\ell}\a s{\ell-s}1 )\\
=\frac{\ell-p}{2(\ell+1)}((1+s)(1+2s)+3(\ell-s)(\ell+1+s)-(\ell-s)(\ell+s))=\frac{(\ell-p)(2\ell+1)}{2}.
\end{multline}
The statements \eqref{eq:al:prop:painrDvN:l-p>0} and \eqref{eq:al:prop:painrDvN:l-p>0 precise} then follow from Lemma~\ref{lem:appB:A1}.

In order to prove the remaining statements, we proceed inductively. Assume that $\ell-p<0$.
Our induction assumption is the following:
For all $1\leq i< \lceil p-\ell \rceil$, the limit $\lim_{v\to\infty}(\rDv)^{\ell-s+i}\alphas_{\ell}(u,v)$ exists, and we have the estimates
\begin{align}\label{eq:al:asyproof2induction1}
\lim_{v\to\infty}(\rDv)^{\ell-s+i}\alphas_{\ell}&=\E\cdot M r_0^{\ell-p+i-1}+M\O(Mr_0^{\ell-p+i-2}+r_0^{\ell-p+i-1-\epsilon}),\\
\label{eq:al:asyproof2induction2}
(\rDv)^{\ell-s+i}\alphas_{\ell}&=\lim_{v\to\infty}(\rDv)^{\ell-s+i}\alphas_{\ell}+M\O(r^{-1}r_0^{\ell-p+i}).
\end{align}
We first establish the case $i=1$: If $\ell-p<0$, then \eqref{eq:al:asyproof2.2} implies that $(\rDv)^{\ell-s+1}\alphas_{\ell}=M\O(r_0^{\ell-p})$. We moreover have that $(\rDv)^{\ell-s}\alphas_{\ell}=\O(r^{\ell-p})$ from \eqref{eq:al:asyproof0}. Inserting both of these estimates into \eqref{eq:cons:commute:dv} with $N=\ell-s+1$ and integrating in $u$, we then obtain that 
\begin{align*}
\Dv(\rDv)^{\ell-s+1}\alphas_{\ell} &=M\O\left(r^{2\ell+2}\int \frac{r_0^{\ell-p}}{r^{2\ell+4}}\dd u'\right)\\
&\leq M\O\left(\max\left(\frac{r_0^{\ell-p+1}}{r^2},r^{\ell-p-1}(1+\delta_{\ell-p,-1}\log r/r_0 )\right)\right).
\end{align*}
Since $\ell-p<0$, the RHS is integrable in $v$, and hence $(\rDv)^{\ell-s+1}\alphas_{\ell}$ attains a limit at $\Scrip$. 
If $1<\lceil p-\ell\rceil$, then \eqref{eq:al:asyproof2induction1} and \eqref{eq:al:asyproof2induction2} follow directly by integrating the above from $\Cin$, also taking into account \eqref{eq:al:transversal_derivatives_beyond}.
On the other hand, if $0>\ell-p\geq -1$, then we obtain that\footnote{Here, we use that, for any $q>0$, $$-(q-1)\int_{r_0}^r \frac{\log (r'/r_0)}{r'^q}\dd r'=\frac{\log (r/r_0)}{r^{q-1}}+\frac{1}{q-1}\left(\frac{1}{r^{q-1}}-\frac{1}{r_0^{q-1}}\right)$$}
\begin{align}\label{eq:al:asyproof2i=1.1}
\lim_{v\to\infty}(\rDv)^{\ell-s+1}\alphas_{\ell}&=\E\cdot M r_0^{\ell-p}+M\O( r_0^{\ell-p+1-2}+r_0^{\ell-p+1-1-\epsilon}),\\
\label{eq:al:asyproof2i=1.2}
(\rDv)^{\ell-s+1}\alphas_{\ell}&=\lim_{v\to\infty}(\rDv)^{\ell-s+1}\alphas_{\ell}+M\O(r^{\ell-p}(1+\delta_{\ell-p,-1}\log r/r_0)).
\end{align}
Inserting the expressions \eqref{eq:al:asyproof2i=1.1} and \eqref{eq:al:asyproof2i=1.2} back into \eqref{eq:cons:commute:dv} with $N=\ell-s+1$ then gives that
\begin{multline}\label{eq:al:asyproof2i=1.3}
(\rDv)^{\ell-s+1}\alphas_{\ell}=\lim_{v\to\infty}(\rDv)^{\ell-s+1}\alphas_{\ell}\\
+\E\cdot M r^{\ell-p}(1+\delta_{\ell-p,-1}\log r/r_0))+M\O(r^{-1}r_0^{\ell-p+1}).
\end{multline}
We have at this point proved \eqref{eq:al:prop:painrDvN:l-p<0:limit}--\eqref{eq:al:prop:painrDvN:l-p<0:2} for $i=1$.

Let now $\ell-p<-1$, and let us inductively assume that there exists $n\in\mathbb N$, $1\leq n< \lceil p-\ell \rceil-1$ such that \eqref{eq:al:asyproof2induction1} and \eqref{eq:al:asyproof2induction2} hold for all $1\leq i \leq n$. 
We will show that \eqref{eq:al:asyproof2induction1} and \eqref{eq:al:asyproof2induction2} then also hold for $i=n+1$.
For this, we insert the estimates \eqref{eq:al:asyproof2induction1}, \eqref{eq:al:asyproof2induction2} into \eqref{eq:cons:commute:dv} with $N=\ell-s+n$ and integrate in $u$; this gives:
\begin{equation}
\rDv(\rDv)^{\ell-s+n}\alphas_{\ell} = r^{2\ell+2n+2} M\cdot \O\left(\int \frac{r_0^{\ell-p+n-1}}{r^{2\ell+2n+2}}\dd u'\right)=M\cdot\O(r_0^{\ell-p+n}).
\end{equation}
We now insert this estimate into \eqref{eq:cons:commute:dv} with $N=\ell-s+n+1$, which similarly gives that
\begin{equation}\label{eq:al:asyproof2.4}
\rDv(\rDv)^{\ell-s+n+1}\alphas_{\ell} = r^{2\ell+2n+4} M\cdot \O\left(\int \frac{r_0^{\ell-p+n}}{r^{2\ell+2n+4}}\dd u'\right)=M\cdot\O(r_0^{\ell-p+n+1}).
\end{equation}
Integrating \eqref{eq:al:asyproof2.4} from $\Cin$ (where we use estimate \eqref{eq:al:transversal_derivatives_beyond}) then establishes that $(\rDv)^{\ell-s+n+1}\alphas_\ell$ attains a limit at $\Scrip$, which satisfies \eqref{eq:al:asyproof2induction1}. Integrating \eqref{eq:al:asyproof2.4} from $\Scrip$ then shows that \eqref{eq:al:asyproof2induction2} holds as well, thus completing the inductive argument. 

In order to show the final missing part of the proposition, namely to show that \eqref{eq:al:prop:painrDvN:l-p<0:limit} holds for $i=\lceil p-\ell \rceil +1$ and that \eqref{eq:al:prop:painrDvN:l-p<0:2} holds, we proceed in the same way in which we proved \eqref{eq:al:asyproof2i=1.1}--\eqref{eq:al:asyproof2i=1.3}. 
\end{proof}
\subsubsection{Asymptotic estimates for \texorpdfstring{$(\rDv)^{j}\alphas_{\ell}$ for $j\leq \ell-s$}{j transversal derivatives of alpha[s]}}
We now compute the asymptotic expressions for all the lower order derivatives.
In particular, this requires explicitly solving the Minkowskian problem, i.e.~it requires explicitly solving the Teukolsky equation for $M=0$.
We begin with the following simple observation:
\begin{prop}\label{prop:al:fundisindieTONNEtreten}
Let $j\in\{0,\dots,\ell-s\}$. The following identity holds for any smooth $\alphas_{\ell}$:
\begin{multline}\label{eq:al:prop:fundi}
(\rDv)^{\ell-s-j} \alphas_{\ell}(u,v)=\sum_{i=0}^j \frac{1}{i!}\left(\frac{1}{r_0}-\frac1r\right)^i  \restr{(\rDv)^{\ell-s-j+i}\alphas_{\ell}}{\Cin}\\
+\underbrace{\int_{v_1}^v \frac{D}{r^2}\cdots \int_{v_1}^{v_{(j)}}\frac{D}{r^2}}_{\text{$j+1$ integrals}} (\rDv)^{\ell-s+1}\alphas_{\ell}\dd v_{(j+1)}\cdots \dd v_{1} .
\end{multline}
Notice that the second line in \eqref{eq:al:prop:fundi} vanishes for $M=0$.
\end{prop}
\begin{proof}
Fundamental theorem of calculus and Lemma~\ref{lem:appB:xxxrrr}.
\end{proof}
We now make use of Proposition~\ref{prop:al:fundisindieTONNEtreten} to obtain an asymptotic estimate for all lower-order derivatives, and, in particular, for $\alphas_{\ell}$ itself. We first turn our attention to the second line of~\eqref{eq:al:prop:fundi}.
\begin{prop}\label{prop:al:prop:asyfull1}
Let $j\in\{0,\dots,\ell-s\}$. The solution of Theorem~\ref{thm:al:0} satisfies the following estimates:
If $\ell-p>0$ and if $j<\ell-p-1$, then
\begin{multline}\label{eq:al:prop:asyfull1}
(\rDv)^{\ell-s-j} \alphas_{\ell}=\sum_{i=0}^j \frac{1}{i!}\left(\frac{1}{r_0}-\frac1r\right)^i  \restr{(\rDv)^{\ell-s-j+i}\alphas_{\ell}}{\Cin}\\
-M\A_{\ell}\ctr{\ell-s} \frac{(\ell-p)(\ell+p+1)!}{(2\ell)!}(\ell-2-p-j)!r^{\ell-1-p-j}\\
\cdot (1+\O((r_0/r)^{\ell-1-p-j}+r_0/r(1+\delta_{\ell-2-p,j}\log r/r_0)+r^{-\epsilon}(1+\delta_{\ell-p-1-\epsilon,j}\log r/r_0))).
\end{multline}

If $\ell-p>0$ and $j=\ell-p-1$, then
\begin{multline}\label{eq:al:prop:asyfull2}
(\rDv)^{\ell-s-j} \alphas_{\ell}=\sum_{i=0}^j \frac{1}{i!}\left(\frac{1}{r_0}-\frac1r\right)^i  \restr{(\rDv)^{\ell-s-j+i}\alphas_{\ell}}{\Cin}\\
-M\A_{\ell}\ctr{\ell-s} \frac{(\ell-p)(\ell+p+1)!}{(2\ell)!}\left(\log(r/r_0) +\O(r_0/r\log r/r_0+r^{-\epsilon}(1+\delta_{\epsilon,1}\log r/r_0))\right).
\end{multline}

If $j>\ell-p-1$, then
\begin{multline}\label{eq:al:prop:asyfull3}
(\rDv)^{\ell-s-j} \alphas_{\ell}=\sum_{i=0}^j \frac{1}{i!}\left(\frac{1}{r_0}-\frac1r\right)^i  \restr{(\rDv)^{\ell-s-j+i}\alphas_{\ell}}{\Cin}
\\
+M\O(r_0^{(\ell-1-p)-j} +(\ell-p)r^{\ell-1-p-j}(1+\delta_{\ell-p\in\mathbb Z_{\neq 0}}\log r/r_0)).
\end{multline}
\end{prop}
\begin{proof}
The proof follows from taking the expression \eqref{eq:al:prop:fundi} and computing the integrals in the second line of \eqref{eq:al:prop:fundi} by inserting the relevant estimates from Proposition~\ref{prop:al:pain:rDvN} and then applying Lemmata~\ref{lem:appB:xxxrrr}--\ref{lem:appB:hardxxx}.

As an example, to prove \eqref{eq:al:prop:asyfull1}, we insert \eqref{eq:al:prop:painrDvN:l-p>0} into \eqref{eq:al:prop:fundi}, and we then compute the arising integral using Lemma~\ref{lem:appB:xxxrrr} and then \eqref{eq:appB:lemhard1} of Lemma~\ref{lem:appB:hardxxx}.
\end{proof}
\begin{rem}\label{rem:al:sneakyPM}
Starting from the more precise estimate \eqref{eq:al:prop:painrDvN:l-p>0 precise} in the case of $\ell-p\in\mathbb N_0$, we can also show in exactly the same way in which we proved the estimates of Proposition~\ref{prop:al:prop:asyfull1} that (setting $\epsilon=\infty$ for simplicity)
\begin{multline*}
\alphas_\ell=\sum_{i=0}^{\ell-s} \frac{1}{i!}\left(\frac{1}{r_0}-\frac1r\right)^i  \restr{(\rDv)^{\ell-s-j+i}\alphas_{\ell}}{\Cin}\\
+M r_0^{s-p-1}\sum_{n=0}^{\ell-s+1} (Q_n+S'_n \log r_0/r+\O(M r_0^{-1})) \left(\frac{r_0}{r}\right)^n
\end{multline*}
for some (potentially vanishing) constants $Q_n$, $S'_n$ that we do not yet determine. We will return to this point later, when we need to compute higher-order asymptotics (cf.~\S\ref{sec:alb}).
\end{rem}
We now turn our attention to the first line of \eqref{eq:al:prop:fundi}:
\begin{prop}\label{prop:al:Minkoswkiansolution}
Let $j\in\{0,\dots,\ell-s\}$. The solution of Theorem~\ref{thm:al:0} satisfies the following estimates:
\begin{equation}\label{eq:al:prop:Minkowskian}
\sum_{i=0}^j \frac{1}{i!}\left(\frac{1}{r_0}-\frac1r\right)^i  \restr{(\rDv)^{\ell-s-j+i}\alphas_{\ell}}{\Cin}=\A_{\ell}r_0^{\ell-j-p}\sum_{n=0}^j \left(\frac{r_0}{r}\right)^n \cdot (S_{\ell,p,j,n,s}+\mathcal O_u(r_0^{-\epsilon}+Mr_0^{-1})),
\end{equation}
where the constants $S_{\ell,p,j,n,s}$ are defined for $0\leq n\leq j\leq \ell-s$ via
\begin{equation}\label{eq:al:SLPJN}
S_{\ell,p,j,n,s}:=\frac{(-1)^{\ell-s+n}(s+p)!(\ell-s)!(2\ell-j+n)!}{n!(\ell+s)!(\ell+p)!}\binom{\ell-p}{\ell-p-j+n}.
\end{equation}
Note that 
\begin{equation}\label{eq:al:SLPJN=0}
S_{\ell,p,j,n,s}=0 \iff \,\,\,( \ell-p\in \mathbb N_{\geq 0}\,\,\,\text{and}\,\,\, n<j+p-\ell).
\end{equation}
\end{prop}
\begin{proof}
For simpler notation, we consider the case where $M=0$ and $\epsilon=\infty$.
Inserting the expressions from Proposition~\ref{prop:al:transversal_derivatives}, we have
\begin{nalign}\label{eq:al:propproof:Minkowskian}
&\sum_{i=0}^j \frac{1}{i!}\left(\frac{1}{r_0}-\frac1r\right)^i  \restr{(\rDv)^{\ell-s-j+i}\alphas_{\ell}}{\Cin}\\
=&\A(-1)^{\ell-s-j}r_0^{\ell-j-p}\frac{(s+p)!(\ell-s)!}{(\ell+s)!}\sum_{i=0}^j \frac{(-1)^i}{i!}\left(1-\frac{r_0}{	r}\right)^i\frac{(2\ell-j+i)!}{(\ell-j+i+p)!}\\
=&\A(-1)^{\ell-s-j}r_0^{\ell-j-p}\frac{(s+p)!(\ell-s)!}{(\ell+s)!}\sum_{n=0}^j \left(\frac{r_0}{r}\right)^n\sum_{i=0}^j\binom{i}{n} \frac{(-1)^{i+n}}{i!}\frac{(2\ell-j+i)!}{(\ell-j+i+p)!}\\
=&\A_{\ell}r_0^{\ell-j-p}\sum_{n=0}^j \left(\frac{r_0}{r}\right)^n \cdot S_{\ell,p,j,n,s},
\end{nalign}
where $S_{\ell,p,j,n,s}$ are constants that we now determine by observing that the expression above, for $j=\ell-s$, solves \eqref{eq:cons:teuk} with $M=0$. This is because it is exactly the first line of \eqref{eq:al:prop:fundi}. In other words, it solves equation \eqref{eq:pm:teukM=0} from \S\ref{sec:PM}.

We can thus appeal to Proposition~\ref{prop:pm:Ansatz=0A}:\footnote{For the reader who cannot be bothered to skip to section \ref{sec:PM}, all we are doing here is inserting the above expression into the Teukolsky equation with $M=0$ and solving the resulting recurrence relations for the coefficients $S_{\ell,p,j=\ell-s,n,s}$.} 
Since $(r^2\Dv)^{\ell-s}\alphas_{\ell}=(r^2\Dv)^{\ell-s}\alphas_{\ell}|_{\Cin}=\A_\ell \ctr{\ell-s}$, we infer the value of $S_{\ell-s}(=S_{\ell,p,\ell-s,n=\ell-s,s})$ in \eqref{eq:pm:prop:Sn} via \eqref{eq:pm:prop:AnsatzM=0A:compare}: $S_{\ell-s}=\frac{(s+p)!}{(\ell+p)!}\frac{(2\ell)!}{(\ell+s)!}$. Equation~\eqref{eq:al:SLPJN} thus follows for $j=\ell-s$.

In order to compute the constants $S_{\ell,p,j,n,s}$ for $j<\ell-s$, we simply observe that $$S_{\ell,p,j,n,s}=(-1)^{\ell-s-j}\frac{(n+\ell-s-j)!}{n!}S_{\ell,p,\ell-s,n+\ell-s-j,s},$$ as can be seen by differentiating the expression \eqref{eq:al:prop:Minkowskian}, which, for $M=0$, is exactly the LHS of \eqref{eq:al:prop:fundi}.

Considering now the case $M\neq0\neq \epsilon$ gives the error terms $\O(Mr_0^{-1}+r_0^{-\epsilon})$.
\end{proof}
\subsection{Asymptotic expansion for \texorpdfstring{$\alphas_{\ell}$}{alpha[s]} and the proofs of Thm.~\ref{thm:al:0} and Cor.~\ref{cor:al:0}}\label{sec:al:final}
We combine the results of the previous sections into the following proposition:
\begin{prop}
The scattering solution $\alphas_{\ell}$ described in Theorem~\ref{thm:al:0} satisfies the following expansion throughout $\DoD$: If $\ell-p\geq0$, then
\begin{multline}\label{eq:al:final}     
\alphas_{\ell}=\A_{\ell}r_0^{s-p}\sum_{n=0}^{\max(\ell-s,\lceil p-s\rceil)} \left(\frac{r_0}{	r}\right)^n (S_{\ell,p,\ell-s,n,s}+\mathcal O_u(r_0^{-\epsilon}+Mr_0^{-1}))\\
-M\A_{\ell}\ctr{\ell-s}\frac{(\ell+p+1)!}{(2\ell)!}\cdot r^{s-1-p}\\
\cdot \begin{cases}  (s-p-2)! \{(\ell-p)&
\\ +\O((\tfrac{r_0}{r})^{\lceil p \rceil -p}+\tfrac{r_0}{r}(1+\delta_{s-2-p,0}\log r/r_0) +r^{-\epsilon}(1+\delta_{s-p-1,\epsilon}\log r/r_0))\},&1+p-s \notin \mathbb N_{\geq 0}\\
\frac{(-1)^{p+1-s}}{(p+1-s)!}((\ell-p)\log r/r_0+\O(\frac{r_0}{r}\log r/r_0+r^{-\epsilon}(1+\delta_{\epsilon,1}\log r/r_0)),&1+p-s \in \mathbb N_{\geq 0}.
\end{cases}
\end{multline}
If, on the other hand, $\ell-p<0$, then we can define $S_{\ell,p,j,n,s}$ to be zero whenever $n>j$, and the first two lines of the above formula still apply. The form of the third line also still applies, that is, the third line still applies if the constants are replaced by some undetermined coefficients.
\end{prop}
\begin{proof}

If $\ell-s\leq \ell-p-1$, i.e., if $s\geq p+1$ and $\ell>p+1$, then the result directly follows from setting $j=\ell-s$ in \eqref{eq:al:prop:asyfull1} or \eqref{eq:al:prop:asyfull2}. 

If $\ell-p>0$ and $\ell-s>\ell-p-1$, then we integrate either \eqref{eq:al:prop:asyfull1} or \eqref{eq:al:prop:asyfull2}, with $j=\ell -\lceil p+1\rceil$,  $\lceil p+1\rceil-s$ times from $\Scrip$ (instead of from $\Cin)$, using \eqref{eq:al:prop:asyfull3} to estimate the boundary terms on $\Scrip$.

 If $\ell-p\leq 0$, then we integrate instead the expressions from Proposition~\ref{prop:al:pain:rDvN} an appropriate amount of times from $\Scrip$.
\end{proof}
This concludes the proof of Thm.~\ref{thm:al:0}.
We finish the section by proving its Corollary~\ref{cor:al:0}:
\begin{proof}[Proof of Cor.~\ref{cor:al:0}]
In order to prove the corollary, we repeat the proof of Theorem~\ref{thm:al:0}, but instead of using the approximate conservation law \eqref{eq:cons:cons1} for $\alphas$, we appeal to the approximate conservation law \eqref{eq:cons:RWcons} for $\Psis$.  
Of course, if $s=0$, then $\Psis=\alphas$ and the result trivially follows. 
But if we now consider the Regge--Wheeler equation for general spin $s\neq 0$, we can simply exploit that the coefficients $\aRW{s}N0$ and $\bRW{s}N0$ as well as the eigenvalues of $\laps-\Lambda_0^{[s]}$ in \eqref{eq:cons:RWcons} are independent of $s$, so the contribution in the first line of \eqref{eq:al:final}, which is entirely Minkowskian, is as for $s=0$.

The contribution in the last two lines of \eqref{eq:al:final}, on the other hand, is  generated by  integrating \eqref{eq:cons:cons1} or \eqref{eq:cons:RWcons} to obtain the leading-order behaviour of $\Dv\Asl$ or $\Dv\mathrm{\mathbf{\Psi}}^{[s]}_\ell$, respectively. This, in turn,  is now governed by the coefficient $2\x s1\ell-\c s{\ell-s} 0$ in the case of $\Dv\Asl$ (cf.~Prop.~\ref{prop:al:pain:Asl}), or, in the case of $\Dv\mathrm{\mathbf{\Psi}}^{[s]}_\ell$, by the coefficient $2\xRW s1\ell-\cRW s\ell 0$. But since $2\xRW s1\ell-\cRW s\ell 0=2\x s1\ell-\c s{\ell-s}0=-(\ell+1)(2\ell+1)$ is independent of $s$, the result is then identical to that of $s=0$.

\end{proof}
\newpage

\newpage
\section{Asymptotics for \texorpdfstring{$\alpha^{[2]}=r^5\Omega^{-2}\overone{\alpha}$}{alpha} arising from the physical data of \S\ref{sec:physd}}\label{sec:alp}
We now apply the results of the previous section to find the asymptotics of $\al$ arising from the physical data described in \S\ref{sec:physd} (cf.~e.g.~Prop.~\ref{prop:physd}). Recall from \eqref{eq:physd:al} that, in the case of \textbf{hyperbolic} orbits (Def.~\ref{defi:physd:Nbodyseed}) these data satisfy the following expansion along $\Cin$:
\begin{equation}\label{eq:alp:data}
r^{-s}\alphass[s=2]_{\Cin}=\frac{r^3}{\Omega^2}\radc\al=\A+\B \frac{\log r_0}{r_0}+\frac{\BB}{r_0}+\O_3(r_0^{-1-\delta}),
\end{equation}
with 
\begin{equation}\label{eq:alp:B}
\B=2\Ds2\Ds1\overline{\D1}\D2 \albdata,\qquad \BB=\mathscr{C}_5
\end{equation}
The presence of the logarithmic term in \eqref{eq:alp:data} necessitates a minor modification of the proof of Theorem~\ref{thm:al:0}. The description of this modification and the result of it is the content of this section.

In contrast, the treatment of the case of \textbf{parabolic} orbits (cf.~Def.~\ref{defi:physd:parabolic}) is already included in Theorem~\ref{thm:al:0}, using \eqref{eq:physd:al:par}, we therefore won't discuss it beyond the statement of Thm.~\ref{thm:alp:1} below.

\subsection{The main theorem (Thm.~\ref{thm:alp:0})}
\begin{thm}\label{thm:alp:0}
Let $\al$ be the unique, smooth scattering solution to \eqref{eq:lin:Teukal} arising from the physical data of Def.~\ref{defi:physd:Nbodyseed} from \S\ref{sec:physd}, i.e.~arising from~\eqref{eq:al:noincomingradiation} with~$s=2$ (no incoming radiation) and~\eqref{eq:alp:data} in the sense of Def.~\ref{defi:al:data}. 
Then each angular mode $\al_\ell$ of $\al$ satisfies the following asymptotic expansion throughout $\DoD$:
\begin{nalign}\label{eq:alp:thmmain}
\frac{r^5\al_{\ell}}{\Omega^2}= \sum_{n=0}^{\ell-2} \left(\frac{r_0}{r}\right)^n\{\A_\ell r_0^2 S_{\ell,0,\ell-2,n,2}&+\B_\ell r_0 \log r_0 S_{\ell,1,\ell-2,n,2}\\
&+\O^u(M\A_\ell r_0)+\O^u(\B_{\ell}r_0)+\O^u(\C_{\ell}r_0)+\O^u(r_0^{1-\delta}) \}\\
- \frac{(-1)^{\ell}2M\A_\ell r}{(\ell-1)(\ell+2)}-\frac{(-1)^{\ell}{3}M\B_\ell \log ^2 r}{\ell(\ell+1)}
&+M\log r (O(r_0\A_\ell)+\O(\B_\ell)+\O(\C_\ell)),
\end{nalign}
where the constants $S_{\ell,p,\ell-s,n,s}$ are defined in~\eqref{eq:al:SLPJN}.

Moreover, the following limits are conserved along $\Scrip$:
\begin{align}\label{eq:alp:thmNP}
 \NP{2-\ell}{0} [\alphas[2]]=\lim_{v\to\infty} r^{2-\ell}\Dv \Asl[2]&=(-1)^{\ell+1}2M\A_{\ell}\frac{(\ell-2)!(\ell+1)}{(\ell+2)!},\\\label{eq:alp:thm:otherlimit}
 \lim_{v\to\infty} r^4\al_{\ell}&=(-1)^{\ell+1}\frac{2M\A_\ell}{(\ell-1)(\ell+2)}.
\end{align}
\end{thm}
For parabolic orbits, we content ourselves with the following
\begin{thm}\label{thm:alp:1}
    Let $\al$ be the unique, smooth scattering solution to \eqref{eq:lin:Teukal} arising from the physical data of Def.~\ref{defi:physd:parabolic} from \S\ref{sec:physd}, i.e.~arising from~\eqref{eq:al:noincomingradiation}, \eqref{eq:al:noincomingradiation2} with~$s=2$ and~\eqref{eq:physd:al:par} (so $p=2/3$ in \eqref{eq:al:decayalongCin}). 
Then each angular mode $\al_\ell$ of $\al$ satisfies
\begin{equation}
    \lim_{v\to\infty}r^{  \frac{14}{3}}\al_{\ell}=(-1)^{\ell+1}\sqrt{3}\pi\conj{\albdata_{\mathrm{par},\ell}}(\ell+\tfrac53)(\ell-\tfrac23).
\end{equation}
\end{thm}

\subsection{Proof of Theorem~\ref{thm:alp:0}}
\begin{proof}
This is essentially an application of Theorem~\ref{thm:al:0} with $s=2$ and a superposition of  $(p=0,\epsilon=1+\delta)$ and $(p=1,\epsilon=\infty)$.
The only modification comes from the logarithmic term in \eqref{eq:alp:data}, $\B r_0^{-1}\log r_0$. 
Throughout the rest of this proof, let us therefore assume scattering data satisfying the no incoming radiation condition (Def.~\ref{defi:al:nir}) and 
\begin{align}\label{eq:alp:logdecaydata}
r^{-s}\radc\alphas=\B r_0^{-p}\log r_0+\BB r_0^{-p},&& p>\max(-1/2,-s+1).
\end{align}
For convenience, we keep the exponent $p$ general for now and only restrict to $p=1$ in the end.
In order to avoid case distinctions, we will in the propositions below also make the additional assumption that $\ell>p$. 

For logarithmically decaying data as in \eqref{eq:alp:logdecaydata}, we first need an analogue of Proposition~\ref{prop:al:transversal_derivatives}:
\begin{prop}\label{prop:alp:transversal_derivatives}
For  $N\leq \ell-s$, transversal derivatives along $\Cin$ satisfy:
\begin{multline}\label{eq:alp:transversal_derivatives}
\restr{(D^{-1}r^2\Dv)^N \alphas_\ell}{\Cin}=\ctr{N} r_0^{N-p+s}\log r_0\cdot\left(\B_\ell+\frac{\B_\ell}{\log r_0}\sum_{k=1}^{N}\frac 1{k+s+p}+\frac{\BB_\ell}{\log r_0}\right)
\\+\mathcal O(Mr_0^{N-p+s-1}\log r_0).
\end{multline}
\end{prop}
\begin{proof}
The proof goes by induction as in Proposition~\ref{prop:al:transversal_derivatives}, using that $$\int_{-\infty}^u r_0^{-k-s-p-1}\log r_0\dd u'=\frac{r_0^{-k-s-p}}{k+s+p}\left(\log r_0+\frac{1}{k+s+p}\right)+M\O(r_0^{-k-s-p-1}\log r_0).$$ 
\end{proof}
Analogously to how we proved Proposition~\ref{prop:al:BS}, we then show that
\begin{equation}
\left|(\rDv)^{\ell-s+1}\alphas_{\ell}\right|\leq M C\max(r^{\ell-p}\log r,r_0^{\ell-p}\log r_0),
\end{equation}
and, integrating this from $\Cin$, we obtain a logarithmically adorned version of \eqref{eq:al:asyproof0}.

We can now prove analogues of Propositions \ref{prop:al:pain:Asl} and \ref{prop:al:pain:rDvN}:
\begin{prop}\label{prop:alp:rDvN and Asl}
Suppose that $\ell-p>0$. 
Then
\begin{equation}\label{eq:alp:Asl}
\lim_{v\to\infty }\frac{r^{p-\ell}}{\log r}\rDv\Asl=-M\B_\ell\ctr{\ell-s} \frac{(\ell-p)!(\ell+p+1)!}{(2\ell)!},
\end{equation}
and 
\begin{multline}\label{eq:alp:rDv (ell-s+1)}
(\rDv)^{\ell-s+1}\alphas_{\ell}=-M \ctr{\ell-s}  \frac{(\ell-p)!(\ell+p+1)!}{(2\ell)!}r^{\ell-p}\log  r\\
\cdot \left(\B_{\ell}\left(1+\frac{1}{\log r}\left(d_{2\ell+3,\ell-p}+\sum_{k={1}}^{\ell-s} \frac 1{k+s+p} \right) \right)+{\BB_\ell}+\O(r_0/r)\right),
\end{multline}
where $d_{2\ell+3,\ell-p}$ is defined in Lemma~\ref{lem:appB:A1}. If $p\in\mathbb N$, then $d_{2\ell+3,\ell-p}=-\sum_{k=\ell-p+1}^{\ell+p+1}\frac1k$.
\end{prop}
\begin{proof}
The proof proceeds in the same manner as the proof of Proposition~\ref{prop:al:pain:Asl}, with the difference that we now need to compute the integral
\begin{align*}
\int \left(\B_{\ell}	 \frac{r_0^{\ell-p}\log r_0}{r^{2\ell+3}}+\left(\B_{\ell}\sum_{k={1}}^{\ell-s} \frac 1{k+s+p} +\BB_{\ell}\right)\frac{r_0^{\ell-p}}{r^{2\ell+3}}\right)\dd u,
\end{align*}
which we do using \eqref{eq:appB:lemA1:1} and \eqref{eq:appB:lemA1:3} for the second and first term, respectively.
\end{proof}
We are left with having to integrate \eqref{eq:alp:rDv (ell-s+1)} $\ell+1$ times from $\Cin$. For the Minkowskian part of the expression, we have
\begin{prop}\label{prop:alp:minkowskian}
Suppose that $\ell-p\geq 0$ and that $p\in\mathbb N$. Then we have:
\begin{nalign}
&\sum_{i=0}^{\ell-s}\frac{1}{i!}\left(\frac{1}{r_0}-\frac{1}{r}\right)^i\restr{(\rDv)^{i}\alphas_{\ell}}{\Cin}\\
=&(\B_\ell \log r_0+\BB_\ell) r_0^{s-p}\sum_{n=0}^{\ell-s} \left(\frac{r_0}{r}\right)^n \cdot S_{\ell,p,\ell-s,n,s}(1+\mathcal O_u(Mr_0^{-1}))\\
&+\B_\ell r_0^{s-p}\sum_{n=0}^{p-s-1}\left(\frac{r_0}{r}\right)^n R_{\ell,p,\ell-s,n,s}\\
&+\B_\ell r_0^{s-p}\sum_{n=p-s}^{\ell-s}\left(\frac{r_0}{r}\right)^nS_{\ell,p,\ell-s,n,s}\left(\sum_{k=0}^{\ell-s-1-n}\frac{1}{\ell-p-k}+\sum_{k=1}^{\ell-s}\frac{1}{k+s+p}\right),
\end{nalign}
where
\begin{equation}
R_{\ell,p,\ell-s,n,s}:=(-1)^{\ell-p+1}\frac{ (\ell-p)! (\ell-s)! (p+s)!}{ (\ell+p)! (\ell+s)!} \frac{(p-n-s-1)!(\ell+n+s)!}{n! (\ell-n-s)!}.
\end{equation}
\end{prop}
\begin{proof}
As in the proof of Prop.~\ref{prop:al:Minkoswkiansolution}, we show the result for~$M=0$: Let $j\in\{0,\dots,\ell-s\}$. Then we have
\begin{nalign}
&\sum_{i=0}^j \frac{1}{i!}\left(\frac{1}{r_0}-\frac1r\right)^i  \restr{(\rDv)^{\ell-s-j+i}\alphas_{\ell}}{\Cin}\\
=&(\B\log r_0+\BB) r_0^{\ell-j-p}\sum_{n=0}^j \left(\frac{r_0}{r}\right)^n \cdot S_{\ell,p,j,n,s}\\
&+(-1)^{\ell-s-j}r_0^{\ell-j-p}\frac{(s+p)!(\ell-s)!}{(\ell+s)!}\\
&\quad\quad \cdot \sum_{n=0}^j \left(\frac{r_0}{r}\right)^n\sum_{i=0}^j\binom{i}{n} \frac{(-1)^{i+n}}{i!}\frac{(2\ell-j+i)!}{(\ell-j+i+p)!}
\left(\sum_{k=1}^{\ell-s-j+i}\frac1{k+p+s}\right).
\end{nalign}
For $j=\ell-s$, we now apply Proposition~\ref{prop:pm:Ansatz=0B} to get a simpler expression for the last line: By comparing \eqref{eq:alp:transversal_derivatives}, with $N=\ell-s$, and \eqref{eq:pm:Ansatz=0B:compare}, we can read off the constants $S_{\ell-s}$ and $R_{\ell-s}$; in particular, we have $R_{\ell-s}/S_{\ell-s}=\sum_{k=1}^{\ell-s}\frac{1}{k+s+p}$. The result then follows.
\end{proof}
Lastly, we prove the analogue of Proposition~\ref{prop:al:prop:asyfull1}:
\begin{prop}\label{prop:alp:asy}
Let $\ell-p\in\mathbb N_{>0}$. Then 
\begin{multline}
\frac{r^5\al_\ell}{\Omega^2}=\sum_{i=0}^{\ell-s} \frac{1}{i!}\left(\frac{1}{r_0}-\frac1r\right)^i  (\rDv)^{i}\alphas_{\ell}|_{\Cin}\\
-M\ctr{\ell-s}\frac{(\ell-p)(\ell+p+1)!}{(2\ell)!}\frac{(-1)^{p+1-s}}{{2}(p+1-s)!}r^{s-1-p}\\
\cdot \left(\B_{\ell}\log^2 r+\log r\left(\BB_{\ell}+\B_\ell\left(1-\sum_{k=1}^{p+1-s}\frac{1}{k}-\sum_{k=\ell-p+1}^{\ell+p+1}\frac1k+  \sum_{k=1+s+p}^{\ell+p}\frac1k\right)\right)+\O(\log^2 r_0)\right).
\end{multline}
\end{prop}
\begin{proof}
The proof is similar to the one of Proposition~\ref{prop:al:prop:asyfull1}, with the difference being that we now have to appeal to Lemma~\ref{lem:appB:hardxxxlog} as well (which is responsible for the $\B_\ell (\log^2 r+\log r(1-\sum_{k=1}^{p+1-s}\frac1k))$-term).
\end{proof}
Superposing Theorem~\ref{thm:al:0} with $p=0$ and $\epsilon=1+\delta$ with  the combined results of Propositions~\ref{prop:al:fundisindieTONNEtreten}, \ref{prop:alp:minkowskian} and \ref{prop:alp:asy} with $p=1$ proves \eqref{eq:alp:thmmain}.
The other two equations~\eqref{eq:alp:thmNP} and \eqref{eq:alp:thm:otherlimit} follow directly from Thm.~\ref{thm:al:0} (as they are independent of $\B$).
\end{proof}

\newpage
\section{Asymptotics for \texorpdfstring{$\Psis[\pm2]$}{Psi and Psibar} arising from the physical data of \S\ref{sec:physd}}
\label{sec:Psi}
We now apply the results of \S\ref{sec:al} to the solutions to the Regge--Wheeler equations $\Psis[\pm2]$ arising from the physical data described in Def.~\ref{defi:physd:Nbodyseed} (cf.~Prop.~\ref{prop:physd}). (Again, the case of Def.~\ref{defi:physd:parabolic} is already contained in Cor.~\ref{cor:al:0}.) 
As we have seen in \eqref{eq:physd:Ps} and \eqref{eq:physd:Psb}, respectively, both Regge--Wheeler quantities, $\Ps=\Psis[2]$ and $\Psb=\Psis[-2]$ (and thus also $\Ds2\Ds1(0,r^3\sig)=\frac14(\Psb-\Ps)$), will have the following expansion along $\Cin$:
\begin{equation}\label{eq:Psi:data}
\Psiss_\Cin=\AP +\BP r_0^{-1}\log r_0+\CP r_0^{-1}+\O_1(r_0^{-1-\delta}),
\end{equation}
where the precise values of $\CP$ do not matter and where
\begin{align}\label{eq:Psi:APBP}
\AP[2]=2\aldata=\overline{\AP[-2]},&&\BP[2]=12\Ds2\Ds1\overline{\D1}\D2\albdata=\overline{\BP[2]}.
\end{align}
\begin{thm}\label{thm:Psi}
For $\ell\geq |s|$, the unique, smooth scattering solution $\Psis$ restricting to \eqref{eq:Psi:data} along $\Cin$ and satisfying the no incoming radiation condition \eqref{eq:al:noincomingradiation} at $\Scrim$ has the following asymptotic expansion for each angular mode $\Psis_\ell$ throughout $\DoD$:
\begin{multline}
\Psis_{\ell}=\sum_{n=0}^\ell \left(\frac{r_0}{r}\right)^n\left(\AP_\ell S_{\ell,0,\ell,n,0}+\BP_{\ell} S_{\ell,1,\ell,n,0}r_0^{-1}\log r_0+\O^u(r_0^{-1})\right)\\
+M(-1)^\ell \left(\ell(\ell+1) \AP_\ell r^{-1}\log r/r_0-{\tfrac12}(\ell-1)(\ell+2)  \BP_\ell r^{-2} \log^2 r\right)
+M\O(r^{-2}r_0\log r) ,
\end{multline}
where the constants $S_{\ell,p,\ell-s,n,s}$ are defined in~\eqref{eq:al:SLPJN}.
\end{thm}
\begin{rem}
Notice, in particular, that (since $S_{\ell,1,\ell,0,0}=0$)
\begin{equation}
\lim_{v\to\infty}\Psis_{\ell}=\AP_\ell S_{\ell,0,\ell,0,0}+\O(r_0^{-1})=(-1)^{\ell}\AP_{\ell}+\O(r_0^{-1}),
\end{equation}
and hence 
\begin{equation}\label{eq:Psi:antipodal}
\lim_{u\to-\infty}\,\lim_{v\to\infty}\Psis_{\ell}=(-1)^{\ell}\lim_{v\to\infty}\,\lim_{u\to-\infty}\Psis_{\ell}.
\end{equation}
This is, in essence, the antipodal matching condition at the level of linearised gravity. See also the discussion below Thm.~\ref{thm:intro:main} in the introduction.
\end{rem}
\begin{proof}
The statement follows by setting $s=0$ in the proof of Theorem~\ref{thm:alp:0}, cf. the proof of Cor.~\ref{cor:al:0} at the end of \S\ref{sec:al}.
\end{proof}

\newpage
\section{Asymptotics for \texorpdfstring{$\alpha^{[-2]}= r\Omega^2\underline{\overone{\alpha}}$}{alphabar} arising from the physical data of \S\ref{sec:physd}}
\label{sec:alb}
We now move to the asymptotic analysis of $\alb$ in the case of the physical data of Def.~\ref{defi:physd:Nbodyseed} of~\S\ref{sec:physd}. (The cases of Definitions~\ref{defi:physd:parabolic} and \ref{defi:physd:graviton}, on the other hand, will only briefly be discussed at the end of the section.) 
Recall from \eqref{eq:physd:alb} that these satisfy the following expansion along $\Cin$:

\begin{equation}\label{eq:alb:data}
r^{-s}\alphass[s=-2]_{\Cin}=r^3\Omega^2\albs_{\Cin}=-6\Bb r^{-p}+\O(r^{-p-\delta})
\end{equation}
with $p=1$.
We thus already see that, since the condition $p>-s-1$ (cf.~\eqref{eq:al:prange}) is violated (we have $p=1$ and $s=-2$), Theorem~\ref{thm:al:0} does not directly apply. 

We recall that the condition of no incoming radiation \eqref{eq:al:noincomingradiation} in the case $s<0$ reads
\begin{equation}\label{eq:alb:noincomingradiation}
\lim_{u\to-\infty}\Dv(\rDv)^{|s|}\alphas\equiv 0.
\end{equation}
It is thus clear that the specification of \eqref{eq:alb:data} and \eqref{eq:alb:noincomingradiation} does not suffice to determine a solution~$\alphas$ to the scattering problem if~$s<0$; one needs to additionally specify all lower-order transversal derivatives $(\rDv)^{i}\alphas$, $i\in\{1,\dots,|s|\}$ as seed data at a sphere along~$\Cin$, cf.~Def.~\ref{defi:al:data}.

In Theorem~\ref{thm:al:0}, we simplified our lives by imposing $\Dv^i(r^5\alb)|_{\Scrim}=0$ for $i=1,\dots s$, but this condition had no physical motivation. 

In the present section, we instead discuss the physically motivated setup for the data of~\S\ref{sec:physd}: 
Since, for general $p$ violating $p>-s-1$ in~\eqref{eq:alb:data}, the integral along~$\Cin$ of~\eqref{eq:cons:teuk} for~$s=-2$ diverges, the seed data for the transversal derivatives cannot be specified at $\Sinfty$. Instead, they  have to be specified at the finite sphere $\Sone$ (or at any other finite sphere).
In the context of our physical data, these data for transversal derivatives are encoded in \eqref{eq:physd:Psb} and \eqref{eq:physd:psb} (recall the definitions \eqref{eq:lin:transformations}). We have:
\begin{equation}\label{eq:alb:1transderivatives}
\rDv\alphas[-2]|_{\Cin}=\radc\psb=6(-2\Ds2\D2)\Bb\frac{\log r_0}{r_0^2}+\mathscr{C}_8 r_0^{-2}+\O(r^{-2-\delta}),
\end{equation}
(note that $(-2\Ds2\D2\Bb)_\ell=\at{-2}{0}{0}\Bb_{\ell}=-(\ell-1)(\ell+2)\Bb_{\ell}$) and
\begin{equation}\label{eq:alb:2transderivatives}
(\rDv)^2\alphas[-2]|_{\Cin}=\radc\Psb=2\Ds2\Ds1\D1\D2(2\Aa+6\Bb r_0^{-1}\log r_0)+\mathscr{C}_{7} r_0^{-1}+\O(r_0^{-1-\delta}),
\end{equation}
where $\Aa$ is such that $2\Ds2\Ds1\D1\D2 \Aa=\overline{\aldata}$. 
Notice that 
\begin{equation}\label{eq:alb:Aa}
(2\Ds2\Ds1\D1\D2 \Aa)_{\ell}=\at{-2}{0}{0}\at{-2}{1}{0}\Aa_\ell=\frac{(\ell+2)!}{(\ell-2)!}\Aa_{\ell}=\overline{\aldata}_{\ell}.
\end{equation}

\subsection{The main theorem (Thm.~\ref{thm:alb:0})}
With the preliminaries discussed in the preceding subsection, let us now state the main result of this section. We will write $6(-2\Ds2\D2 )(\BBb)=\mathscr{C}_8$ (cf.~\eqref{eq:alb:1transderivatives}).
\begin{thm}\label{thm:alb:0}
Let $\alb$ be the unique smooth scattering solution to \eqref{eq:lin:Teukalb} arising from the physical data of Def.~\ref{defi:physd:Nbodyseed}, i.e.~from \eqref{eq:alb:noincomingradiation} with~$s=-2$ and~\eqref{eq:alb:data},~\eqref{eq:alb:1transderivatives}--\eqref{eq:alb:2transderivatives}. Then each angular mode $\alb_{\ell}$ of $\alb$ satisfies the following asymptotic expansion throughout $\DoD$:
\begin{nalign}\label{eq:alb:thm}
\Omega^2r\alb_\ell&=\alphas[-2]_\ell=\Aa_{\ell}r_0^{-2}\sum_{n=2}^{\ell+2}\underline{S}_{\ell,0,\ell+2,n,-2}\left(\frac{r_0}{r}\right)^n+(\Bb_{\ell}\log r_0+\BBb_{\ell})r_0^{-3}\sum_{n=3}^{\ell+2}\underline{S}_{\ell,1,\ell+2,n,-2}\left(\frac{r_0}{r}\right)^n\\
&+r_0^{-3}\sum_{n=0}^{\ell+3}(\Bb_\ell \underline{R}'_{\ell,n}+2M\Aa_{\ell} \underline{Q}'_{\ell,n}+\O^u(r_0^{-\delta}+Mr_0^{-1}))\left(\frac{r_0}{	r}\right)^n\\
&+(-1)^{\ell} M\Aa_{\ell}\frac{\ell(\ell+1)(\ell+2)!}{3(\ell-2)!}r^{-3}(\log r-\log r_0)\\
&-(-1)^{\ell}  M\Bb_{\ell}\frac{(\ell-1)(\ell+2)(\ell+2)!}{{8}(\ell-2)!}r^{-4}\log^2 r+M\O(r^{-4}r_0\log r)),
\end{nalign}
where the constants $\underline{S}_{\ell,p,\ell-s,n,s}$ are defined in~\eqref{eq:alb:SLPJN} and for some constants $\underline{Q}'_{\ell,n}, \underline{R}'_{\ell,n}$, which satisfy, for $n\leq 2$:
\begin{equation}\label{eq:alb:thm:Q'R'}
    \underline{R}'_{\ell, n}=-\frac{2(|s|+1)}{\ell(\ell+1)}\underline{Q}'_{\ell,n}=(-1)^{\ell}\frac{(\ell-1)!(\ell+2)!}{(\ell+1)!(\ell-2)!}(1+|s|)!\frac{(2-n)!(\ell+n-2)!}{n!(\ell-n+2)!}.
\end{equation}

Moreover, the following limits are conserved along $\Scrip$ (cf.~\eqref{eq:alp:thmNP}, \eqref{eq:alp:thm:otherlimit}):
\begin{align}\label{eq:alb:thmNP}
 \NP{2-\ell}{0} [\alphas[-2]]=\lim_{v\to\infty} r^{2-\ell}\Dv \Asl[-2]&=(-1)^{\ell+1}2M\Aa_{\ell}\frac{(\ell+2)!(\ell+1)}{(\ell-2)!}=\frac{(\ell+2)!}{(\ell-2)!} \conj{\NP{2-\ell}{0} [\alphas[2]]},\\\label{eq:alb:thm:otherlimit}
 \lim_{v\to\infty} \frac{1}{r}(\rDv)^{4}(\Omega^2r\alb)&=(-1)^{\ell+1} 2M\Aa_{\ell}\frac{\ell(\ell+1)(\ell+2)!}{(\ell-2)!}=\frac{(\ell+2)!}{(\ell-2)!}\conj{\lim_{v\to\infty}r^4\al_\ell}.
\end{align}

Finally, the limit $\lim_{\Scrip}r\alb_\ell$ satisfies:
\begin{equation}\label{eq:alb:thm:limitofalphabar}
\rad{\alb}{\ell,\Scrip}:=\lim_{v\to\infty}r\alb_\ell=(-1)^\ell \left(\frac{12\Bb_\ell}{\ell(\ell+1)}-4M\Aa_\ell\right) r_0^{-3}+\O(r_0^{-3-\delta}+Mr_0^{-4}).
\end{equation}
\end{thm}
\begin{rem}
    We only provide the values of $\underline{Q}'_{\ell,n}$ and $\underline{S}_{\ell,n}'$ for $n\leq 2$ as those will be the only ones that will play a role in later parts of the paper.
\end{rem}
\subsection{Proof of Theorem~\ref{thm:alb:0}}
Compared to the proof of Theorem~\ref{thm:alp:0}, there are two differences: Firstly, the computation of transversal derivatives along $\Cin$ is very slightly different, as has already been mentioned in the beginning of this section.

The other, more notable, difference is that we here need to understand more refined asymptotic estimates since we also want to compute the coefficient in front of the $r_0^{-3}$-decay of the limit $\lim_{v\to\infty}r\alb$ (cf.~\eqref{eq:alb:thm:limitofalphabar}). 
Indeed, as can be seen from \eqref{eq:alb:thm}, the leading order contributions from the first line all vanish at $\Scrip$, as $\underline{S}_{\ell,0,\ell+2,0,-2}=\underline{S}_{\ell,1,\ell+2,0,-2}$=0. We will come back to this at the end of the proof. 
\begin{proof}
First, we compute the transversal derivatives. In the same way in which we proved Propositions~\ref{prop:al:transversal_derivatives} and \ref{prop:alp:transversal_derivatives}, we have:
\begin{prop}\label{prop:alb:transversal_derivatives}
Let $N\geq 2$, $s=-2$. Define the constants
\begin{equation}
\ctrbar{N}=\frac{(|s|+p)!}{(N+s+p)!}\prod_{j=0}^{N-1}\at{s}{j}{\ell}=\frac{(|s|+p)!}{(N+s+p)!}(-1)^N\frac{(\ell-s)!(N+s+\ell)!}{(\ell-s-N)!(\ell+s)!}.
\end{equation}
Then, for some constants $\underline{E}'_n$ that we shall not yet compute,
\begin{multline}\label{eq:alb:transversal_derivatives}
\restr{(\rDv)^{N}\alphas[s=-2]}{\Cin}=\left.\Aa r_0^{s-p+N}\ctrbar{N}\right|_{p=0}
+\left.\left(\Bb r_0^{s-p+N}\log r_0 +\BBb r_0^{-s-p+N}\right)\ctrbar{N}\right|_{p=1}\\
+\left.\Bb r_0^{s-p+N}		(-\delta_{N,0}\cdot (|s|+p)!+\ctrbar{N}  	\left(+\sum_{i=p+s+2}^{N+s+p}\frac{1}{i}\right)\right|_{p=1}+M\underline{E}'_{n}\Aa r_0^{-3+N}
+ \O( r_0^{-3+N-\delta}).
\end{multline}
In fact, the formula is valid also for $N=0,1$.
\end{prop}
\begin{proof}
For $N\geq 2$, the proof proceeds inductively in the same way as the proof of Proposition~\ref{prop:alp:transversal_derivatives}, starting from \eqref{eq:alb:2transderivatives}. The cases $N=0,1$ follow from direct comparison with \eqref{eq:alb:1transderivatives} and \eqref{eq:alb:2transderivatives}.
\end{proof}
Next, we again compute the Minkowskian part of the solution.
\begin{prop}\label{prop:alb:minkowskian}
Let $s=-2$:
\begin{multline}\label{eq:alb:minkowskian}
\sum_{i=0}^{\ell-s}\frac{1}{i!}\left(\frac{1}{r_0}-\frac{1}{r}\right)^i\restr{(\rDv)^{i}\alphas}{\Cin}
=\left.\Aa r_0^{s-p}\sum_{n=0}^{\ell-s} \left(\frac{r_0}{r}\right)^n \cdot ( \underline{S}_{\ell,p,\ell-s,n,s}+\O^u(Mr_0^{-1}+r_0^{-1-\delta}))\right|_{p=0}\\
+\left.(\Bb\log r_0+\BBb) r_0^{s-p}\sum_{n=0}^{\ell-s} \left(\frac{r_0}{r}\right)^n \cdot \underline{S}_{\ell,p,\ell-s,n,s}\right|_{p=1}\\
+\left.\Bb r_0^{s-p}\sum_{n=0}^{p-s-1}\left(\frac{r_0}{r}\right)^n \underline{R}_{\ell,p,\ell-s,n,s}\right|_{p=1}\\
+\left.\Bb r_0^{s-p}\sum_{n=p-s}^{\ell-s}\left(\frac{r_0}{r}\right)^n\underline{S}_{\ell,p,\ell-s,n,s}\left(\sum_{k=0}^{\ell-s-1-n}\frac{1}{\ell-p-k}+\sum_{k=2}^{\ell-s}\frac{1}{k+s+p}\right)\right|_{p=1},
\end{multline}
where
\begin{equation}
\underline{R}_{\ell,p,\ell-s,n,s}:=(-1)^{\ell-p+1}\frac{ (\ell-p)! (\ell-s)! (p+|s|)!}{ (\ell+p)! (\ell+s)!} \frac{(p-n-s-1)!(\ell+n+s)!}{n! (\ell-n-s)!},
\end{equation}
and
\begin{equation}\label{eq:alb:SLPJN}
\underline{S}_{\ell,p,j,n,s}:=\frac{(-1)^{\ell-s+n}(|s|+p)!(\ell-s)!(2\ell-j+n)!(\ell-p)!}{n!(\ell+s)!(\ell+p)!(\ell-p-j+n)!(j-n)!}.
\end{equation}

\end{prop}
\begin{rem}
    Notice, in particular, that
\begin{equation}
\lim_{v\to\infty}\sum_{i=0}^{\ell-s}\frac{1}{i!}\left(\frac{1}{r_0}-\frac{1}{r}\right)^i\restr{(\rDv)^{\ell-s-j+i}\alphas}{\Cin}=(-1)^{\ell}\frac{12\Bb_{\ell}r_0^{-3}}{\ell(\ell+1)}+M\O(r_0^{-3}).
\end{equation}
This means that we will have to keep track of the $M$-contributions as well in order to obtain an expression for the limit of $r\alb$.
\end{rem}
\begin{proof}
We again set $M=0$ and argue by linearity. 
We have
\begin{multline}
(r^2\Dv)^{\ell-s}\alphas[s=-2]|_{\Cin}=\Aa r_0^{\ell-p}\ctrbar{\ell-s,p}|_{p=0}\\
+ r_0^{\ell-p}(\Bb \log r_0+\BBb) \ctrbar{\ell-s,p}|_{p=1}+\Bb r_0^{\ell-p}\ctrbar{\ell-s,p}\sum_{i=p+s+2}^{N+s+p}\frac 1i|_{p=1}
\end{multline}
The result now follows by applying Prop.~\ref{prop:pm:Ansatz=0A} to the $\Aa$ and the $\BBb$-terms, and by applying Prop.~\ref{prop:pm:Ansatz=0B} to the $\Bb$-terms. (Cf.~the proofs of Propositions~\ref{prop:al:Minkoswkiansolution} and \ref{prop:alp:minkowskian}.)
\end{proof}
Finally, we have 
\begin{prop}
Let $s=-2$. Then
\begin{multline}
\alphas_\ell=\Omega^2 r\alb_\ell=\Aa_{\ell}r_0^{-2}\sum_{n=2}^{\ell+2}\underline{S}_{\ell,0,\ell+2,n,-2}\left(\frac{r_0}{r}\right)^n+(\Bb\log r_0+\BBb)r_0^{-3}\sum_{n=3}^{\ell+2}\underline{S}_{\ell,1,\ell+2,n,-2}\left(\frac{r_0}{r}\right)^n\\
+r_0^{-3}\sum_{n=0}^{\ell+3}(\Bb_\ell \underline{R}'_{\ell,n}+2M\Aa_{\ell} \underline{Q}'_{\ell,n}+\O^u(r_0^{-\delta}+Mr_0^{-1}))\left(\frac{r_0}{	r}\right)^n\\
-\left.M\Aa_{\ell}\ctrbar{\ell-s}\frac{(\ell-p)(\ell+p+1)!}{(2\ell)!}\frac{(-1)^{p+1-s}}{(p+1-s)!}r^{s-p-1}(\log r-\log r_0)\right|_{p=0}\\
-\left.  M\Bb_{\ell}\ctrbar{\ell-s}\frac{(\ell-p)(\ell+p+1)!}{(2\ell)!}\frac{(-1)^{p+1-s}}{{2}(p+1-s)!}r^{s-p-1}\log^2 r\right|_{p=1}
+M\O(r^{-4}r_0\log r).
\end{multline}
\end{prop}
\begin{proof}
Apart from the $r_0^{-3}M\Aa_{\ell}\sum_{n=0}^{\ell+3}\underline{Q}'_{\ell,n}\left(\frac{r_0}{r}\right)^n$-term, the proof proceeds in the same way as in the previous sections, the only difference being that $\ctr{\ell-s}$ is now replaced by $\ctrbar{\ell-s}$:
First, we prove that
\begin{multline*}
(\rDv)^{\ell-s+1}\alphas=\left.-M\Aa_{\ell} \ctrbar{\ell-s}\frac{(\ell-p)!(\ell+p+1)!}{(2\ell)!}
r^{\ell-p}\right|_{p=0}\\-\left.M\Bb_{\ell}\ctrbar{\ell-s}\frac{(\ell-p)!(\ell+p+1)!}{(2\ell)!}r^{\ell-p}\log r\right|_{p=1}+M\O(r_0^{\ell})
\end{multline*}
as in Propositions~\ref{prop:al:pain:rDvN} and~\ref{prop:alp:rDvN and Asl} for the $\Aa$ and the $\Bb$ terms, respectively.
We subsequently compute the nested integrals from Proposition~\ref{prop:al:fundisindieTONNEtreten} as in Propositions~\ref{prop:al:prop:asyfull1} and \ref{prop:alp:asy}, and the result follows.

To also construct the $r_0^{-3}\sum_{n=0}^{\ell+3}\underline{Q}'_{\ell,n}\left(\frac{r_0}{r}\right)^n$-term, with \eqref{eq:alb:thm:Q'R'} (which is, in particular, necessary to compute the $M\neq 0$-contribution to the limit~\eqref{eq:alb:thm:limitofalphabar}), we need to work a bit harder:

Since we have already captured the contribution coming from the $\Bb$- and the $\BBb$-terms, it now suffices to consider the problem where $\Bb$, $\BBb=0$ and $\delta=\infty$. We then do the following Post-Minkowskian type argument:
We split the Teukolsky equation up into equations \eqref{eq:pm:teukM=0}, \eqref{eq:pm:teukM:alt} and \eqref{eq:pm:teukM2}, where we only pose non-trivial data for (the homogeneous) equation \eqref{eq:pm:teukM=0}, and pose trivial data for the other two (inhomogeneous) equations. 

To be precise, the data we pose for \eqref{eq:pm:teukM=0} is that $\alphas_{M=0}|_{\Cin}=0$, that $\Dv\alphas_{M=0}|_{\Cin}=0$ and that $\lim_{u\to-\infty}(\rDv)^2 \alphas_{M=0}|_{\Cin} =4\Ds2\Ds1\D1\D2\Aa$ (cf.~\eqref{eq:alb:data}, \eqref{eq:alb:1transderivatives}, \eqref{eq:alb:2transderivatives}), together with the no incoming radiation condition \eqref{eq:alb:noincomingradiation}.
The solution is then given by
$\alphas_{M=0}=\sum_{i=0}^{\ell-s}\frac{1}{i!}\left(\frac{1}{r_0}-\frac1r\right)^i(\rDv)^i\alphas_{M=0}|_{\Cin}$, and is determined by Proposition~\ref{prop:pm:Ansatz=0A}, the result being the first line of \eqref{eq:alb:minkowskian}; in particular, the coefficients $S_n$ in \eqref{eq:pm:prop:Sn} are given by  \eqref{eq:alb:SLPJN}. In particular, we read off that $\lim_{\Scrip}\alphas_{M=0}=0$.

We now insert $\alphas_{M=0}$ as an \textit{inhomogeneity} into \eqref{eq:pm:teukM}. 
We can easily see from Proposition~\ref{prop:al:pain:rDvN}, eq.~\eqref{eq:al:prop:painAsl:l-p>0 precise} that the solution $\alphas_{M^1}$ must thus be of the form assumed in Proposition~\ref{prop:pm:inhom}, see also Remark~\ref{rem:al:sneakyPM}.
 Applying the result of Proposition~\ref{prop:pm:inhom} then gives us a handy expression for the corresponding solution $\alphas_{M^1}$, namely \eqref{eq:pm:AnsatzM=1}, with $S_{\ell-s}$ given by $\underline{S}_{\ell,p,\ell-s,\ell-s,s}|_{p=0}$. 
 In particular, we have, for $n\leq 2$, the expression:
$$ Q_{n}=\frac{(\ell-s)!}{(\ell+s)!}(-1)^{\ell+1} \frac{(2-n)!(\ell+n-2)!}{n!(\ell+2-n)!}.$$


Finally, it is easy to see that the remaining difference $\alphas_{\O(M^2)}=\alphas-\alphas_{M=0}-\alphas_{M^1}$, as a solution to \eqref{eq:pm:teukM2} with the previous $\alphas_{M=0}$ and $\alphas_{M^1}$ as inhomogeneities  and with trivial data, only produces correction terms at higher order. This concludes the proof of the proposition~\dots
\end{proof}
\dots and thus the proof of Thm.~\ref{thm:alb:0}.
\end{proof}
\subsection{Treatment of the data of Definitions~\ref{defi:physd:parabolic} and \ref{defi:physd:graviton}}
In the above, we discussed the asymptotics arising from data as in Def.~\ref{defi:physd:Nbodyseed}. We now briefly cover the cases of the other definitions:
\begin{thm}
    Let $\alb$ be the unique smooth scattering solution to \eqref{eq:lin:Teukalb} arising from seed data as in Def.~\ref{defi:physd:parabolic}, i.e.~corresponding to parabolic orbits. 
    Then we have, for any $\ell\geq2$,
    \begin{equation}
        \lim_{v\to\infty}r^{-1/3}(\rDv)^4(r\Omega^2\alb_{\ell})=\frac{(\ell+2)!}{(\ell-2)!}\cdot(-1)^{\ell+1}\sqrt{3}\pi{\albdata_{\mathrm{par},\ell}}(\ell+\tfrac53)(\ell-\tfrac23)
    \end{equation}
    as well as
    \begin{equation}
        \lim_{v\to\infty}r\alb_{\ell}=S_{\ell,\tfrac23,\ell+2,0,-2}\albdata_{\mathrm{par},\ell}r_0^{-\frac{8}{3}}+\O(r_0^{-\frac83-\delta}+Mr_0^{-\frac{11}{3}}),
    \end{equation}
    with $S_{\ell,\tfrac23,\ell+2,0,-2}$ being defined in \eqref{eq:al:SLPJN}. 
\end{thm}
Remarkably, the decay rate of $\alb$ along $\Scrip$ is thus slower in the case of parabolic orbits (cf.~\ref{eq:alb:thm:limitofalphabar}).

In the case of definition~\ref{defi:physd:graviton}, Theorem~\ref{thm:alb:0} applies with $\albdata=0$. There is something interesting to observe, however. Clearly, we have the following statement if $M=0$:
\begin{thm}\label{thm:alb:M=0}
     Let $\alb$ be the unique smooth scattering solution to \eqref{eq:lin:Teukalb} arising from seed data as in Def.~\ref{defi:physd:parabolic}, i.e.~corresponding a compactly supported gravitational perturbation along $\Scrim$. \textbf{Moreover, let $M=0$.}
     Then we have, for any $N\in\mathbb N_{\geq 2}$
     \begin{equation}
         \lim_{v\to\infty}r\alb_{\ell}=\left(\sum_{p=2}^{N}\mathscr{E}_p S_{\ell,p,\ell+2,0,-2}r_0^{-2-p}\right)+\O(r_0^{-3-N})
     \end{equation}
     for some generically nonvanishing constants $\mathscr{E}_p$.
\end{thm}
However, we recall from \eqref{eq:al:SLPJN=0} that $S_{\ell,p,\ell+2,0,-2}=0$ for $\ell\geq p$. In other words, this means that if $M=0$ then the limits of $r\alb_{\ell}$ along $\Scrip$ decay like $r_0^{-3-\ell}$ as $u\to-\infty$.
This is a property that is crucially changed if $M\neq 0$, cf.~\eqref{eq:alb:thm:limitofalphabar}.
\newpage
\section{Post-Minkowskian expansions of the Teukolsky equation}\label{sec:PM}
\newcommand{\Dr}{\sll_{r}}
\newcommand{\Dro}{\sll_{r_0}}
This section presents a Post-Minkowskian expansion of \eqref{eq:cons:teuk}. The purpose of this is for us to reduce certain somewhat tedious computations to simply solving Minkowskian Teukolsky equations with inhomogeneities. 
The results of this section have already been used and referenced in the previous sections for the proofs of Prop.~\ref{prop:al:Minkoswkiansolution}, Prop.~\ref{prop:alp:minkowskian} and Prop.~\ref{prop:alb:minkowskian}, i.e.~to compute the Minkowskian parts of the solution.
We have also used them to prove~\eqref{eq:alb:thm:Q'R'}.
\subsection{The expansion}
Recall \eqref{eq:cons:teuk}. We re-write it in coordinates $(r_0,r,\theta,\varphi)$, denoting $\pu r_0=-(1-\frac{2M}{r_0})=-D_0$. Since $\pu=-D_0\partial_{r_0}-D\partial_r$ and $\pv=D\partial_r$, we have
\begin{multline}\label{eq:pm:teuk}
-(D_0\Dro+D\Dr)(r^{-2s} D^{s+1}\Dr \alphas)=\frac{D^{s+1}}{r^{2+2s}}\laps\alphas-\frac{2MD^{s+1}}{r^{3+2s}}(1+s)(1+2s)\alphas.
\end{multline}
Note that 
$$
-D\Dr(r^{-2s}D^{s+1}\Dr\alphas)=(2sr^{-2s-1}D^{s+2}-(s+1)2Mr^{-2s-2}D^{s+1})\Dr\alphas-D^{s+2}r^{-2s}\Dr^2\alphas.
$$
Dividing by $D^{s+1}r^{-2s}$ thus casts \eqref{eq:pm:teuk} into the form
\begin{multline}
-(\Dro+\Dr)\alphas+\frac{D_0-D}{D_0}\Dr^2\alphas+2s\frac{D}{D_0}\Dr\alphas-(s+1)\frac{2M}{r^2}\Dr\alphas\\
=\frac{\laps\alphas}{D_0r^2}-\frac{2M}{D_0r^3}(1+s)(1+2s)\alphas.
\end{multline}

We now decompose the equation above into the following three equations: Firstly:
\begin{equation}\label{eq:pm:teukM=0}
\red{-(\Dro+\Dr)(r^{-2s}\Dr\alphas_{M=0})=\frac{\laps\alphas_{M=0}}{r^{2+2s}}},
\end{equation}
the Minkowskian Teukolsky equation. Secondly:
\begin{multline}\label{eq:pm:teukM}
\purple{-(\Dro+\Dr)(r^{-2s}\Dr\alphas_{M^1})}\\
\red{+r^{-2s}\left(\frac{2M}{r_0}+(s+1)\frac{2M}{r}\right)\Dro\Dr\alphas_{M=0}}\\
\red{+r^{-2s}\left((s+2)\frac{2M}{r}\right)\Dr^2\alphas_{M=0}-\left(2sr^{-2s-1}(s+2)\frac{2M}{r}+(s+1)2Mr^{-2s-2}\right)\Dr\alphas_{M=0}}\\
=\purple{\frac{\laps\alphas_{M^1}}{r^{2+2s}}}\\
-\red{\frac{2M(s+1)\laps\alphas_{M=0}}{r^{3+2s}}-\frac{2M}{r^{3+2s}}(1+s)(1+2s)\alphas_{M=0}},
\end{multline}
the Minkowskian Teukolsky equation with inhomogeneity, 
and thirdly:
\begin{multline}\label{eq:pm:teukM2}
-(D_0\Dro+D\Dr)(r^{-2s}D^{s+1}\Dr\alphas_{\O(M^2)})\\
\purple{-r^{-2s}(D_0D^{s+1}-1)\Dro\Dr\alphas_{M^1}-r^{-2s}(D^{s+2}-1)\Dr^2\alphas_{M^1}}\\
\purple{+2sr^{-2s-1}(D^{s+2}-1)\Dr\alphas_{M^1}-(s+1)2Mr^{-2s-2}\Dr\alphas_{M^1}
}\\
\red{-r^{-2s}(D_0D^{s+1}-(1-\tfrac{2M}{r_0}-(s+1)\tfrac{2M}{r}))\Dro\Dr\alphas_{M=0}-r^{-2s}(D^{s+2}-(1-(s+2)\tfrac{2M}{r})\Dr^2\alphas_{M=0}}\\
\red{+2sr^{-2s-1}(D^{s+2}-(1-(s+2)\tfrac{2M}{r}))\Dr\alphas_{M=0}-(s+1)2Mr^{-2s-2}(D^{s+1}-1)\Dr\alphas_{M=0}
}\\
=\frac{D^{s+1}}{r^{2+2s}}\laps\alphas_{\O(M^2)}-\frac{2MD^{s+1}}{r^{3+2s}}(1+s)(1+2s)\alphas_{\O(M^2)}\\
\purple{+\frac{(D^{s+1}-1)\laps\alphas_{M^1}}{r^{2+2s}}-\frac{2MD^{s+1}}{r^{3+2s}}(1+s)(1+2s)\alphas_{M^1}}\\\red{+\frac{(D^{s+1}-(1-(s+1)\tfrac{2M}{r}))\laps\alphas_{M=0}}{r^{3+2s}}-\frac{2M(D^{s+1}-1)}{r^{3+2s}}(1+s)(1+2s)\alphas_{M=0}},
\end{multline}
the Schwarzschildean Teukolsky equation with inhomogeneity.
Note that adding \eqref{eq:pm:teukM=0}--\eqref{eq:pm:teukM2} gives back \eqref{eq:pm:teuk} with $\alphas=\alphas_{M=0}+\alphas_{M^1}+\alphas_{\O(M^2)}$.
The decomposition $\alphas=\alphas_{M=0}+\alphas_{M^1}+\alphas_{\O(M^2)}$ corresponds to a Post-Minkowskian decomposition in $(r_0,r)$-coordinates, with an exact keeping track of the error term $\alphas_{\O(M^2)}$. 
In applications, equation \eqref{eq:pm:teukM=0}---the Teukolsky equation in Minkowski---is solved first, equation \eqref{eq:pm:teukM} then becomes a Minkowskian Teukolsky equation with inhomogeneity depending on the already solved for $\alphas_{M=0}$, and \eqref{eq:pm:teukM2} is a Schwarzschildean Teukolsky equation with inhomogeneity depending on $\alphas_{M=0}$ and~$\alphas_{M^1}$.

We observe that by eq.~\eqref{eq:pm:teukM=0}, the equation~\eqref{eq:pm:teukM} can be rewritten as
\begin{multline}\label{eq:pm:teukM:alt}
\purple{-(\Dro+\Dr)(r^{-2s}\Dr\alphas_{M^1})=\frac{\laps \alphas_{M^1}}{r^2}}\\
\red{+\frac{2M\laps\alphas_{M=0}}{r^2r_0}-\frac{2M(1+s)(1+2s)\alphas_{M=0}}{r^3}}\\
\red{+(s+1)\frac{2M}{r^2}\Dr\alphas_{M=0}+\frac{2s}{r}(D_0-D)\Dr\alphas_{M=0}-(D_0-D)\Dr^2\alphas_{M=0}}.
\end{multline}
\subsection{Explicitly solving the Minkowskian Teukolsky equation}
We now present explicit solutions, supported on fixed angular frequency, of \eqref{eq:pm:teukM=0}. The following proposition was used in the proof of Propositions~\ref{prop:al:Minkoswkiansolution} and \ref{prop:alb:minkowskian}:
\begin{prop}\label{prop:pm:Ansatz=0A}
Suppose that $\A_{\ell}$ is supported on angular frequency $\ell\geq |s|$, that $\ell\geq p\in\mathbb N_{\geq 0}$, and that 
\begin{equation}\label{eq:pm:prop:AnsatzM=0A}
\alphas_{M=0}=\A_{\ell}r_0^{s-p}\sum_{n=0}^{\ell-s} S_n \left(\frac{r_0}{r}\right)^n
\end{equation}
solves \eqref{eq:pm:teukM=0}. Then
\begin{equation}\label{eq:pm:prop:Sn}
S_n=S_{\ell-s}\frac{(-1)^{\ell-s}(\ell-s)!(\ell-p)!}{(2\ell)!} \cdot(-1)^n \frac{(\ell+s+n)!}{n!(s-p+n)!(\ell-s-n)!},
\end{equation}
and \eqref{eq:pm:prop:AnsatzM=0A} indeed solves \eqref{eq:pm:teukM=0}. \textbf{The result holds true for any $p\in\mathbb R$} if we replace the factor $\frac{(\ell-p)!}{(s-p+n)!(\ell-s-n)!}$ by $\binom{\ell-p}{s-p+n}$ and extend the binomial coefficient to non-integer arguments in the usual fashion.
Moreover, 
\begin{equation}\label{eq:pm:prop:AnsatzM=0A:compare}
(r^2\Dr)^{\ell-s}\alphas=(-1)^{\ell-s}\A_{\ell}r_0^{\ell-p}(\ell-s)!S_{\ell-s}.
\end{equation}
\end{prop}
\begin{proof}
The proof is strictly simpler than the proof of the next proposition. 
\end{proof}
The following proposition was used in the proof of Propositions~\ref{prop:alp:minkowskian} and \ref{prop:alb:minkowskian}.
\begin{prop}\label{prop:pm:Ansatz=0B}
Suppose that $\B_{\ell}$ is supported on angular frequency $\ell\geq |s|$, that $\ell\geq p\in\mathbb N_{\geq 0}$, and that 
\begin{equation}\label{eq:pm:AnsatzM=0}
\alphas_{M=0}=\B_{\ell}r_0^{s-p}\left(\sum_{n=0}^{\ell-s} S_n \left(\frac{r_0}{r}\right)^n\log r_0+\sum_{n=0}^{\ell-s} R_n \left(\frac{r_0}{r}\right)^n\right)
\end{equation}
solves \eqref{eq:pm:teukM=0}. Then
\begin{equation}\label{eq:pm:AnsatzM=0Sn}
S_n=S_{\ell-s}\frac{(-1)^{\ell-s}(\ell-s)!(\ell-p)!}{(2\ell)!} \cdot(-1)^n \frac{(\ell+s+n)!}{n!(s-p+n)!(\ell-s-n)!},
\end{equation}
and, for  $n\leq p-s-1$,
\begin{equation}\label{eq:pm:AnsatzM=0Rn}
R_n=S_{\ell-s}(-1)^{\ell-p+1}\frac{(\ell-s)!(\ell-p)!}{(2\ell)!}\frac{(p-n-1-s)!(n+\ell+s)!}{n!(\ell-s-n)!},
\end{equation}
whereas for $p-s\leq n\leq \ell-s$:
\begin{equation}
R_n=S_n\frac{R_{\ell-s}}{S_{\ell-s}}+S_{n}\sum_{i=0}^{\ell-s-1-n}\frac{1}{\ell-p-i}.
\end{equation}
Moreover, \eqref{eq:pm:AnsatzM=0} indeed solves \eqref{eq:pm:teukM=0}, and
\begin{equation}\label{eq:pm:Ansatz=0B:compare}
(r^2\Dr)^{\ell-s}\alphas_{M=0}=(-1)^{\ell-s}\B_{\ell}r_0^{\ell-p}\left((\ell-s)!S_{\ell-s}\log r_0+(\ell-s)!R_{\ell-s}\right).
\end{equation}
\end{prop}
\begin{proof}
We compute 
\begin{align*}
\Dr\alphas_{M=0}&=\B_{\ell}r_0^{s-p}\sum_{n=0}^{\ell-s}\frac{-n}{r}\left( S_n \left(\frac{r_0}{r}\right)^n\log r_0+ R_n \left(\frac{r_0}{r}\right)^n\right),\\
-\Dr^2\alphas_{M=0}&=\B_{\ell}r_0^{s-p}\sum_{n=0}^{\ell-s}\frac{-n(n+1)}{r^2}\left( S_n \left(\frac{r_0}{r}\right)^n\log r_0+ R_n \left(\frac{r_0}{r}\right)^n\right),
\end{align*}
and
\begin{align*}
-\Dro\Dr\alphas_{M=0}&=r_0^{s-p}\sum_{n=0}^{\ell-s}\left(\frac{r_0}{r}\right)^{n}\left((n+s+p)n\frac{1}{r_0r}(S_n\log r_0+R_n)+\frac{n}{r_0r}S_n\right)\\
&=r_0^{s-p}\sum_{n=0}^{\ell-s-1}\left(\frac{r_0}{r}\right)^{n}\left((n+1+s+p)n\frac{1}{r^2}(S_{n+1}\log r_0+R_{n+1})+\frac{n+1}{r^2}S_{n+1}\right).
\end{align*}
Inserting the above expressions into \eqref{eq:pm:teukM=0} and equating coefficients, we obtain the system of equations (using $n(n+2s+1)-(\ell-s)(\ell+s+1)=(n-(\ell-s))(n+\ell+s+1))$:
\begin{align}
(n-(\ell-s))(n+\ell+s+1)S_n&=S_{n+1}(n+1)(n+1+s-p),\\
(n-(\ell-s))(n+\ell+s+1)R_n&=R_{n+1}(n+1)(n+1+s-p)+(n+1)S_{n+1}\label{eq:pm:AnsatzM=0Rn:relation}.
\end{align}
The first of these inductively implies \eqref{eq:pm:AnsatzM=0Sn}. 

For the second equation \eqref{eq:pm:AnsatzM=0Rn:relation}, we first set $n=p-s-1$:
\begin{equation}
R_{p-s-1}=\frac{p-s}{(p+\ell)(p-\ell-1)}S_{p-s}.
\end{equation}
For $n<p-s-1$, we then simply have $(\ell-s-n)(n+\ell+s+1)R_n=R_{n+1}(n+1)(p-s-1-n)$, which is solved by
\begin{nalign}
R_n&=R_{p-s-1}\frac{(p-s-1)!(\ell-p+1)!}{(\ell+p-1)!}\frac{(p-n-1-s)!(n+\ell+s)!}{n!(\ell-s-n)!}\\
&=-S_{p-s}\frac{(p-s)!(\ell-p)!}{(\ell+p)!}\frac{(p-n-1-s)!(n+\ell+s)!}{n!(\ell-s-n)!}\\
&=S_{\ell-s}(-1)^{\ell-p+1}\frac{(\ell-s)!(\ell-p)!}{(2\ell)!}\frac{(p-n-1-s)!(n+\ell+s)!}{n!(\ell-s-n)!}.
\end{nalign}
To compute the remaining coefficients, we observe that \eqref{eq:pm:AnsatzM=0Rn:relation} implies that $R_n=0$ for $n>\ell-s$. Finally, for  $p-s\leq n\leq \ell-s$, we write $R_n=S_n\cdot P_n$ in \eqref{eq:pm:AnsatzM=0Rn}, which gives:
\begin{align*}
S_nP_n=S_{n}P_{n+1}+\frac{1}{n+1+s-p}S_{n},
\end{align*}
which is solved by 
$
P_{n-j}=P_{n+1}+\sum_{i=0}^j\frac{1}{s-p+n+1-i}
$
or, equivalently,
\begin{equation}
P_{n}=P_{\ell-s}+\sum_{i=0}^{\ell-s-1-n}\frac{1}{\ell-p-i}.
\end{equation}
\end{proof}
\subsection{Explicitly solving the inhomogeneous Minkowskian Teukolsky equation}
We now prove two statements about eq.~\eqref{eq:pm:teukM:alt}. The relevance of the first statement will be to compute the inhomogeneity of \eqref{eq:pm:teukM:alt}, provided that $\alphas_{M=0}$ is as in Proposition~\ref{prop:pm:Ansatz=0A}. 
The purpose of the second proposition will be to provide a simple expression for solutions of~\eqref{eq:pm:teukM:alt}, provided that the inhomogeneity is as computed in the first proposition.
The result of the second proposition has already been used in the proof of \eqref{eq:alb:thm:Q'R'} and \eqref{eq:alb:thm:limitofalphabar}.
\begin{prop}\label{prop:pm:inhomogeneity}
Define the following functional, denoting the inhomogeneity of \eqref{eq:pm:teukM:alt}:
\begin{multline}
g[\alphas_{M=0}]:=\frac{2M\laps\alphas_{M=0}}{r^2r_0}-\frac{2M(1+s)(1+2s)\alphas_{M=0}}{r^3}\\
+(s+1)\frac{2M}{r^2}\Dr\alphas_{M=0}+\frac{2s}{r}(D_0-D)\Dr\alphas_{M=0}-(D_0-D)\Dr^2\alphas_{M=0}.
\end{multline}
Then, for $\alphas_{M=0}$ as in Proposition~\ref{prop:pm:Ansatz=0A}, we have that
\begin{multline}\label{eq:PM:Tn}
g[\alphas_{M=0}]\cdot \frac{r^2r_0^{1+p-s}}{2M}=\sum_{n=0}^{\ell-s+1}\left(\frac{r_0}{r}\right)^n\underbrace{\left((n-\ell+s)(n+\ell+s+1)S_n-(n+s)(n+2s)S_{n-1}\right)}_{:=T_n}.
\end{multline}
\end{prop}
\begin{proof}
By definition of $g$, we have that
\begin{align*}
g[\alphas_{M=0}]\cdot \frac{r^2r_0^{1+p-s}}{2M}=-(\ell-s)(\ell+s+1)\alphas_{M=0}-\frac{r_0}{r}(1+s)(1+2s)\alphas_{M=0}\\
+(3s+1)\Dr\alphas_{M=0}\cdot r_0-2s\Dr\alphas_{M=0}\cdot r+(r^2-rr_0)\Dr^2\alphas_{M=0}.
\end{align*}
We now insert the specific form of $\alphas_{M=0}$ and repeat the calculations done in the proof of the previous proposition to get
\begin{align*}
g[\alphas_{M=0}]\cdot \frac{r^2r_0^{1+p-s}}{2M}=&\sum_{n=0}^{\ell-s}S_n[-(\ell-s)(\ell+s+1)]\left(\frac{r_0}{r}\right)^n-(1+s)(1+2s)S_n\left(\frac{r_0}{r}\right)^{n+1}\\
&-\frac{n}{r}S_n[(3s+1)r_0-2sr]\left(\frac{r_0}{r}\right)^n+n(n+1)\left(1-\frac{r_0}{r}\right)S_n\left(\frac{r_0}{r}\right)^n\\
=&\sum_{n=0}^{\ell-s}S_n\left(\frac{r_0}{r}\right)^n[-(\ell-s)(\ell+s+1)+2sn+n(n+1)]\\
&+\sum_{n=1}^{\ell-s+1}S_{n-1}\left(\frac{r_0}{r}\right)^n[-(1+s)(1+2s)-(n-1)(3s+1)-(n-1)n].
\end{align*}
The result follows from $-(\ell-s)(\ell+s+1)+2sn+n(n+1)=(n-\ell+s)(n+\ell+s+1)$ and $-(1+s)(1+2s)-(n-1)(3s+1)-(n-1)n=-(n+s)(n+2s)$.
\end{proof}

\begin{prop}\label{prop:pm:inhom}
Let $\ell\geq p\in\mathbb N_{\geq0}$, and let $\alphas_{M=0}$ be as in Proposition~\ref{prop:pm:Ansatz=0A}. Then, if $\alphas_{M^1}$ is of the form
\begin{equation}\label{eq:pm:AnsatzM=1}
\alphas_{M^1}=2Mr_0^{s-p-1}\left(\sum_{n=0}^{\ell-s+1} S'_n \left(\frac{r_0}{r}\right)^n\log r_0/r+Q_n  \left(\frac{r_0}{r}\right)^n\right)
\end{equation}
and solves \eqref{eq:pm:teukM:alt}, we have
\begin{equation}\label{eq:pm:prop:S'n}
S'_n=S_{\ell-s}\frac{(-1)^{\ell-s}(\ell-s)!(\ell-p)!}{2(2\ell)!} \cdot \frac{(-1)^n(\ell+s+n)!}{n!(s-p+n-1)!(\ell-s-n)!}=\frac{s-p+n}{2}S_{n}
\end{equation}
and, if $n\leq p-s$:
\begin{equation}\label{eq:pm:prop:Qn}
Q_n=S_{\ell-s}(-1)^{\ell-p+1}\frac{(\ell-s)!(\ell-p)!}{2(2\ell)!}\frac{(p-n-s)!(n+\ell+s)!}{n!(\ell-s-n)!},
\end{equation}
so $Q_n=-\frac{s-p+n}{2}R_n$ if $n\leq p-s-1$. For $\ell-s\geq n>p-s$, the coefficients $Q_n$ can all be explicitly expressed in terms of a free constant, say, $Q_{\ell-s}$. Finally, $Q_{\ell-s+1}=\frac{\ell+s+1}{2}S_{\ell-s}$.
With these expressions, \eqref{eq:pm:AnsatzM=1} indeed solves \eqref{eq:pm:teukM:alt}.
\end{prop}
\begin{proof}
Inserting the ansatz into equation \eqref{eq:pm:teukM:alt}, we obtain
\begin{align*}
&\sum \left(\frac{r_0}{r}\right)^n \log r_0/r[(-n(n+1)-2sn)S'_n+(n+1)(n+s-p)S'_{n-1}]\\
&+\sum \left(\frac{r_0}{r}\right)^n[(-2n-1-2s)S'_n+((n+1)+(n+s-p))S'_{n-1}\\
&+\sum \left(\frac{r_0}{r}\right)^n [(-n(n+1)-2sn)Q_n+(n+1)(n+s-p)Q_{n-1}]\\
=&-(\ell-s)(\ell+s+1)\left(\sum S'_n \left(\frac{r_0}{r}\right)^n\log r_0/r+Q_n  \left(\frac{r_0}{r}\right)^n\right)+g[\alphas_{M=0}]\cdot \frac{r^2r_0^{1+p-s}}{2M},
\end{align*}

Let us, for now, just write 
\begin{align*}
g[\alphas_{M=0}]\cdot \frac{r^2r_0^{1+p-s}}{2M}=\sum_{n=0}^{\ell-s+1}T_n \left(\frac{r_0}{r}\right)^n.
\end{align*}
(Recall that $T_n=0$ if $n<p-s$ from \eqref{eq:PM:Tn}.)
Then we get the following system of equations for the coefficients $S'_n$ and $Q_n$:
\begin{align}
(n-(\ell-s))(n+\ell+s+1)S'_n=&S'_{n+1}(n+1)(n+s-p),\\
(n-(\ell-s))(n+\ell+s+1)Q_n=&Q_{n+1}(n+1)(n+s-p)\nonumber\\
&-(2n+2s+1)S'_n+(2n+1+s-p)S'_{n+1}-T_n\label{eq:pm:inhom:Qnrelations}
\end{align}
We can immediately solve the first equation to obtain \eqref{eq:pm:prop:S'n}.
For the second equation \eqref{eq:pm:inhom:Qnrelations}, we observe that, setting $n=\ell-s$,
\begin{equation}
Q_{\ell-s+1}(\ell-s+1)(\ell-p)=(2\ell+1)S'_{\ell-s}+T_{\ell-s},
\end{equation}
and, setting $n=\ell-s+1$, 
\begin{equation}
(2\ell+2)Q_{\ell-s+1}=Q_{\ell-s+2}(\ell-s+2)(\ell-p+1)-T_{\ell-s+1}.
\end{equation}
In order to ensure that $Q_{\ell-s+2}=0$ (and thus $Q_{\ell-s+i}=0$ for all $i\geq 2$), we demand that
\begin{equation}
Q_{\ell-s+1}=\frac{(2\ell+1)S'_{\ell-s}}{(\ell-s+1)(\ell-p)}+\frac{T_{\ell-s}}{(\ell-s+1)(\ell-p)}=\frac{-T_{\ell-s+1}}{2\ell+2},
\end{equation} 
which is a condition on $S'_{\ell-s}$:
\begin{nalign}
S'_{\ell-s}&=-\frac{(\ell-s+1)(\ell-p)}{(2\ell+1)(2\ell+2)}T_{\ell-s+1}-\frac{T_{\ell-s}}{2\ell+1}\\
&=\left(\frac{(\ell-s+1)(\ell-p)(\ell+1)(\ell+s+1)}{(2\ell+1)(2\ell+2)}-\frac{(\ell-s)(\ell+s)(\ell-p)\ell}{2\ell(2\ell+1)}\right)S_{\ell-s}\\
&=\frac{(\ell-p)S_{\ell-s}}{2(2\ell+1)}((\ell-s+1)(\ell+s+1)-(\ell-s)(\ell+s))=\frac{\ell-p}{2}S_{\ell-s}
\end{nalign}
Here, we used that $T_{\ell-s}=-\ell(\ell+s)S_{\ell-s-1}=(\ell-s)\ell(\ell+s)(\ell-p)S_{\ell-s}/(2\ell)$ and $T_{\ell-s+1}=-(\ell+1)(\ell+s+1)S_{\ell-s}$.

We can now freely specify a value for $Q_{\ell-s}$ and then solve \eqref{eq:pm:inhom:Qnrelations} to obtain all $Q_n$ for $p-s<n\leq\ell-s$. Since we won't need their specific values in this paper, we won't compute them.

Next, assuming that $p-s\geq 0$ (otherwise we'd be done), we insert $n=p-s$ into \eqref{eq:pm:inhom:Qnrelations}:
\begin{equation}
(p-\ell)(p+\ell+1)Q_{p-s}=(p-s+1)S'_{p-s+1}-T_{p-s},
\end{equation}
so $Q_{p-s}$ can be directly determined: Since $S'_{p-s+1}=S_{p-s+1}/2$, and since $$T_{p-s}=(n-\ell-s)(n+\ell+s+1)S_n|_{n=p-s}=S_{n+1}(n+1)(n+1+s-p)|_{n=p-s}=S_{p-s+1}(p-s+1),$$ we get
\begin{equation}
Q_{p-s}=\frac{p-s+1}{2(\ell-p)(\ell+p+1)}S_{p-s+1}.
\end{equation}
Pleasantly, we also have for any $n<p-s$ that
\begin{equation}
(n-(\ell-s))(n+\ell+s+1)Q_n=Q_{n+1}(n+1)(n+s-p),
\end{equation}
so we can inductively determine all $Q_n$ for $n\leq p-s$:
\begin{nalign}
Q_n&=Q_{p-s}\frac{(p-s)!(\ell-p)!}{(\ell+p)!}\frac{(p-n-s)!(n+\ell+s)!}{n!(\ell-s-n)!}\\
&=S_{p-s+1}\frac{(p-s+1)!(\ell-p-1)!}{2(\ell+p+1)!}\frac{(p-n-s)!(n+\ell+s)!}{n!(\ell-s-n)!}\\
&=S_{\ell-s}(-1)^{\ell-p+1}\frac{(\ell-s)!(\ell-p)!}{2(2\ell)!}\frac{(p-n-s)!(n+\ell+s)!}{n!(\ell-s-n)!}.
\end{nalign}
This concludes the proof.
\end{proof}

\newpage
\section{Remarks on summing up the angular modes}
\label{sec:sum}
As we have already said in Remark~\ref{rem:elldependence}, the estimates obtained in this part of the paper cannot directly be summed in $\ell$. For instance, if we write, say, $\al=\sum_{\ell=2}^\infty \al_{\ell}$, and insert the estimates for $\al_{\ell}$ that we have obtained in Theorem~\ref{thm:alp:0}, then the $\ell$-dependent constants hiding in the $\O(\dots)$-terms of Theorem~\ref{thm:alp:0} will not be summable in $\ell$.
To give a concrete example: We cannot directly infer from \eqref{eq:alp:thm:otherlimit} of Thm.~\ref{thm:alp:0} that
\begin{equation}\label{eq:summinginl:limit}
    \lim_{v\to\infty}r^4\al=\sum_{\ell=2}^\infty (-1)^{\ell+1}\frac{2M\A_{\ell}}{(\ell-1)(\ell+2)}.
\end{equation}
The full resolution of this problem is left to upcoming work~\cite{X}, but we here already give an idea how to approach the problem, slightly expanding on the strategy that we have already briefly mentioned in~\S5(e) of \cite{IV}.

Taking $\al$ to be as in \S\ref{sec:alp}, we first write the initial data along $\Cin$ \eqref{eq:alp:data} as
\begin{equation}
    \radc\al={\radc\al}_{,\mathrm{phg}}+{\radc\al}_{,{\Delta}},
\end{equation}
where 
\begin{equation}\label{eq:summinginl:data}
    {\radc\al}_{,\mathrm{phg}}=\A+\B \frac{\log r_0}{r_0}+\frac{\BB}{r_0}.
\end{equation}
We now define $\al_{M=0,\mathrm{phg}}$ as the scattering solution to the \textit{Minkowskian} Teukolsky equation~\eqref{eq:pm:teukM=0} with scattering data ${\radc\al}_{,\mathrm{phg}}$ and no incoming radiation.

From $\al_{M=0,\mathrm{phg}}$, we then define $\al_{M^1,\mathrm{phg}}$ as the scattering solution to the \textit{inhomogeneous Minkowskian} Teukolsky equation~\eqref{eq:pm:teukM} with trivial scattering data and with the inhomogeneity given by $\al_{M=0,\mathrm{phg}}$. (To be precise, the inhomogeneity is given by $g[r^5\Omega^{-2}\al_{M=0,\mathrm{phg}}]$, with $g$ defined in Prop.~\ref{prop:pm:inhomogeneity}.)

Lastly, we define $\al_{\Delta}$ as the scattering solution to the \textit{inhomogeneous Schwarzschildean} Teukolsky equation~\eqref{eq:pm:teukM2}, with scattering data ${\radc\al}_{{\Delta}}$ and with the inhomogeneity sourced by $\al_{M=0,\mathrm{phg}}$ and $\al_{M^1,\mathrm{phg}}$.
We note that $\al=\al_{M=0,\mathrm{phg}}+\al_{M^1,\mathrm{phg}}+\al_{\Delta}$.

Of course, in order for these definitions to be well-defined, one briefly needs to extend out scattering theory to allow for suitably decaying inhomogeneous terms.

The point is now that one can use the estimate \eqref{eq:al:prelim:prop}, the only estimate of this part of the paper where we have sufficiently strong control in $\ell$ in order to admit summation in $\ell$, to prove the following statements:
\begin{itemize}
    \item First, a slightly more refined version of estimate \eqref{eq:al:prelim:prop} directly implies: $|(\al_{M=0,\mathrm{phg}})_{\ell}|\leq C_{\ell} r^{-3}$. 
    \item Inserting this estimate into \eqref{eq:pm:teukM} and proceeding as in the proof of \eqref{eq:al:prelim:prop}, we can then show: $|(\al_{M^1,\mathrm{phg}})_{\ell}|\leq C_{\ell} r^{-4}$.
    \item Inserting both of the above estimates, we can similarly show that $|(\al_{\Delta})_{\ell}|\leq C_{\ell}r^{-4-\delta}$. 
\end{itemize}
In each of the above estimates, the constant $C_{\ell}$ is some constant changing from line to line which, importantly, can be bounded against $\ell^N$ for an $\ell$-independent integer $N$.

The upshot is that the issue of proving, say, \eqref{eq:summinginl:limit} now entirely reduces to proving robust estimates for  the quantity $(\al_{M=0,\mathrm{phg}})_{\ell}$, which can, for instance, be done by proving that $(\al_{M=0,\mathrm{phg}})_{\ell}$ satisfies certain persistence of polyhomogeneity results, cf.~the brief discussion in \S5(e) of \cite{IV}, or, alternatively, by proving ODE estimates for \eqref{eq:pm:teukM=0} etc. (It should be possible to prove that $(\al_{M=0,\mathrm{phg}})_{\ell}$ is bounded by $C_{\ell}|u|^2r^{-5}$ using the expressions from \S\ref{sec:PM}. Such a bound would suffice to prove \eqref{eq:summinginl:limit}.)

In order to prove more refined asymptotic statements, e.g.~that
\begin{equation}
    \lim_{v\to\infty} r^5\log^{-2} r\left(\al-r^{-4}\lim_{v\to\infty}r^4\al \right)=\sum_{\ell=2}^\infty (-1)^{\ell+1}\frac{3M\B_{\ell}}{\ell(\ell+1)},
\end{equation}
one can follow a similar pattern: However, at least if one wants to follow the approach depicted above, one now needs make a stronger assumption on initial data, namely that they admit an expansion up to an error term ${\radc\al}_{,\Delta}=\O(r^{-5})$, and one needs to consider higher-order expansions of the Teukolsky equation (up to order $\O(M^3)$) as well.

\newpage

\section*{Part III:\\ Asymptotic analysis of the remainder of the system\hypertarget{V:partIII}{}}
\addcontentsline{toc}{section}{{\textbf{Part III}: Asymptotic analysis of the remainder of the system}}
In the first two parts of the paper, we have, in particular, defined a class of scattering data describing the far-field region of a system of $N$ infalling masses following approximately hyperbolic Keplerian orbits (Definition~\ref{defi:physd:Nbodyseed}), and we have then found the asymptotic properties of the gauge invariant quantities $\al$, $\Ps$, $\Psb$ and $\alb$ corresponding to the arising solutions in~Theorems~\ref{thm:alp:0}, \ref{thm:Psi} and \ref{thm:alb:0}, respectively. 
We recall that the constants in those theorems are related to those of Def.~\ref{defi:physd:Nbodyseed} and Prop.~\ref{prop:physd} via \eqref{eq:alp:B}, \eqref{eq:Psi:APBP} and \eqref{eq:alb:Aa}, which we recall here for convenience:
\begin{align}
    \Aa_\ell=\frac{(\ell-2)!}{(\ell+2)!}\overline{\A}_\ell,&&\albdata_\ell=\frac{(\ell-2)!}{(\ell+2)!}\overline{\B}_\ell
\end{align}

We will now use these results to deduce the asymptotic properties of the remaining quantities of the system of linearised gravity. 

The case of parabolic orbits (Def.~\ref{defi:physd:parabolic}) will be left as a simple exercise to the reader; the case of compactly supported gravitational perturbations is included in the discussion of hyperbolic orbits (by setting $\A=0$).

\section{Asymptotics for the out- and ingoing null shears \texorpdfstring{$\overone{\hat{\chi}}$}{chihat} and \texorpdfstring{$\overone{\hat{\underline{\chi}}}$}{chihatbar}} \label{sec:xh}
We begin with the asymptotics of the null shears:
Asymptotic expressions for the in- and outgoing null shears follow by straight-forwardly integrating the transport equations \eqref{eq:lin:xh4} and \eqref{eq:lin:xhb4} from $\Cin$ and $\Scrim$, respectively.
\subsection{Asymptotics for \texorpdfstring{$\overone{\hat{\chi}}$}{chihat}}
\begin{prop}\label{prop:xh:xh}
Let $\ell\geq2$. Then the following asymptotic expressions for $\xh$ and its $\Du$-derivative are valid in all of $\DoD$:
    \begin{align}\label{eq:xh:prop:xhexpansion}
        \frac{{r^2\xh}_{\ell}}{\Omega}=\radlp{\xh}-\frac{(-1)^{\ell}2M\A_{\ell}}{(\ell-1)(\ell+2)}\frac1r-\frac{(-1)^{\ell}3M\B_{\ell}}{2\ell(\ell+1)}\frac{\log^2 r}{r^2}+\O\left(M \frac{r_0\log r}{r^2}+\frac{r_0^2}{r^2}\right),\\
    \label{eq:xh:prop:Duxhexpansion}
       \Du( \frac{r^2\xh_{\ell}}{\Omega})=\Du\radlp{\xh}-\frac{(-1)^{\ell}2M\A_{\ell}}{(\ell-1)(\ell+2)}\frac1{r^2}-\frac{(-1)^{\ell}3M\B_{\ell}}{\ell(\ell+1)}\frac{\log^2 r}{r^3}+\O\left(M \frac{r_0\log r}{r^3}+\frac{r_0}{r^2}\right),
    \end{align}
    where the limits $\radlp{\xh}=\radlp{\xh}^{\mathrm{E}}+\radlp{\xh}^{\mathrm{H}}$ in the above are given by:
    \begin{nalign}\label{eq:xh:prop:xhlimits}
      \radlp{\xh}^{\mathrm{E}}:=  \lim_{v\to\infty} r^2\xh_{\ell}^E&=(-\ell(\ell+1)+1)(\radsinf{\xhb})^{\mathrm{E}}_\ell-\sqrt{2}\sum_{m=-\ell}^\ell \YlmE2(\radsinf{\rh})_{\ell,m}(-1+(-1)^\ell +\ell(\ell+1))\\
        &\,\,\,\,\,\,\,\,\,\,-6(\ell+2)(\ell-1)\albdata_\ell^{\mathrm{E}}\cdot\frac{\log r_0}{r_0}+
        \O(r_0^{-1}),\\
       \radlp{\xh}^{\mathrm{H}}:= \lim_{v\to\infty} r^2\xh_{\ell}^H&= (-1)^{\ell}(\radsinf{\xhb})^{\mathrm{H}}_{\ell}+\O(r_0^{-1}).
    \end{nalign}
\end{prop}
In particular, we infer the following
\begin{cor}\label{cor:xh:antipodal} For any $\ell\geq 2$, we have
\begin{equation}
  \lim_{v\to\infty}\lim_{u\to-\infty}r^2 \xhb^{\mathrm{H}}_{\ell}= (-1)^{\ell}\lim_{u\to-\infty}\lim_{v\to\infty}r^2 \xh^{\mathrm{H}}_{\ell}.
\end{equation}
    \end{cor}
This is the antipodal matching condition, cf.~comments below Thm.~\ref{thm:intro:main} in the introduction.
\begin{proof}
By \eqref{eq:lin:xh4}, we have
\begin{equation}
    \frac{r^2\xh}{\Omega}=\frac{r^2\radc\xh}{\Omega}-\int_{v_1}^v r^2\al\dd v'.
\end{equation}
Using the asymptotic estimate \eqref{eq:alp:thmmain} for $\al$, we compute the limit of $r^2\xh_\ell$ at $\Scrip$:
\begin{nalign}\label{eq:xh:xhproof1}
\lim_{v\to\infty}r^2\xh_\ell&=\restr{\frac{r^2\xh_{\ell}}{\Omega}}{\Cin}-\int_{v_1}^\infty \sum_{n=0}^{\ell-2}\frac{r_0^n}{r^{n+3}}\left(\A_{\ell}r_0^2 S_{\ell,0,\ell-2,n,2}+\B_{\ell}r_0\log r_0 S_{\ell,1,\ell-2,n,2}\dd v\right)+\O(r_0^{-1})\\
&=\restr{\frac{r^2\xh_{\ell}}{\Omega}}{\Cin}-(-1)^{\ell}\frac{2(\ell-2)!}{(\ell+2)!}\A_{\ell}\sum_{n=0}^{\ell-2}\frac{(-1)^n}{n+2}\frac{(\ell+2+n)!}{n!(n+2)!(\ell-2-n)!}\\
&-(-1)^{\ell}\frac{6(\ell-2)!(\ell-1)!}{(\ell+2)!(\ell+1)!}\B_{\ell}\sum_{n=0}^{\ell-2}\frac{(-1)^n}{n+2}\frac{(\ell+2+n)!}{n!(n+1)!(\ell-2-n)!}+\O(r_0^{-1}).
\end{nalign}
In the second equality above, we inserted the expressions \eqref{eq:al:SLPJN}.

We now compute the sums in \eqref{eq:xh:xhproof1}  using the following two identities:
\begin{lemma}
Let $\ell\in\mathbb N_{\geq2}$. Then:
\begin{align}\label{eq:xh:mysterysum}
 (-1)^\ell   \sum_{n=0}^{\ell-2}\frac{(-1)^n}{n+2}\frac{(\ell+2+n)!}{n!(n+2)!(\ell-2-n)!}&=-1+(-1)^\ell +\ell(\ell+1),\\
 (-1)^{\ell}\sum_{n=0}^{\ell-2}\frac{(-1)^n}{n!(n+2)!}\frac{(\ell+2+n)!}{(\ell-2-n)!}&=\frac{1}{2}\frac{(\ell+2)!}{(\ell-2)!}.\label{eq:xh:easysum}
\end{align}
\end{lemma}
\begin{proof}
    The first equation is seen most easily from a relation to $\sig$ (namely the second of \eqref{eq:xh:beb=divxh} and \eqref{eq:xh:antipodal:rhosig}, cf.~Rem.~\ref{rem:final:mystery}), so we'll here only give the proof of the second equation. We employ the notation that for any finite series $P(x)=\sum_{n=-N}^N a_n x^n$, $[P(x)]_m=a_m$ denotes the $m$-th coefficient of $P(x)$. Then
    \begin{align*}
   \frac{1}{(\ell+1)(\ell+2)} \sum_{n=0}^{\ell-2}\frac{(-1)^n}{n!(n+2)!}\frac{(\ell+2+n)!}{(\ell-2-n)!} =&\sum_{n=0}^{\ell-2}(-1)^n \binom{\ell+2+n}{n}\binom{\ell}{n+2}\\
   =&\sum_{n=0}^{\ell-2}\left[x^{-n-2}(1-x)^{\ell+2+n}\right]_{-2}\binom{\ell}{n+2}\\
   =&\sum_{n=0}^{\ell-2}\left[(1-x)^{\ell}\binom{\ell}{n+2}\left(\frac{1-x}{x}\right)^{n+2}\right]_{-2}\\
   =&\left[(1-x)^{\ell}\left(\left(1+\frac{1-x}{x}\right)^\ell-1-\ell\frac{1-x}{x}\right)\right]_{-2}\\
   =&\left[(1-x)^{\ell}x^{-\ell}\right]_{-2}=\binom{\ell}{\ell-2}.
    \end{align*}
\end{proof}

Next, we compute the initial data term $\restr{\Omega^{-1}r^2\xh_{\ell}}{\Cin}$ in \eqref{eq:xh:xhproof1}:
Recall that, by \eqref{eq:physd:xh}:
\begin{nalign}
   \restr{ \frac{r^2\xh_{\ell}}{\Omega}}{\Cin}&=((-2\Ds2\D2-1)\radsinf{\xhb})_{\ell}-6\Ds2\D2 \albdata_\ell r_0^{-1}\log r_0+\O(r_0^{-1})\\
    &=(1-\ell(\ell+1))\rad{\xhb}{\ell,\Sinfty}-3(\ell-1)(\ell+2)\albdata_{\ell}r_0^{-1}\log r_0+\O(r_0^{-1}).
\end{nalign}
Now, on the one hand, we have, by \eqref{eq:physd:aldata=rhosigma at infinity}:
\begin{nalign}
    \A=\Ds2\Ds1(\radsinf\rh,-\radsinf{\sig})=\Ds2\Ds1(\radsinf\rh,-\curl\div \radsinf\xhb)\\
    \implies \A_{\ell}=\Ds2\Ds1((\radsinf{\rh})_\ell,0)-\frac{1}{2}\frac{(\ell+2)!}{(\ell-2)!}(\radsinf{\xhb})^{\mathrm{H}}_\ell.
\end{nalign}
On the other hand, we have (cf.~\eqref{eq:alp:B})
\begin{equation}
    \B=\Ds2\Ds1\overline{\D1}\D2\albdata\implies \B_{\ell}=\frac{(\ell+2)!}{(\ell-2)!}(\albdata_\ell^{\mathrm{E}}-\albdata_{\ell}^\mathrm{H}).
\end{equation}

Putting all the above equalities together allows us to infer the expressions \eqref{eq:xh:prop:xhlimits}.

Finally, we can now integrate \eqref{eq:lin:xh4} from $\Scrip$ (inserting again \eqref{eq:alp:thmmain}) to compute the next terms in the expansion of $\xh$ towards $\Scrip$, this gives \eqref{eq:xh:prop:xhexpansion}.

The expression \eqref{eq:xh:prop:Duxhexpansion} follows in an analogous fashion, using that \eqref{eq:alp:thmmain} commutes with~$\Du$. 
\end{proof}

\subsection{Asymptotics for \texorpdfstring{$\overone{\hat{\underline{\chi}}}$}{chihatbar}}
\begin{prop}\label{prop:xh:xhb}
 For any $\ell\geq2$, the following asymptotic expressions for $\xhb$ and its $v$-derivatives are valid throughout all of $\DoD$:
   \begin{align}\label{eq:xh:prop:xhblim0}
      &\radlp{\xhb}:= \lim_{v\to\infty}r\xhb_{\ell}=-\int_{-\infty}^u \radlp{\alb}\dd u'=(-1)^{\ell+1}r_0^{-2}\left(\frac{6\Bb_{\ell}}{\ell(\ell+1)}-2M\Aa_{\ell}\right)+\O(r_0^{-2-\delta}+M r_0^{-3})\\\label{eq:xh:prop:xhblim1}
       &\lim_{v\to\infty}\frac{r^2\Dv}{\Omega^2}(r\Omega\xhb_\ell)=\rad{\xhb}{\ell,\Sinfty}+\Aa_{\ell}\sum_{i=1}^{\ell}\left(\sum_{n=i+2}^{\ell+2}\underline{S}_{\ell,0,\ell+2,n,-2}\right)(i-1)!+\O(r_0^{-1}\log r_0).
   \end{align}
Moreover, we have
    \begin{equation}\label{eq:xh:prop:xhblim2}
    \left(\frac{r^2\Dv}{\Omega^2}\right)^2(r\Omega\xhb_\ell)=(-1)^{\ell}\frac{(\ell+2)!}{(\ell-2)!}2M\Aa_{\ell}\cdot\log r+\O(r_0),
    \end{equation}
and thus, by the fundamental theorem of calculus:
\begin{equation}\label{eq:xh:prop:xhbexpansion1}
    r\Omega\xhb_{\ell}=\lim_{v\to\infty}r\Omega\xhb_{\ell}-\frac{1}{r}\lim_{v\to\infty}\frac{r^2\Dv}{\Omega^2}(r\Omega\xhb_\ell)+\frac{1}{2}\frac{\log r}{r^2}\lim_{v\to\infty}\frac{1}{\log r}\left(\frac{r^2\Dv}{\Omega^2}\right)^2(r\Omega\xhb_\ell)+\O(r^{-2}r_0).
\end{equation}
For certain applications, it is crucial to better resolve the structure near $\Scrim$, so we also record:
\begin{multline}\label{eq:xh:prop:xhbexpansion2}
     \frac{r^2\xhb_\ell}{\Omega} 
    =(\radsinf{\xhb})_\ell-\Aa_{\ell}\sum_{i=1}^{\ell}\left(\sum_{n=i+2}^{\ell+2}\underline{S}_{\ell,0,\ell+2,n,-2}\right)(i-1)!\left(\left(\frac{r_0}{r}\right)^i-1\right)\\-r\int_{-\infty}^u \radlp{\alb}\dd u' +
    \O(\frac{\log r_0}{r_0}) .
\end{multline}
\end{prop}
\begin{proof}
We will, throughout the proof, make use of the fact that
\begin{align}\label{eq:xh:proofxhb:sum}
    \sum_{n=0}^{\ell+2}\underline{S}_{\ell,0,\ell+2,n,-2}= \sum_{n=1}^{\ell}\frac{n!}{(n+2)!}\underline{S}_{\ell,0,\ell,n,-2}=0,&&\sum_{n=1}^\ell \frac{(n-1)!}{(n+2)!}\underline{S}_{\ell,1,\ell,n,-2}=-3(-1)^{\ell}(\ell-1)(\ell+2),
\end{align}
which is proved in the same way as~\eqref{eq:xh:easysum}. (The first identity above can also be seen by restricting~\eqref{eq:alb:thm} to $\Cin$ (where $\frac{r_0}{r}=1$)).

We start by using \eqref{eq:lin:xhb3} to write:
\begin{equation}
    \Dv^n(\frac{r^2\xhb_{\ell}}{\Omega})=\Dv^n\rad{\xhb}{\ell,\Sinfty} -\int_{-\infty}^u \Dv^n(r^2\alb_{\ell})\dd u'.
\end{equation}

Next, we compute, using \eqref{eq:alb:thm} and \eqref{eq:appB:lemA1:2} 
\begin{nalign}
    &-\Dv(\Omega^{-1}r^2\xhb_{\ell})=\int_{-\infty}^u\Dv(r^2\alb)\dd u'\\
    =&\int_{-\infty}^u -\frac{\Aa_{\ell}}{r_0^2}\sum_{n=2}^{\ell+2}(n-1) \underline{S}_{\ell,0,\ell+2,n,-2}\left(\frac{r_0}{r}\right)^n +(\radlp{\alb})+\O\left(\frac{\log r_0}{r^3}+\frac{1}{r_0 r^2}\right)\dd u'\\
    =&-\Aa_{\ell}\sum_{n=2}^{\ell+2}\underline{S}_{\ell,0,\ell+2,n,-2}\frac{(n-2)!(n-1)}{(n-1)!}\frac1r\left(1+\sum_{i=1}^{n-2}\left(\frac{r_0}{r}\right)^{i}i!\right)+\int_{-\infty}^u \radlp{\alb}\dd u' +\O(r^{-2}\log r)\\
    =&-\frac{\Aa_{\ell}}{r}\sum_{i=1}^{\ell}\left(\sum_{n=i+2}^{\ell+2}\underline{S}_{\ell,0,\ell+2,n,-2}\right)i!\left(\frac{r_0}{r}\right)^{i} +\int_{-\infty}^u \radlp{\alb}\dd u'+\O(r^{-2}\log r),
\end{nalign}
where, in the last step, we have used that $\sum_{n=2}^{\ell+2}\underline{S}_{\ell,0,\ell+2,n,-2}=0$ (cf.~\eqref{eq:xh:proofxhb:sum}).
We integrate this in $v$ from $\Cin$ to obtain an expression for $\xhb$:
\begin{nalign}
    &\frac{r^2\xhb_\ell}{\Omega}=    \restr{\frac{r^2{\xhb}_\ell}{{\Omega}}}{\Cin}+\int \Dv(\frac{r^2\xhb_\ell}{\Omega})\dd v' \\
    =&\rad{\xhb}{\ell,\Sinfty}-\Aa_{\ell}\sum_{i=1}^{\ell}\left(\sum_{n=i+2}^{\ell+2}\underline{S}_{\ell,0,\ell+2,n,-2}\right)(i-1)!\left(\left(\frac{r_0}{r}\right)^i-1\right)-r\int_{-\infty}^u \radlp{\alb}\dd u' +
    \O(\frac{\log r_0}{r_0}) .
\end{nalign}
This proves \eqref{eq:xh:prop:xhbexpansion2} and, in particular, also implies \eqref{eq:xh:prop:xhblim0}.
We can hence compute
\begin{nalign}
    \lim_{v\to\infty} \frac{r^2}{\Omega^2}\Dv(r\Omega\xhb_\ell)&=\lim_{v\to\infty}\left(r\Dv(\frac{r^2\xhb_\ell}{\Omega})-(1-\frac{4M}{r})\frac{r^2\xhb_\ell}{\Omega}\right)\\
    =&\rad{\xhb}{\ell,\Sinfty}+\Aa_{\ell}\sum_{i=1}^{\ell}\left(\sum_{n=i+2}^{\ell+2}\underline{S}_{\ell,0,\ell+2,n,-2}\right)(i-1)!+\O(r_0^{-1}\log r_0).
\end{nalign}

It is left to prove \eqref{eq:xh:prop:xhblim2}. For this, we will only compute (and prove the existence of) the limit $\lim \frac{r^2}{\log r}\Dv (\Omega^{-2}r^2\Dv)(r\Omega\xhb_\ell)$; it is then left to the reader to check that the remainder term is bounded by $r_0$.
We write:
\begin{nalign}\label{eq:xh:bigintegral}
\Dv(\frac{r^2}{\Omega^2}\Dv(r\Omega\xhb))&=r\Dv^2(\frac{r^2\xhb}{\Omega})+\frac{2M}{r}\Dv(\frac{r^2\xhb}{\Omega})-\frac{4M\Omega^2}{r^2}\frac{r^2\xhb}{\Omega}\\
=& -r\left(\int		\Dv^2(r^2\alb)	\dd u\right)-\frac{2M}{r}\left(\int	\Dv(r^2\alb)	\dd u\right)+\frac{4M}{r^2}\left(\int r^2\alb \dd u-\radsinf\xhb\right)\\
=& -r\left(\int		\frac{\Omega^2}{r^3}(\rDv)^2(r\Omega^2\alb)-\frac{2M}{r^3}(\rDv)(r\Omega^2\al)+\frac{2M}{r^2\Omega^2}r\Omega^2\alb	\dd u\right)\\
&-\frac{2M}{r}\left(\int\frac{1}{r}(\rDv)(r\Omega^2\alb)+\left(1-\frac{2M}{\Omega^2 r}\right)r\Omega^2\alb	\dd u\right)\\
&+\frac{4M}{r^2}\left(\int\frac{r}{\Omega^2} r\Omega^2\alb \dd u-\radsinf\xhb\right)\dd u'.
\end{nalign}
A simple integration by parts gives that (we use \eqref{eq:alb:thm} in the second line below and only integrate by parts in the third line):
\begin{align*}
   & \lim_{v\to\infty} \frac{r^2}{\log r} \left(-r\int_{-\infty}^u \frac{2M}{r^2\Omega^2}r\Omega^2\alb_\ell	 \dd u'-\frac{2M}{r}\int_{-\infty}^u r\Omega^2\alb_\ell \dd u' +\frac{4M}{r^2}\int_{-\infty}^u r^2\alb_\ell \dd u' \right)\\
    =& \lim_{v\to\infty} \frac{r^2}{\log r}\left(-r\int_{-\infty}^u \frac{2M}{r^2}\lim_{v\to\infty}r\alb_{\ell}	 \dd u'-\frac{2M}{r}\int_{-\infty}^u \lim_{v\to\infty}r\alb_\ell \dd u' +\frac{4M}{r^2}\int_{-\infty}^u r\lim_{v\to\infty}r\alb_\ell \dd u' \right)\\
    =&\lim_{v\to\infty} \frac{r^2}{\log r}\left( -\frac{2M}{r}\int_{-\infty}^u \lim_{v\to\infty}r\alb_\ell\dd u'-\frac{2M}{r}\int_{-\infty}^u \lim_{v\to\infty}r\alb_\ell\dd u'+\frac{4M}{r}\int_{-\infty}^u \lim_{v\to\infty}r\alb_\ell\dd u'
    \right)
    =0,
\end{align*}
even though each of the individual terms inside the brackets only decays like $1/r$ towards $\Scrip$.

Since the terms with only one $\Dv$-derivative in \eqref{eq:xh:bigintegral} only contribute at order $r^{-2}$, we can thus infer that
\begin{equation}\label{eq:xh:proofxhb:limitrelation0}
    \lim_{v\to\infty}\frac{r^2}{\log r}\Dv\left(\frac{r^2}{\Omega^2}\Dv\right)(r\Omega\xhb_\ell)=-\lim_{v\to\infty}\frac{r^3}{\log r}\int_{-\infty}^u\frac{\Omega^2}{r^3}(\rDv)^2(r\Omega^2\alb)\dd u'.
\end{equation}
The limit on the RHS has 3 potential contributions:
\begin{nalign}\label{eq:xh:proofxhb:limitrelation1}
    &\lim_{v\to\infty}\frac{r^3}{\log r}\int_{-\infty}^u\frac{\Omega^2}{r^3}(\rDv)^2(r\Omega^2\alb)\dd u'\\
    = & \lim_{v\to\infty}\frac{r^3}{\log r}\int_{-\infty}^u \frac{\Omega^2}{r^3} \sum_{n=1}^{\ell}(\Aa_\ell \underline{S}_{\ell,0,\ell,n,-2}+\albdata_\ell \underline{S}_{\ell,1,\ell,n,-2} \frac{\log r_0}{r_0})\frac{r_0^n}{r^n}\dd u'\\
    +&  \lim_{v\to\infty}\frac{r^3}{\log r}\int_{-\infty}^u r^{-3}r_0^{-3}\sum_{n=2}^{\ell+3}n(n-1)(\Bb_{\ell}\underline{R}'_{\ell,n}+2M\Aa_{\ell}\underline{Q}'_{\ell,n})\frac{r_0^n}{r^{n-2}}\dd u'\\
    +&  \lim_{v\to\infty}\frac{r^3}{\log r}\int_{-\infty}^u \frac{1}{r^3}(-1)^{\ell}M\Aa_\ell \frac{2\ell(\ell+1)(\ell+2)!}{(\ell-2)!}\frac{\log r-\log r_0}{r}.
\end{nalign}
Let's compute each line on the RHS of~\eqref{eq:xh:proofxhb:limitrelation1} separately. For the first one, we have:
\begin{nalign}\label{eq:xh:popopo}
    & \lim_{v\to\infty}\frac{r^3}{\log r}\int_{-\infty}^u \frac{\Omega^2}{r^3} \sum_{n=1}^{\ell+2}(\Aa_\ell \underline{S}_{\ell,0,\ell,n,-2}+\albdata_{\ell} \underline{S}_{\ell,1,\ell,n,-2} \frac{\log r_0}{r_0})\frac{r_0^n}{r^n}\dd u'\\
    =&\lim_{v\to\infty}\frac{r^3}{\log r}\left(\sum_{n=1}^{\ell}\frac{n!}{(n+2)!}\Aa_\ell \underline{S}_{\ell,0,\ell,n,-2}\frac{1}{r^2}+\frac{(n-1)!\cdot 2}{(n+2)!}\albdata_{\ell}\underline{S}_{\ell,1,\ell,n,-2}\frac{\log r}{r^3}\right)\\
    =&-6(-1)^{\ell}(\ell+2)(\ell-1)\albdata_{\ell},
\end{nalign}
where we have used the integral formulae \eqref{eq:appB:lemA1:1} and \eqref{eq:appB:lemA1:3} in the first equality, and then also \eqref{eq:xh:proofxhb:sum}.

Similarly, we have for the third line of the RHS of \eqref{eq:xh:proofxhb:limitrelation1} that
\begin{nalign}
    \lim_{v\to\infty}\frac{r^3}{\log r}\int_{-\infty}^u \frac{1}{r^3}(-1)^{\ell}M\Aa_\ell \frac{2\ell(\ell+1)(\ell+2)!}{(\ell-2)!}\frac{\log r-\log r_0}{r}=0
\end{nalign}
as a consequence of $\int_{-\infty}^u \frac{\log r-\log r_0}{r^4}\dd u'=\O(r^{-3})$, which follows from \eqref{eq:appB:lemA1:3}. 

Finally, for the second line of the RHS of \eqref{eq:xh:proofxhb:limitrelation1}, we have
\begin{nalign}\label{eq:xh:p[p[p[p}
    \lim_{v\to\infty}\frac{r^3}{\log r}\int_{-\infty}^u r^{-3}r_0^{-1}2(\Bb_{\ell}\underline{R}'_{\ell,2}+2M\Aa_{\ell}\underline{Q}'_{\ell,2})\dd u'=2(\Bb_{\ell}\underline{R}'_{\ell,2}+2M\Aa_{\ell}\underline{Q}'_{\ell,2}).     
\end{nalign}
The result \eqref{eq:xh:prop:xhblim2} then follows from \eqref{eq:xh:p[p[p[p}, \eqref{eq:xh:popopo} and \eqref{eq:xh:proofxhb:limitrelation0} by plugging in the expressions \eqref{eq:alb:thm:Q'R'} for $\underline{Q}_{\ell,2}'$ and $\underline{R}'_{\ell,2}$. In particular, since $R'_{\ell,2}=(-1)^{\ell}3(\ell-1)(\ell+2)$, there is no dependence on $\albdata_\ell$ in~\eqref{eq:xh:prop:xhblim2}.
\end{proof}
\begin{rem}\label{rem:xh:nologsinMink1}
    Here, we have explicitly computed that the Minkowskian contribution to the $\log$-term in \eqref{eq:xh:prop:xhblim2} vanishes. One can also show this indirectly without having to do computations, see also Remark~\ref{rem:xh:nologsinMink2}.
\end{rem}
\section{Remaining asymptotics, Bondi normalisation and the laws of gravitational radiation}

\subsection{Asymptotic behaviour for the curvature components \texorpdfstring{$\overone{\beta}$, $\underline{\overone{\beta}}$, $\overone{\rho}$ and $\overone{\sigma}$}{beta, betabar, rho and sigma}}
\begin{prop}\label{prop:xh:curv}
    For any $\ell\geq2$, the following asymptotic expressions are valid throughout~$\DoD$:
    \begin{align}\label{eq:xh:bebexp}
        r^2\beb_\ell&=\radlp{\beb} +\O(1/r,)\\\label{eq:xh:sigexp}
        (r^3\sig_{\ell},r^3\rh_{\ell})&=(\radlp{\sig}, \radlp{\rh})+\O(\frac{r_0+M\log r}{r}), \\\label{eq:xh:beexp}
        r^4\be_\ell&=\frac{(-1)^{\ell}2M}{(\ell-1)(\ell+2)}\div \A_{\ell}\log r +\O(r_0),
    \end{align}
    where the limits above are given by
    \begin{align}\label{eq:xh:beblimit}
   \radlp{\beb}:&=    \lim_{v\to\infty}r^2\beb_{\ell}=-\int_{-\infty}^u\div\radlp{\alb}\dd u',\\\label{eq:xh:rhosiglimit}
        (\radlp\rh,\radlp\sig):&=\lim_{v\to\infty}(r^3\sig_{\ell},r^3\rh_{\ell})=(-1)^{\ell}(\radsinf{\sig},\radsinf\rh)_\ell-\int_{-\infty}^u(\div\radlp\beb,\curl\radlp\beb)\dd u'.
    \end{align}
    In particular, we have the relations
    \begin{align}\label{eq:xh:beb=divxh}
        \radlp\beb=\div\radlp\xhb,&&\radlp\sig=\curl\div\radlp\xh
    \end{align}
\end{prop}
\begin{proof}
   Equation \eqref{eq:xh:bebexp} follows from \eqref{eq:lin:alb4} and the asymptotic results \eqref{eq:alb:thm}, \eqref{eq:xh:prop:xhbexpansion1} for $\alb$, $\xhb$, respectively.
   The expression for the limit then follows from \eqref{eq:lin:beb3}.

   Similarly, equation~\eqref{eq:xh:sigexp} follows from applying Theorem~\ref{thm:Psi} to $\sig$ via \eqref{eq:lin:Ps-Psb=sig} in the case of $\sig$, whereas the result for $\rh$ follows from considering $\Ps+\Psb$ together with the asymptotics for $\xh$ and $\xhb$ from Propositions~\ref{prop:xh:xh} and \ref{prop:xh:xhb}.
   The expressions for the limits then follow from \eqref{eq:lin:sig3}, \eqref{eq:lin:rh3} together with the fact that $\lim_{v\to\infty}\trxb_{\ell}=0$, which we will independently show in Prop.~\ref{prop:final:bondiprop}, and lastly using that
   \begin{equation}\label{eq:xh:antipodal:rhosig}
       \lim_{u\to-\infty}\lim_{v\to\infty}(r^3\rh_\ell,r^3\sig_\ell)=(-1)^{\ell}\lim_{u\to-\infty}\lim_{v\to\infty}(r^3\rh_\ell,r^3\sig_\ell)=(-1)^{\ell}(\radsinf{\rh}, \radsinf{\sig})_\ell,
   \end{equation}
   which follows from \eqref{eq:Psi:antipodal} and \eqref{eq:physd:limitsof:rho,sig,xhb}.

   Finally, the expression \eqref{eq:xh:beexp} for $\be$ follows from integrating \eqref{eq:lin:be4}.
\end{proof}
\begin{rem}\label{rem:final:mystery}
    Notice that \eqref{eq:xh:antipodal:rhosig}, combined with \eqref{eq:xh:beb=divxh}, implies the identity~\eqref{eq:xh:mysterysum}.
\end{rem}
\subsection{Asymptotic behaviour for the connection coefficients \texorpdfstring{$\overone{\eta}$, $\overone{\underline{\eta}}$, $\overone{\omega}$ and $\overone{\Omega}$}{eta, etabar, omega and Omega}}
We can deduce immediately from the relation \eqref{eq:lin:xhb4} and the asymptotic expressions for $\xh$ and $\xhb$ from Propositions~\ref{prop:xh:xh} and \ref{prop:xh:xhb} the following:
\begin{cor}\label{cor:xh:etabar}
    For any $\ell\geq2$, we have that\footnote{Recall that e.g.~$(\Ds2\etb)_{\ell}^{\mathrm{E}}= \sum_{m=-\ell}^{\ell}(\Ds2\etb)_{\ell,m}^{\mathrm{E}}\cdot \YlmE2 =\sum_{m=-\ell}^{\ell}\sqrt{\tfrac12(\ell-1)(\ell+2)}\etb^{\mathrm{E}}_{\ell,m}\cdot \YlmE2$.}
    \begin{equation}
       r^2 \Ds2\etb_\ell =\Ds2\radlp{\etb}+\O(r_0/r), 
    \end{equation}
    where 
    \begin{multline}\label{eq:xh:limitetabar}
     -2 \Ds2\radlp\etb:=  \lim_{v\to\infty}-2r^2\Ds2\etb_{\ell}=\lim_{v\to\infty}(r^2\Dv(r\Omega\xhb_\ell)-r^2\xh_\ell)\\
        =\rad{\xhb}{\ell,\Sinfty}+\Aa_{\ell}\sum_{i=1}^{\ell}\left(\sum_{n=i+2}^{\ell+2}\underline{S}_{\ell,0,\ell+2,n,-2}\right)(i-1)!-(-\ell(\ell+1)+1)(\radsinf{\xhb})^{\mathrm{E}}_\ell\\
        -\sqrt{2}\sum_{m=-\ell}^\ell \YlmE2(\radsinf{\rh})_{\ell,m}(-1+(-1)^\ell +\ell(\ell+1))-(-1)^{\ell}\rad{\xhb}{\ell,\Sinfty}^{\mathrm{H}}+\O(r_0^{-1}\log r_0).
    \end{multline}
\end{cor}
Similarly, by combining \eqref{eq:lin:xhb4} with \eqref{eq:lin:etb4}, we obtain
\begin{cor}\label{cor:xh:om}
    For any $\ell\geq2$, we have
    \begin{equation}
        4\Ds2\sl\om_\ell=(-1)^{\ell+1}2M(\overline{\A}_\ell+\A_{\ell})\frac{\log r}{r^3}+\O(\frac{r_0}{r^3}).
    \end{equation}
\end{cor}
\begin{proof}
    In view of the asymptotics for $\etb$, $\be$, $\xhb$ and $\xh$ found so far, we deduce from \eqref{eq:lin:xhb4} and \eqref{eq:lin:etb4} that
    \begin{nalign}
        \lim_{v\to\infty}\frac{2r^3}{\log r}\Ds2\sl\om_\ell&=\lim_{v\to\infty}\frac{r^2}{\log r}\Dv\Ds2(r^2\etb_\ell)-\frac{r^4}{\log r}\Ds2\be_{\ell}\\
        &=\lim_{v\to\infty}-\frac{1}{2}\frac{r^2}{\log r}\Dv\frac{r^2\Dv}{\Omega^2}(r\Omega\xhb_\ell)-\frac{r^4}{\log r}\Ds2\be_\ell.
    \end{nalign}
The result  then follows by inserting \eqref{eq:xh:beexp} and \eqref{eq:xh:prop:xhblim2}.
\end{proof}
\begin{rem}\label{rem:xh:nologsinMink2}
    Recall Remark~\ref{rem:xh:nologsinMink1}: If there were a $\albdata$-dependent Minkowskian contribution to the limit, then this would appear in the expression for $\lim\frac{r^3}{\log r}\Ds2\sl\om_{\ell}$ as well. One could then deduce a contradiction from the fact that $\Ds2\sl\om_{\ell}$ is only supported on electric angular modes.
\end{rem}
Next, we use \eqref{eq:lin:xh3} to deduce
\begin{cor}\label{cor:xh:eta}
For any $\ell\geq2$, we have that
\begin{equation}
    r \Ds2\et_\ell =\Ds2\radlp{\et}+\O(1/r),
\end{equation}
where
\begin{equation}\label{eq:xh:limiteta}
  -2\Ds2\radlp{\et}:=  \lim_{v\to\infty}  -2 r\Ds2\et_\ell=\Du\radlp{\xh}+\radlp\xhb=-6(\ell+2)(\ell-1)\albdata_\ell^{\mathrm{E}}\cdot \frac{\log r_0}{r_0^2}+\O(r_0^{-2}).
\end{equation}
\end{cor}
\begin{rem}
    In fact, the limit \eqref{eq:xh:limiteta} is supported entirely on electric angular modes. This can be seen, for instance, from \eqref{eq:lin:curleta}, the fact that magnetic angular modes are not in the kernel of $\curl$, and $\lim_{v\to\infty} r^2\sig=0$.
\end{rem}

Finally, we have:
\begin{cor}\label{cor:xh:Omega}
    For any $\ell\geq2$ and any $m\in\{-\ell,\dots,\ell\}$, we have $\lim_{v\to\infty}{\sl\Om}_\ell=\radlp\et$ and thus:
    \begin{equation}
        \lim_{v\to\infty}\Omega_{\ell,m}=-3\sqrt{2\frac{(\ell-1)(\ell+2)}{\ell(\ell+1)}}\albdata_{\ell,m}^{\mathrm{E}}\cdot \frac{\log r_0}{r_0^2}+\O(r_0^{-2}).
    \end{equation}
\end{cor}
\begin{proof}
    This follows from \eqref{eq:lin:OmmA} and Corollaries \ref{cor:xh:etabar} and \ref{cor:xh:eta}.
\end{proof}

\subsection{Future Bondi normalisation of the solution}
Corollary \ref{cor:xh:Omega} shows that the solution has a nonvanishing limit of $\Om$ along $\Scrip$, that is, the Bondi normalisation of the solution towards the past (cf.~Def.~\ref{def:setup:Bondi}) does not imply Bondi normalisation towards the future, which we define as follows:
\begin{defi}\label{defi:xh:Bondinorm}
    A solution to~\fullsystem~is called \textit{Bondi normalised towards the future/$\Scrip$} if
    \begin{equation}\label{eq:xh:Bondinorm}
       \lim_{v\to\infty}\Om =\lim_{v\to\infty}r\trxb=0=\lim_{v\to\infty} r^2\K=\lim_{v\to\infty}\trg=\lim_{v\to\infty}\gsh.
    \end{equation}
    (Notice that the vanishing of the limits of $r\trxb$ and $\trg$ is automatically enforced by the vanishing of the other three limits.)
\end{defi}
\begin{rem}
It may at first not be entirely clear that this definition is just the future version of Def.~\ref{def:setup:Bondi}; this is because in the latter, we demand the vanishing of $\lim_{u\to-\infty}r^{-1}\b$ along $\Scrim$ and the vanishing of $\lim_{u\to-\infty}\gsh$ only for $v=1$. In essence, however, this is equivalent to demanding that $\lim_{u\to-\infty}\gsh$ vanish along all of $\Scrim$.
\end{rem}
Corollary~\ref{cor:xh:Omega} shows that our solution is not Bondi normalised towards the future, although we will see below that the last three quantities of \eqref{eq:xh:Bondinorm} do, in fact, vanish (as a consequence of our solution being Bondi normalised towards the past).

We can remedy the failure of the solution to be Bondi-normalised towards the future by adding an ingoing gauge solution that kills of $\lim_{v\to\infty}\Om$ and $\lim_{v\to\infty}r\trxb$ while simultaneously decaying sufficiently fast towards $\Scrim$ to leave the Bondi normalisation near $\Scrim$ intact. 

\begin{defi}
    Define the ingoing gauge function $\underline{f}^{\Scrip}(u, \theta,\varphi)$ via:
    \begin{equation}
        \underline{f}^{\Scrip}_{\ell,m}(u)=\int_{-\infty}^u \lim_{v\to\infty}\Om_{\ell,m}\dd u',
    \end{equation}
    and define $\mathfrak{S}_{\underline{f}^{\Scrip}}$ to be the ingoing gauge solution induced by $\underline{f}^{\Scrip}$ as in Def.~\ref{prop:gauge:in}.

    Finally, define $\mathfrak{S}_{\mathrm{Bondi}}:=\mathfrak{S}-\mathfrak{S}_{\underline{f}^{\Scrip}}$, where $\mathfrak{S}$ is the solution of Proposition~\ref{prop:physd}.
\end{defi}
\textit{\textbf{All statements made henceforth will concern the solution $\mathfrak{S}_{\mathrm{Bondi}}$.}}

Notice that the subtraction of $\mathfrak{S}_{\underline{f}^{\Scrip}}$ modifies equations \eqref{eq:physd:trxb}--\eqref{eq:physd:om} at next-to-leading order.

Concerning this new difference solution, we then have the following statement:

\begin{prop}\label{prop:final:bondiprop}
The results of Propositions~\ref{prop:xh:xh}--\ref{prop:xh:curv} and Corollaries~\ref{cor:xh:antipodal}--\ref{cor:xh:om} still apply to $\mathfrak{S}_{\mathrm{Bondi}}$; i.e.~all changes induced by subtraction of $\mathfrak{S}_{\underline{f}^{\Scrip}}$ are subleading.
In addition, $\mathfrak{S}_{\mathrm{Bondi}}$ satisfies for any $\ell\geq2$
 \begin{equation}
       \lim_{v\to\infty}\Om_\ell=\lim_{v\to\infty}r\et_\ell =\lim_{v\to\infty}r\trxb_\ell=0=\lim_{v\to\infty} r^2\K_\ell=\lim_{v\to\infty}\trg_\ell=\lim_{v\to\infty}\gsh_\ell.
    \end{equation}
Finally, $\mathfrak{S}_{\mathrm{Bondi}}$ obeys, for any $\ell\geq2$:
    \begin{equation}\label{eq:xh:b}
        \lim_{v\to\infty} \b_\ell=0.
    \end{equation}
\end{prop}
\begin{proof}
The fact that Propositions~\ref{prop:xh:xh}--\ref{prop:xh:curv} and Corollaries~\ref{cor:xh:antipodal}--\ref{cor:xh:om} remain unchanged follows directly from Prop.~\ref{prop:gauge:in}.

    Next, the limits of $\Om_{\ell}$ and $r\et_{\ell}$ vanish by construction (we have chosen $\underline{f}^{\Scrip}$ to kill off $\Om$, after all.)

    To see that the limit of $r\trxb_\ell$ vanishes, we consider \eqref{eq:lin:divxhb} together with \eqref{eq:xh:beb=divxh}; the result then follows from the vanishing of $\lim_{v\to\infty}r\et_{\ell}$.

    Moving on, we note that \eqref{eq:lin:divxh}, together with \eqref{eq:xh:prop:xhexpansion}, \eqref{eq:xh:beexp} and \eqref{eq:xh:limitetabar} implies that $\trx_{\ell}\lesssim r^{-2}$. 
    Thus, integrating \eqref{eq:lin:trg3} gives that $\trg_{\ell}\lesssim r^{-1}$; in particular, $\lim_{v\to\infty}\trg_\ell=0$. 

    We can similarly infer that $\gsh_{\ell}\lesssim r^{-1}$ by integrating \eqref{eq:xh:prop:xhbexpansion2}.
    Hence, by \eqref{eq:lin:Kdef}, the limit $\lim_{v\to\infty}r^2\K_{\ell}$ also vanishes.

Lastly, \eqref{eq:xh:b} follows from integrating \eqref{eq:lin:b3}.

\end{proof}

\subsection{Computing a few higher-order limits}
For the sake of completeness, we here record the following higher-order limits:
\begin{prop}
    The Bondi normalised solution $\mathfrak{S}_{\mathrm{Bondi}}$ satisfies for any $\ell\geq2$:
    \begin{align}\label{eq:xh:hol:etb}
      \Ds2\radlp{\etb}:=  \lim_{v\to\infty} r^2\Ds2\etb_\ell &=\frac12\lim_{v\to\infty}(r^2\xh_\ell-r^2\Dv(r\Omega\xhb_{\ell})),\\\label{eq:xh:hol:et}
      \lim_{v\to\infty}r^2\et_{\ell}&=-\radlp\etb,\\\label{eq:xh:hol:trx}
      \lim_{v\to\infty}r^2\sl\trx_{\ell}&=2(\radlp{\etb}^{\mathrm{E}}+\div\radlp{\xh}^{\mathrm{E}}),\\\label{eq:xh:hol:trxb}
      \lim_{v\to\infty}r^2\sl\trxb_{\ell}&=-2\lim_{v\to\infty}r^2\Dv(r\Omega\div\xhb^{\mathrm{E}}_\ell)-2\sl\radlp{\rh}+2\radlp\etb^{\mathrm{E}},\\\label{eq:xh:hol:K}
      \lim_{v\to\infty}r^3\sl\K_{\ell}&=(\div\radlp\xh)^{\mathrm{E}}+\lim_{v\to\infty}r^2\Dv(r\Omega\div\xhb^{\mathrm{E}}_{\ell}),\\\label{eq:xh:hol:gsh}
     \radlp\gsh:= \lim_{v\to\infty}r\gsh_{\ell}&=2\Aa_{\ell}\sum_{i=1}^{\ell}\sum_{n=i+2}^{\ell+2}\underline{S}_{\ell,0,\ell+2,n,-2}\frac{i!}{i+1}+\O(\frac{\log r_0}{r_0}),\\\label{eq:xh:hol:trg}
     \lim_{v\to\infty}(\lap+2)(r\trg_{\ell})&=2\div\div\radlp\gsh-4\lim_{v\to\infty}r^3\K_{\ell},\\\label{eq:xh:hol:b}
     \lim_{v\to\infty}\Ds2 r\b_{\ell}&=-\radlp\gsh -2\radlp\xh.
    \end{align}
    All of these limits are to leading order constant in $u$, with the third one and the last two ones being exactly constant in $u$. (Cf.~Conclusion 17.0.5 in \cite{CK93}.)
\end{prop}
\begin{proof}
Equation~\eqref{eq:xh:hol:etb} is as in \eqref{eq:xh:limitetabar}; notice that the difference coming from the subtraction of $\mathfrak{S}_{\underline{f}^{\Scrip}}$ is subleading (it's $\O(r_0^{-1}\log r_0)$).

Equation \eqref{eq:xh:hol:et} then follows from integrating \eqref{eq:lin:et4} from $\Scrip$. Alternatively, it follows from \eqref{eq:lin:OmmA} and the fact that $\om=\O(r^{-3}\log r)$ towards $\Scrip$ (cf.~Cor.~\ref{cor:xh:om}).

Equation \eqref{eq:xh:hol:trx} follows from \eqref{eq:lin:divxh}; similarly, \eqref{eq:xh:hol:trxb} follows from \eqref{eq:lin:divxhb}, and then also using \eqref{eq:xh:beb=divxh} as well as $\lim_{v\to\infty}r^2\Dv(r^2\beb_{\ell})^{\mathrm{E}}=-\sl\radlp\rh$. One similarly deduces that the RHS of \eqref{eq:xh:hol:trx} is constant in $u$.

Equation \eqref{eq:xh:hol:K} follows from \eqref{eq:lin:K} and the previous expressions.

Equation \eqref{eq:xh:hol:gsh} follows from integrating \eqref{eq:xh:prop:xhbexpansion2} using \eqref{eq:appB:lemA1:3} and taking the limit as $r\to\infty$. Equation \eqref{eq:xh:hol:trg} then follows from \eqref{eq:lin:Kdef}.

Equation \eqref{eq:xh:hol:b} follows from multiplying \eqref{eq:lin:gsh4} with $r^2$ and using that $\lim_{v\to\infty}r^2\Dv\gsh_{\ell}=-\lim_{v\to\infty}r\gsh_{\ell}$.

The fact that the last two limits are conserved can either be seen directly, or by using \eqref{eq:lin:trg3}, \eqref{eq:lin:b3}, respectively.
    
\end{proof}

\newpage
\appendix
\part*{Appendices}

\section{Spin-weighted functions, the Christodoulou--Klainerman formalism and the Newman--Penrose formalism}\label{app:SpinSphereNPCK}
In this section, we relate the formalism of the present paper to the Newman--Penrose formalism by defining spin-weighted quantities, the spin-weighted spherical harmonics and by explicitly writing down a dictionary between the notation used by Newman and Penrose and that of the present paper (which goes back to Christodoulou and Klainerman). We begin with a few calculations.
\subsection{Miscellaneous calculations on \texorpdfstring{$\Stwo$}{the sphere}}\label{app:misc}
We work with the metric $\mathring{\slashed{g}}=d\theta^2+\sin^2d\varphi^2$ on $\Stwo$. 
We recall from \eqref{eq:SS:stfsrep} that any symmetric tracefree two-tensor $\alpha$ can be written as $\Ds2\Ds1(f,g)$ for two smooth functions $f$ and $g$ on $\Stwo$ that are uniquely specified up to $\ell\leq 1$-modes. 
Let us spell out what this operator looks like in components: The only non-vanishing Christoffel symbols are
\begin{align*}
\Gamma_{\varphi\varphi}^\theta=-\cos\theta\sin\theta,&&\Gamma_{\varphi\theta}^\varphi=\cot\theta=\Gamma_{\theta\varphi}^\varphi.
\end{align*}
Then, working in in the orthonormal frame $e_1=\partial_\theta$, $e_2=\tfrac{1}{\sin\theta} \partial_\varphi$, the Christoffel symbols ($\sl_{e_i}e_j=\Gamma_{ij}^ke_k$) become
\begin{equation}\label{eq:app:Stwo:Christoffel}
\Gamma_{21}^2=\cot\theta,\quad \Gamma_{22}^1=-\cot\theta,\quad \Gamma_{12}^2=0=\text{all others}.
\end{equation}
We can thus compute (using that $\sl_1\sl_2 f=\sl_2\sl_1 f$)
\begin{align*}
\Ds2\Ds1(f,g)_{11}&=\sl_1\sl_1 f-\frac12\lap f-\sl_1\sl_2 g=-\Ds2\Ds1(f,g)_{22},\\
\Ds2\Ds1(f,g)_{12}&=\sl_1\sl_2 f-\frac{1}{2}(\sl_2\sl_2 g-\sl_1\sl_1 g),
\end{align*}
and, using \eqref{eq:app:Stwo:Christoffel}, 
\begin{align*}
\sl_1\sl_1 f&=\partial_\theta^2 f,\\
\sl_1\sl_2 f&=\frac{1}{\sin\theta}\left(\partial_\theta\partial_\varphi f-\cot\theta \partial_{\varphi}f\right)=\sl_2\sl_1 f,\\
\sl_2\sl_2 f&=\frac{1}{\sin^2\theta}\left(\partial_{\varphi}^2f+\cos\theta\sin\theta\partial_\theta f\right).
\end{align*}
In summary, we have
\begin{align*}
\Ds2\Ds1(f,g)_{11}&=\partial_\theta^2 f-\frac12\lap f-\frac{1}{\sin\theta}(\partial_\theta\partial_\varphi g-\cot\theta\partial_\varphi g)=-(\Ds2\Ds1(f,g))_{22,}\\
\Ds2\Ds1(f,g)_{12}&=\frac{\partial_\theta\partial_\varphi f}{\sin \theta}-\frac{\cot\theta}{\sin\theta}\partial_{\varphi}f-\frac1{2\sin^2\theta}(\partial_{\varphi}^2g+\cos\theta\sin\theta\partial_\theta g)+\frac{1}{2}\partial_\theta^2 g.
\end{align*}

We can similarly compute that the components of the Laplacian $\lap$ acting on symmetric traceless two-tensors are given by
\begin{nalign}\label{eq:app:laplacian}
(\lap \alpha)_{11}&=\lap(\alpha_{11})-4\cot^2\theta \alpha_{11}-4\frac{\cot\theta}{\sin\theta}\alpha_{12},\\
(\lap \alpha)_{12}&=\lap(\alpha_{12})-4\cot^2\theta \alpha_{12}+4\frac{\cot\theta}{\sin\theta}\alpha_{11},
\end{nalign}
where we used
\begin{align*}
    (\sl_1\sl_1\alpha)_{11}&=\partial_\theta^2\alpha_{11},\\
    (\sl_2\sl_2\alpha)_{11}&=\frac{1}{\sin^2\theta}\partial_\varphi^2\alpha_{11}-4\frac{\cot\theta}{\sin\theta} \partial_\varphi\alpha_{11}+\cot\theta\partial_\theta\alpha_{11}-4\cot^2\theta\alpha_{11},
    \\
    (\sl_1\sl_1\alpha)_{12}&=\partial_\theta^2\alpha_{12},\\
    (\sl_2\sl_2\alpha)_{12}&=\frac{1}{\sin^2\theta}\partial_{\varphi}^2\alpha_{12}+4\frac{\cot\theta}{\sin\theta}\partial_{\varphi}\alpha_{11}+\cot\theta\partial_\theta\alpha_{12}-4\cot^2\theta\alpha_{12}.
\end{align*}
\subsection{Spin-weighted functions, \texorpdfstring{the operators $\eth$, $\eth'$, and the spin-weighted spherical harmonics ${}_sY_{\ell,m}$}{the eth-operators and the spin-weighted spherical harmonics}}\label{app:spinweightedfunctions}
Rather than working in the simple geometric framework of 1-forms and stf two-tensor fields, large parts of the literature work within the closely related framework of spin $s$-weighted functions.
We here define the spaces of spin $s$-weighted functions following \cite{Sbierski22}. Introduce the following complex frame vector fields on $\Stwo$: 
\begin{align}
    m=\tfrac{1}{\sqrt{2}}(e_1+i e_2),&& \conj{ m} =\tfrac{1}{\sqrt{2}}(e_1-i e_2).
\end{align} 
We again stress that the notation $\conj{m}$ does not refer to magnetic conjugation (cf.~\S\ref{def:SS:magneticconjugate}), $\conj{m}$ is simply the complex conjugate of $m$.
\begin{defi}
We define
\begin{itemize}
\item 
the space of smooth spin 2-weighted functions on $\Stwo$ as the image of $\stfs$ under the map sending $\alpha\in\stfs$ to $\alpha(m,m)=\alpha_{11}+i\alpha_{12}$.
\item
 the space of spin (-2)-weighted functions on $\Stwo$ as the image of $\stfs$ under the map sending $\alpha$ to $\alpha(\conj{m},\conj{m})=\alpha_{11}-i\alpha_{12}$. 
\item 
the space of smooth spin 1-weighted functions on $\Stwo$ as the image of $\oneforms$ under the map sending $\beta\in\oneforms$ to $\sqrt{2}\beta(m)=\beta_1+i\beta_2$, 
\item the space of spin (-1)-weighted functions as the image of $\oneforms$ under the map sending $\beta$ to $\sqrt{2}\beta(\conj m)$.
\item the space of smooth spin 0-weighted functions on $\Stwo$ as the image of $\functions^2$ under the map $\mathfrak{c}$ sending $(f,g)\mapsto\mathfrak{c}(f,g):=(f+ig)$, or, equivalently, under the map  $\conj{\mathfrak{c}}$ sending $(f,g)\mapsto\conj{\mathfrak{c}}(f,g):=(f-ig)$.
\end{itemize}
\end{defi}
For intrinsic definitions and properties of these spaces, see also the original work \cite{NP66Note}, or \cite{DHR19} or \cite{Sbierski22}.
In large parts of the physics literature on spin $s$-weighted functions, readers will likely also encounter the  eth operators (eth referring to the Old English letter $\eth$):

We here provide an extrinsic, geometric definition of these operators in terms of the formalism of the present paper:
\begin{defi}\label{def:app:spin:eth}
Let $f\in\functions$, let $\beta\in\oneforms$, and let $\alpha\in\stfs$. We define the operators $\eth_{s\to s+1}$ (sending spin $s$-weighted functions to spin $s+1$-weighted functions) and $\eth'_{s\to s-1}$ (sending spin $s$-weighted functions to spin $s-1$-weighted functions, respectively, via:
\begin{nalign}\label{eq:app:spin:Ds1toedth}
    \sqrt{2}(\Ds1(f,0))(m)=:\eth_{0\to 1}f,&&\sqrt{2}(\Ds1(f,0))(\conj m)=:\eth'_{0\to-1}f,\\
    \sqrt{2}(\Ds2(\beta))(m,m)=:\eth_{1\to 2}(\beta(m)),&&\sqrt{2}(\Ds2(\beta))(\conj m,\conj m)=:\eth'_{-1\to- 2}(\beta(\conj m)),\\
   \sqrt{2} (\D2 \alpha)(m)=:-\eth'_{2\to 1}(\alpha(m,m)),&&\sqrt{2}(\D2 \alpha)(\conj m)=:-\eth_{-2\to -1}(\alpha(\conj m,\conj m)),
\end{nalign}
and, finally,
\begin{align}
 \mathfrak{c}(\D1\beta)=   \div \beta+i\curl \beta=:-\sqrt{2}\eth'_{1\to 0}(\beta(m)),&&\conj{\mathfrak{c}}({\D1\beta})= \div \beta-i\curl \beta:=-\sqrt{2}\eth_{-1\to 0}(\beta(\conj m)).
\end{align}
\end{defi}
\begin{rem}
Using that ${}^\ast e_1=-e_2$, ${}^\ast e_2=-e_1$ and thus ${}^\ast m=i\cdot m$, ${}^\ast \conj m=-i\conj m$, \eqref{eq:app:spin:Ds1toedth} implies that $(\Ds1(0,g))(m)=i\eth_{0\to-1} g$, $(\Ds1(0,g))(\conj m)=-i\eth'_{0\to 1}g$, so
\begin{align}
\sqrt{2}(\Ds1(f,g))(m)=\eth_{0\to-1}( \mathfrak{c}(f,g)),&&\sqrt{2}(\Ds1(f,g))(\bar m)=\eth'_{0\to 1}(\conj{\mathfrak{c}}({f,g})). 
\end{align}
\end{rem}
\begin{rem}
It is easy to check that, in $(\theta,\varphi)$ coordinates, Def.~\ref{def:app:spin:eth} implies 
\begin{align}\label{eq:app:spin:eth:intrinsicdef}
    \eth_{s\to s+1}=-\partial_\theta-\frac{i}{\sin\theta}\partial_\varphi +s\cot\theta,&&
     \eth'_{s\to s-1}=-\partial_\theta+\frac{i}{\sin\theta}\partial_\varphi -s\cot\theta.
\end{align}
This is how most readers will have encountered the operator before (with the ${}_{s\to s+1}$ subscript suppressed).
For instance, we compute
\begin{align*}
    \frac{1}{\sqrt{2}}\eth_{1\to2}\beta(m)&:=(\Ds2\beta)(m,m)=(\Ds2\beta)_{11}+i(\Ds2\beta)_{12}\\
    &=-\frac12 \left(\partial_\theta \beta_1-\frac{1}{\sin\theta}\partial_\varphi\beta_2\right)-\frac{i}{2}\left(\partial_\theta+\frac{1}{\sin\theta}\partial_\varphi \beta_1-\cot\theta \beta_2\right)\\
    &=-\frac12 \left(\partial_\theta(\beta_1+i\beta_2)+\frac{i}{\sin\theta}\partial_{\varphi}(\beta_1+i\beta_2)-\cot\theta(\beta_1+i\beta_2)\right).
\end{align*}
\end{rem}
\begin{rem}\label{rem:app:eth4}
Recall the angular operator $\Ds2\Ds1\overline{\D1}\D2$ appearing on the RHS of \eqref{eq:lin:TSI-}. We can rewrite this operator in terms of spin-weighted functions as follows:
\begin{align*}
   2 (\Ds2\Ds1\overline{\D1}\D2\alpha)(\conj m,\conj m)=&\sqrt{2}\eth'_{-1\to-2}((\Ds1\overline{\D1}\D2\alpha)(\conj m))\\
    =&\eth'_{-1\to-2}\eth'_{0\to-1}(\conj{\mathfrak{c}}(\overline{\D1}\D2\alpha)) \\
    =&\eth'_{-1\to-2}\eth'_{0\to-1}({\mathfrak{c}}({\D1}\D2\alpha))\\
    =&-\sqrt{2}\eth'_{-1\to-2}\eth'_{0\to-1}\eth'_{1\to0}((\D2\alpha)(m))\\
    =&\eth'_{-1\to-2}\eth'_{0\to-1}\eth'_{1\to0}\eth'_{2\to 1}(\alpha(m,m)).
\end{align*}
\end{rem}
We now use the $\eth$-operators to define the spin $s$-weighted spherical harmonics:
\begin{defi}
Defining\footnote{In most parts of the literature, one typically starts with $\mathfrak{c}(\Ylm)=\Ylm+i\Ylm$ as  ${}_0Y_{\ell,m}$. } ${}_0Y_{\ell,m}=\Ylm$, the spin $s$-weighted spherical harmonics ${}_s\Ylm$ are inductively defined for any $s\in\mathbb Z$ and $\ell\geq|s|$ via:
\begin{nalign}
    {}_{s+1}Y_{\ell,m}&:=\frac{1}{\sqrt{(\ell-s)(\ell+s+1)}}\eth_{s\to s+1}{}_sY_{\ell,m},\\
     {}_{s-1}Y_{\ell,m}&:=-\frac{1}{\sqrt{(\ell+s)(\ell-s+1)}}\eth'_{s\to s-1}{}_sY_{\ell,m}.
\end{nalign}
\end{defi}

To give an example, we have
\begin{equation*}
    {}_{\pm 2}Y_{\ell,m}=\frac{1}{\sqrt{(\ell-1)\ell(\ell+1)(\ell+2)}}\left(\mp \partial_\theta-\frac{i}{\sin\theta}\partial_\varphi\pm\cot\theta\right)\left(\mp\partial_\theta-\frac{i}{\sin\theta}\partial_\varphi\right)\Ylm={}\conj{_{\mp2}Y_{\ell,m}}.
\end{equation*}
 
 \begin{rem}
With these definitions at hand, we can now relate the spin $s$-weighted spherical harmonics to the tensorial spherical harmonics of Def.~\ref{def:SS:sphericalharmonics}:
\begin{align}
    {}_1Y_{\ell,m}=\sqrt{2}\YlmE1(m)&=-i{\sqrt{2}}\YlmH1( m)=-\conj{{}_{-1}Y_{\ell,m}},\\
    {}_2Y_{\ell,m}=\sqrt{2}\YlmE2(m,m)&=-i\sqrt{2}\YlmH2( m, m)=\conj{{}_{-2}Y_{\ell,m}}.
\end{align}
\end{rem}
The spin $s$-weighted spherical harmonics are eigenfunctions of the spin $s$-weighted spherical Laplacian:
\begin{defi}
 For $|s|=2$, the spin $s$-weighted Laplacian is defined via
\begin{align}
    \lap^{[+2]}(\alpha(m,m)):=((\lap+2)\alpha)(m,m),&&\lap^{[-2]}(\alpha(\conj m,\conj m)):=((\lap-2)\alpha)(\conj m,\conj m),
\end{align}
and analogously for $|s|=1,0$. In particular, $\lap^{[0]}$ is simply the Laplacian acting on scalar functions. 
\end{defi}
The eigenvalues of $\lap^{[s]}$ are $-\ell(\ell+1)+s(s+1)$. 
In coordinates $(\theta, \varphi)$, the spin $s$-weighted Laplacian for $s=2$ reads (cf.~to \eqref{eq:app:laplacian})
\begin{equation}
    \lap^{[s]}=\lap^{[0]} -\frac{2si\cot\theta}{\sin\theta}\partial_{\varphi}-4\cot^2\theta+s.
\end{equation}
Finally, we remark that it can also be written as
\begin{equation}
    \lap^{[s]}=\eth'_{s+1\to s}\eth_{s\to s+1}=\eth_{s-1\to s}\eth'_{s\to s-1}+2s.
\end{equation}

\subsection{A dictionary between the Christodoulou--Klainerman formalism and the Newman--Penrose formalism}\label{app:dictionaryCKNP}
In this part of the appendix, we provide an explicit translation of the Christodoulou--Klainerman framework and notation, employed in the present paper, to the very closely related Newman--Penrose formalism.

Since many people are familiar with only one of them, even though both formalisms are essentially the same, we hope that this will make it easier for people to read papers written in the respective other formalism.

\paragraph{The C--K formalism} In the Christodoulou--Klainerman formalism~\cite{CK93}, given a spacetime $(\mathcal M,g)$, a local orthonormal null frame $(e_1,e_2,e_3,e_4)$ is picked, with $e_1$, $e_2$ spacelike and $e_3$, $e_4$ null. 
That is to say, letting upper-case Latin letters denote indices $\in\{1,2\}$, $g(e_A,e_B)=\delta_{AB}$, $g(e_A,e_3)=0=g(e_A,e_4)$, and $g(e_3,e_4)=-2$. We further introduce the notation that $\slashed{g}$ and $\slashed{\varepsilon}$ denote the metric and volume form on the orthogonal complement $\langle e_3,\e_4\rangle^{\perp}$ induced by the spacetime metric and spacetime volume form, respectively. 

Then, the Ricci/connection coefficients are decomposed as follows:
\begin{nalign}
\chi_{AB}=g(\nabla_A e_4,e_B),&&\underline{\chi}_{AB}=g(\nabla_Ae_3,e_B),\\
\zeta_A=\frac12 g(\nabla_A e_4,e_3),&&\underline{\zeta}_A=\frac12 g(\nabla_A e_3,e_4)=-\zeta_A,\\
\omega=\frac12 g(\nabla_4e_3,e_4),&& \underline{\omega}=\frac12 g(\nabla_{3}e_4,e_3),\\
\eta=-\frac12 g(\nabla_3 e_A,e_4),&&\underline\eta=-\frac12 g(\nabla_4e_A,e_3),\\
\xi_A=\frac12 g(\sl_4e_4,e_A),&& \underline\xi_{A}=\frac12g(\nabla_3 e_3,e_A),
\end{nalign}
and, finally\footnote{We are not aware of a standard convention for the notation for the Christoffel symbols on the sphere within the C--K framework, so we here  introduce the notation $\slashed{\Gamma}$.},
\begin{equation}
    \slashed{\Gamma}_A=g(\nabla_A e_1,e_2),\quad \slashed{\Gamma}_3=g(\nabla_3e_1,e_2),\quad \slashed{\Gamma}_4=g(\nabla_4 e_1,e_2).
\end{equation}
All other connection coefficients either vanish or can be derived from the ones above by an application of the Leibniz rule, using the orthonormality of the frame.
Notice that, in general, neither $\chi$ nor $\underline\chi$ are symmetric.

Similarly, the Weyl curvature tensor $C$ (which equals the Riemann curvature tensor if $(\mathcal M,g)$ solves the Einstein vacuum equations) is decomposed as follows:
\begin{nalign}
\alpha_{AB}=C(e_A,e_4,e_B,e_4),&&\underline{\alpha}_{AB}=C(e_A,e_3,e_B,e_3)\\
\beta_A=C(e_A,e_4,e_3,e_4),&&\underline{\beta}_A=C(e_A,e_3,e_4,e_3)\\
-\rho \slashed{g}_{AB}+\sigma\slashed{\varepsilon}_{AB}=C(e_A,e_3,e_B,e_4).
\end{nalign}
Notice that $\alpha$ and $\underline{\alpha}$ are symmetric, and also tracefree by definition of the Weyl tensor.

\paragraph{The N--P formalism} On the other hand, in the Newman--Penrose formalism~\cite{NP62Approach}, one works with a null tetrad $(l,n,m,\conj m)$ which is related to a frame $(e_1,e_2,e_3,e_4)$ as in the C--K formalism via
$e_4=l$, $e_3=n$, $\tfrac{1}{\sqrt2}(e_1+ie_2)=m$, $\tfrac{1}{\sqrt2}(e_1-ie_2)=\conj m$.
 In order to closely resemble the standard works employing this formalism, we will now use $;$ subscripts to denote covariant differentiation.
Notice that the C--K and the N--P formalisms use an overlapping set of symbols for different quantities, it should however always be clear from context which ones we are talking about.
The Ricci coefficients are then decomposed into the following complex scalars:
\begin{nalign}
\rho=l_{\mu; \nu}m^\mu\conj m^\nu,&& \sigma=l_{\mu;\nu}m^\mu m^\nu,\\
\mu=-n_{\mu;\nu}\conj m^\mu m^\nu,&& \lambda=-n_{\mu;\nu}\conj m^\mu\conj m^\nu,\\
\alpha=\tfrac12(l_{\mu;\nu}n^\mu\conj m^\nu-m_{\mu;\nu}\conj m^\mu\conj m^\nu),&&\beta =\tfrac12 (l_{\mu;\nu}n^\mu m^\nu-m_{\mu;\nu}\conj m^\mu m^\nu),\\
\epsilon=\tfrac12(l_{\mu;\nu}n^\mu l^\nu-m_{\mu;\nu}\conj m^\mu l^\nu),&&\gamma =\tfrac12 (l_{\mu;\nu}n^\mu n^\nu-m_{\mu;\nu}\conj m^\mu n^\nu),\\
\tau=l_{\mu;\nu} m^\mu n^\nu,&& \pi=-n_{\mu;\nu}\conj m^\mu l^\nu,\\
\kappa=l_{\mu;\nu}m^\mu l^\nu,&& \nu=-n_{\mu;\nu}\conj m^\mu n^\nu.
\end{nalign}
Finally, Newman and Penrose decompose the Weyl curvature coefficients according to
\begin{nalign}
\Psi_0=C(l,m,l,m),&&\Psi_4=C(\conj m, n, \conj m,n),\\
\Psi_1=C(m,l,n,l),&& \Psi_3=C(\conj m, n,l,n),\\
\Psi_2=-C(m,n,\conj m,l).&&
\end{nalign}

\paragraph{The dictionary}
We explicitly relate the definitions of Newman and Penrose to the C--K quantities as follows:
\begin{nalign}
\rho=\chi(\conj m,m)=\tfrac12(\slashed{g}\cdot \chi+i\slashed{\varepsilon}\cdot \chi),&& \sigma=\chi(m,m)=\hat{\chi}(m,m),\\
\mu=-\underline{\chi}(m,\conj m)=-\tfrac12(\slashed{g}\cdot \underline{\chi}-i\slashed{\epsilon}\cdot \underline\chi),&& \lambda= -\underline{\chi}(\conj m, \conj m)=-\hat{\underline{\chi}}(\conj m, \conj m),\\
\conj \alpha+\beta=2\zeta(m)=-2\underline{\zeta}(m),&& -\conj\alpha+\beta=i\cdot\slashed{\Gamma}(m)=\slashed{\Gamma}({}^\ast m),\\
\epsilon+\conj\epsilon=-2\omega,&& \gamma+\conj\gamma=2\underline\omega,\\
\epsilon-\conj\epsilon=i\slashed{\Gamma}_4,&&\gamma-\conj\gamma=i\slashed{\Gamma}_3,\\
\tau=2\eta(m),&&\pi=-2\underline\eta(\conj m),\\
\kappa=2\xi(m),&&\nu=-2\underline{\xi}(\conj m).
\end{nalign}
Let us give one detailed example for the required computations:
\begin{multline*}
    -\conj \alpha+\beta =\tfrac12 \conj m_{\mu;\nu}m^\mu m^\nu -\tfrac12 m_{\mu;\nu}\conj m^\mu m^\nu=\tfrac12 g(\nabla_m \conj m,m)-\tfrac12 g(\nabla_m m,\conj m)=-g(\nabla_m m,\conj m)\\
    =-\tfrac{1}{2\sqrt2}g(\nabla_{e_1+ie_2}e_1+ie_2,e_1-ie_2)=-\tfrac{1}{2\sqrt2}(-ig(\nabla_1e_1,e_2)+ig(\nabla_1 e_2,e_1)+g(\nabla_2 e_1,e_2)-g(\nabla_2 e_2,e_1))\\
    =-\tfrac{1}{\sqrt2} (g(\nabla_2 e_1,e_2)-ig(\nabla_1 e_1,e_2))=\tfrac{i}{\sqrt2}(g(\nabla_1 e_1,e_2)+ig(\nabla_2 e_1,e_2))=i\slashed\Gamma_A m^A=\slashed{\Gamma}({}^\ast m)
\end{multline*}
Finally, we evidently have
\begin{nalign}
\Psi_0=\alpha(m,m),&&\Psi_4=\underline{\alpha}(\conj m, \conj m),\\
\Psi_1=\beta(m),&& \Psi_3=\beta(\conj m),\\
\Psi_2=\rho+i\sigma.
\end{nalign}

\newpage
\section{Useful integral identities}\label{app:integrals}
Recall the notation $r_0(u)=r(u,v=v_1)$. Note that $r_0/r\leq 1$.  Recall further the notation $D=1-2M/r=\pv r=-\pu r$, and $D_0=1-2M/r_0$. 
Finally, we  write, for any $q\in\mathbb R$ such that $-q\notin\mathbb N_{>0}$, $q!$ to mean $q!:=\Gamma(q+1)$. 
\subsection{Integrals in the \texorpdfstring{$u$}{u}-direction}
The following lemmata are used for integrating the approximate conservation law \eqref{eq:cons:cons0} or similar equations from $\Scrim$:
\begin{lemma}\label{lem:appB:A1}
Let $q\in\mathbb R$, $N\in\mathbb N$, with $-1< q<N-1$. Then
\begin{equation}\label{eq:appB:lemA1:1}
\int_{-\infty}^u \frac{r(u',1) ^q}{r(u',v)^N}\dd u'= \frac{q!(N-q-2)!}{(N-1)!} \frac{1}{r(u,v)^{N-q-1}}+\O\left(\frac{r_0}{r^{N-q}}+\frac{r_0^{q+1}}{r^N}+\frac{M}{r^{N-q}}\right).
\end{equation}
In the special case where $q\in \mathbb N$, we have more precisely that
\begin{equation}\label{eq:appB:lemA1:2}
\int_{-\infty}^u \frac{r_0 ^q}{r^N}\dd u'= \frac{q!(N-q-2)!}{(N-1)!} \frac{1}{r^{N-q-1}}\cdot\left(1+\sum_{i=1}^q \frac{r_0^i}{r^i}\prod_{j=1}^i (N-q-2+j)\right)+\O\left(\frac{M}{r^{N-q}}\right).
\end{equation}
Finally, we have
\begin{equation}\label{eq:appB:lemA1:3}
\int_{-\infty}^u \frac{r_0 ^q \log r_0}{r^N}\dd u'= \frac{q!(N-q-2)!}{(N-1)!} \frac{\log r}{r^{N-q-1}}+\frac{d_{N,q}}{r^{N-q-1}}+\O\left(\frac{r_0\log r_0}{r^{N-q}}+\frac{r_0^{q+1}\log r_0}{r^N}+\frac{M\log r}{r^{N-q}}\right),
\end{equation}
where the constant $d_{N,q}$ is given by $\int_0^\infty \frac{x^q\log x}{(1+x)^N}\dd x=\frac{q!(N-q-2)!}{(N-1)!}(H_q-H_{N-q-2}) $.
\end{lemma}
\begin{proof}
The second statement, \eqref{eq:appB:lemA1:2}, follows from
\begin{align*}
\int \frac{r_0^q}{r^N}\dd u&=\int \frac{\dd }{\dd u} \left(\frac{1}{N-1}\frac{r_0^q}{r^{N-1}}\right)+\frac{qD_0}{N-1} \frac{r_0^{q-1}}{r^{N-1}}-\frac{r_0^q}{r^N}(1-D)\\
&=\frac{1}{N-1}\frac{r_0^q}{r^{N-1}}+\frac{q}{N-1}\int \frac{r_0^{q-1}}{r^{N-1}}\dd u+M\cdot\O\left(\int \frac{r_0^q}{r^{N+1}}+\frac{r_0^{q-2}}{r^{N-1}}\dd u\right)
\end{align*}
and complete induction.

Moving to the first statement, \eqref{eq:appB:lemA1:1}, we first prove it in the case $M=0$ where $r-r_0$ is independent of $u$:
\begin{align*}
\int_{-\infty}^u \frac{r_0^q}{r^N}\dd u'=\int_{-\infty}^u \frac{r_0^q}{((r-r_0)+r_0)^N}\dd u'=
\frac{1}{(r-r_0)^{N-q-1}}\int_{r_0/(r-r_0)}^\infty \frac{x^q}{(1+x)^N}\dd x\\=\frac{1}{(r-r_0)^{N-q}}\int_{0}^\infty \frac{x^q}{(1+x)^N}\dd x- \frac{1}{(r-r_0)^{N-q-1}}\int_{0}^{r_0/(r-r_0)} \frac{x^q}{(1+x)^N}\dd x.
\end{align*}
The first integral on the RHS can be looked up or be computed by \textit{mathematica} and evaluates to $\int_{0}^\infty \frac{x^q}{(1+x)^N}\dd x=q!(N-q-2)!/(N-1)!$. On the other hand, the second integral can be bounded as
\begin{align*}
\int_{0}^{r_0/(r-r_0)} \frac{x^q}{(1+x)^N}\dd x=\frac{(r_0/(r-r_0))^{q+1}}{q+1} +\O\left((r_0/(r-r_0))^{q+2}\right).
\end{align*}
This proves the claim if $M=0$. 

In order to prove the claim for $M>0$, we copy the above proof and use that $r-r_0=v+M\O(\log r)$. This immediately gives the claim, but with a logarithmic error bound of order $M\O(r^{q-N}\log r)$. 
We remove this logarithmic error term by slightly deforming the path of integration: Instead of integrating along constant $v$, we integrate along constant $r-r_0$:
We work in coordinates $z=r-r_0$, $y=r_0$. Then $\pu=-D_0\partial_y+(D_0-D)\partial_z$, $\pv= D\partial_z$. Clearly, the integral $\int_{-\infty}^u r_0^p/r^N\dd u'$ then satisfies
\begin{align*}
-D_0\partial_y \left(\int_{-\infty}^u r_0^q/r^N \dd u'\right)=r_0^q/r^N+(D-D_0)\partial_z \left(\int_{-\infty}^u r_0^q/r^N \dd u'\right).
\end{align*}
We can now divide by $D_0$ and integrate in $y$ (along constant $z$). By using the estimate with the logarithmic error term for the second term on the RHS, we can show that that term is subleading:
\begin{align*}
\int_{-\infty}^u \frac{r_0^q}{r^N} \dd u'=\int_{r_0}^\infty \frac{y^q}{(z+y)^N} +M\cdot\O\left((z+y)^{-N}y^{q-1}\right)\dd y.
\end{align*}
The first integral on the RHS is now the same one we had for $M=0$:
\begin{align*}
\int_{r_0}^\infty \frac{y^q}{(z+y)^N}\dd y=\frac{q!(N-q-2)!}{(N-1)!(r-r_0)^{N-q-1}}+\O\left(\frac{r_0^{q+1}}{(r-r_0)^{N}}\right).
\end{align*}
The proof concludes by expanding $1/(r-r_0)^{N-q-1}=1/(r^{N-q-1}\cdot (1-r_0/r)^{N-q-1})$.

Finally, we prove \eqref{eq:appB:lemA1:3}. We only consider the case $M=0$:
Then
\begin{multline*}
\int_{-\infty}^u \frac{r_0^q\log r_0}{r^N}\dd u'=\int_{r_0}^\infty \frac{y^q\log y}{(z+y)^N}\dd y\\
=\frac{1}{(r-r_0)^{N-q-1}}\int_{r_0/(r-r_0)}^\infty \frac{x^q}{(1+x)^N}\dd x+\frac{\log(r-r_0)}{(r-r_0)^{N-q-1}}\int_{r_0/(r-r_0)}^\infty \frac{x^q}{(1+x)^N}\dd x.
\end{multline*}
The proof concludes by again writing both integrals as $\int_{r_0/(r-r_0)}^\infty=\int_0^\infty-\int_0^{r_0/(r-r_0)}$, and by using that $\log(r-r_0)=\log r+\O(r_0/r)$.
\end{proof}

\begin{lemma}\label{lem:appB:careerlog}
Let $N\in\mathbb N_{>0}$. Then
\begin{equation}
\int_{-\infty}^u \frac{r_0^{-1}}{r^N}\dd u=r^{-N}\log( r/r_0)+\O(r^{-N}).
\end{equation}
\end{lemma}
\begin{proof}
We only give the proof for $M=0$. The extension to $M>0$ is immediate using $r-r_0=v+\O(\log r).$

The integral is computed using partial fractions.  Denote again $r_0=y$, $r-r_0=z$. Then
\begin{equation*}
\frac{1}{r_0r^N}=\frac{1}{y(z+y)^N}=\frac{1}{z^N}\left(\frac{1}{y}+\frac{1}{y+z}\right)-\sum_{i=1}^{N-1}\frac{1}{z^{N+1-i}}\frac{1}{(z+y)^{i+1}}.
\end{equation*}
We now compute
\begin{align*}
z^{-N}\int_{-\infty}^u \left(\frac{1}{y}+\frac{1}{y+z}\right)\dd u'=z^{-N}\int_{r_0}^\infty \left(\frac{1}{y}+\frac{1}{y+z}\right)\dd y= \frac{1}{(r-r_0)^N}\log r/r_0.
\end{align*}
The statement then follows from $(r-r_0)^{-N}\log(r/r_0)=r^{-N}\log (r/r_0)(1+r_0/r)^N$ and the fact that $-x \log x$ is bounded by $\e^{-1}$ on $x\in[0,1]$.
\end{proof}
The previous lemmata can be extended to the entire range of $q\in\mathbb R$ using integration by parts:
\begin{lemma}
Suppose that $N\in\mathbb N$, $q\in\mathbb R$, with $-q<-1$. Then
\begin{multline}
\int \frac{x^{-q}}{(x+y)^N}\dd x'=\frac{1}{(-q+1)}\sum_{i=0}^{\lceil q-2\rceil}\frac{x^{-q+1+i}}{(x+y)^{N+i}}\prod_{j=0}^{i-1}\frac{(N+j)}{(-q+j+2)}\\
+\prod_{j=0}^{\lceil q-2\rceil }\frac{(N+j)}{(-q+j+1)}\int \frac{x^{-q+\lceil q-1\rceil}}{(x+y)^{N+\lceil q-1\rceil}}\dd x.
\end{multline}
\end{lemma}
\begin{proof}
Induction.
\end{proof}
\begin{lemma}\label{lem:appB:funnylog}
We have
\begin{equation}
\int_{-\infty}^u \frac{\log r-\log r_0}{r}\dd u'=1+\O(r_0/r).
\end{equation}
\end{lemma}
\begin{proof}
We again first prove this for $M=0$. Using the dilogarithmic identity $\sum_{k=1}^\infty \frac{x^k}{k!}=-\int_{0}^x \frac{\log(1-x')}{x'}\dd x' (=\mathrm{Li}_2(x))$ for $|x|\leq 1$ (page 1004 in \cite{AbramowitzStegun}), we have that
\begin{align*}
\int_{-\infty}^u \frac{\log(v-u')-\log|u'|}{(v-u')}\dd u'=-\int_{0}^{v/(v-u)} 	\frac{\log(1-x')}{x'}\dd x'=\frac{v}{v-u}+\sum_{k=2}^\infty (v/(v-u))^{k-1}/k^2,
\end{align*}
so the result follows. In order to elevate the result to $M\neq 0$, we follow the same procedure as in the proof of Lemma~\ref{lem:appB:A1}.
\end{proof}
\subsection{Integrals in the \texorpdfstring{$v$}{v}-direction}
The next lemmata are relevant for integrating in $v$ from $\Cin$.
\begin{lemma}\label{lem:appB:xxxrrr}
Define $x=1/r$ and $x_0=1/r_0$. 
Let $N\in\mathbb N$, and let $q\in\mathbb R$. Then
\begin{multline}
\underbrace{\int_{v_{\Cin}}^v \frac{D}{r^2}(u,v_{(1)})\cdots\int_{v_{\Cin}}^{v_{(n-1)}} \frac{D}{r^2}(u,v_{(n)})}_{\text{$N$ integrals}} \int_{v_{\Cin}}^{v_{(n)}} \frac{D}{r^2} r^q (u,v_{(n+1)}) \dd v_{(n+1)}\dd v_{(n)} \cdots \dd v_{(1)}\\
=(-1)^N\int_{x_0(u)}^x \cdots\int_{x_0}^{x_{(n-1)}} \int_{x_0}^{x_{(n)}} x_{(n+1)}^{-q} \dd x_{(n+1)}\cdots \dd x_{(1)}.
\end{multline} 
\end{lemma}
\begin{proof}
Direct substitution of the variable $x=1/r$.
\end{proof}
\begin{lemma}\label{lem:appB:easyxxx} Let $N\in\mathbb N$. Then
\begin{equation}
(-1)^N\underbrace{\int_{x_0}^x\cdots\int_{x_0}^{x_{(n-1)}}}_{\text{$N$ times}}\dd x_{(n)}\cdots \dd x_{(1)}=\frac{(x_0-x)^N}{N!}.
\end{equation}
\end{lemma}
\begin{proof}
Induction.
\end{proof}
\begin{lemma}\label{lem:appB:hardxxx}
Let $N\in\mathbb N$, and let $q\in\mathbb R$.
If $N\leq q-1$, or if $q\notin \mathbb N_{>0}$, then
\begin{multline}\label{eq:appB:lemhard1}
(-1)^N\underbrace{\int_{x_0}^x\cdots\int_{x_0}^{x_{(n-1)}}}_{\text{$N$ times}} x_{(n)}^{-q}\dd x_{(n)}\cdots \dd x_{(1)}\\
=\frac{(q-N-1)!}{(q-1)!}x^{N-q}-x_0^{N-q}\sum_{i=0}	^{N-1}\frac{(q+i-N-1)!}{(q-1)!i!}\left(1-\frac{x}{x_0}\right)^i.
\end{multline}
Otherwise, if $N>q-1$ and $q\in\mathbb N_{>0}$, then
\begin{multline}\label{eq:appB:lemhard2}
(-1)^N\underbrace{\int_{x_0}^x\cdots\int_{x_0}^{x_{(n-1)}}}_{\text{$N$ times}} x_{(n)}^{-q}\dd x_{(n)}\cdots \dd x_{(1)}\\
=\frac{(-1)^{N+1-q}}{(q-1)!(N-q)!}x^{N-q}\log (x/x_0)-x_0^{N-q}\sum_{i=0}	^{N-1} d_{N,q,i}\left(1-\frac{x}{x_0}\right)^i,
\end{multline}
with some coefficients $d_{N,q,i}\in\mathbb Q$.
\end{lemma}
\begin{proof}
The first statement \eqref{eq:appB:lemhard1} is easy to prove (via induction). 

In order to show \eqref{eq:appB:lemhard2}, one first shows that
\begin{multline*}
\underbrace{\int_{x_0}^x\cdots\int_{x_0}^{x_{(n-1)}}}_{\text{$N$ times}} x_{(n)}^{-1}\dd x_{(n)}\cdots \dd x_{(1)}=\frac{ x^{N-1}\log (x/x_0)}{(N-1)!}-x_0^{N-1}\sum_{i=1}^{N-1}\frac{(\frac{x}{x_0}-1)^i	}{i!(N-1-i)!}\sum_{j=0}^{i-1}\frac{1}{(N-1-j)}.
\end{multline*}
Combining this with the first statement of the lemma, and appropriately splitting up the integrals, we then obtain
\begin{align*}
(-1)^N\underbrace{\int_{x_0}^x\cdots\int_{x_0}^{x_{(n-1)}}}_{\text{$N$ times}} x_{(n)}^{-q}\dd x_{(n)}\cdots \dd x_{(1)}=\frac{(-1)^{N+1-q}}{(q-1)!(N-q)!}x^{N-q}\log (x/x_0)\\
-\frac{(-1)^{N+1-q}}{(q-1)!(N-q)!}x_0^{N-q}\sum_{i=1}^{N-q}\binom{N-q}{i}\left(\tfrac{x}{x_0}-1\right)^i\sum_{j=0}^{i-1}\frac{1}{N-q-j}\\
-\frac{1}{(q-1)!}\sum_{i=0}^{q-2}\frac{i!}{(N+1-q+i)!}x_0^{-1-i}(x_0-x)^{N+1-q+i}.
\end{align*}
We note that it's much easier to prove only the schematic form of \eqref{eq:appB:lemhard2}, without keeping track of the precise expressions of the coefficients.
\end{proof}
We finally record a version of the above Lemma when logarithms are present:
\begin{lemma}\label{lem:appB:hardxxxlog}
Let $N\in\mathbb N$, and let $q\in\mathbb R$.
If $N\leq q-1$, or if $q\notin \mathbb N_{>0}$, then
\begin{multline}\label{eq:appB:lemhard1log}
(-1)^N\underbrace{\int_{x_0}^x\cdots\int_{x_0}^{x_{(n-1)}}}_{\text{$N$ times}} x_{(n)}^{-q}\log(x_{(n)})\dd x_{(n)}\cdots \dd x_{(1)}\\
=\frac{(q-N-1)!}{(q-1)!}x^{N-q} \log x-x_0^{N-q}\sum_{i=0}	^{N-1}\frac{(q+i-N-1)!}{(q-1)!i!}\left(1-\frac{x}{x_0}\right)^i\log x_0+d'_{N,q,i}\left(\frac{x}{x_0}\right)^i.
\end{multline}
Otherwise, if $N>q-1$ and $q\in\mathbb N_{>0}$, then
\begin{multline}\label{eq:appB:lemhard2log}
(-1)^N\underbrace{\int_{x_0}^x\cdots\int_{x_0}^{x_{(n-1)}}}_{\text{$N$ times}} x_{(n)}^{-q} \log x_{(n)}\dd x_{(n)}\cdots \dd x_{(1)}\\
=\frac{(-1)^{N+1-q}}{{2}(q-1)!(N-q)!}x^{N-q}\left(\log^2x+\log x\left(1-\sum_{i=1}^{N-q}\frac{1}{i}\right)\right)\\
-x_0^{N-q} \sum_{i=0}	^{N-1} (d''_{N,q,i}+d'''_{N,q,i}\log x_0+d_{N,q,i}\log^2x_0)\left(1-\frac{x}{x_0}\right)^i,
\end{multline}
with some coefficients $d_{N,q,i}\in\mathbb Q$.
\end{lemma}
\begin{proof}
The proof is similar to that of Lemma~\ref{lem:appB:hardxxx}, using also that $\int x^{-1}\log x \dd x={\frac12}\log^2 x$ and
\begin{equation*}
\int x^n \log^2 x \dd x=\frac{x^{n+1}}{n+1}\left(	\log^2 x-\frac{2}{n+1}\log x +\frac{2}{(n+1)^2}	\right),\qquad \forall n\neq -1.
\end{equation*}
\end{proof}

\begin{spacing}{0.9}
\bibliographystyle{alpha}
\bibliography{references_all}
\end{spacing}
\end{document}